\documentclass{ituthesis}

\settitle{Algorithms for Similarity Search and Pseudorandomness$^\dagger$}
\setauthor{Tobias Christiani}
\setsupervisor{Rasmus Pagh}
\setdate{May 2018}

\usepackage[utf8]{inputenc}
\usepackage[T1]{fontenc}
\usepackage{amsthm,amssymb,amsfonts}
\usepackage{multirow}
\usepackage{graphicx}
\usepackage{microtype}
\usepackage{enumitem}
\usepackage{dsfont} 
\usepackage{booktabs}
\usepackage{float}
\usepackage{verbatim}
\usepackage{mathrsfs}
\usepackage{array}
\usepackage{booktabs, multicol, multirow}
\usepackage{cellspace}
\usepackage{makecell}
\usepackage{diagbox}
\usepackage[caption = false]{subfig}

\usepackage{xspace} 
\usepackage[ruled, vlined, linesnumbered]{algorithm2e}

\usepackage{csvsimple}
\usepackage{siunitx}

\usepackage[hidelinks,pdfpagelabels]{hyperref}

\theoremstyle{plain}
\newtheorem{theorem}{Theorem}[chapter]
\newtheorem{lemma}{Lemma}[chapter]
\newtheorem{corollary}{Corollary}[chapter]

\newtheorem{innercustomthm}{Theorem}
\newenvironment{customthm}[1]{\renewcommand\theinnercustomthm{#1}\innercustomthm}{\endinnercustomthm}

\theoremstyle{definition}
\newtheorem{definition}{Definition}[chapter]

\theoremstyle{remark}
\newtheorem{remark}{Remark}[chapter]

\DeclareMathOperator*{\E}{\mathbb{E}}
\newcommand{\ip}[2]{\langle{#1},{#2}\rangle}
\DeclareMathOperator{\poly}{poly}
\DeclareMathOperator{\sign}{sign}
\DeclareMathOperator*{\argmin}{arg\,min}

\newcommand{\norm}[1]{\left\lVert #1 \right\rVert}
\DeclareMathOperator{\dist}{dist}
\DeclareMathOperator{\simil}{sim}

\newcommand{\LSH}{\mathcal{H}}

\newcommand{\LSF}{\mathcal{F}}
\newcommand{\DSH}{\mathcal{D}}
\newcommand{\ALSH}{\mathcal{A}}
\newcommand{\LSM}{\mathcal{M}}

\newcommand{\cube}[1]{\{-1,1\}^{#1}}
\newcommand{\bitcube}[1]{\{0,1\}^{#1}}

\newcommand{\F}{\mathbf{F}}
\newcommand{\Q}{\mathbf{Q}}
\newcommand{\U}{\mathbf{U}}

\newcommand{\RR}{\mathcal{V}}
\newcommand{\TLSF}{\mathbf{F}^{\otimes \tau}\!}  
\newcommand*\conc{\mathbin{\|}}
\DeclareMathOperator{\id}{\mathrm{I}}

\newcommand{\real}{\mathbb{R}}
\newcommand{\sphere}[1]{\mathbb{S}^{#1}}
\newcommand{\corrsub}[1]{\substack{(x, y) \\ {#1}\text{-correlated}}}
\newcommand*\mcap{\mathbin{\mathpalette\mcapinn\relax}}
\newcommand*\mcapinn[2]{\vcenter{\hbox{$\mathsurround=0pt\ifx\displaystyle#1\textstyle\else#1\fi\bigcap$}}}

\newcommand{\bit}[1]{\langle{#1}\rangle}

\newcommand{\Var}{\mathrm{Var}}
\newcommand{\1}{\mathds{1}}

\newcommand{\colvec}[1]{\begin{bmatrix}#1\end{bmatrix}} 
\newcommand{\K}{\mathbb{K}}

\newcommand{\sen}[1]{\K_{\LSH}(#1)} 

\newcommand{\family}{There exists a randomized data structure that takes as input positive integers $u$, $r$, $k$, $t$ and selects a family of functions $\mathcal{F}$ from $[u]$ to $[r]$. 
In the word RAM model with word size $w$ the data structure satisfies the following:}

\newcommand{\R}{\mathbb{R}}
\newcommand{\A}{\mathcal{A}}

\newcommand{\GF}{\mathbb{F}}

\DeclareMathOperator{\diag}{diag}

\newcommand{\simjoin}{{\;\bowtie_\lambda\;}}
\newcommand{\cp}{\textsc{Chosen Path}\xspace}
\newcommand{\cpsj}{\textsc{CPSJoin}\xspace}
\newcommand{\mh}{\textsc{MinHash}\xspace}
\newcommand{\blsh}{\textsc{BayesLSH}\xspace}
\newcommand{\all}{\textsc{AllPairs}\xspace}
\newcommand{\pp}{\textsc{PPJoin}\xspace}

\newcommand{\sectionquote}[1]{
\vspace{-1.5cm}
\begin{center}
\begin{minipage}{0.8\textwidth}
\textit{`#1'} 
\end{minipage}
\end{center}
\vspace{0.66cm}
}

\begin{document}
\frontmatter

\thetitlepage
\newpage

\begin{abstract}
\noindent We study the problem of approximate near neighbor (ANN) search and show the following results: 
\begin{itemize}
	\item An improved framework for solving the ANN problem using locality-sensitive hashing, 
		reducing the number of evaluations of locality-sensitive hash functions and the word-RAM complexity compared to the standard framework.
	\item A framework for solving the ANN problem with space-time tradeoffs as well as tight upper and lower bounds for the space-time tradeoff of framework solutions to the ANN problem under cosine similarity.
	\item A novel approach to solving the ANN problem on sets along with a matching lower bound, improving the state of the art.
		  A self-tuning version of the algorithm is shown through experiments to outperform existing similarity join algorithms. 
	\item Tight lower bounds for asymmetric locality-sensitive hashing which has applications to the approximate furthest neighbor problem, orthogonal vector search, and annulus queries.
	\item A proof of the optimality of a well-known Boolean locality-sensitive hashing scheme.
\end{itemize}
We study the problem of efficient algorithms for producing high-quality pseudorandom numbers and obtain the following results:
\begin{itemize}
	\item A deterministic algorithm for generating pseudorandom numbers of arbitrarily high quality in constant time using near-optimal space.
	\item A randomized construction of a family of hash functions that outputs pseudorandom numbers of arbitrarily high quality with space usage and running time nearly matching known cell-probe lower bounds.
\end{itemize}
\end{abstract}

\newpage

\begin{otherlanguage}{danish}
\begin{abstract}
	Vi undersøger et grundlæggende problem indenfor approksimativ søgning: tilnærmelsesvis nær nabo (TNN) problemet, og viser følgende resultater:
	\begin{itemize}
		\item En forbedret generel løsning af TNN problemet som reducerer antal evalueringer af afstandsfølsomme spredefunktioner.
		\item En generel løsning af TNN problemet som giver mulighed for tid-plads afvejning samt tætte øvre og nedre grænser for TNN problemet med tid-plads afvejning under kosinuslighed.
		\item En ny tilgang til løsning af TNN problemet på mængder samt en matchende nedre grænse. 
			  En adaptiv version af algoritmen til approksimativ sammenføjning vises ved eksperimenter at være konkurrencedygtig. 
		\item Tætte nedre grænser for asymmetrisk afstandsfølsom spredning som har anvendelser til approksimativ søgning efter fjerne naboer, ortogonale vektorer, og annulus forespørgsler.
		\item Et optimalitetsbevis for en velkendt familie af Boolske afstandsfølsomme spredefunktioner.
	\end{itemize}
	Vi undersøger problemet at finde effektive algoritmer til produktion af pseudotilfældighed af høj kvalitet og opnår følgende resultater:
	\begin{itemize}
		\item En deterministisk algoritme til generation af pseudotilfældige tal af vilkårlig høj kvalitet i konstant tid og med tæt på optimalt pladsforbrug.
		\item En randomiseret konstruktion af en familie af spredefunktioner som afbilder til pseudotilfældige tal af vilkårlig høj kvalitet, med evalueringstid og pladsforbrug tæt på den nedre grænse.
	\end{itemize}
\end{abstract}
\end{otherlanguage}

\newpage

\begin{acknowledgements}
I am very grateful to Rasmus Pagh for advising me for the past almost five years.
Most of my favorite results have come through my collaboration with Rasmus where his overview and technical strength complements my intuition.
I always feel that Rasmus is ready to listen to my ideas, and to encourage me and guide my research in the right direction.
I cannot imagine a better advisor.

I would like to thank my collegues in the 4B corridor at ITU for creating a friendly academic environment where I enjoy spending my time.
In particular I would like to thank my current and former office mates Johan Sivertsen, Matteo Dusefante, Thomas Ahle, Martin Aumüller, Morten Stöckel, and Ninh Pham.
Thore Husfeldt also deserves special thanks for his stimulating lunch discussions on issues ranging from superintelligence to immigration policy. 

I would like to thank Greg Valiant for hosting my stay at Stanford in the fall of 2015 and Michael Mitzenmacher for hosting my stay at Harvard in the fall of 2017.
Specifically I would like to thank Josh Alman, Michael Kim, Zhao Song, and Aviad Rubinstein for making my time spent abroad a pleasant and social experience.  

Finally I want to thank my parents Tom and Kirsten for supporting me, 
my brother Anders for tolerating living with a somewhat disorganized PhD student, 
and my girlfriend Elisabeth for listening to me and helping me through difficult times.
\end{acknowledgements}

\newpage
\begin{center}
\begin{minipage}{0.57\textwidth}
\begin{itshape} 
All that is gold does not glitter, \\
Not all those who wander are lost; \\
The old that is strong does not wither, \\
Deep roots are not reached by the frost. \\ 

From the ashes, a fire shall be woken, \\
A light from the shadows shall spring; \\
Renewed shall be blade that was broken, \\
The crownless again shall be king. \\
\end{itshape}
\end{minipage}
\end{center}
\begin{minipage}{0.85\textwidth}
\begin{flushright} \footnotesize
		J. R. R. Tolkien (1892--1973)
\end{flushright}
\end{minipage}

\cleardoublepage
\setcounter{tocdepth}{1}
\tableofcontents

\mainmatter

\midsloppy
\sloppybottom

\chapter{Introduction}
\section{Part I: Similarity search}
Similarity search in large collections of high-dimensional objects is a problem that is well-motivated by numerous applications.
Consider for example the representation of an image by a $d$-dimensional feature vector $x$, where each entry $x_i$ denotes the fraction of pixels of color $i$ in the image. 
Given a collection of images $P$ and a query image $q$, we could for example be interested in finding the nearest neighbor of $q$:
the image $x \in P$ such that the distance $\dist(q,x)$ is minimized, for some appropriate choice of distance function. 
Applications of near neighbor search include:
\begin{itemize}
	\item Classification: Given a collection $P$ of labelled objects and an unlabelled object $q$, classify $q$ according to the label of its nearest neighbor in $P$.
	\item Recommender systems: Find similar users, movies, songs, books etc.\ to be used for recommendation.
	\item Duplicate detection: Remove near-identical objects from a collection, for example duplicate web pages from the index of a search engine.  
\end{itemize}
The trivial solution to the near neighbor problem would be to iterate through every $x \in P$ and compute $\dist(q, x)$ while keeping track of the nearest neighbor found so far.
If we let $n = |P|$ denote the size of our collection and assume that it takes time $O(d)$ to compute the distance between a pair of objects, then the trivial solution uses time $O(dn)$.

Suppose we are interested in preprocessing the collection $P$ into a data structure that supports answering queries faster than the trivial solution.
In two-dimensional Euclidean space there exists a solution based on the Voronoi diagram of $P$ with space usage $O(n)$ and query time $O(\log n)$~\cite{berg2008}.
In higher dimensions, the best known solutions to the nearest neighbor problem either suffer from space usage or query time that is exponential in $d$~\cite{har-peled2012}. 
This phenomenon is known as the ``curse of dimensionality'' and has recently been substantiated by conditional hardness results~\cite{alman2015, david2016curse, williams2018difference, rubinstein2018hardness}, showing for example that the problem of finding all nearest neighbors in a collection of $n$ points in $d$-dimensional Euclidean space cannot be solved in time subquadratic in $n$ when $d = \omega(\log n)$ unless the Strong Exponential Time Hypothesis (SETH) is false~\cite{williams2004}. 

In order to efficiently solve similarity search problems in high dimensional spaces, researchers and practitioners have turned to approximate solutions. 
Instead of finding the exact nearest neighbor of a query point, we settle for finding a point that is some approximation factor $c > 1$ times further away than the nearest neighbor.
The algorithms and data structures for similarity search in this thesis are primarily aimed at providing efficient solutions to the \emph{approximate near neighbor} problem defined as follows:
\begin{definition}\label{intro:def:ann}
Let $P \subseteq X$ be a collection of $|P| = n$ points in a distance space $(X, \dist)$.
A solution to the \emph{$(r, cr)$-near neighbor} problem is a data structure that supports the following query operation:
Given a query $q \in X$, if there exists $x \in P$ with $\dist(q, x) \leq r$ return $x' \in P$ with $\dist(q, x') < cr$.  
\end{definition}
The $(r, cr)$-near neighbor problem differs from the nearest neighbor problem by searching for any point within a fixed radius $r$ of the query point and allowing us to return points at distance up to $cr$ even though better candidates exist.
It will be convenient to also define the \emph{$(s_1, s_2)$-similarity problem} as the natural equivalent of the $(r, cr)$-near neighbor problem where we measure similarities rather than distances, 
i.e., we wish to report find a point with similarity $\simil(q, x) \geq s_1$ and we are willing to accept points with similarity $s_2 < s_1$. 

By allowing an approximation factor $c > 1$ it is possible to solve the $(r, cr)$-near neighbor problem in Euclidean space (and many other spaces) with query time that is sublinear in $n$ and polynomial in $d$ using space polynomial in $d$ and $n$~\cite{indyk1998, datar2004}.
However, even approximation has its limits when it comes to alleviating the curse of dimensionality. 
Rubinstein~\cite{rubinstein2018hardness} has recently shown that unless SETH is false, for every choice of constants $\gamma, \delta > 0$ there exists $\varepsilon > 0$ such that a solution to the $(1+\varepsilon)$-approximate near neighbor problem with $O(n^\gamma)$ preprocessing time must use query time $\Omega(n^{\delta})$. 
\subsection{Locality-sensitive hashing}
One of the most successful approaches for finding solutions to the approximate near neighbor problem in various spaces is known as \emph{locality-sensitive hashing}, commonly abbreviated as LSH (see~\cite{andoni2008, wang2014} for more information).
The idea behind locality-sensitive hashing is to construct a distribution $\LSH$ over functions $h \colon X \to R$ that are used to partition the space $X$.
This randomized partitioning scheme is locality-sensitive in the sense that close points $x, y \in X$ are more likely to hash to the same part of a randomly sampled partition.
When discussing locality-sensitive hashing, we will sometimes refer to a distribution of locality-sensitive hash functions as a family.
\begin{definition}[Locality-sensitive hashing {\cite{indyk1998}}]\label{intro:def:lsh}
Let $(X, \dist)$ be a distance space and let $\LSH$ be a distribution over functions $h \colon X \to R$.
We say that $\LSH$ is \emph{$(r, cr, p_1, p_2)$-sensitive} if for $x, y \in X$ and $h \sim \LSH$ we have that:
\begin{itemize}
\item If $\dist(x, y) \leq r$ then $\Pr[h(x) = h(y)] \geq p_1$.
\item If $\dist(x, y) \geq cr$ then $\Pr[h(x) = h(y)] \leq p_2$.
\end{itemize}
\end{definition}
We can speed up approximate near neighbor searches at the cost of some additional preprocessing by partitioning the set of points $P$ according to $L$ randomly sampled locality-sensitive hash functions $h_1, \dots, h_L$.
A query for a point $q$ proceeds by considering the points of $P$ that collide with $q$ under $h_1, \dots, h_L$.
Intuitively we want to sample enough hash functions such that the ball of radius $r$ around every potential query point $q \in X$ is covered by the union of the parts $h_{1}^{-1}(q), \dots, h_{L}^{-1}(q)$.
This approach yields the following general LSH framework for solving the approximate near neighbor problem (for more details see Chapter~\ref{ch:fast}).
\begin{theorem}[Indyk-Motwani {\cite{indyk1998, har-peled2012}, simplified}]\label{intro:thm:lsh_im_simple}
Let $\LSH$ be $(r, cr, p_1, p_2)$-sensitive and let $\rho = \frac{\log(1/p_1)}{\log(1/p_2)}$, then there exists a solution to the $(r, cr)$-near neighbor problem using $O(n^{1+\rho})$ words of space and with query time dominated by $O(n^{\rho} \log n)$ evaluations of functions from~$\LSH$.
\end{theorem}
\subsection{Examples} \label{intro:sec:examples}
To further introduce locality-sensitive hashing and the approach to solving the approximate near neighbor problem used in this thesis, 
we will present three simple and powerful families of locality-sensitive hash functions: 
Bit-sampling by Indyk and Motwani~\cite{indyk1998}, MinHash by Broder~\cite{broder1997syntactic}, and SimHash by Charikar~\cite{charikar2002}.
Indyk, Broder, and Charikar received the 2012 ACM Paris Kanellakis Theory and Practice Award ``\emph{for their groundbreaking work on Locality-Sensitive Hashing that has had great impact in many fields of computer science including computer vision, databases, information retrieval, machine learning, and signal processing}''~\cite{kanellakis}.
We proceed by describing each of these families in turn, introducing relevant notation as we go along.

\paragraph{Bit-sampling.}
Indyk and Motwani introduced a simple family of locality-sensitive hash functions $\LSH_{H}$ for the $d$-dimensional Boolean hypercube $\bitcube{d}$ under Hamming distance $\dist_{H}(x, y) = |\{ i \in [d] \mid x_i \neq y_i \}|$ where $[d]$ denotes the set $\{1, 2, \dots, d \}$.
We sample a function $h \sim \LSH$ by sampling $i$ uniformly at random in $[d]$ and setting $h(x) = x_i$.
It is easy to see that a pair of points fail to collide under a random hash function $h(x) = x_i$ if and only if $i$ is sampled from the set of coordinates where $x$ and $y$ differ. 
\begin{equation*}
	\Pr_{h \sim \LSH_{H}}[h(x) = h(y)] = 1 - \Pr_{h \sim \LSH_{H}}[h(x) \neq h(y)] =  1 - \dist_{H}(x,y)/d.
\end{equation*}
Suppose we want to use this function to solve the $(r, cr)$-near neighbor problem in Hamming space $(\bitcube{d}, \dist_{H})$.
Then, from Theorem \ref{intro:thm:lsh_im_simple} we optain a query exponent of
\begin{equation*}
	\rho = \frac{\log(1/(1 - r/d))}{\log(1/(1 - cr/d))} \leq 1/c 
\end{equation*}
where the details behind the last inequality can be found in~\cite{har-peled2012}.
In conclusion, bit-sampling gives a solution to the $(r, cr)$-near neighbor problem in Hamming space with query time roughly $n^{1/c}$ and space usage and preprocessing time roughly $n^{1 + 1/c}$.

\paragraph{MinHash.}
MinHash is a family of locality-sensitive hash functions with applications to similarity search and similarity estimation on sets under Jaccard similarity.
Given sets $x, y \subseteq [d]$ their Jaccard similarity is defined by $\simil_{J}(x, y) = |x \cap y| / |x \cup y|$.

A random hash function $h$ from the MinHash distribution $\LSH_{J}$ is specified by a random permutation of $[d]$ and hashes a set $x$ to the first element of $x$ in this permutation.
The permutation can be specified by a uniformly random hash function $f \colon [d] \to [0,1]$ where $[0,1]$ denotes the closed interval from $0$ to $1$.
Specifically, we sample a random $h \sim \LSH_{J}$ by sampling a uniformly random hash function $f \colon [d] \to [0,1]$ and setting 
\begin{equation*}
	h(x) = \argmin_{i \in x} f(i).
\end{equation*}
Two sets $x$ and $y$ collide under a random hash function $h \sim \LSH_{J}$ if and only if the smallest element of $x \cup y$ is contained in $x \cap y$.
Otherwise, the smallest element of $x$ is in $x {\setminus} y$ or the smallest element of $y$ is in $y {\setminus} x$ and there is no way the sets hash to the same element.
Since the smallest element of $x \cup y$ is uniformly distributed we get that
\begin{equation*}
	\Pr_{h \sim \LSH_{J}}[h(x) = h(y)] =  \frac{|x \cap y|}{|x \cup y|} = \simil_{J}(x, y).
\end{equation*}
MinHash gives a solution to the $(s_1, s_2)$-similarity problem with exponent $\rho = \log(1/s_1) / \log(1/s_2)$.

\paragraph{SimHash.}
SimHash is a family of Boolean-valued locality-sensitive hash functions for $\real^d$ under cosine similarity $\simil_{C}(x, y) = \cos(\theta(x, y))$ where $\theta(x, y)$ denotes the angle between $x$ and $y$.
We sample a function $h \sim \LSH_{C}$ by sampling a $d$-dimensional standard normal random variable $z \sim \mathcal{N}^{d}(0, 1)$ and setting
\begin{equation*}
	h(x) = \sign(\ip{x}{z}).
\end{equation*}
Intuitively, we sample a random hyperplane that goes through the origin and hash points depending on which side of the hyperplane they are on (the sign of the inner product $\ip{x}{y} = \sum_i x_i z_u$).
Due to the rotational invariance of the standard normal distribution the properties of this scheme can be analyzed in two dimensions.
The probability that two points on the unit circle are separated by a random line through the origin is exactly
\begin{equation*}
	\Pr_{h \sim \LSH_{C}}[h(x) = h(y)] = 1 - \theta(x, y)/\pi.
\end{equation*}
This scheme yields a solution to the $(s_1, s_2)$-similarity problem under cosine similarity with $\rho = \log(1 - \arccos(s_1)/\pi) / \log(1 - \arccos(s_2)/\pi)$. 
\subsection{Lower bounds}
Given a space $(X, \dist)$ and distance thresholds $r$, $cr$ we are interested in finding a $(r, cr, p_1, p_2)$-sensitive family with a value of $\rho = \log(1/p_1)/\log(1/p_2)$ that is as small as possible.
The primary technique for deriving locality-sensitive hashing lower bounds has been Fourier analysis of Boolean functions under noisy inputs (see the excellent book by O'Donnell for a comprehensive introduction~\cite{odonnell2014analysis}).
Lower bounds for locality-sensitive hashing schemes (distributions over functions) often follow from lower bounds on the behaviour of a single function  $f \colon \cube{d} \to R$ under randomly $\alpha$-correlated inputs, defined as follows:
\begin{definition}\label{intro:def:correlation}
	For $-1 \leq \alpha \leq 1$ and $x, y \in \cube{d}$ we say that $(x, y)$ is randomly $\alpha$-correlated 
	if the pairs $(x_i, y_i)$ are i.i.d.\ with $x_i$ uniform in $\cube{d}$ and 
	\begin{equation*}
		y_i =
		\begin{cases}
			x_i & \text{with probability } \frac{1 + \alpha}{2}, \\
			-x_i & \text{with probability } \frac{1 - \alpha}{2}. 
		\end{cases}
	\end{equation*}
\end{definition}
\noindent If two vectors $(x, y)$ are randomly $\alpha$-correlated their expected cosine similarity is $\alpha$, and their expected Hamming distance is given by $(1-\alpha)d/2$. 
As the dimensionality increases, the empirical correlation between $x$ and $y$ will be tightly concentrated around $\alpha$.

Let $0 \leq \beta < \alpha < 1$ and consider a $((1-\alpha)d/2, (1-\beta)d/2, p_1, p_2)$-sensitive family $\LSH$ for Hamming space $(\cube{d}, \dist_{H})$.
Combining lower bounds by O'Donnell et al.~\cite{odonnell2014optimal} and Andoni and Razenshteyn~\cite{andoni2016tight} (building on work by Motwani et al.~\cite{motwani2007}), we have that
\begin{equation} \label{eq:combined_lower}
	\rho = \frac{\log(1/p_1)}{\log(1/p_2)} \geq \max\left( \frac{\log(1/\alpha)}{\log(1/\beta)}, \frac{1 - \alpha}{1 + \alpha - 2\beta} \right) - o_{d}(1).
\end{equation}
The lower bounds require that $p_2$ is not too small as a function of $d$. 
In particular, only the trivial lower bound of $\rho \geq 0$ holds if $p_2$ can be exponentially small in $d$, but such families are typically not of interest for high-dimensional similarity search where we want $p_2 \approx 1/n$. 
For a more comprehensive discussion of this issue see~\cite{odonnell2014optimal}.

Compared against different constructions of locality-sensitive hash families, the two lower bounds comprising equation \eqref{eq:combined_lower} reveal interesting properties of the Boolean hypercube.
As $\alpha, \beta$ approach $1$ the lower bound of $\log(1/\alpha)/\log(1/\beta)$ is the larger of the two bounds. 
If we convert the lower bound to Hamming distance we get that $\rho \geq \log(1/(1 - 2r/d))/\log(1/(1 - 2cr/d) \approx 1/c$ for an $(r, cr, p_1, p_2)$-sensitive family when $r, cr \ll d$.
This lower bound is tight against the bit-sampling LSH of Indyk and Motwani.
The bit-sampling family can be described as randomly partitioning the Boolean hypercube into subcubes, 
so in a sense subcubes are an optimal ``shape'' for distinguishing between very short random walks and slightly longer random walks in the Boolean hypercube.
The lower bound of $1/c$ in Hamming space gives a lower bound of $1/c^p$ for $\ell_p^d$-spaces (vectors in $\real^d$ under the $\ell_p$-norm $\norm{x-y}_p = (\sum_i |x_i - y_i|^p)^{1/p}$).
This follows from a direct embedding of the Boolean hypercube in $\ell_p$-space. 

As $\beta$ approaches $0$ the lower bound of $(1 - \alpha)/(1 + \alpha - 2\beta)$ dominates.
Converted to Hamming distance this bound becomes $\rho \geq 1/(2c - 1)$.
For $\beta = 0$ this is tight against existing constructions that use balls to partition the hypercube~\cite{dubiner2010bucketing, andoni2014beyond, andoni2015practical}.
Loosely speaking, in this regime we see that balls in Hamming space are optimal for simultaneously minimizing volume (capturing $0$-correlated points) while maximizing the probability of capturing positively correlated points. 

In Hamming space, the family of locality-sensitive hash functions that give the best known upper bound on the $\rho$-value can essentially be described as follows: 
We sample a function $h \sim \LSH$ by sampling a sequence of $d$ balls of radius slightly below $d/2$ with the center of each ball being sampled uniformly at random from $\cube{d}$. 
A point $x \in \cube{d}$ is then hashed to the index of the first ball in the sequence that contains $x$.
As we increase $d$ and decrease the radius of the balls, this scheme has a $\rho$-value for the $((1 - \alpha)d/2, (1- \beta)d/2)$-near neighbor problem of
\begin{equation} \label{eq:best_upper}
	\rho = \left. \frac{1 - \alpha}{1 + \alpha} \middle/ \frac{1-\beta}{1 + \beta} \right. + o_{d}(1).
\end{equation}
This scheme also works on the unit sphere if we replace the balls by spherical caps~\cite{terasawa2007spherical, andoni2015practical}.
The size of the gap between the lower bound in equation \eqref{eq:combined_lower} and the upper bound \eqref{eq:best_upper} is shown in Figure \ref{fig:gap}.
Since the gap is less than $0.06$ it is difficult to argue that closing the gap would have huge practical implications,
especially since the lower order terms in existing constructions exceed this for most realistic applications~\cite{andoni2015practical}. 
Nevertheless, considering the tools that have gone into proving the existing lower bounds, 
we believe that it is of fundamental mathematical interest to understand how to best separate $\beta$-correlated points from $\alpha$-correlated points.
\begin{figure} 
	\centering
	\includegraphics[width=0.85\textwidth]{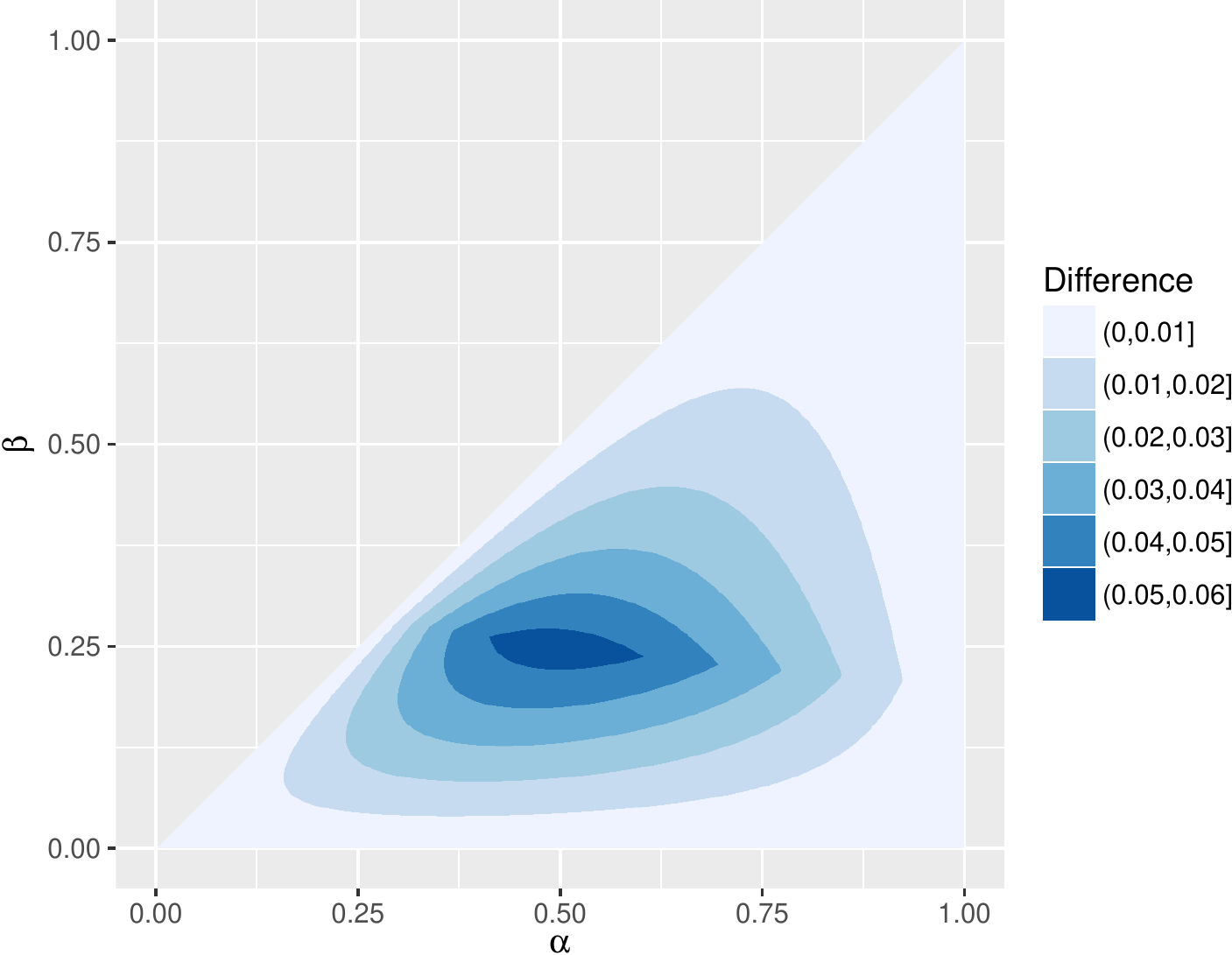}
	\caption{The gap in the $\rho$-value between the best known upper and lower bounds for families of $((1-\alpha)d/2, (1-\beta)d/2, p_1, p_2)$-sensitive hash functions in $d$-dimensional Hamming space.}
	\label{fig:gap}
\end{figure}
\subsection{Beyond locality-sensitive hashing}
A common theme among recent advances in the area of theoretical approximate similarity search has been to move beyond standard locality-sensitive hashing~\cite{andoni2014beyond, laarhoven2015, shrivastava2014asymmetric, becker2016, andoni2017optimal, christiani2017framework, christiani2017set, aumuller2017distance}.
The results in this direction usually modify part of the framework, for example by constructing the locality-sensitive family by looking at the data, 
but the underlying approach of using locality-sensitive mappings from points to buckets remains the same.
This thesis explores several variations of standard locality-sensitive hashing and we therefore briefly introduce some of this work here.

\paragraph{Data-dependent locality-sensitive hashing.}
A sequence of papers~\cite{andoni2014beyond, andoni2015optimal, andoni2016tight, andoni2017practical, andoni2017optimal} has explored the idea of data-dependent locality-sensitive hashing: 
If we allow the construction of $\LSH$ to depend on the set of data points $P$, how fast can we then solve the approximate near neighbor problem?
Andoni and Razenshteyn was able to show matching upper and lower bounds of $\rho = 1/(2c^2 - 1) + o_{d}(1)$ in Euclidean space~\cite{andoni2015optimal, andoni2016tight}. 
This matches standard LSH upper and lower bounds in the case of random instances on the unit sphere, and indeed the construction by Andoni and Razenshteyn is based on a reduction to this case. 
Unfortunately the construction and its analysis is complicated and suffer from large lower order terms~\cite{andoni2017optimal}, 
although recent work has found some success in striking a balance between algorithmic simplicity and theoretical optimality using data-dependence in Hamming space~\cite{andoni2017practical}. 

\paragraph{Asymmetric locality-sensitive hashing.}
Asymmetric locality-sensitive hashing extends the concept of standard locality-sensitive hashing to cover distributions over pairs of functions $(h, g) \sim \ALSH$ and studies how the probability of collision between pairs of points can be made to depend on the distance/similarity between the points~\cite{shrivastava2014asymmetric, aumuller2017distance}.
This modification to standard locality-sensitive hashing opens up new applications such as approximate search for furthest neighbors, orthogonal vectors~\cite{vijayanarasimhan2014hashing}, and annulus queries (see \cite{aumuller2017distance} for an overview).
In Chapter~\ref{ch:asymmetric} we show lower bounds for asymmetric locality-sensitive hashing.

\paragraph{Space-time tradeoffs.}
The standard locality-sensitive hashing framework offers a balanced space-time tradeoff that is the result of a symmetric query and update procedure: 
Every data point is stored in $O(n^\rho)$ buckets and during queries we probe $O(n^\rho)$ buckets.
A line of work has investigated how the query and update parts of the algorithm can be modified to yield different tradeoffs between space usage and query time~\cite{panigrahy2006, lv2007, andoni2009, kapralov2015, laarhoven2015, christiani2017framework, andoni2017optimal}.
Typically the performance of such solutions is expressed by two exponents: $\rho_u$ and $\rho_q$. 
During updates we store points in $O(n^{\rho_u})$ buckets and during queries we probe $O(n^{\rho_q})$ buckets.

Early work in this area focused on how to modify the standard locality-sensitive hashing query and update algorithms using an idea known as multi-probing~\cite{lv2007}.
Regular locality-sensitive hashing uses $L = O(n^\rho)$ hash functions $h_1, \dots, h_L$.
Suppose $h_{l}(q)$ denotes the $l$th bucket to be probed during the standard LSH query algorithm.
By inspecting buckets in the neighborhood of $h_{l}(q)$, for example by adding some noise $z$ to $q$ and probing $h(q + z)$, we can increase the probability of finding a near neighbor of $q$, which in turn allows us to reduce $L$ while maintaining correctness. 

Recent breakthroughs in this area have come by abandoning the locality-sensitive hashing framework in favor of a more direct approach based on locality-sensitive filtering~\cite{laarhoven2015, christiani2017framework}.
Finally, Andoni et al.~\cite{andoni2017optimal} have combined their data-dependent approach to locality-sensitive hashing with the best known space-time tradeoff solutions for random data to obtain optimal space-time tradeoffs, as shown by matching lower bounds.
The optimal trade-off between $\rho_q, \rho_u \geq 0$ for the $(r, cr)$-near neighbor problem in Euclidean space can be described by the equation
\begin{equation*}
	c^2 \sqrt{\rho_q} + (c^2 - 1)\sqrt{\rho_u} = \sqrt{2c^2 - 1}.
\end{equation*}
For a balanced tradeoff this collapses to $1/(2c^2 -1 )$ which is tight for data-dependent locality-sensitive hashing, but the bound has been shown to be tight for every choice of $\rho_q, \rho_u$ that satisfies the equation.

\paragraph{Locality-sensitive filters and maps.}
Locality-sensitive filtering~\cite{becker2016} differs from locality-sensitive hashing in that it uses locality-sensitive subsets of space (filters) rather than locality-sensitive partitions (hash functions) to solve the approximate near neighbor problem.
An example of a locality-sensitive filter family is the distribution over balls of some fixed radius in Hamming space.
This idea is further extended to allow asymmetry by using different filters for queries and updates~\cite{laarhoven2015, christiani2017framework}. 
It turns out that the filter family of consisting of pairs of concentric balls in Hamming space can be used to solve the approximate near-neighbor problem with optimal space-time tradeoffs, matching the lower bound of Andoni et al.~\cite{andoni2017optimal} for random data. 
Chapter~\ref{ch:filters} further introduces locality-sensitive filtering and space-time tradeoffs.

In even greater generality we can think of locality-sensitive hashing and filtering as being approaches to constructing randomized mappings $M \colon X \to 2^{R}$ (where $2^R$ denotes the power set of $R$) from a space $(X, \dist)$ to a collection of $|R|$ buckets that satisfy certain properties. 
Recent work on set similarity search (Chapter \ref{ch:sparse}) and improvements to the standard locality-sensitive hashing framework (Chapter \ref{ch:fast}) explores these ideas and obtains efficient search algorithms by deviating from the standard approach. 
\section{Part II: Pseudorandom hashing and generators}
The second part of this thesis contains results on efficient pseudorandom hash functions and pseudorandom number generators.
We are interested in replacing the use of true randomness in randomized algorithms and data structures with the output of a pseudorandom hash function or generator, 
stretching a small random seed into a much larger output of pseudorandom values, while retaining guarantees on the performance of these algorithms.
For a primer on the general study of pseudorandom generators see~\cite{goldreich2010}.

\paragraph{Universal hashing.}
The pseudorandomness part of this thesis focuses on one specific type of pseudorandomness known as $k$-wise independence or $k$-independence,
first introduced to the field of computer science through the concept of universal hashing by Carter and Wegman~\cite{carter1977}.
\begin{definition} \label{intro:def:kindependence}
Let $k$ be a positive integer and let $\mathcal{F}$ be a family of functions from $U$ to $R$.
We say that~$\mathcal{F}$ is a \emph{$k$-independent} family of functions if for every choice of $k$ distinct keys $x_{1}, \dots, x_{k}$ and arbitrary values $y_{1}, \dots, y_{k}$ we have that 
\begin{equation*}
\Pr_{f \sim \mathcal{F}}[f(x_{1}) = y_{1} \land f(x_{2}) = y_{2} \land \dots \land f(x_{k}) = y_{k}] = |R|^{-k}.
\end{equation*}
Furthermore, we say that $f$ is $k$-independent when it is selected uniformly at random from a family of $k$-independent functions.
\end{definition}
We can sample a $k$-independent hash function $f(x) = \sum_{i = 0}^{k-1} a_i x^i \bmod p$ by sampling each $a_i$ uniformly at random from the set $\{0, 1, \dots, p-1\}$ where $p$ is prime. 
In fact, the family of polynomials of degree at most $k-1$ over a finite field is $k$-independent~\cite{joffe1974}.
We are typically interested in applications where the size of the universe $u = |U|$ is much larger than the degree of independence $k$.

Different types of hashing-based dictionaries work for $k$-independent hash functions with $k$ much smaller than the number of elements in the dictionary which we denote by $n$.
For example, it was shown that $5$-independence suffices for linear probing to ensure expected constant time per operation~\cite{pagh2011linear}.
It is known that $\Theta(\log n)$-independence suffices for Cuckoo hashing~\cite{pagh2004cuckoo}, but $5$-independence is not enough to ensure constant amortized cost per operation~\cite{cohen2009bounds}.
For a brief introduction to the use of random hashing in algorithms and data structures see~\cite{dietzfelbinger2012}. 

\paragraph{Fast hashing and lower bound.}
For applications that require super-constant independence, the time to evalute the hash function can be a performance bottleneck.
A $k$-independent polynomial hash function can be stored using $O(k)$ words and evaluated using time $O(k)$ on a word-RAM, assuming constant time arithmetic over the finite field.  
What if we are willing to use more space to represent a $k$-independent hash function $f \sim \mathcal{F}$ in order to reduce the evaluation time? 
Siegel~\cite{siegel2004} gave a powerful cell-probe lower bound for this problem, showing that for a $k$-independent hash function with domain size $u$, even if we use space roughly $O(ku^{1/t})$ for some $t \geq 1$ the evaluation time has to be $\Omega(t)$.

Siegel also showed the existence of a matching upper bound based on highly unbalanced bipartite expander graphs $G = (U \cup V, E)$ with left vertex set $U$ corresponding to the domain of the hash function,
right vertex set $V$ of size $|V| = O(k u^{1/t})$, and left outdegree $d = O(t)$.
Given an appropriate expander graph $G$ we can sample a $k$-independent hash function $f \colon U \to R$ by associating each vertex $v \in V$ with a random element from $R$ where we assume that $(R, +)$ is an abelian group, such as the the integers under modular arithmetic.
To compute $f(x)$ we take the sum of the random elements associated with the neighbors of the vertex $x \in U$ and return the result.

Unfortunately we only know of the existence of such optimal expander graphs by the probablistic method: a random bipartite graph has the right properties for optimal $k$-independent hashing with overwhelming probability if we parameterize the graph generation process correctly. 
Several works, Siegel's original paper included, attempt to approach the performance of such optimal bipartite expander graph by the use of probabilistic constructions~\cite{dietzfelbinger2003, pagh2008, thorup2013, christiani2015}.
In Chapter \ref{ch:expanders} we show a probablistic construction with space usage and evaluation time that almost matches the lower bound.
Finding optimal explicit constructions remains a major open problem.

Other approaches to the problem of finding fast hash functions with theoretical guarantees include the study of tabulation hashing and its variations which has guarantees beyond what can be derived from the degree of $k$-independence~\cite{thorup2017fast}, to simulate uniformly random hashing in constant time on a subset of the universe~\cite{pagh2008}, reusing randomness by splitting the problem into sub-problems that share a single highly random hash function~\cite{dietzfelbinger2009applications}, or extracting additional randomness from the input to the hash function~\cite{mitzenmacher2013}.

\paragraph{Generating $k$-independent random variables.}
The generation of $k$-independent random variables differs from random hashing by allowing the algorithm designer to specify where to evaluate a $k$-independent function $f$ in order to generate a sequence of variables $f(x_1), f(x_2), \dots$ that is $k$-independent.   
The problem of generating a sequence of $k$-independent random variables is therefore easier than the problem of constructing a data structure to represent a random $k$-independent hash function that an adversary can choose to evaluate in an arbitrary point.

We can take a standard $k$-independent polynomial hash function and evaluate it in $k$ points in time $k \poly \log k$ using fast multipoint evaluation algorithms~\cite{bostan2005, gathen2013}, giving us a generator of $k$-independent random variables with generation time $\poly \log k$ per variable that uses space $O(k)$.
This in itself shows that the task of generation is easier than hashing, 
as it would be impossible to evaluate a $k$-independent hash function in time $\poly \log k$ using space $O(k)$, if for example $k = \Theta(\log u)$.
In Chapter \ref{ch:rng} we show how to generate $k$-independent variables in constant time, independent of $k$, using space $k \poly \log k$.
\section{Overview and contributions}
This thesis is divided into two parts. The first part presents algorithms and lower bounds for various problems related to similarity search. 
The second part presents algorithms for the efficient generation of high-quality pseudorandom numbers, as well as efficient hash functions. 
The chapters are based on the following papers:
\begin{itemize}
	\item[I.] Similarity search. 
	\item[2.] Tobias Christiani: Fast locality-sensitive hashing frameworks for approximate near neighbor search~\cite{christiani2017fast}. 2017. Unpublished. 
	\item[3.] Tobias Christiani: A Framework for Similarity Search with Space-Time Tradeoffs using Locality-Sensitive Filtering~\cite{christiani2017framework}. SODA 2017.
	\item[4.] Tobias Christiani and Rasmus Pagh: Set similarity search beyond MinHash~\cite{christiani2017set}. STOC 2017.
	\item[5.] Tobias Christiani, Rasmus Pagh and Johan Sivertsen: Scalable and robust set similarity join~\cite{christiani2017scalable}. ICDE 2018.
	\item[6.] Martin Aumüller, Tobias Christiani, Rasmus Pagh and Francesco Silvestri: Distance-sensitive hashing~\cite{aumuller2017distance}. PODS 2018.
	\item[7.] Tobias Christiani: Optimal Boolean locality-sensitive hashing. 2018. Unpublished.
	\item[II.] Pseudorandomness.
	\item[8.] Tobias Christiani and Rasmus Pagh: Generating $k$-independent variables in constant time~\cite{christiani2014}. FOCS 2014.
	\item[9.] Tobias Christiani, Rasmus Pagh and Mikkel Thorup: From Independence to Expansion and Back Again~\cite{christiani2015}. STOC 2015.
\end{itemize}
We proceed by giving a brief description of the contribution of each chapter.
\subsection{Part I: Similarity search}
\paragraph{Chapter \ref{ch:fast}: Fast locality-sensitive hashing frameworks.}
This chapter begins by surveying different techniques for constructing a solution to the approximate near neighbor problem from a family of locality-sensitive hash functions. 
Given a family $\LSH$ of locality-sensitive hash functions, the standard Indyk-Motwani framework (Theorem \ref{intro:thm:lsh_im_simple}) uses $O(n^\rho \log n)$ functions from $\LSH$ to solve the approximate near neighbor problem.
During a query all of these hash functions are evaluated, dominating the query time.
For many LSH schemes the time to evaluate a single function is $O(d)$ or greater, as witnessed for example by SimHash or MinHash, further exacerbating the problem.
Building on recent work by Dahlgaard et al.~\cite{dahlgaard2017fast} we show that the number of locality-sensitive hash functions can be reduced to $O(\log^2 n)$ in general, yielding an improved LSH framework. 
We combine this result with a technique from another LSH framework by Andoni and Indyk~\cite{andoni2006efficient} to reduce the word-RAM complexity of this improved framework by a logarithmic factor to $O(n^\rho)$.

\paragraph{Chapter \ref{ch:filters}: Space-time tradeoffs for similarity search.}
This chapter introduces a framework for solving the approximate near neighbor problem with space-time tradeoffs using locality-sensitive filtering. 
We show concrete solutions on the unit sphere under cosine similarity with extensions to $\ell_p$-space for every $0 < p \leq 2$.
These results improve and generalize prior work~\cite{laarhoven2015, kapralov2015}. 
We also include a lower bound on space-time tradeoff that is tight, but suffers from some important restrictions.
A paper by Andoni et al.~\cite{andoni2017optimal} has since shown a strengthened lower bound and an improved upper bound through the use of data-dependent techniques.
An early version of the paper behind this chapter formed part of my master's thesis.
At the end of the chapter we have added an improved locality-sensitive filtering framework compared to the one in the main text, building on ideas introduced in Chapter \ref{ch:fast} and \ref{ch:sparse}.

\paragraph{Chapter \ref{ch:sparse}: Set similarity search beyond MinHash.}
In this chapter we consider the problem of set similarity search under Braun-Blanquet similarity $\simil_{B}(x, y) = |x \cap y| / \max(|x|, |y|)$.
We show that the $(s_1, s_2)$-similarity problem in this setting can be solved with an exponent of $\rho = \log(1/s_1)/\log(1/s_2)$ and that this is tight among solutions based on data-independent locality-sensitive maps. 
The upper bound is based on a novel construction inspired by branching processes and interestingly, although it is data-independent, it outperforms the best known data-dependent techinques for a large portion of the parameter space $0 \leq s_2 < s_1 < 1$.
The lower bound follows from a reduction to the standard $(r, cr)$-near neighbor problem in Hamming space for $r, cr \ll d/2$.
In this setting the lower bound by O'Donnell et al.~\cite{odonnell2014optimal} is tight and we are able to show that it extends to Braun-Blanquet similarity for \emph{every} choice of $0 \leq s_2 < s_1 < 1$. 
This is interesting in the light of the gap in our knowledge when it comes to the usual $(\alpha, \beta)$-similarity problem for cosine similarity, as explained in the introduction.

\paragraph{Chapter \ref{ch:asymmetric}: Lower bounds for asymmetric locality-sensitive hashing.}
In this chapter we derive lower bounds (on the $\rho$-value) for asymmetric locality-sensitive hashing.
Our lower bound covers the case of asymmetric families for approximate near neighbor search, as well as the case of approximate furthest neighbor search where we are interested in having the collision probability of $(h, g) \sim \ALSH$ increase in the distance between points.
We show that our lower bounds are tight against existing symmetric constructions in the case of the application to near neighbor search, and that this construction can easily be modified to yield an optimal asymmetric construction for furthest neighbor search. 

\paragraph{Chapter \ref{ch:bool}: Optimal Boolean locality-sensitive hashing.}
In this chapter we show that, among the class of Boolean locality-sensitive hash functions $h \colon \cube{d} \to \{-1, 1\}$, 
bit-sampling is an optimal LSH (minimizes the $\rho$-value) for the $((1-\alpha)d/2, (1-\beta)d/2)$-near neighbor problem in Hamming space for every choice of $0 \leq \beta < \alpha < 1$.
This stands in contrast to the lower bound by O'Donnell et al.~\cite{odonnell2014optimal} which is unrestricted with respect to the range of the locality-sensitive hash functions.
Bit-sampling only matches this unrestricted lower bound in the case where $\alpha, \beta$ approach $1$.
Our result settles the question of optimal Boolean locality-sensitive hashing for Hamming space and shows that we have to look towards families of hash functions with a larger range in order to further improve the $\rho$-value compared to bit-sampling.
Andoni et al.~\cite{andoni2015practical} have shown lower bounds on the $\rho$-value on the unit sphere as a function of the size of the range of the hash function. 
\subsection{Part II: Pseudorandom hashing and number generation}
\paragraph{Chapter \ref{ch:rng}: Generating $k$-independent random variables in constant time.}
We investigate the problem of efficiently generating $k$-independent random variables and give an explicit generator of $k$-independent random variables with constant generation time, independent of $k$. 
The explicit construction combines multipoint evaluation of polynomials over finite fields with a cascading construction of explicit bipartite expander graphs by Capalbo et al.~\cite{capalbo2002}.
The space usage of this construction is $k \poly \log k$ with a very large exponent in the polynomial.
We also show a randomized version of the same construction that uses a randomly generated bipartite graph.
This reduces the space overhead to $O(\log^3 k)$ at the cost of introducing an error probability (the generated sequence may fail to be $k$-independent) that is polynomially small in $k$. 
We implement a version of the generator that combines a random bipartite expander with fast multipoint evaluation of polynomials over $\GF_{2^{64}}$ and show that it scales well, even for generating $k = 2^{20}$-independent variables.

\paragraph{Chapter \ref{ch:expanders}: Near-optimal $k$-independent hashing.}
In this chapter we attack the problem of constructing fast $k$-independent random hash functions.
We use the fact that there is a sort of duality between randomized bipartite expander graphs and $k$-independent random hash functions.
A bipartite expander graph that expands on subsets of size $k$ can be used to construct a $k$-independent family of functions, and a $k$-independent function is likely to represent a bipartite expander that expands on subsets of size $k$.
We take a small bipartite expander graph and apply an inefficient graph product that preserves its expansion properties while increasing the size of the left vertex set (the size of the domain of the resulting hash function). 
Then we use this resulting bipartite expander graph to construct a $k$-independent random hash function that now represents a new expander on a larger domain with optimal properties. 
By applying this strategy recursively using different graph products we are able to give randomized constructions of $k$-independent hash functions in the word-RAM model that almost match Siegel's cell probe lower bound~\cite{siegel2004}.
\section{Conclusion and open problems}
\subsection{Similarity search}
We have shown new upper and lower bounds for problems related to approximate similarity search in high-dimensional spaces,
showing improved locality-sensitive hashing frameworks, lower bounds for Boolean locality-sensitive hashing, 
and going beyond locality-sensitive hashing in several different directions with asymmetric locality-sensitive hashing, 
space-time tradeoffs through locality-sensitive filtering, and locality-sensitive maps for set similarity search.

\paragraph{Optimal data-independent locality-sensitive hashing.}
It remains open to close the gap between the upper and lower bounds on the $\rho$-value of $((1-\alpha)d/2, (1-\beta)d/2, p_1, p_2)$-sensitive families in Hamming space (shown in Figure \ref{fig:gap}).
Existing lower bounds seem to have explored the limits of what can be shown with our current understanding of hypercontractive inequalities and Fourier analysis of Boolean functions. 
We conjecture that the ball-based LSH construction with the $\rho$-value given in equation \ref{eq:best_upper} is asymptotically optimal for every choice of $0 \leq \beta < \alpha < 1$. 

\paragraph{Orthogonal search.}
Suppose we are interested in an asymmetric locality-sensitive hashing scheme for the unit sphere under cosine similarity that can be used to search for orthogonal vectors.
For this purpose we want the probability of collison to be as high as possible for $0$-correlated (orthogonal) vectors and have the probability of collision decrease at the correlation becomes positive or negative. 
Let $p(\alpha)$ denote the probability of collision of the asymmetric locality-sensitive hashing scheme for a pair of $\alpha$-correlated vectors.
The current best upper bound on $\rho = \log(1/p(0))/\log(1/\max(p(\alpha), p(-\alpha)))$ is given by $(1 - \alpha^2)/(1 + \alpha^2)$~\cite{aumuller2017distance}.
The lower bound presented in Chapter~\ref{ch:asymmetric} only implies $\rho \geq (1 - |\alpha|)/(1 + |\alpha|)$. 
Obtaining a ``two-sided'' lower bound that simultanously relates $p(0)$ to both $p(\alpha)$ and $p(-\alpha)$ has close ties to the open symmetric Gaussian problem~\cite{odonnell2012open}.
It is conjectured that the upper bound is tight.

\paragraph{Simple data-dependent constructions.}
It is an important open problem to find simpler data-dependent solutions to approximate near neighbor search.
Despite the intuitive appeal of using the data to inform the construction of the solution, 
relatively few people have succeded in making theoretical progress in this area~\cite{andoni2014beyond, andoni2015optimal, andoni2017practical}.
Perhaps by relaxing the problem slightly, for example by only requiring that queries that follow a specific distribution succeed with constant probability, progress can be made.
An example of such a query distribution could be to sample one of the $n$ data points uniformly at random and sample the query from a ball around the data point.
Attacking the problem for data structures that use near-linear space in $n$ also seems like a promising approach. 
\subsection{$k$-independent hashing and generation}
We have shown near-optimal results for $k$-independent hashing and generation.

\paragraph{Optimal explicit unbalanced bipartite expander graphs.}
The main open problem in this area is the explicit construction of highly unbalanced bipartite expander graphs with optimal properties.
We would like to be able to evaluate the neighbor function $\Gamma \colon U \to V^d$ of a left $d$-regular bipartite expander graph with optimal parameters (matching Siegel's lower bound for $k$-independent hashing) using time that is at most polynomial in the bit-length of the input. 
For the application to random hashing we would furthermore like to be able to list the $d$ neighbors of a vertex in time $O(d)$.
The construction in Chapter~\ref{ch:expanders} is essentially able to solve this task in time $O(d \log d)$, 
so it would require a very clean explicit construction to yield an improvement to the efficiency of random hashing in practice. 
Results on the construction of explicit bipartite expanders by Guruswami et al.~\cite{guruswami2009} and preprocessing polynomials~\cite{kedlaya2008} are based directly on results such as the fundamental theorem of algebra and the Chinese remainder theorem and give hope that there exists a simple explicit construction. 

\paragraph{Constant time generators with minimal space.}
The fast generators in Chapter \ref{ch:rng} uses polynomials over finite fields and require space $k \poly \log k$.
Through the sequential evaluation of hash functions presented in Chapter \ref{ch:expanders} we can remove the need for arithmetic over finite fields, but it seems that if we want to use minimal space the evalution time will still be $O(\log u)$ with space usage $k \poly(\log u, \log k)$.
Is it possible to get constant-time generation in a restricted word-RAM model without multiplication using space $O(k)$?

\part{Similarity search}
\chapter{Fast locality-sensitive hashing frameworks} \label{ch:fast} 
\sectionquote{Renewed shall be blade that was broken}
\noindent The Indyk-Motwani Locality-Sensitive Hashing (LSH) framework (STOC 1998) is a general technique for constructing a data structure to answer approximate near neighbor queries by using a distribution $\LSH$ over locality-sensitive hash functions that partition space.
For a collection of $n$ points, after preprocessing, the query time is dominated by $O(n^{\rho} \log n)$ evaluations of hash functions from $\LSH$ and $O(n^{\rho})$ hash table lookups and distance computations where $\rho \in (0,1)$ is determined by the locality-sensitivity properties of $\LSH$. 
It follows from a recent result by Dahlgaard et al.~(FOCS 2017) that the number of locality-sensitive hash functions can be reduced to $O(\log^2 n)$, leaving the query time to be dominated by $O(n^{\rho})$ distance computations and $O(n^{\rho} \log n)$ additional word-RAM operations.
We state this result as a general framework and provide a simpler analysis showing that the number of lookups and distance computations closely match the Indyk-Motwani framework. 
Using ideas from another locality-sensitive hashing framework by Andoni and Indyk (SODA 2006) we are able to reduce the number of additional word-RAM operations to $O(n^\rho)$.
\section{Introduction}
The $(r_1, r_2)$-approximate near neighbor problem is the problem of preprocessing a collection $P$ of $n$ points in a space $(X, \dist)$ into a data structure that after preprocessing supports the following query operation: Given a query point $q \in X$, if there exists a point $x \in P$ with $\dist(q, x) \leq r_1$, then the data structure is guaranteed to return a point $x' \in P$ such that $\dist(q, x') < r_2$.

Indyk and Motwani \cite{indyk1998} introduced a general framework for constructing solutions to the approximate near neighbor problem using a technique known as locality-sensitive hashing~(LSH). 
The framework takes a distribution over hash functions $\LSH$ with the property that near points are more likely to collide under a random $h \sim \LSH$. 
During preprocessing a number of locality-sensitive hash functions are sampled from $\LSH$ and used to hash the points of $P$ into buckets. 
The query algorithm evaluates the same hash functions on the query point and looks into the associated buckets to find an approximate near neighbor.

The locality-sensitive hashing framework of Indyk and Motwani has had a large impact in both theory and practice (see surveys \cite{andoni2008} and \cite{wang2014} for an introduction), and many of the best known (data-independent) solutions to the approximate near neighbor problem in high-dimensional spaces, such as Euclidean space \cite{andoni2006}, the unit sphere under inner product similarity \cite{andoni2015practical}, and sets under Jaccard similarity \cite{broder2000} come in the form of families of locality-sensitive hash functions that can be plugged into the Indyk-Motwani LSH framework.
Recent work on data-dependent locality-sensitive hashing has further improved solutions for $\ell_p$-spaces and cosine similarity~\cite{andoni2014beyond,andoni2015optimal,andoni2017optimal}, 
but these solutions typically do not come directly in the form of a distribution over locality-sensitive hash functions and as such it is unclear whether the techniques in this paper can yield further speedups to these results.
\begin{definition}[Locality-sensitive hashing {\cite{indyk1998}}]\label{def:lsh}
Let $(X, \dist)$ be a distance space and let $\LSH$ be a distribution over functions $h \colon X \to R$.
We say that $\LSH$ is $(r_1, r_2, p_1, p_2)$-sensitive if for $x, y \in X$ and $h \sim \LSH$ we have that:
\begin{itemize}
\item If $\dist(x, y) \leq r_1$ then $\Pr[h(x) = h(y)] \geq p_1$.
\item If $\dist(x, y) \geq r_2$ then $\Pr[h(x) = h(y)] \leq p_2$.
\end{itemize}
\end{definition}

\noindent
The Indyk-Motwani framework takes a $(r_1, r_2, p_1, p_2)$-sensitive family $\LSH$ and constructs a data structure that solves the approximate near neighbor problem for parameters $r_1 < r_2$ with some positive constant probability of success. 
We will refer to this randomized approximate version of the near neighbor problem as the $(r_1, r_2)$-near neighbor problem, where we require queries to succeed with probability at least $1/2$ (see Definition \ref{def:ann}).  
To simplify the exposition we will assume throughout the introduction, unless otherwise stated, that $0 < p_1 < p_2 < 1$ are constant, that a hash function $h \in \LSH$ can be stored in $n / \log n$ words of space, and for $\rho = \log(1/p_1) / \log(1/p_2) \in (0,1)$ that a point $x \in X$ can be stored in $O(n^\rho)$ words of space.
The assumption of a constant gap between $p_1$ and $p_2$ allows us to avoid performing distance computations by instead using the $1$-bit sketching scheme of Li and K{\"o}nig~\cite{li2011theory} together with the family $\LSH$ to approximate distances (see Section \ref{sec:sketching} for details).
In the remaining part of the paper we will state our results without any such assumptions to ensure, for example, that our results hold in the important case where $p_1, p_2$ may depend on $n$ or the dimensionality of the space~\cite{andoni2006, andoni2015practical}.  
\begin{theorem}[Indyk-Motwani {\cite{indyk1998, har-peled2012}, simplified}]\label{thm:lsh_im_simple}
Let $\LSH$ be $(r_1, r_2, p_1, p_2)$-sensitive and let $\rho = \frac{\log(1/p_1)}{\log(1/p_2)}$, then there exists a solution to the $(r_1, r_2)$-near neighbor problem using $O(n^{1+\rho})$ words of space and with query time dominated by $O(n^{\rho} \log n)$ evaluations of functions from~$\LSH$.
\end{theorem}
The query time of the Indyk-Motwani framework is dominated by the number of evaluations of locality-sensitive hash functions. 
To make matters worse, almost all of the best known and most widely used locality-sensitive families have an evalution time that is at least linear in the dimensionality of the underlying space~\cite{broder2000, charikar2002, datar2004, andoni2006, andoni2015practical}.
Significant effort has been devoted to the problem of reducing the evaluation complexity of locality-sensitive hash families~\cite{terasawa2007spherical, eshgi2008, dasgupta2011fast, andoni2015practical, kennedy2017fast, shrivastava2016simple, shrivastava2017optimal, dahlgaard2017fast}, while the question of how many independent locality-sensitive hash functions are actually needed to solve the $(r_1, r_2)$-near neighbor problem has received relatively little attention~\cite{andoni2006efficient, dahlgaard2017fast}.  

This paper aims to bring attention to, strengthen, generalize, and simplify results that reduce the number of locality-sensitive hash functions used to solve the $(r_1, r_2)$-near neighbor problem. In particular, we will extract a general framework from a technique introduced by Dahlgaard et al.~\cite{dahlgaard2017fast} in the context of set similarity search under Jaccard similarity, showing that the number of locality-sensitive hash functions can be reduced to $O(\log^2 n)$ in general. 
Reducing the number of locality-sensitive hash functions allows us to spend time $O(n^\rho / \log^2 n)$ per hash function evaluation without increasing the overall complexity of the query algorithm --- something which is particularly useful in Euclidean space where the best known LSH upper bounds offer a tradeoff between the $\rho$-value that can be achieved and the evaluation complexity of the locality-sensitive hash function~\cite{andoni2006, andoni2015practical, kennedy2017fast}.

The main technical contribution of this paper is to reduce the word-RAM complexity of the general LSH framework from $O(n^\rho \log n)$ to $O(n^\rho)$ by combining techniques from Dahlgaard et al. and Andoni and Indyk~\cite{andoni2006efficient}.
\subsection{Related work}
\paragraph{Indyk-Motwani.}
The Indyk-Motwani framework uses $L = O(n^\rho)$ independent partitions of space, each formed by overlaying $k = O(\log n)$ random partitions induced by $k$ random hash functions from a locality-sensitive family $\LSH$.
The parameter $k$ is chosen such that a random partition has the property that a pair of points $x,y \in X$ with $\dist(x, y) \leq r_1$ has probability $n^{-\rho}$ of ending up in the same part of the partition, while a pair of points with $\dist(x, y) \geq r_2$ has probability $n^{-1}$ of colliding. 
By randomly sampling $L = O(n^\rho)$ such partitions we are able to guarantee that a pair of near points will collide with constant probability in at least one of them. 
Applying these $L$ partitions to our collection of data points $P$ and storing the result of each partition of $P$ in a hash table we obtain a data structure that solves the $(r_1, r_2)$-near neighbor problem as outlined in Theorem \ref{thm:lsh_im_simple} above. 
Section \ref{sec:frameworks} and \ref{sec:im} contains a more complete description of LSH-based frameworks and the Indyk-Motwani framework.

\paragraph{Andoni-Indyk.}
As previously mentioned, many locality-sensitive hash functions happen to have a super-constant evaluation time. 
This motivated Andoni and Indyk to introduce a replacement to the Indyk-Motwani framework in a paper on substring near neighbor search~\cite{andoni2006efficient}.  
The key idea is to re-use hash functions from a small collection of size $m \ll L$ by forming all combinations of $\binom{m}{t}$ hash functions. 
This technique is also known as tensoring and has seen some use in the work on alternative solutions to the approximate near neighbor problem, in particular the work on locality-sensitive filtering~\cite{dubiner2010bucketing, becker2016, christiani2017framework}.
By applying the tensoring technique the Andoni-Indyk framework reduces the number of hash functions to $O(\exp(\sqrt{\rho \log n \log \log n})) = n^{o(1)}$ as stated in Theorem \ref{thm:lsh_ai_simple}.
\begin{theorem}[Andoni-Indyk {\cite{andoni2006efficient}}, simplified]\label{thm:lsh_ai_simple}
	Let $\LSH$ be $(r_1, r_2, p_1, p_2)$-sensitive and let $\rho = \frac{\log(1/p_1)}{\log(1/p_2)}$, then there exists a solution to the $(r_1, r_2)$-near neighbor problem using $O(n^{1+\rho})$ words of space and with query time dominated by $O(\exp(\sqrt{\rho \log n \log \log n}))$ evaluations of functions from $\LSH$ and $O(n^\rho)$ other word-RAM operations.
\end{theorem}

The paper by Andoni and Indyk did not state this result explicitly as a theorem in the same form as the Indyk-Motwani framework; the analysis made some implicit restrictive assumptions on $p_1, p_2$ and ignored integer constraints. 
Perhaps for these reasons the result does not appear to have received much attention, although it has seen some limited use in practice~\cite{sundaram2013streaming}.
In Section \ref{sec:ai} we present a slightly different version of the Andoni-Indyk framework together with an analysis that satisfies integer constraints, providing a more accurate assessment of the performance of the framework in the general, unrestricted case.   

\paragraph{Dahlgaard-Knudsen-Throup.}
The paper by Dahlgaard et al.~\cite{dahlgaard2017fast} introduced a different technique for constructing the $L$ hash functions/partitions from a smaller collection of $m$ hash functions from $\LSH$. 
Instead of forming all combinations of subsets of size $t$ as the Andoni-Indyk framework they instead sample $k$ hash functions from the collection to form each of the $L$ partitions.
The paper focused on a particular application to set similarity search under Jaccard similarity, and stated the result in terms of a solution to this problem. 
In Section \ref{sec:dkt} we provide a simplified and tighter analysis to yield a general framework:
\begin{theorem}[Dahlgaard-Knudsen-Thorup {\cite{dahlgaard2017fast}}, simplified]\label{thm:lsh_dkt_simple}
	Let $\LSH$ be $(r_1, r_2, p_1, p_2)$-sensitive and let $\rho = \frac{\log(1/p_1)}{\log(1/p_2)}$, then there exists a solution to the $(r_1, r_2)$-near neighbor problem using $O(n^{1+\rho})$ words of space and with query time dominated by $O(\log^2 n)$ evaluations of functions from $\LSH$ and $O(n^\rho \log n)$ other word-RAM operations.
\end{theorem}
The analysis of \cite{dahlgaard2017fast} indicates that the Dahlgaard-Knudsen-Thorup framework, when compared to the Indyk-Motwani framework, would use at least $50$ times as many partitions (and a corresponding increase in the number of hash table lookups and distance computations) to solve the $(r_1, r_2)$-near neighbor problem with success probability at least $1/2$. 
Using elementary tools, the analysis in this paper shows that we only have to use twice as many partitions as the Indyk-Motwani framework to obtain the same guarantee of success.

\paragraph{Number of hash functions.}
To provide some idea of the number of hash functions $H$ used by the different frameworks, Figure \ref{fig:comparison} shows the value of $\log_2 H$ that is obtained by the Indyk-Motwani (IM), Andoni-Indyk (AI), and Dahlgaard-Knudsen-Thorup (DKT) frameworks according to the analysis in Section \ref{sec:frameworks} for $p_1 = 1/2$ and every value of $0 < p_2 < 1/2$ for a solution to the $(r_1, r_2)$-near neighbor problem on a collection of $n = 2^{30}$ points with success probability at least $1/2$. 
Note that Figure \ref{fig:comparison} shows an upper bound on the number of hash functions used by the frameworks according to the analysis in order to provide a solution with theoretical guarantees to the approximate near neighbor problem for any data set, and not the actual setting required for a particular data set (we haven't actually performed an experiment on $2^{30}$ points).  
In the analysis behind Figure \ref{fig:comparison} we have attempted to minimize $H$ within each respective framework.
\begin{figure} 
	\centering
	\includegraphics[width=0.85\textwidth]{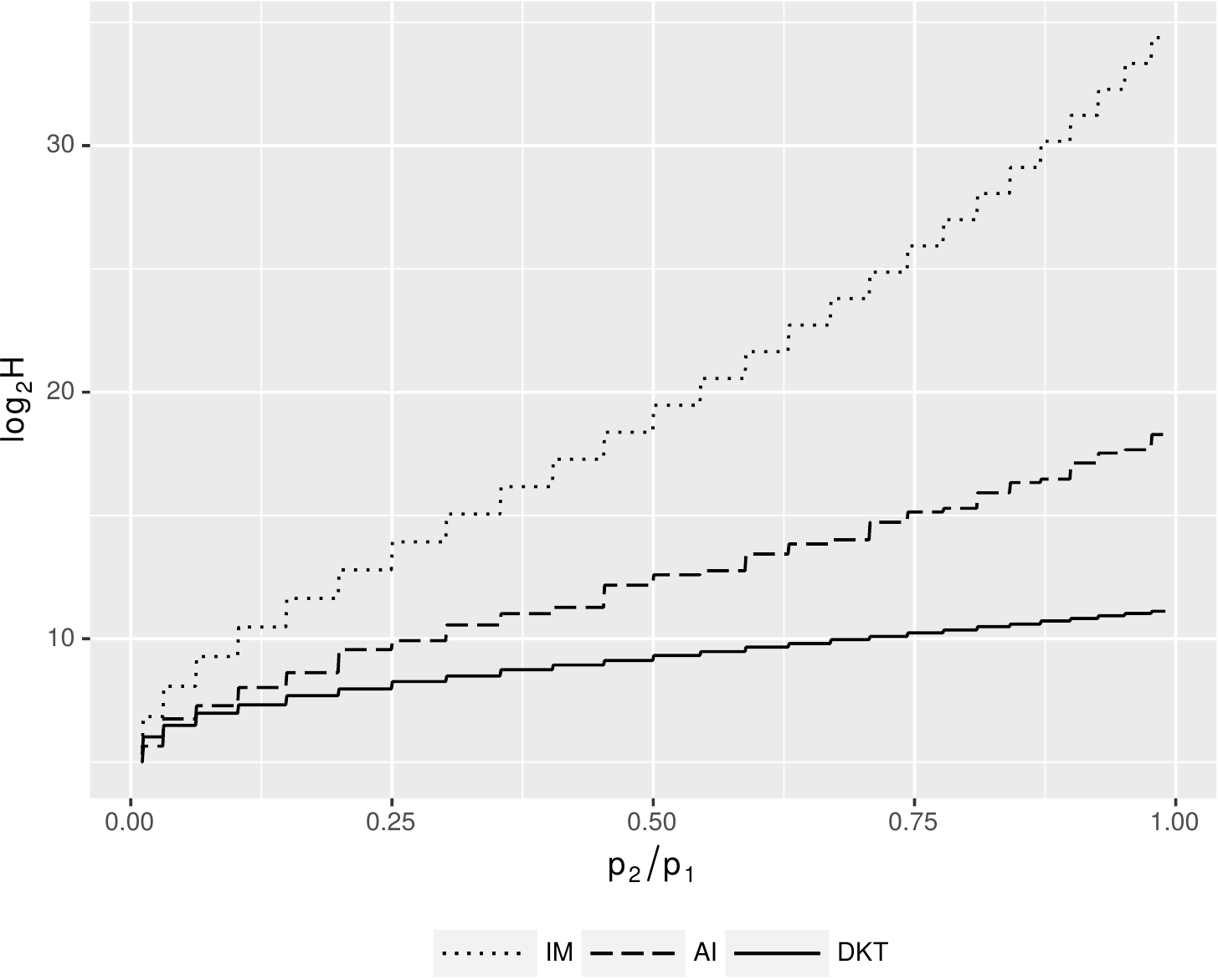}
	\caption{Upper bounds on the number of locality-sensitive hash functions from a $(r_1, r_2, 0.5, p_2)$-sensitive family used by different frameworks to solve the $(r_1, r_2)$-near neighbor problem on a collection of $2^{30}$ points according to the analysis in this paper.}
	\label{fig:comparison}
\end{figure}

Figure \ref{fig:comparison} reveals that the number of hash functions used by the Indyk-Motwani framework exceeds $2^{30}$, the size of the collection of points $P$, as $p_2$ approaches $p_1$. 
In addition, locality-sensitive hash functions used in practice such as Charikar's SimHash~\cite{charikar2002} and $p$-stable LSH~\cite{datar2004} have evaluation time $O(d)$ for points in $\mathbb{R}^d$. 
These two factors might help explain why a linear scan over sketches of the entire collection of points is a popular approach to solve the approximate near neighbor problem in practice~\cite{weiss2008spectral, gong2012angular}.
The Andoni-Indyk framework reduces the number of hash functions by several orders of magnitude, and the Dahlgaard-Knudsen-Thorup framework presents another improvement of several orders of magnitude. 
Since the word-RAM complexity of the DKT framework matches the the number of hash functions used by the IM framework, the gap between the solid line (DKT) and the dotted line (IM) gives some indication of the time we can spend on evaluating a single hash function in the DKT framework without suffering a noticeable increase in the query time. 
\subsection{Contribution}
\paragraph{Improved word-RAM complexity.}
In addition to our work on the Andoni-Indyk and Dahlgaard-Knudsen-Thorup frameworks as mentioned above, we show how the word-RAM complexity of the DKT framework can be reduced by a logarithmic factor. 
The solution is a simple combination of the DKT sampling technique and the AI tensoring technique: 
First we use the DKT sampling technique twice to construct two collections of $\sqrt{L}$ partitions.
Then we use the AI tensoring technique to form $L = \sqrt{L} \times \sqrt{L}$ pairs of partitions from the two collections. 
Below we state our main Theorem \ref{thm:lsh_dkt_ram} in its general form where we make no implicit assumptions about $\LSH$ ($p_1$ and $p_2$ are not assumed to be constant and can depend on for example $n$) or about the complexity of storing a point or a hash function, or computing the distance between pairs of points in the space $(X, \dist)$.
\begin{theorem}\label{thm:lsh_dkt_ram}
	Let $\LSH$ be $(r_1, r_2, p_1, p_2)$-sensitive and let $\rho = \log(1/p_1) / \log(1/p_2)$, then there exists a solution to the $(r_1, r_2)$-near neighbor with the following properties:
\begin{itemize}
	\item The query complexity is dominated by $O(\log_{1/p_2}^{2}(n)/p_1)$ evaluations of functions from $\LSH$, $O(n^\rho)$ distance computations, and $O(n^{\rho} / p_1)$ other word-RAM operations.
	\item The solution uses $O(n^{1 + \rho} /p_1)$ words of space in addition to the space required to store the data and $O(\log_{1/p_2}^{2}(n)/p_1)$ functions from $\LSH$.
\end{itemize}
\end{theorem}

Under the same simplifying assumptions used in the statements of Theorem \ref{thm:lsh_im_simple}, \ref{thm:lsh_ai_simple}, and \ref{thm:lsh_dkt_simple}, our main Theorem \ref{thm:lsh_dkt_ram} can be stated as Theorem \ref{thm:lsh_dkt_simple} with the word-RAM complexity reduced by a logarithmic factor to $O(n^\rho)$. 
This improvement in the word-RAM complexity comes at the cost of a (rather small) constant factor increase in the number of hash functions, lookups, and distance computations compared to the DKT framework.
By varying the size $m$ of the collection of hash functions from $\LSH$ and performing independent repetitions we can obtain a tradeoff between the number of hash functions and the number of lookups.
In Section~\ref{sec:corner} we remark on some possible improvements in the case where $p_2$ is large.

\paragraph{Distance sketching using LSH.}
Finally, we combine Theorem \ref{thm:lsh_dkt_ram} with the 1-bit sketching scheme of Li and K{\"o}nig~\cite{li2011theory} where we use the locality-sensitive hash family to create sketches that allow us to leverage word-level parallelism and avoid direct distance computations. 
This sketching technique is well known and has been used before in combination with LSH-based approximate similarity search~\cite{christiani2017scalable}, but we believe there is some value in the simplicity of the analysis and in a clear statement of the combination of the two results as given in Theorem~\ref{thm:lsh_dkt_ram_sketch}, for example in the important case where $0 < p_2 < p_1 < 1$ are constant.
\begin{theorem}\label{thm:lsh_dkt_ram_sketch}
	Let $\LSH$ be $(r_1, r_2, p_1, p_2)$-sensitive and let $\rho = \log(1/p_1) / \log(1/p_2)$, then there exists a solution to the $(r_1, r_2)$-near neighbor with the following properties:
\begin{itemize}
	\item The complexity of the query operation is dominated by $O(\log^2(n)/(p_1 - p_2)^2)$ evaluations of hash functions from $\LSH$ and $O(n^{\rho}/(p_1 - p_2)^2)$ other word-RAM operations.
	\item The solution uses $O(n^{1 + \rho}/p_1 + n/(p_1 - p_2)^2)$ words of space in addition to the space required to store the data and $O(\log^2(n)/(p_1 - p_2)^2)$ hash functions from $\LSH$.
\end{itemize}
\end{theorem}
\section{Preliminaries}\label{fast:sec:preliminaries}
\paragraph{Problem and dynamization.}
We begin by defining the version of the approximate near neighbor problem that the frameworks presented in this paper will be solving:
\begin{definition}\label{def:ann}
Let $P \subseteq X$ be a collection of $|P| = n$ points in a distance space $(X, \dist)$.
A solution to the $(r_1, r_2)$-near neighbor problem is a data structure that supports the following query operation:
Given a query point $q \in X$, if there exists a point $x \in P$ with $\dist(q, x) \leq r_1$, 
then, with probability at least $1/2$, return a point $x' \in P$ such that $\dist(q, x') < r_2$.  
\end{definition}
We aim for solutions with a failure probability that is upper bounded by $1/2$.
The standard trick of using $\eta$ independent repetitions of the data structure allows us to reduce the probability of failure to $1/2^\eta$.
For the sake of simplicity we restrict our attention to static solutions, meaning that we do not concern ourselves with the complexity of updates to the underlying set $P$, although it is simple to modify the static solutions presented in this paper to dynamic solutions where the update complexity essentially matches the query complexity~\cite{overmars1981, har-peled2012} 

\paragraph{LSH powering.}
The Indyk-Motwani framework and the Andoni-Indyk framework will make use of the following standard powering technique described in the introduction as ``overlaying partitions''.
Let $k \geq 1$ be an integer and let $\LSH$ denote a locality-sensitive family of hash functions as in Definition \ref{def:lsh}. 
We will use the notation $\LSH^k$ to denote the distribution over functions $h' \colon X \to R^k$ where
\begin{equation*}
	h'(x) = (h_1(x), \dots, h_k(x))
\end{equation*}
and $h_1, \dots, h_k$ are sampled independently at random from $\LSH$.
It is easy to see that $\LSH^k$ is $(r_1, r_2, p_1^k, p_2^k)$-sensitive.
To deal with some special cases we define $\LSH^0$ to be the family consisting of a single constant function.

\paragraph{Model of computation.}
We will work in the standard word-RAM model of computation~\cite{hagerup1998} with a word length of $\Theta(\log n)$ bits where $n$ denotes the size of the collection $P$ to be searched in the $(r_1, r_2)$-near neighbor problem.
During the preprocessing stage of our solutions we will assume access to a source of randomness that allows us to sample independently from a family $\LSH$ and to seed pairwise independent hash functions~\cite{carter1979, wegman1981}. 
The latter can easily be accomplished by augmenting the model with an instruction that generates a uniformly random word in constant time and using that to seed the tables of a Zobrist hash function~\cite{zobrist1970new}.
\section{Frameworks}\label{sec:frameworks}
\paragraph{Overview.}
We will describe frameworks that take as input a $(r_1, r_2, p_1, p_2)$-sensitive family $\LSH$ and a collection $P$ of $n$ points and constructs a data structure that solves the $(r_1, r_2)$-near neighbor problem.
The frameworks described in this paper all use the same high-level technique of constructing $L$ hash functions $g_{1},\dots,g_{L}$ that are used to partition space such that a pair of points $x, y$ with $\dist(x, y) \leq r_1$ will end up in the same part of one of the $L$ partitions with probability at least $1/2$. 
That is, for $x, y$ with $\dist(x, y) \leq r_1$ we have that $\Pr[\exists l \in [L] \colon g_{l}(x) = g_{l}(y)] \geq 1/2$ where $[L]$ is used to denote the set $\{1,2,\dots,L\}$.
At the same time we ensure that the expected number of collisions between pairs of points $x,y$ with $\dist(x, y) \geq r_2$ is at most one in each partition.

\paragraph{Preprocessing and queries.}
During the preprocessing phase, for each of the $L$ hash functions $g_{1}, \dots,g_{L}$ we compute the partition of the collection of points $P$ induced by $g_{l}$ and store it in a hash table in the form of key-value pairs $(z, \{ x \in P \mid g_{l}(x) = z \})$.
To reduce space usage we store only a single copy of the collection $P$ and store references to $P$ in our $L$ hash tables.
To guarantee lookups in constant time we can use the perfect hashing scheme by Fredman et al.~\cite{fredman1984storing} to construct our hash tables.
We will assume that hash values $z = g_{l}(x)$ fit into $O(1)$ words. 
If this is not the case we can use universal hashing~\cite{carter1977} to operate on fingerprints of the hash values.

We perform a query for a point $q$ as follows: for $l = 1, \dots, L$ we compute $g_{l}(q)$, 
retrieve the set of points $\{ x \in P \mid g_{l}(x) = g_{l}(q) \}$, and compute the distance between $q$ and each point in the set.
If we encounter a point $x'$ with $\dist(q, x') < r_2$ then we return $x'$ and terminate.
If after querying the $L$ sets no such point is encountered we return a special symbol $\varnothing$ and terminate.

We will proceed by describing and analyzing the solutions to the $(r_1, r_2)$-near neighbor problem for different approaches to sampling, storing, and computing the $L$ hash functions $g_{1}, \dots, g_{L}$, resulting in the different frameworks as mentioned in the introduction.
\subsection{Indyk-Motwani}\label{sec:im}
To solve the $(r_1, r_2)$-near neighbor problem using the Indyk-Motwani framework we sample $L$ hash functions $g_{1}, \dots,g_{L}$ independently at random from the family $\LSH^k$ where we set $k = \lceil \log(n) / \log(1/p_2) \rceil$ and $L = \lceil (\ln 2)/p_1^k \rceil$.
Correctness of the data structure follows from the observation that the probability that a pair of points $x, y$ with $\dist(x,y) \leq r_1$ does not collide under a randomly sampled $g_l \sim \LSH^k$ is at most $1 - p_1^k$. 
We can therefore upper bound the probability that a near pair of points does not collide under any of the hash functions by $(1-p_1^k)^L \leq \exp(-p_1^k L) \leq 1/2$ using a standard bound stated as Lemma~\ref{lem:exp_upper} in Appendix~\ref{app:inequalities}.

In the worst case, the query operation computes $L$ hash functions from $\LSH^k$ corresponding to $Lk$ hash functions from $\LSH$. 
For a query point $q$ the expected number of points $x' \in P$ with $\dist(q, x') \geq r_2$ that collide with $q$ under a randomly sampled $g_l \sim \LSH^k$ is at most $np_2^k \leq np_2^{\log(n) / \log (1/p_2)} = 1$.
It follows from linearity of expectation that the total expected number of distance computations during a query is at most $L$.
The result is summarized in Theorem \ref{thm:lsh_im_exact} from which the simplified Theorem \ref{thm:lsh_im_simple} follows.
\begin{theorem}[Indyk-Motwani {\cite{indyk1998, har-peled2012}}]\label{thm:lsh_im_exact}
Given a $(r_1, r_2, p_1, p_2)$-sensitive family $\LSH$ we can construct a data structure that solves the $(r_1, r_2)$-near neighbor problem such that for
$k = \lceil \log(n) / \log(1/p_2) \rceil$ and $L = \lceil (\ln 2)/p_1^k \rceil$ the data structure has the following properties:
\begin{itemize}
\item The query operation uses at most $Lk$ evaluations of hash functions from $\LSH$, 
	expected $L$ distance computations, and $O(Lk)$ other word-RAM operations.
\item The data structure uses $O(nL)$ words of space in addition to the space required to store the data and $Lk$ hash functions from $\LSH$.
\end{itemize}
\end{theorem}
Theorem \ref{thm:lsh_im_exact} gives a bound on the expected number of distance computations while the simplified version stated in Theorem \ref{thm:lsh_im_simple} uses Markov's inequality and independent repetitions to remove the expectation from the bound by treating an excessive number of distance computations as a failure.
\subsection{Andoni-Indyk}\label{sec:ai}
In 2006 Andoni and Indyk, as part of a paper on the substring near neighbor problem, introduced an improvement to the Indyk-Motwani framework that reduces the number of locality-sensitive hash functions~\cite{andoni2006efficient}.
Their improvement comes from the use of a technique that we will refer to as tensoring: setting the hash functions $g_1, \dots, g_L$ to be all $t$-tuples from a collection of $m$ functions sampled from $\LSH^{k/t}$ where $m \ll L$.
The analysis in~\cite{andoni2006efficient} shows that by setting $m = n^{\rho/t}$ and repeating the entire scheme $t!$ times, the total number of hash functions can be reduced to $O(\exp(\sqrt{\rho \log n \log \log n}))$ when setting $t = \sqrt{\frac{\rho \log n}{\log \log n}}$.
This analysis ignores integer constraints on $t$, $k$, and $m$, and implicitly place restrictions on $p_1$ and $p_2$ in relation to $n$ (e.g.\ $0 < p_2 < p_1 < 1$ are constant).
We will introduce a slightly different scheme that takes into account integer constraints and analyze it without restrictions on the properties of $\LSH$.

Assume that we are given a $(r_1, r_2, p_1, p_2)$-sensitive family $\LSH$.
Let $\eta, t, k_1, k_2, m_1, m_2$ be non-negative integer parameters.
Each of the $L$ hash functions $g_1, \dots, g_L$ will be formed by concatenating one hash function from each of $t$ collections of $m_1$ hash functions from $\LSH^{k_1}$ and concatenating a last hash function from a collection of $m_2$ hash functions from $\LSH^{k_2}$.
We take all $m_1^t m_2$ hash functions of the above form and repeat $\eta$ times for a total of $L = \eta m_1^t m_2$ hash functions constructed from a total of $H = \eta(m_1 k_1 t + m_2 k_2)$ hash functions from $\LSH$.
In Appendix \ref{app:ai} we set parameters, leaving $t$ variable, and provide an analysis of this scheme, showing that $L$ matches the Indyk-Motwani framework bound of $O(1/p_1^k)$ up to a constant where $k = \lceil \log(n) / \log(1/p_2) \rceil$ as in Theorem \ref{thm:lsh_im_exact}.

\paragraph{Setting $t$.}
It remains to show how to set $t$ to obtain a good bound on the number of hash functions $H$.
Note that in practice we can simply set $t = \argmin_t H$ by trying $t = 1,\dots,k$. 
If we ignore integer constraints and place certain restrictions of $\LSH$ as in the original tensoring scheme by Andoni and Indyk we want to set $t$ to minimize the expression $t^t n^{\rho/t}$. 
This minimum is obtained when setting $t$ such that $t^2 \log t = \rho \log n$.
We therefore cannot do much better than setting $t = \sqrt{\rho \log(n) / \log \log n}$ which gives the bound $H = O(\exp(\sqrt{\rho \log(n) \log \log n}))$ as shown in \cite{andoni2006efficient}. 
To allow for easy comparison with the Indyk-Motwani framework without placing restrictions on $\LSH$ we set $t = \lceil \sqrt{k} \rceil$, resulting in Theorem \ref{thm:tensoring}. 
\begin{theorem}\label{thm:tensoring}
Given a $(r_1, r_2, p_1, p_2)$-sensitive family $\LSH$ we can construct a data structure that solves the $(r_1, r_2)$-near neighbor problem such that for $k = \lceil \log(n) / \log(1/p_2) \rceil$, $H = k(\sqrt{k}/p_1)^{\sqrt{k}}$, and $L = \lceil 1/p_1^k \rceil$ the data structure has the following properties:
\begin{itemize}
	\item The query operation uses $O(H)$ evaluations of functions from $\LSH$, $O(L)$ distance computations, and $O(L + H)$ other word-RAM operations.
	\item The data structure uses $O(nL)$ words of space in addition to the space required to store the data and $O(H)$ hash functions from $\LSH$.
\end{itemize}
\end{theorem}
Thus, compared to the Indyk-Motwani framework we have gone from using $O(k(1/p_1)^k)$ locality-sensitive hash functions to $O(k(\sqrt{k}/p_1)^{\sqrt{k}})$ locality-sensitive hash functions. 
Figure~\ref{fig:comparison} shows the actual number of hash functions of the revised version of the Andoni-Indyk scheme as analyzed in Appendix \ref{app:ai} when $t$ is set to minimize $H$.
\subsection{Dahlgaard-Knudsen-Thorup}\label{sec:dkt}
In a recent paper Dahlgaard et al.~\cite{dahlgaard2017fast} introduce a different technique for reducing the number of locality-sensitive hash functions.
The idea is to construct each hash value $g_{l}(x)$ by sampling and concatenating $k$ hash values from a collection of $km$ pre-computed hash functions from~$\LSH$.
Dahlgaard et al.\ applied this technique to provide a fast solution the approximate near neighbor problem for sets under Jaccard similarity.
In this paper we use the same technique to derive a general framework solution that works with every family of locality-sensitive hash functions, reducing the number of locality-sensitive hash functions compard to the Indyk-Motwani and Andoni-Indyk frameworks.

Let $[n]$ denote the set of integers $\{1,2,\dots,n\}$.
For $i \in [k]$ and $j \in [m]$ let $h_{i,j} \sim \LSH$ denote a hash function in our collection.
To sample from the collection we use $k$ pairwise independent hash functions~\cite{wegman1981} of the form $f_i \colon [L] \to [m]$ and set
\begin{equation*}
g_{l}(x) = (h_{1,f_{1}(l)}(x), \dots, h_{k,f_{k}(l)}(x)).
\end{equation*}
To show correctness of this scheme we will use make use of an elementary one-sided version of Chebyshev's inequality stating that for a random variable $Z$ with mean $\mu > 0$ and variance $\sigma^2 < \infty$ we have that $\Pr[Z \leq 0] \leq \sigma^2 / (\mu^2 + \sigma^2)$. 
For completeness we have included the proof of this inequality in Lemma~\ref{lem:cantelli} in Appendix~\ref{app:inequalities}.
We will apply this inequality to lower bound the probability that there are no collisions between close pairs of points.
For two points $x$ and $y$ let $Z_l = \1\{g_{l}(x) = g_{l}(y)\}$ so that $Z = \sum_{l = 1}^L Z_l$ denotes the sum of collisions under the $L$ hash functions. 
To apply the inequality we need to derive an expression for the expectation and the variance of the random variable $Z$.
Let $p = \Pr_{h \sim \LSH}[h(x) = h(y)]$ then by linearity of expectation we have that $\mu = \E[Z] = L p^k$.
To bound $\sigma^2 = \E[Z^2] - \mu^2$ we proceed by bounding $\E[Z^2]$ where we note that $Z_l = \Pi_{i = 1}^{k} Y_{l, i}$ for $Y_{l, i} = 1\{h_{i,f_{i}(l)}(x) = h_{i,f_{i}(l)}(x)\}$ and make use of the independence between $Y_{l,i}$ and $Y_{l', i'}$ for $i \neq i'$.
\begin{align*}
	\E[Z^2] &= \sum_{\substack{l, l' \in [L] \\ l \neq l'}} \E[Z_l Z_{l'}] + \sum_{l = 1}^L \E[Z_l] \\ 
			&= (L^2 - L) \E[Z_l Z_{l'}] + \mu \\
			&\leq L^2 \E\left[ \Pi_{i = 1}^{k} Y_{l, i} Y_{l', i} \right] + \mu \\ 
			&= L^2 \left( \E[Y_{l, i} Y_{l', i}] \right)^k + \mu. 
\end{align*}
We have that $\E[Y_{l, i} Y_{l', i}] = \Pr[f_i(l) = f_i(l')] p + \Pr[f_i(l) \neq f_i(l')] p^2 = (1/m)p + (1-1/m)p^2$ which follows from the pairwise independence of $f_i$. 
Let $\varepsilon > 0$ and set $m = \lceil \frac{1-p_1}{p_1} \frac{k}{\ln(1+\varepsilon)} \rceil$ then for $p \geq p_1$ we have that $\left( \E[Y_{l, i} Y_{l', i}] \right)^k \leq (1 + \varepsilon)p^{2k}$.
This allows us to bound the variance of $Z$ by $\sigma^2 \leq \varepsilon \mu^2 + \mu$ resulting in the following lower bound on the probability of collision between similar points.
\begin{lemma} \label{lem:dkt_success}
For $\varepsilon > 0$ let $m \geq \lceil \frac{1-p_1}{p_1} \frac{k}{\ln(1+\varepsilon)} \rceil$, then for every pair of points $x, y$ with $\dist(x, y) \leq r_1$ we have that 
\begin{equation*}
	\Pr[\exists l \in [L] \colon g_{l}(x) = g_{l}(y)] \geq \frac{1 + \varepsilon \mu}{1 + (1 + \varepsilon)\mu}.
\end{equation*}
\end{lemma}
By setting $\varepsilon = 1/4$ and $L = \lceil (2 \ln(2))/p_1^k \rceil$ we obtain an upper bound on the failure probability of $1/2$.
Setting the size of each of the $k$ collections of pre-computed hash values to $m = \lceil 5k/p_1 \rceil$ is sufficient to yield the following solution to the $(r_1, r_2)$-near neighbor problem where provide exact bounds on the number of lookups $L$ and hash functions $H$:
\begin{theorem}[Dahlgaard-Knudsen-Thorup {\cite{dahlgaard2017fast}}]\label{thm:lsh_dkt_exact}
Given a $(r_1, r_2, p_1, p_2)$-sensitive family $\LSH$ we can construct a data structure that solves the $(r_1, r_2)$-near neighbor problem such that for $k = \lceil \log(n) / \log(1/p_2) \rceil$, $H = k \lceil 5k / p_1 \rceil$, and $L = \lceil (2 \ln(2))/p_1^k \rceil$ the data structure has the following properties:
\begin{itemize}
\item The query operation uses at most $H$ evaluations of hash functions from $\LSH$, 
	expected $L$ distance computations, and $O(Lk)$ other word-RAM operations.
\item The data structure uses $O(nL)$ words of space in addition to the space required to store the data and $H$ hash functions from $\LSH$.
\end{itemize}
\end{theorem}
Compared to the Indyk-Motwani framework we have reduced the number of locality-sensitive hash functions $H$ from $O(k (1/p_1)^k)$ to $O(k^2 / p_1)$ at the cost of using twice as many lookups. 
To reduce the number of lookups further we can decrease $\varepsilon$ and perform several independent repetitions.
This comes at the cost of an increase in the number of hash functions $H$.
\section{Reducing the word-RAM complexity} \label{section:word-RAM}
One drawback of the DKT framework is that each hash value $g_{l}(x)$ still takes $O(k)$ word-RAM operations to compute, even after the underlying locality-sensitive hash functions are known. 
This results in a bound on the total number of additional word-RAM operations of $O(Lk)$.
We show how to combine the DKT universal hashing technique with the AI tensoring technique to ensure that the running time is dominated by $O(L)$ distance computations and $O(H)$ hash function evaluations. 
The idea is to use the DKT scheme to construct two collections of respectively $L_1$ and $L_2$ hash functions, and then to use the AI tensoring approach to form $g_1, \dots, g_L$ as the $L = L_1 \times L_2$ combinations of functions from the two collections.
The number of lookups can be reduced by applying tensoring several times in independent repetitions, but for the sake of simplicity we use a single repetition.
For the usual setting of $k = \lceil \log(n) / \log(1/p_2) \rceil$ let $k_1 =  \lceil k/2 \rceil$ and $k_2 = \lfloor k/2 \rfloor$.
Set $L_1 = \lceil 6 (1/p_1)^{k_1} \rceil$ and $L_2 = \lceil 6 (1/p_1)^{k_2} \rceil$.
According to Lemma \ref{lem:dkt_success} if we set $\varepsilon = 1/6$ the success probability of each collection is at least $3/4$ and by a union bound the probability that either collection fails to contain a colliding hash function is at most $1/2$.
This concludes the proof of our main Theorem \ref{thm:lsh_dkt_ram}.
\subsection{Sketching}\label{sec:sketching}
The theorems of the previous section made no assumptions on the word-RAM complexity of distance computations and instead stated the number of distance computations as part of the query complexity. 
We can use a $(r_1, r_2, p_1, p_2)$-sensitive family $\LSH$ to create sketches that allows us to efficiently approximate the distance between pairs of points, provided that the gap between $p_1$ and $p_2$ is sufficiently large. 
In this section we will re-state the results of Theorem~\ref{thm:lsh_dkt_ram} when applying the family $\LSH$ to create sketches using the 1-bit sketching scheme of Li and König~\cite{li2011theory}.
Let $b$ be a positive integer denoting the length of the sketches in bits.
The advantage of this scheme is that we can use word level parallelism to evaluate a sketch of $b$ bits in time $O(b/\log n)$ in our word-RAM model with word length $\Theta(\log n)$. 

For $i = 1, \dots, b$ let $h_i \colon X \to R$ denote a randomly sampled locality-sensitive hash function from $\LSH$ and let $f_i \colon R \to \{0,1\}$ denote a randomly sampled universal hash function.
We let $s(x) \in \{0,1\}^b$ denote the sketch of a point $x \in X$ where we set the $i$th bit of the sketch $s(x)_i = f_i(h(x))$.
For two points $x, y \in X$ the probability that they agree on the $i$th bit is $1$ if the points collide under $h_i$ and $1/2$ otherwise. 
\begin{align*}
	\Pr[s(x)_i = s(y)_i] &= \Pr[h_i(x) = h_i(y)] + (1 - \Pr[h_i(x) = h_i(y)])/2 \\
				 &= (1 + \Pr[h_i(x) = h_i(y)])/2.
\end{align*}
We will apply these sketches during our query procedure instead of direct distance computations when searching through the points in the $L$ buckets, comparing them to our query point $q$.
Let $\lambda \in (0,1)$ be a parameter that will determine whether we report a point or not.
For sketches of length $b$ we will return a point $x$ if $\norm{s(q) - s(x)}_1 > \lambda b$.
An application of Hoeffiding's inequality gives us the following properties of the sketch:
\begin{lemma}\label{lem:sketching}
	Let $\LSH$ be a $(r_1, r_2, p_1, p_2)$-sensitive family and let $\lambda = (1+p_2)/2 + (p_1-p_2)/4$, 
	then for sketches of length $b \geq 1$ and for every pair points $x, y \in X$:  
	\begin{itemize}
		\item If $\dist(x ,y) \leq r_1$ then $\Pr[\norm{s(x) - s(y)}_1 \leq \lambda b] \leq e^{-b(p_1 - p_2)^2 / 8 }$.
		\item If $\dist(x ,y) \geq r_2$ then $\Pr[\norm{s(x) - s(y)}_1 > \lambda b] \leq e^{-b(p_1 - p_2)^2 / 8 }$.
	\end{itemize}
\end{lemma}
If we replace the exact distance computations with sketches we want to avoid two events: 
Failing to report a point with $\dist(q,x) \leq r_1$ and reporting a point $x$ with $\dist(q, x) \geq r_2$.
By setting $b = O(\ln(n) / (p_1 - p_2)^2)$ and applying a union bound over the $n$ events that the sketch fails for a point in our collection $P$ we obtain Theorem~\ref{thm:lsh_dkt_ram_sketch}. 
\section{The number of hash functions in corner cases}\label{sec:corner}
When the collision probabilities of the $(r_1, r_2, p_1, p_2)$-sensitive family $\LSH$ are close to one we get the behavior displayed in Figure \ref{fig:p1_090} where we have set $p_1 = 0.9$. 
Here it may be possible to reduce the number of hash functions by applying the DKT framework to the family $\LSH^\tau$ for some positive integer $\tau$. That is, instead of applying the DKT technique directly to $\LSH$ we first apply the powering trick to produce the family $\LSH^\tau$.
The number of locality-sensitive hash functions from $\LSH$ used by the DKT framework is given by $H = O( (\log(n) / \log(1/p_2))^2 / p_1)$.
If we instead use the family $\LSH^\tau$ the expression becomes $H = O( \tau(\log(n) / \log(1/p_2^\tau))^2 / p_1^\tau) = O((\log(n) / \log(1/p_2))^2 / \tau p_1^\tau)$. 
Ignoring integer constraints, the value of $\tau$ that maximizes $\tau p_1^\tau$, thereby minimizing $H$, is given by $\tau = 1 / \ln(1/p_1)$.
Discretizing, the resulting number of hash functions when setting $\tau = \lceil 1 / \ln(1/p_1)\rceil$ is given by $H = O(\rho (\log n)^2 / (p_1 \log(1/p_2)))$.
For constant $\rho$ and large $p_2$ this reduces the number of hash functions by a factor $1/\log(1/p_2)$.
\begin{figure}
	\centering
	\includegraphics[width=0.85\textwidth]{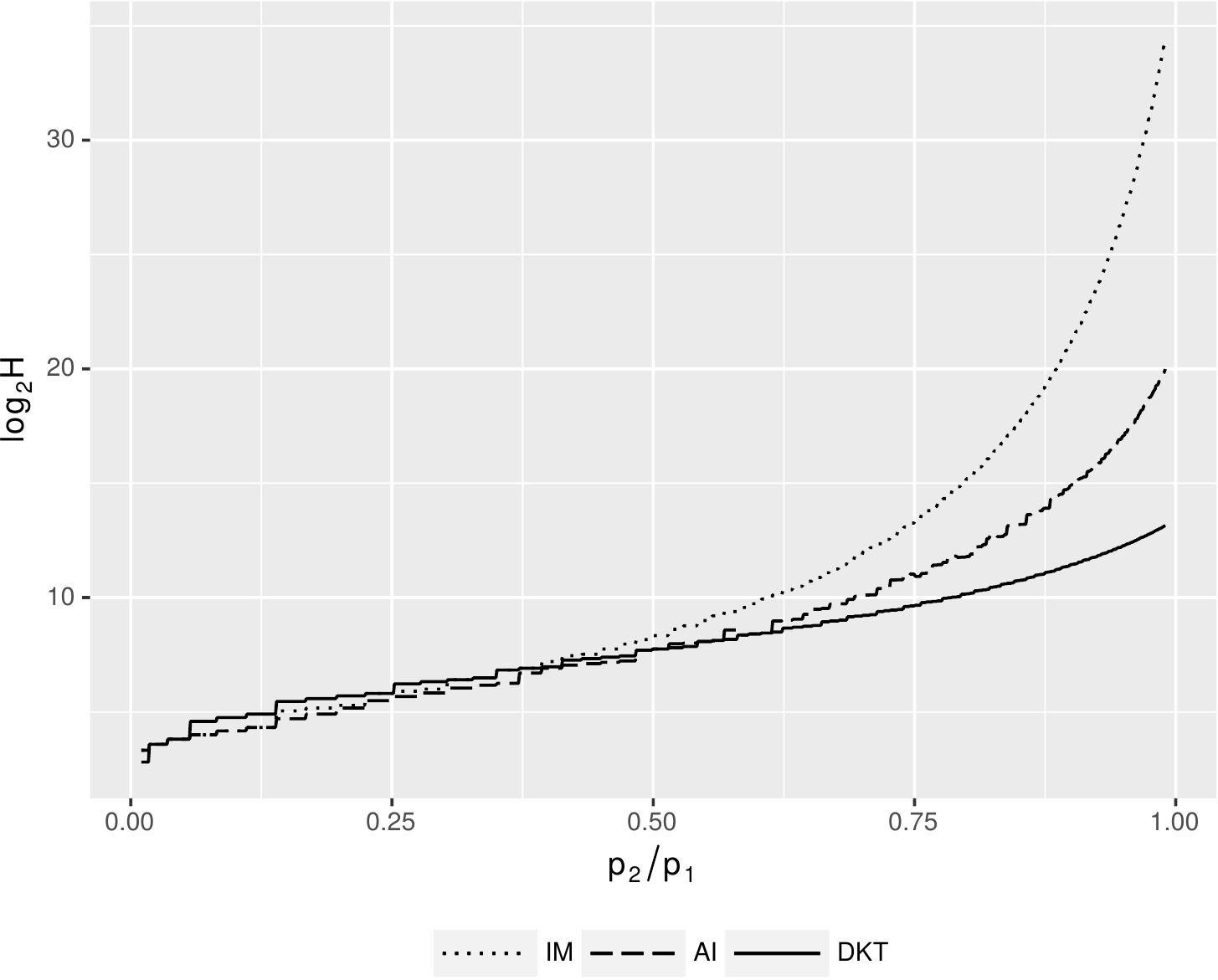}
	\caption{The number of locality-sensitive hash functions from a $(r_1, r_2, 0.9, p_2)$-sensitive family used by different frameworks to solve the $(r_1, r_2)$-near neighbor problem on a collection of $2^{30}$ points.}
	\label{fig:p1_090}
\end{figure}
The behavior for small values of $p_1$ is displayed in Figure \ref{fig:p1_010} where we have set $p_1 = 0.1$.
\begin{figure}
	\centering
	\includegraphics[width=0.85\textwidth]{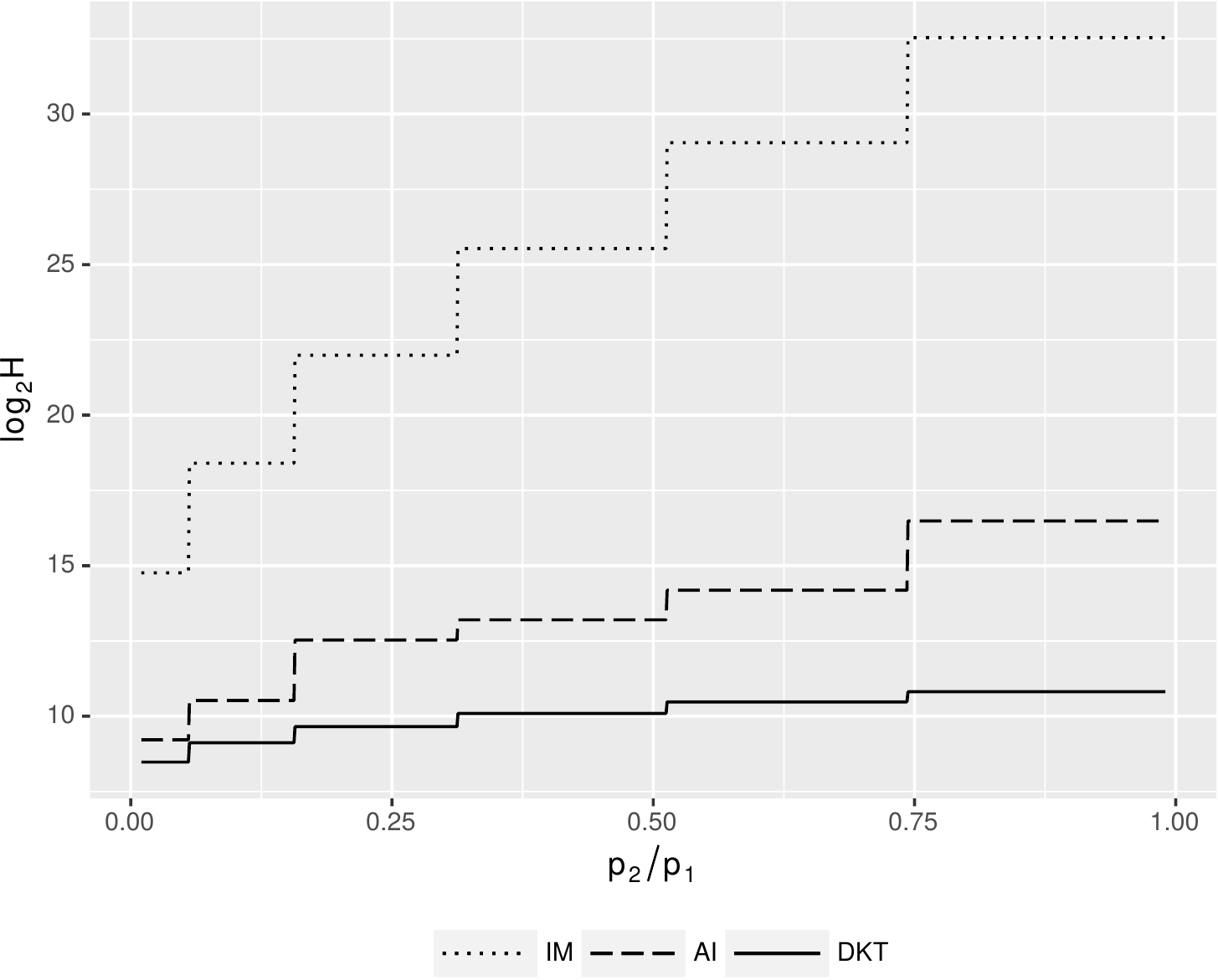}
\caption{The number of locality-sensitive hash functions from a $(r_1, r_2, 0.1, p_2)$-sensitive family used by different frameworks to solve the $(r_1, r_2)$-near neighbor problem on a collection of $2^{30}$ points.}
\label{fig:p1_010}
\end{figure}
\section{Conclusion and open problems}
We have shown that there exists a simple and general framework for solving the $(r_1, r_2)$-near neighbor problem using only few locality-sensitive hash functions and with a reduced word-RAM complexity matching the number of lookups.
The analysis in this paper indicates that the performance of the Dahlgaard-Knudsen-Thorup framework is highly competitive compared to the Indyk-Motwani framework in practice, especially when locality-sensitive hash functions are expensive to evaluate, as is often the case.

An obvious open problem is to provide a framework that uses fewer than $O(k^2 / p_1)$ locality-sensitive hash function.
Another direction would be to find a lower bound on the number of independent locality-sensitive hash functions required to solve the ANN problem using LSH in a suitably restricted model.
\section*{Acknowledgement}
I want to thank Rasmus Pagh commenting on an earlier version of this manuscript and for making me aware of the application of the tensoring technique in~\cite{sundaram2013streaming} that led me to the Andoni-Indyk framework~\cite{andoni2006efficient}.
\section{Appendix: Inequalities}\label{app:inequalities}
We make use of the following standard inequalities for the exponential function. 
See \cite[Chapter 3.6.2]{mitrinovic1970} for more details.
\begin{lemma}\label{lem:exp_upper}
	Let $n, t \in \mathbb{R}$ such that $n \geq 1$ and $|t| \leq n$ then $e^{-t}(1-t^2 /n) \leq (1 - t/n)^n \leq e^{-t}$.
\end{lemma}
\begin{lemma}\label{lem:exp_taylor}
	For $t \geq 0$ we have that $e^{-t} \leq 1 - t + t^2 / 2$.
\end{lemma}

We make use of a one-sided version of Chebyshev's inequality to show correctness of the Dahlgaard-Knudsen-Thorup LSH framework. 
\begin{lemma}[Cantelli's inequality]\label{lem:cantelli}
Let $Z$ be a random variable with $\E[Z] = \mu > 0$ and $\Var[Z] = \sigma^2 < \infty$ then $\Pr[Z \leq 0] \leq \sigma^2/(\mu^2 + \sigma^2)$. 
\end{lemma}
\begin{proof}
	For every $s \in \mathbb{R}$ we have that 
\begin{equation*}
\Pr[Z \leq 0] = \Pr[-(Z - \mu) + s \geq \mu + s] \leq  \Pr[(-(Z - \mu) + s)^2 \geq (\mu + s)^2]. 
\end{equation*}
Next we apply Markov's inequality 
\begin{align*}
\Pr[(-(Z - \mu) + s)^2 \geq (\mu + s)^2] &\leq \E[(-(Z - \mu) + s)^2]/ (\mu + s)^2 \\
										 &= (\sigma^2 + s^2)/ (\mu + s)^2 
\end{align*}
Set $s = \sigma^2 / \mu$ and use that $\sigma^2 = s\mu$ to simplify 
\begin{equation*}
(\sigma^2 + s^2)/ (\mu + s)^2 = (s\mu + s^2)/ (\mu + s)^2 = \sigma^2 / (\mu^2 + \sigma^2).
\end{equation*}
\end{proof}

To analyze the 1-bit sketching scheme by Li and K{\"o}nig we make use of Hoeffding's inequality:
\begin{lemma}[{Hoeffding \cite[Theorem 1]{hoeffding1963}}] \label{fast:lem:hoeffding}
	Let $X_1, X_2, \dots, X_n$ be independent random variables satisfying $0 \leq X_i \leq 1$ for $i \in [n]$.
	Define $\bar{X} = (X_1 + X_2 + \dots + X_n)/n$ and $\mu = \E[\bar{X}]$, then:
	\begin{itemize}
		\item[-] For $0 < \varepsilon < 1 - \mu$ we have that $\Pr[\bar{X} - \mu \geq \varepsilon] \leq e^{-2n\varepsilon^{2}}$.
		\item[-] For $0 < \varepsilon < \mu$ we have that $\Pr[\bar{X} - \mu \leq - \varepsilon] \leq e^{-2n\varepsilon^{2}}$.
	\end{itemize}
\end{lemma}
\section{Appendix: Analysis of the Andoni-Indyk framework}\label{app:ai}
Let $\varphi$ denote the probability that a pair of points $x, y$ with $\dist(x,y) \leq r_1$ collide in a single repetition of the scheme.
A collision occurs if and only if there there exists at least one hash function in each of the underlying $t + 1$ collections where the points collide.
It follows that
\begin{equation*}
	\varphi = (1 - (1-p_1^{k_1})^{m_1})^t (1 - (1-p_1^{k_2})^m_2).
\end{equation*}
To guarantee a collision with probability at least $1/2$ it suffices to set $\eta = \lceil \ln(2) / \varphi \rceil$.

We will proceed by analyzing this scheme where we let $t \geq 1$ be variable and set parameters as followers:
\begin{align*}
	k &= \lceil \log(n) / \log(1/p_2) \rceil \\
	k_1 &= \lfloor k/t \rfloor \\
	k_2 &= k - tk_1 \\
	m_1 &= \lceil 1/ t p_1^{k_1} \rceil \\
	m_2 &= \lceil 1/ p_1^{k_2} \rceil \\
    \eta &= \lceil \ln(2) / \varphi \rceil.
\end{align*}
To upper bound $L$ we begin by lower bounding $\varphi$.
The second part of $\varphi$ can be lower bounded using Lemma \ref{lem:exp_upper} to yield $(1 - (1-p_1^{k_2})^{m_2}) \geq 1 - 1/e$.
To lower bound $(1 - (1-p_1^{k_1})^{m_1})^t$ we first note that in the case where $p_1^{k_1} > 1/t$ we have $m_1 = 1$ and the expression can be lower bounded by $p_1^{k_1 t} = (p_1^{k_1} m_1)^t \geq  (p_1^{k_1} m_1)^t / 2e$.
The same lower bound holds in the case there $t = 1$.
In the case where $p_1^{k_1} \leq 1/t$ and $t \geq 2$ we make use of Lemma \ref{lem:exp_upper} and \ref{lem:exp_taylor} to derive the lower bound. 
\begin{align*}
	1 - (1-p_1^{k_1})^{m_1} &\geq 1 - e^{-p_1^{k_1 m_1}} \\
 &\geq 1 - (1 - p_1^{k_1} m_1 + (p_1^{k_1} m_1)^2/2) \\
 &\geq p_1^{k_1} m_1 (1 - p_1^{k_1}(1/t p_1^{k_1} + 1)/2) \\
 &\geq p_1^{k_1} m_1 (1 - 1/t). 
\end{align*}
Using the bound $(p_1^{k_1} m_1 (1 - 1/t))^t \geq (p_1^{k_1} m_1)^t / 2e$ we have that 
\begin{equation*}
	\varphi \geq (p_1^{k_1} m_1)^t / 4e \geq (1/t)^t / 4e.
\end{equation*}
We can then bound the number of lookups and the expected number of distance computations
\begin{equation*}
	L = \eta m_1^t m_2 \leq (4e / (p_1^{k_1} m_1)^t + 1) m_1^t (1/p_1^{k_2} + 1) \leq 16e (1/p_1^k). 
\end{equation*}
Note that this matches the upper bound of the Indyk-Motwani LSH framework up to a constant factor.

To bound the number of hash functions from $\LSH$ we use that $k_1 \leq k/t \leq k$ and $k_2 < t$.
\begin{equation*}
	H = \eta (m_1 k_1 t + m_2 k_2) \leq 8e t^t \left(\frac{k}{t p_1^{k/t}} + \frac{t-1}{p_1^{t-1}}\right).
\end{equation*}
\chapter{Space-time tradeoffs for similarity search} \label{ch:filters}
\sectionquote{All that is gold does not glitter}
\noindent We present a framework for similarity search based on Locality-Sensitive Filtering~(LSF),
generalizing the Indyk-Motwani (STOC 1998) Locality-Sensitive Hashing~(LSH) framework to support space-time tradeoffs.    
Given a family of filters, defined as a distribution over pairs of subsets of space that satisfies certain locality-sensitivity properties, 
we can construct a dynamic data structure that solves the approximate near neighbor problem on a collection of $n$ points in $d$-dimensional space 
with query time $dn^{\rho_q + o(1)}$, update time $dn^{\rho_u + o(1)}$, 
and space usage $dn + n^{1 + \rho_u + o(1)}$.
The space-time tradeoff is tied to the tradeoff between query time and update time (insertions/deletions), 
controlled by the exponents $\rho_q, \rho_u$ that are determined by the filter family. \\   
Locality-sensitive filtering was introduced by Becker et al. (SODA 2016) together with a framework yielding a single, 
balanced, tradeoff between query time and space, further relying on the assumption of an efficient oracle for the filter evaluation algorithm.
We extend the LSF framework to support space-time tradeoffs 
and through a combination of existing techniques we remove the oracle assumption. \\
Laarhoven (arXiv 2015), building on Becker et al., introduced a family of filters with space-time tradeoffs 
for the high-dimensional unit sphere under inner product similarity and analyzed it for the important special case of random data.
We show that a small modification to the family of filters gives a simpler analysis that we use, 
together with our framework, to provide guarantees for worst-case data. 
Through an application of Bochner's~Theorem from harmonic analysis by Rahimi \& Recht (NIPS 2007), 
we are able to extend our solution on the unit sphere to $\real^d$ under the class of 
similarity measures corresponding to real-valued characteristic functions.
For the characteristic functions of $s$-stable distributions we obtain a solution 
to the $(r, cr)$-near neighbor problem in $\ell_s^d$-spaces with query and update exponents 
$\rho_q = \frac{c^s (1+\lambda)^2}{(c^s + \lambda)^2}$ and $\rho_u = \frac{c^s (1-\lambda)^2}{(c^s + \lambda)^2}$
where $\lambda \in [-1,1]$ is a tradeoff parameter. 
This result improves upon the space-time tradeoff of Kapralov (PODS 2015) and is shown to be optimal in the case of a balanced tradeoff, 
matching the LSH lower bound by O'Donnell et al.~(ITCS 2011) and a similar LSF lower bound proposed in this paper.
Finally, we show a lower bound for the space-time tradeoff on the unit sphere that matches Laarhoven's and our own upper bound in the case of random data.
\section{Introduction}
Let $(X, \dist)$ denote a space over a set~$X$ equipped with a symmetric measure of dissimilarity~$\dist$ (a distance function in the case of metric spaces).
We consider the \emph{$(r, cr)$-near neighbor problem} first introduced by Minsky and Papert~\cite[p. 222]{minsky1969} in the~1960's.  
A solution to the $(r, cr)$-near neighbor problem for a set $P$ of $n$ points in $(X, \dist)$ 
takes the form of a data structure that supports the following operation: 
given a query point $x \in X$, if there exists a data point $y \in P$ such that $\dist(x, y) \leq r$ 
then report a data point $y' \in P$ such that $\dist(x, y') \leq cr$.
In some spaces it turns out to be convenient to work with a measure of similarity rather than dissimilarity.
We use $\simil$ to denote a symmetric measure of similarity and define the \emph{$(\alpha, \beta)$-similarity problem}
to be the $(-\alpha, -\beta)$-near neighbor problem in~$(X, -\simil)$.

A solution to the $(r, cr)$-near neighbor problem can be viewed as a fundamental building block that yields solutions 
to many other similarity search problems such as the $c$-approximate \emph{nearest} neighbor problem \cite{indyk2004, har-peled2012}.
In particular, the $(r, cr)$-near neighbor problem is well-studied in $\ell_s^d$-spaces  
where the data points lie in~$\real^d$ and distances are measured by $\dist(x, y) = \norm{x - y}_s = (\sum_{i=1}^{d}|x_i - y_i|^s)^{1/s}$.
Notable spaces include the Euclidean space~$(\real^d, \norm{\cdot}_2)$, Hamming space~$(\{0,1\}^{d}, \norm{\cdot}_1)$, 
and the $d$-dimensional unit sphere~$\sphere{d} = \{ x \in \real^d \mid \norm{x}_2 = 1 \}$
under inner product (cosine) similarity \mbox{$\simil(x, y) = \ip{x}{y} = \sum_{i=1}^{d}x_iy_i$}.

\paragraph{Curse of dimensionality.}
All known solutions to the $(r, cr)$-near neighbor problem for $c = 1$ (the exact near neighbor problem) 
either suffer from a space usage that is exponential in $d$ or a query time that is linear in $n$ \cite{har-peled2012}.
This phenomenon is known as the ``curse of dimensionality'' and has been observed both in theory and practice.
For example, Alman and Williams \cite{alman2015} recently showed that the existence of an algorithm for determining 
whether a set of $n$ points in $d$-dimensional Hamming space contains a pair of points that are exact near neighbors 
with a running time strongly subquadratic in $n$ would refute the Strong Exponential Time Hypothesis (SETH) \cite{williams2004}.
This result holds even when $d$ is rather small, $d = O(\log n)$.
From a practical point of view, Weber et al. \cite{weber1998} showed that the performance of many of the tree-based approaches to similarity search 
from the field of computational geometry \cite{berg2008} degrades rapidly to a linear scan as the dimensionality increases.

\paragraph{Approximation to the rescue.}
If we allow an approximation factor of $c > 1$ then there exist solutions to the $(r, cr)$-near neighbor problem 
with query time that is strongly sublinear in~$n$ and space polynomial in~$n$
where both the space and time complexity of the solution depends only polynomially on $d$.
Techniques for overcoming the curse of dimensionality through approximation were discovered 
independently by Kushilevitz et al. \cite{kushilevitz2000} and Indyk and Motwani \cite{indyk1998}.
The latter, classical work by Indyk and Motwani \cite{indyk1998, har-peled2012} introduced a general framework 
for solving the $(r, cr)$-near neighbor problem known as Locality-Sensitive Hashing (LSH). 
The introduction of the LSH framework has inspired an extensive literature (see e.g. \cite{andoni2008, wang2014} for surveys) 
that represents the state of the art in terms of solutions to the $(r, cr)$-near neighbor problem in high-dimensional spaces \cite{indyk1998, charikar2002, datar2004, panigrahy2006, andoni2006, andoni2008, andoni2009, andoni2015practical, andoni2015optimal, kapralov2015, andoni2016tight, becker2016, laarhoven2015}.

\paragraph{Hashing and filtering frameworks.}
The LSH framework and the more recent LSF framework 
introduced by Becker et al. \cite{becker2016} produce data structures that solve the $(r, cr)$-near neighbor problem 
with query and update time $dn^{\rho + o(1)}$ and space usage $dn + n^{1 + \rho + o(1)}$. 
The LSH (LSF) framework takes as input a distribution over partitions (subsets) of space with the locality-sensitivity property 
that close points are more likely to be contained in the same part (subset) of a randomly sampled element from the distribution.
The frameworks proceed by constructing a data structure that associates each point in space with a number of memory locations
or ``buckets'' where data points are stored. 
During a query operation the buckets associated with the query point are searched by computing the distance to every data point in the bucket, 
returning the first suitable candidate.
The set of memory locations associated with a particular point is independent of whether an update operation or a query operation is being performed.
This symmetry between the query and update algorithm results in solutions to the near neighbor problem with a balanced space-time tradeoff.
The exponent~$\rho$ is determined by the locality-sensitivity properties of the family of partitions/hash functions (LSH) or subsets/filters (LSF) 
and is typically upper bounded by an expression that depends only on the aproximation factor $c$.
For example, Indyk and Motwani \cite{indyk1998} gave a simple locality-sensitive family of hash functions for Hamming space 
with an exponent of~$\rho \leq 1/c$. 
This exponent was later shown to be optimal by O'Donnell et al. \cite{odonnell2014optimal} who gave a lower bound of $\rho \geq 1/c - o_{d}(1)$ 
in the setting where $r$ and $cr$ are small compared to $d$.
The advantage of having a general framework for similarity search lies in the reduction of the $(r, cr)$-near neighbor problem 
to the, often simpler and easier to analyze, problem of finding a locality-sensitive family of hash functions or filters for the space of interest.    

\paragraph{Space-time tradeoffs.}
Space-time tradeoffs for solutions to the $(r, cr)$-near neighbor problem is an active line of research that can be 
motivated by practical applications where it is desirable to choose the tradeoff between query time and update time (space usage) 
that is best suited for the application and memory hierarchy at hand~\cite{panigrahy2006, lv2007, andoni2009, kapralov2015, laarhoven2015}. 
Existing solutions typically have query time $dn^{\rho_q + o(1)}$, update time (insertions/deletions) $dn^{\rho_u + o(1)}$, and use space $dn + n^{1 + \rho_u + o(1)}$ 
where the query and update exponents $\rho_q, \rho_u$ that control the space-time tradeoff depend on the approximation factor $c$ and on a tradeoff parameter $\lambda \in [-1,1]$.

This paper combines a number of existing techniques \cite{becker2016, laarhoven2015, dubiner2010bucketing} to provide a general framework for similarity search with space-time tradeoffs. 
The framework is used to show improved upper bounds on the space-time tradeoff in the well-studied setting of $\ell_s$-spaces and the unit sphere under inner product similarity.
Finally, we show a new lower bound on the space-time tradeoff for the unit sphere that matches an upper bound for random data on the unit sphere by Laarhoven \cite{laarhoven2015}.
We proceed by stating our contribution and briefly surveying the relevant literature in terms of frameworks, 
upper bounds, and lower bounds as well as some recent developments.
See table Table \ref{lsf:tab:results} for an overview.
\begin{table*}[t]
	\centering
\setlength{\tabcolsep}{0.2em}
	\renewcommand\arraystretch{2.3}
	\caption{Overview of data-independent locality-sensitive hashing (LSH) and filtering (LSF) results}
	\tiny
    \begin{tabular}{|c|c|c|c|c|} \hline
		\textbf{Reference} & \textbf{Setting} &  $\rho_q$ & $\rho_u$ \\ \hline\hline
		LSH \cite{indyk1998, har-peled2012}, LSF \cite{becker2016} &
		\multirow{2}{*}{$(X, \dist)$, $(X, \simil)$}  
		& \multicolumn{2}{c|}{$\dfrac{\log (1/p)}{\log (1/q)}$} \\[3pt] \cline{1-1} \cline{3-4} 
		\textbf{Theorem \ref{thm:lsfvanilla}} & 
		& $\dfrac{\log (p_q /p_1)}{\log (p_q / p_2)}$ & $\dfrac{\log (p_u /p_1)}{\log (p_q / p_2)}$ \\[4pt] 
	\hline\hline 
	
	Cross-poly.\@ LSH \cite{andoni2015practical} & $(\alpha, \beta)$-sim.\@, $(\sphere{d},\ip{\cdot}{\cdot})$ & 
	\multicolumn{2}{c|}{
		$\left. \dfrac{1-\alpha}{1+\alpha}  \middle/  \dfrac{1-\beta}{1+\beta} \right.$
	} \\[3pt]  \hline
	
	Spherical cap LSF \cite{laarhoven2015} & $(\alpha, o_{d}(1))$-sim.\@, $(\sphere{d},\ip{\cdot}{\cdot})$ & 
	$\dfrac{(1-\alpha^{1+\lambda})^2}{1-\alpha^2}$ &
	$\dfrac{(\alpha^{\lambda} -\alpha)^2}{1-\alpha^2}$ 
	\\[2pt]  \hline
	
	\textbf{Theorem \ref{thm:spherevanilla}} & $(\alpha, \beta)$-sim.\@, $(\sphere{d},\ip{\cdot}{\cdot})$ & 
	$\left. \dfrac{(1-\alpha^{1+\lambda})^2}{1-\alpha^2}  \middle/  \dfrac{(1- \alpha^{\lambda} \beta)^2}{1-\beta^2} \right.$ &
	$\left. \dfrac{(\alpha^{\lambda} - \alpha)^2}{1-\alpha^2}  \middle/  \dfrac{(1-\alpha^{\lambda} \beta)^2}{1-\beta^2} \right.$ 
	\\[3pt]  \hline \hline

	Ball-carving LSH \cite{andoni2006} & \multirow{3}{*}{ $(r, cr)$-nn.\@ in $\ell_2^d$}  
	& \multicolumn{2}{c|}{$1/c^2$}   \\  \cline{1-1} \cline{3-4} 
	
	Ball-search LSH* \cite{kapralov2015} && 
	$\dfrac{c^2 (1 + \lambda)^2}{(c^2 + \lambda)^2 - c^2 (1+\lambda^2)/2 - \lambda^2}$ & 
	$\dfrac{c^2 (1 - \lambda)^2}{(c^2 + \lambda)^2 - c^2 (1+\lambda^2)/2 - \lambda^2}$ \\[2pt] \cline{1-1} \cline{3-4} 
	
	\textbf{Theorem \ref{thm:lsvanilla}}  && 
	$\dfrac{c^2 (1 + \lambda)^2}{(c^2 + \lambda)^2}$ & 
	$\dfrac{c^2 (1 - \lambda)^2}{(c^2 + \lambda)^2}$ \\[3pt] 
	\hline \hline

	Lower bound \cite{odonnell2014optimal} & LSH in $\ell_2^d$ &
	\multicolumn{2}{c|}{$\geq 1/c^2$} \\  \hline
	\textbf{Theorem \ref{thm:lsflshvanilla}} & LSF in $\ell_2^d$ &
	\multicolumn{2}{c|}{$\geq 1/c^2$} \\  \hline
	Lower bound \cite{motwani2007, andoni2016tight} & LSH in $(\sphere{d},\ip{\cdot}{\cdot})$ &
	\multicolumn{2}{c|}{
		$\geq \dfrac{1-\alpha}{1+\alpha}$
	} \\[2pt]  \hline
	\textbf{Theorem \ref{thm:lowertradeoffvanilla}}, \cite{andoni2017optimal} & LSF in $(\sphere{d},\ip{\cdot}{\cdot})$ &
	$\geq \dfrac{(1-\alpha^{1+\lambda})^2}{1-\alpha^2}$ &
	$\geq \dfrac{(\alpha^{\lambda} -\alpha)^2}{1-\alpha^2}$ 
	\\[2pt]  \hline
    \end{tabular}
\begin{minipage}{1\textwidth} 
\scriptsize
\vspace{0.6em}
\textsc{Table notes:} Space-time tradeoffs for dynamic randomized solutions to similarity search problems in the LSH and LSF frameworks 
with query time $dn^{\rho_q + o(1)}$, update time $dn^{\rho_u + o(1)} + dn^{o(1)}$ and space usage $dn + n^{1 + \rho_u + o(1)}$.
Lower bounds are for the exponents $\rho_q, \rho_u$ within their respective frameworks.
Here $\varepsilon > 0$ denotes an arbitrary constant and $\lambda \in [-1,1]$ controls the space-time tradeoff. 
We have hidden $o_{n}(1)$ terms in the upper bounds and $o_d(1)$ terms in the lower bounds.
\newline *Assumes $c^2 \geq (1+\lambda)^2 / 2 + \lambda + \varepsilon$.
\end{minipage}
\label{lsf:tab:results}
\end{table*}

\subsection{Contribution}
Before stating our results we give a definition of locality-sensitive filtering that supports asymmetry in the framework query and update algorithm, 
yielding space-time tradeoffs.
\begin{definition} \label{def:lsf}
	Let $(X, \dist)$ be a space and let $\LSF$ be a probability distribution over \mbox{$\{(Q, U) \mid Q \subseteq X, U \subseteq X \}$}.
	We say that $\LSF$ is $(r, cr, p_1, p_2, p_q, p_u)$-sensitive if for all points $x, y \in X$ 
	and $(Q, U)$ sampled randomly from $\LSF$ the following holds:
	\begin{itemize}
		\item If $\dist(x, y) \leq r$ then $\Pr[x \in Q, y \in U] \geq p_1$.
		\item If $\dist(x, y) > cr$ then $\Pr[x \in Q, y \in U] \leq p_2$.
		\item $\Pr[x \in Q] \leq p_q$ and $\Pr[x \in U] \leq p_u$.
	\end{itemize}
	We refer to $(Q, U)$ as a filter and to $Q$ as the query filter and $U$ as the update filter.
	%
\end{definition}

Our main contribution is a general framework for similarity search with space-time tradeoffs 
that takes as input a locality-sensitive family of filters.
\begin{theorem}\label{thm:lsfvanilla}
	Suppose we have access to a family of filters that is $(r, cr, p_1, p_2, p_q, p_u)$-sensitive.
	Then we can construct a fully dynamic data structure 
	that solves the $(r, cr)$-near neighbor problem with 
	query time $dn^{\rho_q + o(1)}$, update time $dn^{\rho_u + o(1)}$, and space usage $dn + n^{1 + \rho_u + o(1)}$ where
	$\rho_q = \frac{\log(p_{q} / p_{1}) }{\log(p_{q} / p_{2})}$ and $\rho_u = \frac{\log(p_{u} / p_{1})}{\log(p_{q} / p_{2}) }$. 
\end{theorem}  

We give a worst-case analysis of a slightly modified version of Laarhoven's \cite{laarhoven2015} filter family for the unit sphere 
and plug it into our framework to obtain the following theorem.
\begin{theorem}\label{thm:spherevanilla}
	For every choice of $0 \leq \beta < \alpha < 1$ and $\lambda \in [-1,1]$ 
	there exists a solution to the $(\alpha, \beta)$-similarity problem in $(\sphere{d}, \ip{\cdot}{\cdot})$ 
	that satisfies the guarantees from Theorem \ref{thm:lsfvanilla} with exponents  
	$\rho_q = \left. \frac{(1-\alpha^{1+\lambda})^2}{1-\alpha^2}  \middle/  \frac{(1- \alpha^{\lambda} \beta)^2}{1-\beta^2} \right.$ and  
	$\rho_u = \left. \frac{(\alpha^{\lambda} - \alpha)^2}{1-\alpha^2}  \middle/  \frac{(1-\alpha^{\lambda} \beta)^2}{1-\beta^2} \right.$. 
\end{theorem}

We show how an elegant and powerful application of Bochner's Theorem \cite{rudin1990} by Rahimi and Recht \cite{rahimi2007} 
allows us to extend the solution on the unit sphere to a large class of similarity measures, yielding as a special case solutions for $\ell_s$-space.
\begin{theorem} \label{thm:lsvanilla}
	For every choice of $c \geq 1$, $s \in (0, 2]$, and $\lambda \in [-1,1]$ 
	there exists a solution to the~$(r, cr)$-near neighbor problem in $\ell_s^d$ 
	that satisfies the guarantees from Theorem \ref{thm:lsfvanilla} with exponents
	$\rho_q = \frac{c^s (1 + \lambda)^2}{(c^s + \lambda)^2}$ and $\rho_u = \frac{c^s (1 - \lambda)^2}{(c^s + \lambda)^2}$.		
\end{theorem}
This result improves upon the state of the art for every choice of asymmetric query/update exponents $\rho_q \neq \rho_u$ \cite{panigrahy2006, andoni2006, andoni2009, kapralov2015}.
We conjecture that this tradeoff is optimal among the class of algorithms that independently of the data determine which locations in memory to probe during queries and updates.
In the case of a balanced space-time tradeoff where we set $\rho_q = \rho_u$ our approach matches existing, optimal \cite{odonnell2014optimal}, 
data-independent solutions in $\ell_s$-spaces~\cite{indyk1998, datar2004, andoni2006, nguyen2014}.

The LSF framework is very similar to the LSH framework, 
especially in the case where the filter family is symmetric ($Q = U$ for every filter in $\LSF$).
In this setting we show that the LSH lower bound by O'Donnell et al.\ applies to the LSF framework as well \cite{odonnell2014optimal}, 
confirming that the results of Theorem \ref{thm:lsvanilla} are optimal when we set $\rho_q = \rho_u$.
\begin{theorem}[informal] \label{thm:lsflshvanilla}
	Every filter family that is symmetric and $(r, cr, p_1, p_2, p_q, p_u)$-sensitive in $\ell_{s}^{d}$ 
	must have $\rho = \frac{\log (p_u /p_1)}{\log (p_q / p_2)} \geq 1/c^{s} - o_{d}(1)$ when $r = \omega_{d}(1)$ is chosen to be sufficiently small.
\end{theorem}

Finally we show a lower bound on the space-time tradeoff that can be obtained in the LSF framework.
Our lower bound suffers from two important restrictions. 
First the filter family must be regular, meaning that all query filters and all update filters are of the same size. 
Secondly, the size of the query and update filter cannot differ by too much.
\begin{theorem}[informal] \label{thm:lowertradeoffvanilla}
    Every regular filter family that is $((1-\alpha)d/2, (1-\beta)d/2, p_1, p_2, p_q, p_u)$-sensitive in $d$-dimensional Hamming space 
	with asymmetry controlled by $\lambda \in [-1,1]$ 
	cannot simultanously have that
	\mbox{$\rho_q < \frac{(1-\alpha^{1+\lambda})^2}{1-\alpha^2} - o_{d}(1)$} 
	and $\rho_u < \frac{(\alpha^{\lambda} -\alpha)^2}{1-\alpha^2} - o_{d}(1)$. 
\end{theorem}
Together our upper and lower bounds imply that the filter family of concentric balls in Hamming space is asymptotically optimal for random data.

\paragraph{Techniques.}
The LSF framework in Theorem \ref{thm:lsfvanilla} relies on a careful combination of ``powering'' and ``tensoring'' techniques.
For positive integers $m$ and $\tau$ with $m \gg \tau$ the tensoring technique, 
a variant of which was introduced by Dubiner \cite{dubiner2010bucketing}, allows us to simulate a collection of $\binom{m}{\tau}$ filters 
from a collection of $m$ filters by considering the intersection of all $\tau$-subsets of filters.
Furthermore, given a point $x \in X$ we can efficiently list the simulated filters that contain $x$. 
This latter property is crucial as we typically need $\poly(n)$ filters to split our data 
into sufficiently small buckets for the search to be efficient.
The powering technique lets us amplify the locality-sensitivity properties of a filter family 
in the same way that powering is used in the LSH framework \cite{indyk1998, andoni2008, odonnell2014optimal}.

To obtain results for worst-case data on the unit sphere we analyze a filter family based on standard normal projections using the same techniques as Andoni et al. \cite{andoni2015practical} together with existing tail bounds on bivariate Gaussians.
The approximate kernel embedding technique by Rahimi and Recht \cite{rahimi2007} is used to extend the solution on the unit sphere 
to a large class of similarity measures, yielding Theorem \ref{thm:lsvanilla} as a special case. 

The lower bound in Theorem \ref{thm:lsflshvanilla} relies on an argument of contradiction against the LSH lower bounds 
by O'Donnell \cite{odonnell2014optimal} and uses a theoretical, inefficient, construction of a locality-sensitive family of hash functions 
from a locality-sensitive family of filters that is similar to the spherical LSH by Andoni et al. \cite{andoni2014beyond}.

Finally, the space-time tradeoff lower bound from Theorem \ref{thm:lowertradeoffvanilla} is obtained through an application of an 
isoperimetric inequality by O'Donnell \cite[Ch. 10]{odonnell2014analysis} and is similar in spirit to the LSH lower bound by Motwani et al.~\cite{motwani2007}.
\subsection{Related work}
The LSH framework takes a distribution $\LSH$ over hash functions that partition space with the property 
that the probability of two points landing in the same partition is an increasing function of their similarity.
\begin{definition} \label{filters:def:lsh}
	Let $(X, \dist)$ be a space and let $\LSH$ be a probability distribution over functions $h \colon X \to R$.
	We say that $\LSH$ is $(r, cr, p, q)$-sensitive if for all points $x, y \in X$ and $h$ sampled randomly from $\LSH$ the following holds:
	\begin{itemize}
		\item If $\dist(x, y) \leq r$ then $\Pr[h(x) = h(y)] \geq p$.
		\item If $\dist(x, y) > cr$ then $\Pr[h(x) = h(y)] \leq q$.
	\end{itemize}
\end{definition}
The properties of $\LSH$ determines a parameter $\rho < 1$ that governs the space and time complexity 
of the solution to the $(r, cr)$-near neighbor problem.
\begin{theorem}[LSH framework \cite{indyk1998, har-peled2012}]\label{filters:thm:lsh}
	Suppose we have access to a $(r, cr, p, q)$-sensitive hash family.
	Then we can construct a fully dynamic data structure 
	that solves the $(r, cr)$-near neighbor problem with query time $dn^{\rho + o(1)}$, update time~$dn^{\rho + o(1)}$, 
	and with a space usage of $dn + n^{1 + \rho + o(1)}$ where \mbox{$\rho = \frac{\log (1 / p) }{\log (1 / q) }$}. 
\end{theorem}

The LSF framework by Becker et al. \cite{becker2016} takes a symmetric $(r, cr, p_1, p_2, p_q, p_u)$-sensitive 
filter family $\LSF$ and produces a data structure that solves the 
$(r, cr)$-near neighbor problem with the same properties as the one produced by the LSH framework 
where instead we have \mbox{$\rho = \frac{\log(p_{q}/p_{1})}{\log(p_{q}/p_{2})}$}.
In addition, the framework assumes access to an oracle that is able to efficiently list the relevant filters containing a point $x \in X$
out of a large collection of filters.
The LSF framework in this paper removes this assumption, showing how to construct an efficient oracle as part of the framework. 

In terms of frameworks that support space-time tradeoffs, Panigrahy \cite{panigrahy2006} developed a 
framework based on LSH that supports the two extremes of the space-time tradeoff. 
In the language of Theorem \ref{thm:lsfvanilla}, 
Panigrahy's framework supports either setting $\rho_u = 0$ for a solution that uses near-linear space at the cost of a slower query time,
or setting $\rho_q = 0$ for a solution with query time $n^{o(1)}$ at the cost of a higher space usage.
To obtain near-linear space the framework stores every data point in $n^{o(1)}$ partitions 
induced by randomly sampled hash functions from a $(r, cr, p, q)$-sensitive LSH family $\LSH$.
In comparison, the standard LSH framework from Theorem \ref{filters:thm:lsh} uses $n^{\rho}$ such partitions where $\rho$ is determined by $\LSH$.
For each partition induced by $h \in \LSH$ the query algorithm in Panigrahy's framework generates a number of random points $z$ in a ball 
around the query point $x$ and searches the parts of the partition $h(z)$ that they hash to.
The query time is bounded by $n^{\hat{\rho} + o(1)}$ where $\hat{\rho} = \frac{I(h(z) | x, h)}{\log(1/q)}$ 
and $I(h(z) | x, h)$ denotes conditional entropy, i.e. the query time is determined by how hard it is to guess where 
$z$ hashes to given that we know $x$ and $h$.
Panigrahy's technique was used in a number of follow-up works that improve on solutions for specific spaces,
but to our knowledge none of them state a general framework with space-time tradeoffs \cite{lv2007, andoni2009, kapralov2015}.

\paragraph{Upper bounds.}
As is standard in the literature we state results in $\ell_s$-spaces in terms of the properties of a solution to the $(r, cr)$-near neighbor problem.
For results on the unit sphere under inner product similarity $(\sphere{d}, \ip{\cdot}{\cdot})$
we instead use the $(\alpha, \beta)$-similarity terminology, defined in the introduction, as we find it to be 
cleaner and more intuitive while aligning better with the analysis.
The $\ell_s$-spaces, particularly $\ell_1$ and $\ell_2$, as well as $(\sphere{d}, \ip{\cdot}{\cdot})$ are some of 
most well-studied spaces for similarity search and are also widely used in practice \cite{wang2014}. 
Furthermore, fractional norms ($\ell_s$ for $s \neq 1, 2$) have been shown to perform better than the standard norms 
in certain use cases \cite{aggarwal2001} which motivates finding efficient solutions to the near neighbor problem in general $\ell_s$-space. 

In the case of a balanced space-time tradeoff the best data-independent upper bound for the $(r, cr)$-near neighbor problem in $\ell_s^d$ 
are solutions with an LSH exponent of $\rho = 1/c^s$ for $0 < s \leq 2$.
This result is obtained through a combination of techniques.
For \mbox{$0 < s \leq 1$} the LSH based on $s$-stable distributions by Datar et al. \cite{datar2004} 
can be used to obtain an exponent of $(1 + \varepsilon)/c^s$ for an arbitrarily small constant $\varepsilon > 0$.
For \mbox{$1 < s \leq 2$} the ball-carving LSH by Andoni and Indyk~\cite{andoni2006} for Euclidean space can be extended to $\ell_s$
using the technique described by Nguyen \cite[Section 5.5]{nguyen2014}.
Theorem \ref{thm:lsvanilla} matches (and potentially improves in the case of $0 < s < 1$) 
these results with a single unified technique and analysis that we find to be simpler.

For space-time tradeoffs in Euclidean space (again extending to $\ell_s$ for $1 < s < 2$) Kapralov \cite{kapralov2015},
improving on Panigrahy's results \cite{panigrahy2006} in Euclidean space and using similar techniques, obtains a solution with 
query exponent \mbox{$\rho_q = \frac{c^2 (1 + \lambda)^2}{(c^2 + \lambda)^2 - c^2 (1+\lambda^2)/2 - \lambda^2}$}
and update exponent \mbox{$\rho_u = \frac{c^2 (1 - \lambda)^2}{(c^2 + \lambda)^2 - c^2 (1+\lambda^2)/2 - \lambda^2}$} 
under the condition that $c^2 \geq (1+\lambda)^2 / 2 + \lambda + \varepsilon$ where $\varepsilon > 0$ is an arbitrary positive constant.
Comparing to our Theorem \ref{thm:lsvanilla} it is easy to see that we improve upon Kapralov's space-time tradeoff 
for all choices of $c$ and $\lambda$. 
In addition, Theorem \ref{thm:lsvanilla} represents the first solution to the $(r, cr)$-near neighbor problem in Euclidean space 
that for every choice of constant $c > 1$ obtains sublinear query time ($\rho_q < 1$) using only near-linear space ($\rho_u = 0$). 
Due to the restrictions on Kapralov's result he is only able to obtain sublinear query time for $c > \sqrt{3}$ 
when the space usage is restricted to be near-linear. 
It appears that our improvements can primarily be attributed to our techniques allowing a more direct analysis.
Kapralov uses a variation of Panigrahy's LSH-based technique of, depending on the desired space-time tradeoff, 
either querying or updating additional memory locations around a point $x \in X$ in the partition induced by $h \in \LSH$.
For a query point $x$ and a near neighbor $y$ his argument for correctness is based on guaranteeing that both the query algorithm 
and update algorithm visit the part $h(z)$ where $z$ is a point lying between $x$ and $y$, 
possibly leading to a loss of efficiency in the analysis.
More details on the comparison of Theorem \ref{thm:lsvanilla} to Kapralov's result can be found in Appendix~\ref{app:kapralov}. 

In terms of space-time tradeoffs on the unit sphere, Laarhoven \cite{laarhoven2015} modifies a filter family 
introduced by Becker et al. \cite{becker2016} to support space-time tradeoffs,
obtaining a solution for random data on the unit sphere (the $(\alpha, \beta)$-similarity problem with $\beta = o_{d}(1)$) 
with query exponent \mbox{$\rho_q = \frac{(1 - \alpha^{1 + \lambda})^2}{1 - \alpha^2}$} 
and update exponent \mbox{$\rho_u = \frac{(\alpha^{\lambda} - \alpha)^2}{1 - \alpha^2}$}.
Theorem \ref{thm:spherevanilla} extends this result to provide a solution to the $(\alpha, \beta)$-similarity problem on the unit sphere
for every choice of $0 \leq \beta < \alpha < 1$.
This extension to worst case data is crucial for obtaining our results for $\ell_s$-spaces in Theorem \ref{thm:lsvanilla}. 
We note that there exist other data-independent techniques (e.g. Valiant \cite[Alg. 25]{valiant2015}) 
for extending solutions on the unit sphere to $\ell_2$, but they also require a solution for worst-case data on the unit sphere to work.

\paragraph{Lower bounds}
The performance of an LSH-based solution to the near neighbor problem in a given space that uses
a $(r, cr, p, q)$-sensitive family of hash functions $\LSH$ is summarized by the value of the exponent $\rho = \frac{\log(1/p)}{\log(1/q)}$. 
It is therefore of interest to lower bound $\rho$ in terms of the approximation factor $c$.
Motwani et al.~\cite{motwani2007} proved the first lower bound for LSH families in $d$-dimensional Hamming space. 
They show that for every choice of $c \geq 1$ then for some choice of $r$ it must hold that $\rho \geq 0.462/c$ as $d$ goes to infinity
under the assumption that $q$ is not too small ($q \geq 2^{-o(d)}$). 

As part of an effort to show lower bounds for data-dependent locality-sensitive hashing,
Andoni and Razenshteyn \cite{andoni2016tight} strengthened the lower bound by Motwani et al.\ to $\rho \geq 1/(2c - 1)$ in Hamming space.
These lower bounds are initially shown in Hamming space and can then be extended to $\ell_s$-space and the unit sphere 
by the fact that a solution in these spaces can be used to yield a solution in Hamming space, 
contradicting the lower bound if $\rho$ is too small.
Translated to $(\alpha, \beta)$-similarity on the unit sphere, 
which is the primary setting for the lower bounds on LSF space-time tradeoffs in this paper,
the lower bound by Andoni and Razenshteyn shows that an LSH on the unit sphere must have $\rho \geq \frac{1-\alpha}{1+\alpha}$ 
which is tight in the case of random data~\cite{andoni2015practical}. 

The lower bound uses properties of random walks over a partition of Hamming space: 
A random walk starting from a random point $x \in \cube{d}$ is likely to ``walk out'' 
of the the part identified by $h(x)$ in the partition induced by $h$.
The space-time tradeoff lower bound in Theorem \ref{thm:lowertradeoffvanilla} relies on a similar argument 
that lower bounds the probability that a random walk starting from a subset $Q$ ends up in another subset $U$, 
corresponding nicely to query and update filters in the LSF framework. 

Using related techniques O'Donnell \cite{odonnell2014optimal} showed tight LSH lower bounds for $\ell_s$-space of $\rho \geq 1/c^s$.
The work by Andoni et al.~\cite{andoni2006lower} and Panigrahy et al.~\cite{panigrahy2008geometric, panigrahy2010} 
gives cell probe lower bounds for the $(r, cr)$-near neighbor problem, 
showing that in Euclidean space a solution with a query complexity of $t$ probes require space at least $n^{1 + \Omega(1/tc^2)}$. 
For more details on these lower bounds and how they relate to the upper bounds on the unit sphere see \cite{andoni2017optimal, laarhoven2015}.

\paragraph{Data-dependent solutions}
The solutions to the $(r, cr)$-near neighbor problems considered in this paper are all \emph{data-independent}.
For the LSH and LSF frameworks this means that the choice of hash functions or filters used by the data structure,
determining the mapping between points in space and the memory locations that are searched during the query and update algorithm, 
is made without knowledge of the data.
Data-independent solutions to the $(r, cr)$-near neighbor problem for worst-case data have been the state of the art
until recent breakthroughs by Andoni et al.~\cite{andoni2014beyond} and Andoni and Razenshteyn~\cite{andoni2015optimal}
showing improved solutions to the $(r, cr)$-near neighbor problem in Euclidean space using \emph{data-dependent} techniques.
In this setting the solution obtained by Andoni and Razenshteyn has an exponent of $\rho = 1/(2c^2 - 1)$ 
compared to the optimal data-independent exponent of $\rho = 1/c^2$.
Furthermore, they show that this exponent is optimal for data-dependent solutions in a restricted model \cite{andoni2016tight}.

\paragraph{Recent developments} 
Recent work by Andoni et al.~\cite{andoni2017optimal}, done independently of and concurrently with this paper, 
shows that Laarhoven's upper bound for random data on the unit sphere can be combined with data-dependent techniques~\cite{andoni2015optimal} 
to yield a space-time tradeoff in Euclidean space with $\rho_u, \rho_q$ satisfying $c^2 \sqrt{\rho_q} + (c - 1)\sqrt{\rho_u} = \sqrt{2c^2 - 1}$.
This improves the result of Theorem~\ref{thm:lsvanilla} and matches the lower bound in Theorem~\ref{thm:lowertradeoffvanilla}.
In the same paper they also show a lower bound matching our lower bound in Theorem \ref{thm:lowertradeoffvanilla}.
Their lower bound is set in a more general model that captures both the LSH and LSF framework 
and they are able to remove some of the technical restrictions such as the filter family being regular that weaken the lower bound in this paper.
In spite of these results we still believe that this paper presents an important contribution 
by providing a general and simple framework with space-time tradeoffs as well as improved data-independent solutions 
to nearest neighbor problems in $\ell_s$-space and on the unit sphere. 
We would also like to point out the simplicity and power of using Rahimi and Recht's~\cite{rahimi2007} result 
to extend solutions on the unit sphere to spaces with similarity measures corresponding to real-valued characteristic functions, 
further described in Appendix~\ref{app:characteristic}.
\section{A framework with space-time tradeoffs}
We use a combination of powering and tensoring techniques to amplify the locality-sensitive properties of our initial filter family, 
and to simulate a large collection of filters that we can evaluate efficiently.
We proceed by stating the relevant properties of these techniques which we then combine to yield our Theorem \ref{thm:lsfvanilla}. 
\begin{lemma}[powering]\label{lem:powering}
	Given a $(r, cr, p_1, p_2, p_q, p_u)$-sensitive filter family $\LSF$ for $(X, \dist)$ 
	and a positive integer $\kappa$ define the family $\LSF^{\kappa}$ as follows:
	we sample a filter $F = (Q, U)$ from $\LSF^{\kappa}$ by sampling 
	$(Q_1, U_1), \dots, (Q_\kappa, U_\kappa)$ independently from $\LSF$
	and setting $(Q, U) = (\bigcap_{i=1}^{\kappa} Q_i, \bigcap_{i=1}^{\kappa} U_i)$. 
	The family $\LSF^{\kappa}$ is $(r, cr, p_{1}^{\kappa}, p_{2}^{\kappa}, p_{q}^{\kappa}, p_{u}^{\kappa})$-sensitive for $(X, \dist)$.
\end{lemma}

Let $\F$ denote a collection (indexed family) of $m$ filters and let $\Q$ and $\U$ denote the corresponding collections of query and update filters, 
that is, for $i \in \{1, \dots, m\}$ we have that $\F_i = (\Q_i, \U_i)$. 
Given a positive integer $\tau \leq m$ (typically $\tau \ll m$) we define $\TLSF$ to be the 
collection of filters formed by taking all the intersections of $\tau$-combinations of filters from $\F$, 
that is, for every $I \subseteq \{1, \dots, m \}$ with~$|I| = \tau$ we have that 
\begin{equation*}
\F_{I}^{\otimes \tau}\! = \left(\mcap_{i \in I} \Q_{i}, \mcap_{i \in I} \U_{i}\right) 
\end{equation*}
The following properties of the tensoring technique will be used to provide correctness, running time, 
and space usage guarantees for the LSF data structure that will be introduced in the next subsection.
We refer to the evaluation time of a collection of filters $\F$ as the time it takes, given a point $x \in X$ 
to prepare a list of query filters $\Q(x) \subseteq \Q$ containing $x$ and a list of update filters $\U(x) \subseteq \U$ containing $x$
such that the next element of either list can be reported in constant time.
We say that a pair of points $(x, y)$ is contained in a filter $(Q, U)$ if~$x \in Q$ and~$y \in U$.
\begin{lemma}[tensoring]\label{lem:tensoring}
	Let $\LSF$ be a filter family that is $(r, cr, p_1, p_2, p_q, p_u)$-sensitive in $(X, \dist)$.
	Let $\tau$ be a positive integer and let $\F$ denote a collection of \mbox{$m = \lceil \tau / p_1 \rceil$} 
	independently sampled filters from~$\LSF$.
	Then the collection $\TLSF$ of $\binom{m}{\tau}$ filters has the following properties:
	\begin{itemize}
		\item If $(x, y)$ have distance at most $r$ then with probability at least $1/2$ 
			there exists a filter in $\TLSF$ containing~$(x, y)$.
		\item If $(x, y)$ have distance greater than $cr$ then the expected number of filters in $\TLSF$ 
			containing $(x, y)$ is at most~$p_{2}^{\tau}\binom{m}{\tau}$.
		\item In expectation, a point $x$ is contained in at most $p_{q}^{\tau}\binom{m}{\tau}$ query filters 
			and at most $p_{u}^{\tau}\binom{m}{\tau}$ update filters in $\TLSF$.
		\item The evaluation time and space complexity of $\TLSF$ is dominated by the time 
			it takes to evaluate and store $m$ filters from $\LSF$. 
	\end{itemize}
\end{lemma}
\begin{proof}
	To prove the first property we note that there exists a filter in $\TLSF$ containing 
	$(x, y)$ if at least $\tau$ filters in $\F$ contain $(x, y)$.
	The binomial distribution has the property that the median is at least as great as the mean rounded down \cite{buhrman1980}.
	By the choice of $m$ we have that the expected number of filters in $\F$ containing $(x, y)$ is at least $\tau$ and the result follows. 
	The second and third properties follow from the linearity of expectation and the fourth is trivial.
\end{proof}
\subsection{The LSF data structure} \label{sec:lsfds}
We will introduce a dynamic data structure that solves the $(r, cr)$-near neighbor problem on a set of points $P \subseteq X$.
The data structure has access to a $(r, cr, p_1, p_2, p_q, p_u)$-sensitive filter family $\LSF$ in the sense that
it knows the parameters of the family and is able to sample, store, and evaluate filters from $\LSF$ in time $dn^{o(1)}$.

The data structure supports an initialization operation that initializes a collection of filters $\F$ 
where for every filter we maintain a (possibly empty) set of points from $X$.
After initialization the data structure supports three operations: \textsc{insert}, \textsc{delete}, and \textsc{query}.
The \textsc{insert} (\textsc{delete}) operation takes as input a point $x \in X$ and adds (removes) the point from the set of 
points associated with each update filter in $\F$ that contains $x$.
The \textsc{query} operation takes as input a point $x \in X$. 
For each query filter in $\F$ that contains $x$ we proceed by 
computing the dissimilarity $\dist(x, y)$ to every point $y$ associated with the filter.
If a point $y$ satisfying $\dist(x, y) \leq cr$ is encountered, then $y$ is returned and the query algorithm terminates.
If no such point is found, the query algorithm returns a special symbol ``$\varnothing$'' and terminates.

The data structure will combine the powering and tensoring techniques in order to simulate the collection of filters $\F$ from two smaller collections:
$\F_1$ consisting of $m_1$ filters from $\LSF^{\kappa_1}$ and $\F_2$ consisting of $m_2$ filters from $\LSF^{\kappa_2}$.
The collection of simulated filters $\F$ is formed by taking all filters $(Q_1 \cap Q_2, U_1 \cap U_2)$ 
where $(Q_1, U_1)$ is a member of $\F_{1}^{\otimes \tau}\!$ and $(Q_2, U_2)$ is a member of $\F_2$.
It is due to the integer constraints on the parameter $\tau$ in the tensoring technique 
and the parameter $\kappa$ in the powering technique that we simulate our filters from two underlying collections 
instead of just one.
This gives us more freedom to hit a target level of amplification of the simulated filters which in turn
makes it possible for the framework to support efficient solutions for a wider range of parameters of LSF families. 

The initialization operation takes $\LSF$ and parameters $m_1, \kappa_1, \tau, m_2, \kappa_2$ and samples and stores $\F_1$ and $\F_2$.
The filter evaluation algorithm used by the insert, delete, and query operation takes a point $x \in X$ and computes for $\F_1$ and $\F_2$, 
depending on the operation, the list of update or query filters containing $x$.  
From these lists we are able to generate the list of filters in $\F$ containing~$x$.

Setting the parameters of the data structure to guarantee correctness while balancing the contribution to the query time from 
the filter evaluation algorithm, the number of filters containing the query point, and the number of distant points examined, 
we obtain a partially dynamic data structure that solves the $(r, cr)$-near neighbor problem with failure probability $\delta \leq 1/2 + 1/e$.
Using a standard dynamization technique by Overmars and Leeuwen \cite[Thm. 1]{overmars1981} we obtain a fully dynamic data structure 
resulting in Theorem \ref{thm:lsfvanilla}. The details of the proof have been deferred to Appendix~\ref{app:framework}.
\section{Gaussian filters on the unit sphere} \label{sec:gaussianlsf}
In this section we show properties of a family of filters for the unit sphere $\sphere{d}$ under inner product similarity.
Later we will show how to make use of this family to solve the near neighbor problem in other spaces, including $\ell_s$ for $0 < s \leq 2$.
\begin{lemma} \label{lem:gaussianlsf}
	For every choice of $0 \leq \beta < \alpha < 1$, $\lambda \in [-1, 1]$, and $t > 0$ 
	let $\mathcal{G}$ denote the family of filters defined as follows:
	we sample a filter $(Q, U)$ from $\mathcal{G}$ by sampling $z \sim \mathcal{N}^{d}(0,1)$ and setting
\begin{align*}
	Q &= \{ x \in \real^d \mid \ip{x}{z} > \alpha^{\lambda} t \}, \\ 
	U &= \{ x \in \real^d \mid \ip{x}{z} >  t \}.
\end{align*}
Then $\mathcal{G}$ is locality-sensitive 
on the unit sphere under inner product similarity with exponents
\small
\begin{align*}
	\rho_q &\leq \left.\left( \frac{(1-\alpha^{1+\lambda})^2}{1-\alpha^2} + 
\frac{\ln(2\pi(1+t/\alpha)^{2})}{t^2 / 2}\right) \middle/  \frac{(1-\alpha^{\lambda} \beta)^2}{1-\beta^2} \right., \\
\rho_u &\leq \left.\left( \frac{(\alpha^{\lambda} - \alpha)^2}{1-\alpha^2} + 
\frac{\ln(2\pi(1+t/\alpha)^{2})}{t^2 / 2}\right) \middle/  \frac{(1-\alpha^{\lambda} \beta)^2}{1-\beta^2} \right..
\end{align*}
\normalsize
\end{lemma}
Laarhoven's filter family~\cite{laarhoven2015} is identical to $\mathcal{G}$ 
except that he normalizes the projection vectors $z$ to have unit length. 
The properties of $\mathcal{G}$ can easily be verified with a simple back-of-the-envelope analysis using two facts:
First, for a standard normal random variable $Z$ we have that $\Pr[Z > t] \approx e^{-t^{2}/2}$.
Secondly, the invariance of Gaussian projections $\ip{x}{z}$ to rotations, allowing us to analyze the projection of arbitrary points 
$x, y \in \sphere{d}$ with inner product $\ip{x}{y} = \alpha$ in a two-dimensional setting $x = (1, 0)$ 
and $y = (\alpha, \sqrt{1-\alpha^{2}})$ without any loss of generality.
The proof of Lemma~\ref{lem:gaussianlsf} as well as the proof of Theorem~\ref{thm:spherevanilla} has been deferred to Appendix~\ref{app:gaussian}.
\section{Space-time tradeoffs under kernel similarity} \label{sec:kernel}
In this section we will show how to combine the Gaussian filters for the unit sphere with kernel approximation techniques in order to solve the 
$(\alpha, \beta)$-similarity problem over $(\real^d, S)$ for the class of similarity measures of the form $S(x, y) = k(x-y)$ 
where $k \colon \real^d \times \real^d \to \real$ is a real-valued characteristic function \cite{ushakov1999}.
For this class of functions there exists a feature map $\psi$ into a (possibly infinite-dimensional) 
dot product space such that $k(x, y) = \ip{\psi(x)}{\psi(y)}$.
Through an elegant combination of Bochner's Theorem and Euler's Theorem, detailed in Appendix \ref{app:characteristic},
Rahimi and Recht \cite{rahimi2007} show how to construct approximate feature maps, 
i.e., for every $k$ we can construct a function $v$ with the property that $\ip{v(x)}{v(y)} \approx \ip{\psi(x)}{\psi(y)} = k(x - y)$.
We state a variant of their result for a mapping onto the unit sphere.
\begin{lemma} \label{lem:rahimirecht}
	For every real-valued characteristic function $k$ and every positive integer $l$ there exists
	a family of functions $\RR \subseteq \{ v \mid v \colon \real^{d} \to \sphere{l} \}$ such that 
	for every $x, y \in \real^d$ and $\varepsilon > 0$ we have that 
	\begin{equation*}
	\Pr_{v \sim \RR}[|\ip{v(x)}{v(y)} - k(x, y)| \geq \varepsilon] \leq e^{-\Omega(l \varepsilon^2)}.
	\end{equation*}
\end{lemma}
Theorem \ref{thm:characteristiclsf} in Appendix \ref{app:characteristic} shows that Theorem \ref{thm:spherevanilla} holds
with the space $(\sphere{d}, \ip{\cdot}{\cdot})$ replaced by $(\real^d, k)$.
\subsection{Tradeoffs in $\ell_s^d$-space} \label{sec:ls}
Consider the $(r, cr)$-near neighbor problem in $\ell_s^d$ for $0 < s \leq 2$. 
We solve this problem by first applying the approximate feature map from Lemma \ref{lem:rahimirecht} for the characteristic function 
of a standard $s$-stable distribution~\cite{zolotarev1986}, mapping the data onto the unit sphere, 
and then applying our solution from Theorem \ref{thm:spherevanilla} to solve the appropriate $(\alpha, \beta)$-similarity problem on the unit sphere.  
The characteristic functions of $s$-stable distributions take the following form:
\begin{lemma}[{Lévy \cite{levy1925}}]\label{lem:levy}
	For every positive integer $d$ and $0 < s \leq 2$ there exists a characteristic function 
	$k \colon \real^d \times \real^d \to [0,1]$ of the form
	\begin{equation*}
	k(x, y) = k(x-y) = e^{-\norm{x-y}_{s}^{s}}.
	\end{equation*}
\end{lemma}
A result by Chambers et al. \cite{chambers1976} shows how to sample efficiently from an $s$-stable distributions. 

To sketch the proof of Theorem \ref{thm:lsvanilla} we proceed by upper bounding the exponents 
$\rho_q$, $\rho_u$ from Theorem \ref{thm:spherevanilla} when applying Lemma \ref{lem:rahimirecht} 
to get $\alpha \geq e^{-r^s} - \varepsilon$ and $\beta \leq e^{-c^s r^s} - \varepsilon$.
We make use of the following standard fact (see e.g. \cite{savage1962}) that can be derived from the 
Taylor expansion of the exponential function: for $x \geq 0$ it holds that $1 - x \leq e^{-x} \leq 1 - x + x^{2}/2$.
Scaling the data points such that $r^s = o(1)$ and inserting the above values of $\alpha \approx 1 - r^s$ and $\beta \approx 1 - c^s r^s$ 
into the expressions for $\rho_q$, $\rho_u$ in Lemma \ref{lem:gaussianlsf} 
we can set parameters $t$ and $l$ such that Theorem \ref{thm:lsvanilla} holds. 
\section{Lower bounds}
We begin by stating the lower bound on the LSH exponent $\rho = \log(1/p) / \log(1/q)$ by O'Donnell et al.~\cite{odonnell2014optimal}.
\begin{theorem}[O'Donnell et al. \cite{odonnell2014optimal}]\label{thm:odlower}
	Fix $d \in \mathbb{N}$, $1 < c < \infty$, $0 < s < \infty$ and $0 < q < 1$. 
	Then for a certain choice of $r = \omega_{d}(1)$ and under the assumption that $q \geq 2^{-o(d)}$ 
	we have that every $(r, cr, p, q)$-sensitive family of hash functions for $\ell_s^d$ 
	must satisfy 
	\begin{equation*}
		\rho = \frac{\log(1/p)}{\log(1/q)} \geq \frac{1}{c^{s}} - o_{d}(1).
	\end{equation*}	
\end{theorem}
The following lemma shows how to use a filter family $\LSF$ to construct a hash family $\LSH$.
\begin{lemma}\label{lem:lsftolsh}
	Given a symmetric family of filters that is $(r, cr, p_1, p_2, p_q, p_u)$-sensitive in $(X, \dist)$
	we can construct a $(r, cr, p_{1}/(2p_{q}), p_{2}/p_{q})$-sensitive family of hash functions in $(X, \dist)$. 
\end{lemma}
\begin{proof}
	Given the filter family $\LSF$ we sample a random function $h$ from the hash family $\LSH$  
	taking an infinite sequence of independently sampled filters $(F_i)_{i=0}^{\infty}$ from~$\LSF$ 
	and setting $h(x) = \min \, \{i \mid x \in F_i\}$.
	The probability of collision is given by
	\begin{equation*}	
	\Pr_{h \sim \LSH}[h(x) = h(y)] = \frac{\Pr_{F \sim \LSF}[x \in F \land y \in F]}{\Pr_{F \sim \LSF}[x \in F \lor y \in F]}
	\end{equation*}
	and the result follows from the properties of $\LSF$.
\end{proof}
If the LSH family in Lemma \ref{lem:lsftolsh} had $p = p_{1}/p_{q}$ and $q = p_{2}/p_{q}$ then the lower bound would follow immediately.
We apply the powering technique from Lemma \ref{lem:powering} to the underlying filter family 
in order make the factor $2$ in $p_{1}/(2p_{q})$ disappear in the statement of $\rho$ as $d$ tends to infinity.
\begin{customthm}{1.4} \label{thm:lower1}
	Every symmetric $(r, cr, p_1, p_2, p_q, p_u)$-sensitive filter family $\LSF$ for $\ell_s^d$ must satisfy the lower bound
	of Theorem \ref{thm:odlower} with $p = p_{1}/p_{q}$ and $q = p_{2}/p_{q}$.
\end{customthm}
\begin{proof}
	Given a family $\LSF$ that satisfies the requirements from Theorem \ref{thm:odlower} there exists an integer $\kappa = \omega_{d}(1)$
	such the hash family $\LSH$ that results from applying Lemma \ref{lem:lsftolsh} to the powered family $\LSF^{\kappa}$ 
	also satisfies the requirements from Theorem \ref{thm:odlower}.
	The constructed family $\LSH$ is $(r, cr, p, q)$-sensitive for $p = (1/2) \cdot (p_{1}/p_{q})^{\kappa}$ and $q = (p_{2}/p_{q})^{\kappa}$.
	By our choice of $\kappa$ we have that $\log(1/p)/\log(1/q) = \log(p_{q}/p_{1})/\log(p_{q}/p_{2}) + o_{d}(1)$ 
	and the lower bound on $\log(1/p)/\log(1/q)$ from Theorem \ref{thm:odlower} applies.
\end{proof}
\subsection{Asymmetric lower bound}
The lower bound is based on an isoperimetric-type inequality that holds for randomly correlated points in Hamming space.
We say that the pair of points $(x, y)$ is $\alpha$-correlated if $x$ is a random point in $\cube{d}$ 
and $y$ is formed by taking $x$ and independently flipping each bit with probability $(1-\alpha)/2$.    
We are now ready to state O'Donnell's generalized small-set expansion theorem.
Notice the similarity to the value of $p_1$ for the Gaussian filter family described in 
Section~\ref{sec:gaussianlsf} and Appendix~\ref{app:gaussian}.
\begin{lemma}[{\cite[p. 285]{odonnell2014analysis}}] \label{lem:odhyper}
	For every \mbox{$0 \leq \alpha < 1$}, $-1 \leq \lambda \leq 1$, and \mbox{$Q, U \subseteq \cube{d}$} 
	satisfying that \mbox{$|Q|/2^{d} = (|U|/2^{d})^{\alpha^{2\lambda}}$} we have 
	\begin{equation*}
		\Pr_{\corrsub{\alpha}}[x \in Q, y \in U] \leq (|U|/2^{d})^\frac{1 + \alpha^{2\lambda} - 2 \alpha^{1 + \lambda}}{1 - \alpha^{2}}.
	\end{equation*}
\end{lemma}
The argument for the lower bound assumes a regular $(r, cr, p_1, p_2, p_q, p_u)$-sensitive filter family $\LSF$ for Hamming space
where we set $r = (1-\alpha)d/2$ and \mbox{$cr = (1-\beta)d/2$} for some choice of $0 < \beta < \alpha < 1$.
We then proceed by deriving constraints on $p_1$, $p_2$, $p_q$, $p_u$, and minimize $\rho_q$ and $\rho_u$ subject to those constrains.
The proof of Theorem \ref{thm:lower2} is provided in Appendix~\ref{app:lowertradeoffproof}.
\begin{customthm}{1.5} \label{thm:lower2}
Fix $0 < \beta < \alpha < 1$. 
Then for every regular $((1-\alpha)d/2, (1-\beta)d/2, p_1, p_2, p_q, p_u)$-sensitive filter family in $d$-dimensional Hamming space 
with and $|Q|/2^d = (|U|/2^d)^{\alpha^{2\lambda}}$ where $\lambda$ satisfies 
$\alpha + 2\sqrt{\ln(d) / d} \leq \alpha^{\lambda} \leq 1/(\alpha - 2\sqrt{\ln(d)/d})$ it must hold that
\begin{align*}
	\rho_q &= \frac{\log(p_q / p_1)}{\log(p_q / p_2)} \geq \frac{(1-\alpha^{1+\lambda})^{2}}{1-\alpha^{2}} - o_{d}(1), \\ 
	\rho_u &= \frac{\log(p_u / p_1)}{\log(p_q / p_2)} \geq \frac{(\alpha^{\lambda} - \alpha)^{2}}{1-\alpha^{2}} - o_{d}(1) 
\end{align*}
when $p_q$ is set to minimize $\rho_q$ and we assume that $|U|/2^{d} \geq 2^{-o_{d}(1)}$.
\end{customthm}
\section{Open problems}
An important open problem is to find simple and practical data-dependent solutions to the $(r, cr)$-near neighbor problem.
Current solutions, the Gaussian filters in this paper included, suffer from $o(1)$ terms in the exponents that decrease very slowly in $n$.
A lower bound for the unit sphere by Andoni et al. \cite{andoni2015practical} indicates that this might be unavoidable.

Another interesting open problem is finding the shape of provably exactly optimal filters in different spaces.
In the random data setting in Hamming space, this problem boils down to maximizing the number of pairs of points 
below a certain distance threshold that is contained in a subset of the space of a certain size.
This is a fundamental problem in combinatorics that has been studied by among others \cite{kahn1988}, but a complete answer remains elusive.
The LSH and LSF lower bounds~\cite{motwani2007, odonnell2014optimal, andoni2016tight}, 
along with classical isoperimetric inequalities such as Harper's Theorem 
and more recent work summarized in the book by O'Donnell \cite{odonnell2014analysis} hints that the answer is somewhere between a subcube and a generalized sphere.

A recent result by Chierichetti and Kumar \cite{chierichetti2015} characterizes the set of 
transformations of LSH-able similarity measures as the set of probability-generating functions. 
This seems to have deep connections to result of this paper that uses characteristic functions that allow well-known kernel transformations. 
It seems possible that this paper can be viewed as a semi-explicit construction of their result, 
or that their result can be described as an application of Bochner's Theorem.  
\section*{Acknowledgment}
I would like to thank Rasmus Pagh for suggesting the application of Rahimi \& Recht's result \cite{rahimi2007} 
and the MinHash-like~\cite{broder1998} connection between LSF and LSH used in Theorem \ref{thm:lower1}. 
I would also like to thank Gregory Valiant and Udi Wieder for useful discussions about locality-sensitive filtering 
and the analysis of boolean functions during my stay at Stanford.
Finally, I would like to thank the Scalable Similarity Search group at the IT University of Copenhagen for feedback during the writing process, 
and in particular Martin Aumüller for pointing out the importance of a general framework for locality-sensitive filtering with space-time tradeoffs. 
\section{Appendix: Framework} \label{app:framework}
We state a version of Theorem \ref{thm:lsfvanilla} where the parameters of the filter family are allowed to depend on $n$.
\begin{customthm}{\ref{thm:lsfvanilla}} \label{thm:lsf}
	Suppose we have access to a filter family that is $(r, cr, p_1, p_2, p_q, p_u)$-sensitive.
	Then we can construct a fully dynamic data structure that solves the $(r, cr)$-near neighbor problem. 
	Assume that $1/p_{1}$, $1 / \log(p_{q}/p_{2})$, and $\exp(\log(1/p_{1})/\log(\min(p_{q},p_{u})/p_{1}))$ are $n^{o(1)}$, then the data structure has
	\begin{itemize}
		\item[--] query time $dn^{\rho_q + o(1)}$,
		\item[--] update time $n^{\rho_u + o(1)} + dn^{o(1)}$, 
		\item[--] space usage $n^{1 + \rho_u + o(1)} + dn + dn^{o(1)}$
	\end{itemize}
	where
	\begin{equation*}
		\rho_q = \frac{\log p_{q} / p_{1} }{\log p_{q} / p_{2} }, \quad  
		\rho_u = \frac{\log p_{u} / p_{1} }{\log p_{q} / p_{2} }. 
	\end{equation*}
\end{customthm}

To prove Theorem \ref{thm:lsf}, we begin by setting the parameters mentioned in the description of 
the LSF data structure in Section \ref{sec:lsfds}. 
\begin{align*}
\kappa_{1} &= \left \lceil \frac{ \min(\rho_q, \rho_u) \log n}{\log(1/p_{1})} \right \rceil \\
\tau  &= \left \lfloor \frac{\log n}{\kappa_1 \log(p_{q}/p_{2})} \right \rfloor \leq \frac{\log(1/p_1)}{\log(\min(p_{q}, p_{u})/p_{1})} \\
m_{1} &= \lceil \tau/p_{1}^{\kappa_{1}}  \rceil \\
\kappa_{2} &= \max(0, \lceil \log(n) / \log(p_{q}/p_{2}) \rceil - \tau \kappa_1) \\
m_{2} &= \lceil 1/p_{1}^{\kappa_{2}} \rceil
\end{align*}
We will now briefly explain the reasoning behind the parameter settings.
Begin by observing that the powering and tensoring techniques both amplify the filters from $\LSF$. 
Let $m = \binom{m_{1}}{\tau} \cdot m_{2}$ denote the number of simulated filters in our collection $\F$ 
and let $a = \tau\kappa_1 + \kappa_2$ be an integer denoting the number of times each filter has been amplified.
Ignoring the time it takes to evaluate the filters, the query time is determined by the sum of the number of filters that contain a query point 
and the number of distant points associated with those filters that the query algorithm inspects.
The expected number of activated filters is given by $mp_{q}^{a}$ while the worst case expected number of distant points to be inspected by the query algorithm is given by $nmp_{2}^{a}$.
Balancing the contribution to the query time from these two effects (ignoring the $O(d)$ factor from distance computations) 
results in a target value of $a = \lceil \log(n) / \log(p_{q}/p_{2}) \rceil$.
Compared to having an oracle that is able to list the filters from a collection that contains a point, 
there is a small loss in efficiency from using the tensoring technique due to the increase in the number of filters required to guarantee correctness.
The parameters of the LSF data structure are therefore set to minimize the use of tensoring such that
the time spent evaluating our collection of filters roughly matches the minimum of the query and update time.

Consider the initialization operation of the LSF data structure with the parameters setting from above.
We have that $\kappa_2 \leq \kappa_1$ implying that $m_2 = O(m_1)$. 
The initialization time and the space usage of the data structure prior to any insertions 
is dominated by the time and space used to sample and store the filters in $\F_1$.
By the assumption that a filter from $\LSF$ can be sampled in $O(d)$ operations and stored using $O(d)$ words, 
we get a space and time bound on the initialization operation of
\begin{equation*}
	O(d \kappa_1 m_1) = O\left(dn^{\min(\rho_q, \rho_u)} \frac{p_{1} \log(n)}{\log(p_{q}/p_{2})}\right).
\end{equation*}
Importantly, this bound also holds for the running time of the filter evaluation algorithm, that is, 
the preprocessing time required for constant time generation of the next element in the list of filters in $\F$ containing a point.
In the following analysis of the update and query time we will temporarily ignore the running time of the filter evaluation algorithm.

The expected time to insert or delete a point is dominated by the number of update filters in $\F$ that contains it.
The probability that a particular update filter in $\F$ contains a point is given by $p_{u}^{a}$.
Using a standard upper bound on the binomial coefficient we get that $m = O(e^{\tau}/p_{1}^{a})$ resulting in an expected update time of
\begin{equation*}
O(m p_{u}^{a} + d) = O(n^{\rho_{u}} (p_{u}/p_{1}) e^{\tau} + d).
\end{equation*}
In the worst case where every data point is at distance greater than $cr$ from the query point and has collision probablity $p_2$
the expected query time can be upper bounded by
\begin{equation*}
O(mp_{q}^{a} + dnmp_{2}^{a}) = O(n^{\rho_q}e^{\tau}(p_{q}/p_{1} + d)). 
\end{equation*}
With respect to the correctness of the query algorithm, if a near neighbor $y$ to the query point $x$ exists in $P$,
then it is found by the query algorithm if $(x, y)$ is contained in a filter in $\F_{1}^{\otimes \tau}\!$ as well as in a filter in $\F_2$.
By Lemma \ref{lem:tensoring} the first event happens with probability at least $1/2$ and by the choice of $m_2$, 
the second event happens with probability at least $1 - (1-p_{1}^{\kappa_2})^{p_{1}^{\kappa_2}} \geq 1 - 1/e$.
From the independence between $\F_1$ and $\F_2$ we can upper bound the failure probability $\delta \leq (1/2)(1+1/e)$.
This completes the proof of Theorem \ref{thm:lsf}.
\section{Appendix: Gaussian filters} \label{app:gaussian}
In this section we upper and lower bound the probability mass in the tail of the bivariate standard normal distribution
when the correlation between the two standard normals is at most $\beta$ (upper bound) or at least $\alpha$ (lower bound).
We make use of the following upper and lower bounds on the univariate standard normal as well as an upper bound for the multivariate case.
\begin{lemma}[Follows Szarek \& Werner \cite{szarek1999}] \label{lem:univariatebounds}
	Let $Z$ be a standard normal random variable. Then, for every $t \geq 0$ we have that
	\small
	\begin{equation*}
		\frac{1}{\sqrt{2\pi}}\frac{1}{t+1}e^{-t^{2}/2} \leq \Pr[Z \geq t] \leq \frac{1}{\sqrt{\pi}}\frac{1}{t+1}e^{-t^{2}/2}. 
	\end{equation*}
	\normalsize
\end{lemma}
\begin{lemma}[Lu \& Li \cite{lu2009}]\label{lem:luli}
Let $z$ be a $d$-dimensional vector of i.i.d. standard normal random variables 
and let $D \subset \real^d$ be a closed convex domain that does not contain the origin.
Let~$\Delta$ denote the Euclidean distance to the unique closest point in $D$, then we have that 
\begin{equation*}
	\Pr[z \in D] \leq e^{-\Delta^{2}/2}.
\end{equation*}
\end{lemma}
\begin{lemma}[Tail upper bound] \label{lem:upper}
	For $\alpha, \lambda, t, \beta$ satisfying \mbox{$0 < \alpha < 1$}, $-1 \leq \lambda \leq 1$, $t > 0$, and $-1 < \beta < \alpha$ 
	every pair of standard normal random variables $(X, Y)$ with correlation $\beta' \leq \beta$ satisfies 
\begin{equation*}
	\Pr[X \geq t \land Y \geq \alpha^{\lambda} t] \leq e^{-\Delta^{2}/2}
\end{equation*}
where $\Delta^2 = (1 + \frac{(\alpha^{\lambda} - \beta)^2}{1-\beta^2})t^2$.    
\begin{proof}
For $\beta' = -1$ the result is trivial. 
For values of $\beta'$ in the range \mbox{$-1 < \beta' \leq \beta$} we use the $2$-stability of the normal distribution to analyze a tail bound for 
$(X, Y)$ in terms of a Gaussian projection vector $z = (Z_1, Z_2)$ applied to unit vectors \mbox{$x, y \in \real^2$}.
That is, we can define \mbox{$X = \ip{z}{x}$} and $Y = \ip{z}{x}$ for some appropriate choice of $x$ and $y$.
Without loss of generality we set $x = (1, 0)$ and note that for $\E[XY] = \beta'$ we must have that \mbox{$y = (\beta', \sqrt{1-\beta'^2})$}. 
If we consider the region of $\real^2$ where $z$ satisfies $X \geq t \land Y \geq \alpha^{\lambda} t$ 
we get a closed domain $D$ defined by $z = (Z_1, Z_2)$ such that $Z_1 \geq t$ 
and \mbox{$Z_2 \geq (\alpha^{\lambda} t - \beta' Z_1)/(\sqrt{1-\beta'^2})$}.
The squared Euclidean distance from the origin to the closest point in $D$ at least 
$\Delta^2$ as can be seen by the fact that $\Delta^2$ decreasing in $\beta$.
Combining this observation with Lemma \ref{lem:luli} we get the desired result.
\end{proof}
\end{lemma}
\begin{lemma}[Tail lower bound] \label{lem:lower}
For $\alpha, \lambda, t$ satisfying $0 < \alpha < 1$, $-1 \leq \lambda \leq 1$, 
and $t > 0$ every pair of standard normal random variables $(X, Y)$ with correlation $\alpha' \geq \alpha$ satisfies 
\begin{equation*}
	\Pr[X \geq t \land Y \geq \alpha^{\lambda} t] \geq \frac{e^{-\Delta^{2}/2}}{2\pi(1+ t/\alpha)^{2}}    
\end{equation*}
where $\Delta^2 = (1 + \frac{(\alpha^{\lambda} - \alpha)^2}{1-\alpha^2})t^2$.    
\end{lemma}
\begin{proof}
	For $\alpha' = 1$ the result follows directly from Lemma \ref{lem:univariatebounds}. 
	For $\alpha' < 1$ we use the trick from the proof of Lemma \ref{lem:upper} and define $X = \ip{z}{x}$ and $Y = \ip{z}{x}$ 
	where $x = (1, 0)$ and $y = (\alpha, \sqrt{1-\alpha^2})$ and $z = (Z_1, Z_2)$ is a vector of two i.i.d. standard normal random variables.
	This allows us to rewrite the probability as follows:
	\begin{align*}
		&\Pr[Z_1 \geq t \land \alpha Z_1 + \sqrt{1-\alpha^2}Z_2 \geq \alpha^{\lambda} t] \\
		&= \Pr[Z_1 \geq t]\Pr[\alpha Z_1 + \sqrt{1-\alpha^2}Z_2 \geq \alpha^{\lambda} t \mid Z_1 \geq t] \\
		&\geq \Pr[Z_1 \geq t]\Pr[\alpha t + \sqrt{1-\alpha^2}Z_2 \geq \alpha^{\lambda} t]  
	\end{align*}
	By the restrictions on $\alpha$ and $\lambda$ we have that \mbox{$(\alpha^{\lambda} - \alpha)t / \sqrt{1-\alpha^2} \leq t/\alpha$}.
	The result follows from applying the lower bound from Lemma \ref{lem:univariatebounds} and noting that the bound is increasing in $\alpha$.
\end{proof}
\subsection{Space-time tradeoffs on the unit sphere}
Summarizing the bound from the previous section, the family $\mathcal{G}$ from Lemma \ref{lem:gaussianlsf} satisfies that
\begin{align*}
	p_1 &\geq \frac{e^{-(1 + \frac{(\alpha^{\lambda} - \alpha)^2}{1-\alpha^2})t^2/2}}{2\pi(1+ t/\alpha)^{2}} \\
	p_2 &\leq e^{-(1 + \frac{(\alpha^{\lambda} - \beta)^2}{1-\beta^2})t^2/2} \\
	p_q &\leq e^{-\alpha^{2\lambda}t^{2}/2} \\
	p_u &\leq e^{-t^{2}/2}.
\end{align*}

We combine the Gaussian filters with Theorem \ref{thm:lsf} to show that we can solve the $(\alpha, \beta)$-similarity problem efficiently 
for the full range of space/time tradeoffs, even when $\alpha, \beta$ are allowed to depend on $n$, 
as long as the gap $\alpha - \beta$ is not too small. 
\begin{customthm}{\ref{thm:spherevanilla}}\label{thm:sphere}
	For every choice of $0 \leq \beta < \alpha < 1$ and $\lambda \in [-1, 1]$ 
	we can construct a fully dynamic data structure that solves the $(\alpha, \beta)$-similarity problem in $(\sphere{d}, \ip{\cdot}{\cdot})$.
	Suppose that $\alpha - \beta \geq (\ln n)^{-zeta}$ for some constant $zeta < 1/2$, 
	that satisfies the guarantees from Theorem \ref{thm:lsfvanilla} with exponents  
	$\rho_q = \left. \frac{(1-\alpha^{1+\lambda})^2}{1-\alpha^2}  \middle/  \frac{(1- \alpha^{\lambda} \beta)^2}{1-\beta^2} \right.$ and  
	$\rho_u = \left. \frac{(\alpha^{\lambda} - \alpha)^2}{1-\alpha^2}  \middle/  \frac{(1-\alpha^{\lambda} \beta)^2}{1-\beta^2} \right.$. 
\end{customthm}
\begin{proof}
	Assuming that $\alpha - \beta \geq (\ln n)^{-zeta}$ there exists a constant $\varepsilon > 0$ where by setting
	the parameter $t$ of $\mathcal{G}$ such that $t^2 / 2 = \frac{1-\beta^2}{(1-\alpha^{\lambda} \beta)^2}(\ln n)^{\varepsilon}$
	the family of filters satisfies the assumptions in Theorem \ref{thm:lsf} while guaranteeing 
	that the second term in $\rho_q$ and $\rho_u$ from Lemma \ref{lem:gaussianlsf} are $o(1)$.
\end{proof}
\begin{remark}
	Theorem \ref{thm:sphere} aims for simplicity and generality while allowing $\alpha$ and $\beta$ to depend on $n$. 
	For specific values of $\alpha, \beta, \lambda$ it is easy to find better bounds on the probabilties 
	(e.g. the bounds by Savage \cite{savage1962}) and to adjust $t$ in Lemma \ref{lem:gaussianlsf} 
	to avoid powering (setting $\kappa_{1} = 1,  \kappa_{2} = 0$) in the LSF framework.
\end{remark}
\section{Appendix: Approximate feature maps, characteristic functions, and Bochner's Theorem} \label{app:characteristic}
We begin by defining what a characteristic function is and listing some properties that are useful for our application.
More information about characteristic functions can be found in the books by Lukacs~\cite{lukacs1970} and Ushakov~\cite{ushakov1999}.
\begin{lemma}[\cite{lukacs1970, ushakov1999}] \label{lem:characteristicproperties}
	Let $Z$ denote a random variable with distribution function $\mu$.
	Then the characteristic function $k(\Delta)$ of $Z$ is defined as 
		\begin{equation*}
			k(\Delta) = \int_{-\infty}^{\infty}\mu(t)e^{i \Delta t} dt
		\end{equation*}
	and it has the following properties:
	\begin{itemize}
		\item[-] A distribution function is symmetric if and only if its characteristic function is real and even.
		\item[-] Every characteristic function $k(\Delta)$ is uniformly continuous, has $k(0) = 1$, and $|k(\Delta)| \leq 1$ for all real $\Delta$.
		\item[-] Suppose that $k(\Delta)$ denotes the characteristic function of an absolutely continuous distribution
			     then \mbox{$\lim_{\Delta \rightarrow \infty}|k(\Delta)| = 0$}.
		\item[-] Let $X$ and $Y$ be independent random variables
			with characteristic functions $k_{X}$ and $k_{Y}$. 
			Then the characteristic function of $Z = (X,Y)$ is given by 
			$k(x, y) = k_{X}(x) k_{Y}(y)$. 
	\end{itemize}
\end{lemma}

Bochner's Theorem reveals the relation between characteristic functions and the class of real-valued functions 
$k(x, y)$ that admit a feature space representation $k(x, y) = \ip{\phi(x)}{\phi(y)}$
\begin{theorem}[Bochner's Theorem \cite{rudin1990}]
	A function $k : \real^d \times \real^d \to [0,1]$ is positive definite if and only if it can be written on the form
	\begin{equation*}
		k(x, y) = \int_{\real^d} \mu(v) e^{i \ip{v}{x-y}}  dv
	\end{equation*}
	where $\mu$ is the probability density function of a symmetric distribution.
\end{theorem}

Rahimi \& Recht's \cite{rahimi2007} family of approximate feature maps $\RR$ is constructed from Bochner's Theorem by making use of Euler's Theorem as follows:
\begin{align*}
	k(x, y) &= \int_{\real^d} \mu(v) e^{i \ip{v}{x-y}}  dv \\
			  &= \int_{\real^d} \mu(v)( \cos(\ip{v}{x-y}) + i \sin(\ip{v}{x-y})) dv \\
	&= \E_{v}[\cos(\ip{v}{x-y})]  \\ 
	&= \E_{v, b}[\cos(\ip{v}{x-y}) + \cos(\ip{v}{x} + \ip{v}{y}+ 2b)]  \\ 
	&= 2\E_{v, b}[\cos(\ip{v}{x} + b) \cdot \cos(\ip{v}{y} + b)]. 
\end{align*}
Where the third equality makes use of the fact that $k(x, y)$ is real-valued to remove the complex part of the integral
and the fifth equality uses that $2\cos(x)\cos(y) = \cos(x + y) + \cos(x - y)$.

Now that we have an approximate feature map onto the sphere for the class of shift-invariant kernels, 
we will take a closer look at what functions this class contains, and what their applications are for similarity search. 
Given an arbitrary similarity function, we would like to be able to determine whether it is indeed a characteristic function.
Unfortunately, there are no known simple techniques for answering this question in general.
However, the machine learning literature contains many applications of different shift-invariant kernels \cite{scholkopf2002} 
and many common distributions have real characteristic functions (see Appendix B in \cite{ushakov1999} for a long list of examples).
Characteristic functions are also well studied from a mathematical perspective \cite{lukacs1970, ushakov1999},
and a number of different necessary and sufficient conditions are known.
A classical result by Pólya \cite{polya1949} gives simple sufficient conditions for a function to be a characteristic function.
Through the vectorization property from Lemma \ref{lem:characteristicproperties}, 
Pólya's conditions directly imply the existence of a large class of similarity measures on $\real^d$ that can fit into the above framework. 
\begin{theorem}[Pólya \cite{polya1949}]
	Every even continuous function $k : \real \to \real$ satisfying the properties
	\begin{itemize}
		\item[-] $k(0) = 1$ 
		\item[-] $\lim_{\Delta \to \infty}k(\Delta) = 0$ 
		\item[-] $k(\Delta)$ is convex for $\Delta > 0$ 
	\end{itemize}
	is a characteristic function.
\end{theorem}
Based on the results of Section \ref{sec:ls} one could hope for the existence of characteristic functions of the form 
$k(\Delta) = e^{-|\Delta|^s}$ for $s > 2$ but it is known that such functions cannot exist \cite[Theorem D.8]{benyamini1998}.
Furthermore, Marcinkiewicz \cite{marcinkiewicz1939} shows that a function of the form $k(\Delta) = \exp(-\poly(\Delta))$ 
cannot be a characteristic function if the degree of the polynomial is greater than two.

We state a more complete, constructive version of Lemma \ref{lem:rahimirecht} as well as the proof here.
\begin{lemma} \label{lem:rahimirecht2}
	Let $k$ be a real-valued characteristic function with associated distribution function $\mu$ and let $l$ be a positive integer.
	Consider the family of functions $\RR \subseteq \{ v \mid v \colon \real^{d} \to \sphere{l} \}$ 
	where a randomly sampled function $v$ is defined by, independently for $j = 1,\dots,l$, 
	sampling $v$ from $\mu$ and $b$ uniformly on $[0, 2\pi]$, letting $\hat{v}(x)_{j} = \sqrt{(2/l)}\cos(\ip{v}{x} + b)$
	and normalizing $v(x)_{j} = \frac{\hat{v}(x)}{\norm{\hat{v}(x)}}$.
	The family $\RR$ has the property that for every $x, y \in \real^d$ and $\varepsilon > 0$ we have that 
	\begin{equation*}
		\Pr_{v \sim \RR}[|\ip{v(x)}{v(y)} - k(x, y)| \geq \varepsilon] \leq 6e^{-l \varepsilon^2 / 128}.
	\end{equation*}
\end{lemma}
\begin{proof}
	Since $l \cdot \hat{v}(x)_j \hat{v}(y)_j$ is bounded between $2$ and $-2$, and we have independence for different values of $j$, 
	Hoeffding's inequality \cite{hoeffding1963} can be applied to show that for every fixed pair of points $x, y$
	and $\hat{\varepsilon} > 0$ it holds that
	\begin{equation*}
		\Pr[|\ip{\hat{v}(x)}{\hat{v}(y)} - k(x, y)| \geq \hat{\varepsilon}] \leq 2e^{-l\hat{\varepsilon}^{2}/8}.
	\end{equation*}
	From the properties of characteristic functions we have that $k(x, x) = 1$ and $k(x, y) \leq 1$.
	The bound on the deviation of 
	\begin{equation*}
		\ip{v(x)}{v(y)} = \frac{\ip{\hat{v}(x)}{\hat{v}(y)}}{\sqrt{\ip{\hat{v}(x)}{\hat{v}(x)}\ip{\hat{v}(y)}{\hat{v}(y)}}}
	\end{equation*}
	from $k(x, y)$ follows from setting $\hat{\varepsilon} = \varepsilon/4$ 
	and using a union bound over the probabilities that the deviation of one of the inner products is too large.
\end{proof}
Combining the approximate feature map onto the unit sphere with Theorem \ref{thm:spherevanilla} we obtain the following:
\begin{theorem} \label{thm:characteristiclsf}
	Let $k \colon \real^d \to \real$ be a characteristic function and define the similarity measure $S(x, y) = k(x-y)$.
	Assume that we have access to samples from the distribution associated with $k$, 
	then Theorem \ref{thm:sphere} holds with $(\sphere{d}, \ip{\cdot}{\cdot})$ replaced by $(\real^d, S)$.
\end{theorem}
\begin{proof}
	According to Lemma \ref{lem:rahimirecht2}, we can set $l = n^{o(1)}$ to obtain a map $v \colon \real^d \to \sphere{l}$
	such that the the inner product on $\sphere{l}$ preserves the pairwise similarity between $n^{O(1)}$ 
	points with additive error $\varepsilon = o(1)$.
	This map has a space and time complexity of $O(dl) = dn^{o(1)}$.
	After applying $v$ to the data we can solve the $(\alpha, \beta)$-similarity problem on $(\real^d, k(x-y))$ by solving the 
	$(\alpha - \varepsilon, \beta + \varepsilon)$-similarity problem on $(\sphere{d}, \ip{\cdot}{\cdot})$.
	We can use Theorem \ref{thm:sphere} to construct a fully dynamic data structure for solving this problem, 
	adjusting the parameter $\lambda$ so that it continues to lie in the admissible range.
	The space and time complexities follow.
\end{proof}
\section{Appendix: Proof of tradeoff lower bound} \label{app:lowertradeoffproof}
Consider $\rho_q = \frac{\log(p_{q} / p_{1})}{\log(p_{q} / p_{2})}$.
Subject to the (implicit) LSF constraint that $p_q, p_u > p_1 > p_2 > 0$ we see that $\rho_q$ is minimized by setting 
$p_q, p_2$ as small as possible and $p_1$ as large as possible.
We will therefore derive lower bounds on $p_q, p_2$ and an upper bound on $p_1$.
For every value of $p_1$ and $p_2$ we minimize $\rho_q, \rho_u$ by choosing $p_q$ as small as possible.

For a random point $x \in \cube{d}$ it must hold that $\Pr_{\LSF}[x \in Q] = |Q|/2^{d}$.
This implies the existence of a fixed point $y \in \cube{d}$ with the property that $\Pr_{\LSF}[y \in Q] \geq |Q|/2^{d}$.
A regular filter family must therefore satisfy that $p_q \geq |Q|/2^{d}$ and $p_u \geq |U|/2^{d}$. 
Let $\lambda$ be defined as in Lemma \ref{lem:odhyper} then by a similar argument we have that $p_2 \geq (U/2^d)^{1+\alpha^{2\lambda}}$. 

In order to upper bound $p_1$ we make use of Lemma \ref{lem:odhyper} together with the following lemma that follows directly 
from an application of Hoeffding's inequality~\cite{hoeffding1963}.
\begin{lemma}\label{lem:correlationconcentration}
For every $0 < \varepsilon < (1-\alpha)/2$ we have that
\begin{equation*}
	\Pr_{\substack{(x, y) \\ \alpha+\varepsilon\text{-correlated}}}
	\left[\frac{1}{d}\sum_{i=1}^{d} x_i y_i \leq \alpha\right] \leq e^{-\varepsilon^{2}d/2}.
\end{equation*}
\end{lemma}
In the following derivation, assume that $\alpha, \varepsilon$ satisfies $0 < \varepsilon < (1-\alpha)/2$, 
let $x, y$ denote randomly $(\alpha + \varepsilon)$-correlated vectors in $\cube{d}$, 
and assume that $\alpha + \varepsilon \leq \alpha^{\lambda} \leq 1/(\alpha + \varepsilon)$, then
\begin{align*}
	&(|U|/2^{d})^\frac{1 + \alpha^{2\lambda} - 2 \alpha^{\lambda} (\alpha + \varepsilon)}{1 - (\alpha+\varepsilon) ^{2}} 
	\geq \Pr[x \in Q, y \in U] \\
	&\quad \geq \Pr[x \in Q, y \in U \mid \ip{x}{y} \geq \alpha] \Pr[\ip{x}{y} \geq \alpha] \\
	&\quad \geq  p_1 (1 - e^{-\varepsilon^{2} d/2})   
\end{align*}
Summarizing the bounds:
\begin{align*}
	p_1 &\leq \frac{(|U|/2^{d})^\frac{1 + \alpha^{2\lambda} - 2\alpha^{\lambda}(\alpha+\varepsilon)}{1-\alpha^{2}}}{1-e^{-\varepsilon^{2}d/2}} \\ 
	p_2 &\geq (|U|/2^{d})^{1 + \alpha^{2\lambda}}\\ 
	p_q &\geq |Q|/2^{d} \\
	p_u &\geq |U|/2^{d}.
\end{align*}

When minimizing $\rho_q$ we have that \mbox{$\log(p_{q}/p_{2}) = -\log(|U|/2^{d})$}.
Setting $\varepsilon = 2\sqrt{\ln(d) /d}$ results in 
$\log(1/p_1) \geq -\frac{1+\alpha^{2\lambda} - 2\alpha^{\lambda}(\alpha + \varepsilon)}{1-\alpha^{2}}\log(|U|/2^{d}) - O(1/d^{2})$.
Putting things together:
\begin{align*}
	\frac{\log(p_q / p_1)}{\log(p_q / p_2)} &\geq 
	-\frac{\alpha^{2\lambda} \log(|U|/2^{d})}{\log(|U|/2^d)} \\ 
	&\quad +  
	\frac{\frac{1+\alpha^{2\lambda} - 2\alpha^{\lambda}(\alpha + \varepsilon)}{1-\alpha^{2}}\log(|U|/2^{d})
	+ O(1/d^{2})}{\log(|U|/2^d)} \\
	&= \frac{(1 - \alpha^{1 + \lambda})^{2} -  2\alpha^{\lambda}\varepsilon}{1-\alpha^{2}} + \frac{ O(1/d^{2})}{\log(|U|/2^d)} \\
	&= \frac{(1 - \alpha^{1 + \lambda})^{2}}{1-\alpha^{2}} - O(\sqrt{\log(d)/d}).
\end{align*}
The derivation of the lower bound for $\rho_u$ is almost the same and the resulting expression is
\begin{equation*}
	\frac{\log(p_u / p_1)}{\log( p_q / p_2)} \geq \frac{(\alpha^{\lambda} - \alpha)^{2}}{1-\alpha^{2}} - O(\sqrt{\log(d)/d}).
\end{equation*}
\section{Appendix: Comparison to Kapralov} \label{app:kapralov}
Kapralov uses $\alpha$ to denote a parameter controlling the space-time tradeoff 
for his solution to the $(r, cr)$-near neighbor problem in Euclidean space.
For every choice of tradeoff parameter $\alpha \in [0,1]$, 
assuming that \mbox{$c^2 \geq 3(1-\alpha)^2 - \alpha^2 + \varepsilon$} for arbitrarily small constant $\varepsilon > 0$, 
Kapralov \cite{kapralov2015} obtains query and update exponents
\begin{align*}
	\rho_q &= \frac{4(1-\alpha)^2}{c^2 + (1-\alpha)^2 - 3\alpha^2}, \\
	\rho_u &= \frac{4\alpha^2}{c^2 + (1-\alpha)^2 - 3\alpha^2}.
\end{align*}
We convert Kapralov's notation to our own by setting $\lambda = 1 - 2\alpha$.
To compare, Kapralov sets $\alpha = 0$ for near-linear space and we set $\lambda = 1$.
We want to write Kapralov's exponents on the form
\begin{equation*}
	\rho_q = \frac{c^2(1+\lambda)^2}{(c^2 + \lambda)^2 + x}, \quad \rho_u = \frac{c^2(1-\lambda)^2}{(c^2 + \lambda)^2 + x}
\end{equation*}
for some $x$ that we will proceed to derive.
We have that $(1-\alpha)^2 = (1+\lambda)^{2}/4$ and $\alpha^2 = (1-\lambda)^{2}/4$.
Multiplying the numerator and denominator in Kapralov's exponents by $c^2$ we can write Kapralov's exponents as
\begin{align*}
	\rho_q &= \frac{c^2(1+\lambda)^2}{c^4 + c^2 (1+\lambda)^2 / 4 - 3c^2 (1-\lambda)^2 / 4}, \\
	\rho_u &= \frac{c^2(1-\lambda)^2}{c^4 + c^2 (1+\lambda)^2 / 4 - 3c^2 (1-\lambda)^2 / 4}.
\end{align*}
We have that 
\begin{align*}
	x &= c^4 + c^2 (1+\lambda)^2 / 4 - 3c^2 (1-\lambda)^2 / 4 - (c^2 + \lambda)^2\\ 
	  &= -c^2(1+\lambda^2)/2 - \lambda^2. 
\end{align*}
For every choice of $\lambda \in [-1,1]$, and under the assumption that $c^2 \geq (1+\lambda)^2 / 2 + \lambda + \varepsilon$ 
for an arbitrarily small constant $\varepsilon > 0$, this allows us to write Kapralov's exponents as 
\begin{align*}
	\rho_q &= \frac{c^2(1+\lambda)^2}{(c^2 + \lambda)^2 - c^2(1+\lambda^2)/2 - \lambda^2}, \\
	\rho_u &= \frac{c^2(1-\lambda)^2}{(c^2 + \lambda)^2 - c^2(1+\lambda^2)/2 - \lambda^2}.
\end{align*}
To compare Kapralov's result against our own for search in $\ell_s$-spaces we consider the exponents from Theorem \ref{thm:lsvanilla},
ignoring additive $o(1)$ terms:
\begin{equation*}
		\rho_q = \frac{c^s (1 + \lambda)^2}{(c^s + \lambda)^2}, \quad 
		\rho_u = \frac{c^s (1 - \lambda)^2}{(c^s + \lambda)^2}.		
\end{equation*}
Setting $\lambda = 1$ we obtain a data structure that uses near-linear space and we get a query exponent $\rho_q = 16/25$ 
while Kapralov obtains an exponent of $\rho_q = 16/20$, ignoring $o(1)$ terms.
At the other end of the tradeoff, setting $\lambda = -1$, we get a data structure with query time $n^{o(1)}$ and update exponent $\rho_u = 16/9$ 
while Kapralov gets an update exponent of $\rho_u = 4$, again ignoring additive $o(1)$ terms.

The assumption made by Kapralov that $c^2 \geq (1+\lambda)^2 / 2 + \lambda + \varepsilon$ means that in the case of a 
near-linear space data structure ($\lambda = 1$) sublinear query time can only be obtained for $c > \sqrt{3}$.
In contrast, Theorem \ref{thm:lsvanilla} gives sublinear query time for every constant $c > 1$.
\section{Appendix: Details about dynamization and the model of computation}
In order to obtain fully dynamic data structures we apply a powerful dynamization 
result of Overmars and Leeuwen \cite{overmars1981} for decomposable searching problems.
Their result allows us to turn a partially dynamic data structure into a fully dynamic data structure, 
supporting arbitrary sequences of queries and updates, at the cost of a constant factor in the space and running time guarantees. 
Suppose we have a partially dynamic data structure that solves the $(r, cr)$-near neighbor problem on a set of $n$ points.
By partially dynamic we mean that, after initialization on a set $P$ of $n$ points, 
the data structure supports $\Theta(n)$ updates without changing the query time by more than a constant factor.
Let $T_{q}(n)$, $T_{u}(n)$, and $T_{c}(n)$ denote the query time, update time, and construction time of such a data structure containing $n$ points.
Then, by Theorem 1 of Overmars and Leeuwen \cite{overmars1981}, there exists a fully dynamic version of the data structure 
with query time $O(T_{q}(n))$ and update time $O(T_{u}(n) + T_{c}(n)/n)$ that uses only a constant factor additional space.   
The data structures presented in this paper, as well as most related constructions from the literature, 
have the property that $T_{c}(n)/n = O(T_{u}(n))$, allowing us to go from a partially dynamic 
to a fully dynamic data structure ``for free'' in big O notation.  

In terms of guaranteeing that the query operation solves the $(r, cr)$-near neighbor problem on the set of points $P$ 
currently inserted into the data structure, we allow a constant failure probability $\delta < 1$, 
typically around $1/2$, and omit it from our statements.
We make the standard assumption that the adversary does not have knowledge of the randomness used by the data structure.
Say we have a data structure with constant failure probability and a bound on the expected space usage. 
Then, for every positive integer $T$ we can create a collection of $O(\log T)$ independent repetitions of the data structure such that
for every sequence of $T$ operations it holds with high probability in $T$ that the space usage will never 
exceed the expectation by more than a constant factor and no query will fail.

\subsection{Model of computation}
We use the standard word RAM model as defined by Hagerup \cite{hagerup1998} with a word size of $\Theta(\log n)$ bits. 
Unless otherwise stated, we make the assumption that a point in $(X, D)$ can be stored in $d$ words and 
that the dissimilarity between two arbitrary points can be computed in $d$ operations where $d$ is a positive integer 
that corresponds to the dimension in the various well-studied settings mentioned in the main text.
Furthermore, when describing framework-based solutions to the $(r, cr)$-near neighbor problem, 
we make the assumption that we can sample, evaluate, and represent elements from $\LSF$ and $\LSH$ 
with neglible error using space and time $dn^{o(1)}$.

Many of the LSH and LSF families rely on random samples from the standard normal distribution.
We will ignore potential problems resulting from rounding due to the fact that our model only supports finite precision arithmetic.
This approach is standard in the literature and can be justified by noting that the error introduced by rounding is neglible.
Furthermore, there exists small pseudorandom standard normal distributions that support sampling 
using only few uniformly distributed bits as noted by Charikar \cite{charikar2002}. 
In much of the related literature the model of computation is left unspecified and statements 
about the complexity of solutions to the $(r, cr)$-near neighbor problem are usually made with respect to particular operations 
such as the hash function computations, distance computations, etc., leaving out other details \cite{indyk1998, har-peled2012}.  
\section{Addendum: An improved framework}
The LSF framework in Theorem \ref{thm:lsf} suffers from large lower-order terms that depend on the $(r, cr, p_1, p_2, p_q, p_u)$-sensitivity properties of $\LSF$.
With the parameterization in Appendix \ref{app:framework} the framework uses $O(\tau n^{\min(\rho_q, \rho_u)})$ filters from $\LSF$ where $\tau \leq \log(1/p_1)/\log(\min(p_q, p_u)/p_1)$.
In addition, the query and update time have a multiplicative factor $e^\tau$ which can potentially be very large and where we have to assume explicitly that $e^\tau = n^{o(1)}$.
We will use a combination of techniques in recent work on set similarity search~\cite{christiani2017set} and fast locality-sensitive hashing frameworks~\cite{dahlgaard2017fast, christiani2017fast} to give an improved LSF framework with more precise complexity bounds. 

The data structure produced by the framework follows the high-level approach as outlined in Section \ref{sec:lsfds}: queries and updates are mapped to a collection of buckets that are searched for similar points in the case of a query, or updated to store a reference to the point in the case of an update.
Let $V \colon X \to R$ denote the mapping from query points to buckets and $W \colon X \to R$ denote the corresponding map for updates. 
The set of buckets $V(x)$ will be identified by the ``survivors'' of $w$ branching processes through $k$ collections of $m$ filters, similarly to the Chosen Path algorithm~\cite{christiani2017set}.

The data structure is initialized by sampling $k$ collections of $m$ filters.
We will use the notation $Q_{i,j}$ ($U_{i,j}$) to denote the $j$th query (update) filter in the $i$th collection.
For $i = 1, \dots, k$ let $h_i \colon [w] \times [m]^{i} \to [0,1]$ denote a pairwise independent random hash function.
Let $\lambda \in [0,1]$ be a parameter to be determined later and let $\circ$ denote vector concatenation,
then the locality-sensitive map $V$ is defined recursively as follows:
\begin{equation*}
	V_{i}(x) = 
	\begin{cases}
		\{ p \circ j \mid p \in V_{i-1}(x) \land h_{i}(p \circ j) < \lambda \land x \in Q_{i,j} \} & \text{if $i > 0$} \\ 
		[w]	& \text{if $i = 0$.}
	\end{cases}
\end{equation*}
The map $W$ is defined in the same way except it uses $U_{i,j}$ instead of $Q_{i,j}$.

\paragraph{Properties.}
To show that the maps $V, W$ provide an efficient solution to the $(r, cr)$-near neighbor problem we need to show the following:
\begin{itemize}
	\item An upper bound on the expected size of $V(x)$ and $W(x)$ to bound the expected number of buckets probed during queries/updates.
	\item An upper bound on the expected size of $V(x) \cap W(y)$ when $\dist(x, y) > cr$ to bound the expected number of distant points that will be encountered during the linear scan part of the query algorithm.
	\item That $V(x) \cap W(y)$ is non-empty with constant probability when $\dist(x, y) \leq r$ to guarantee that the query algorithm encounters a point at distance at most $r$ with constant probability, provided such a point exists. 
\end{itemize}
By the independence between the different levels $i = 1, \dots, k$ in the branching process we have that
\begin{equation*}
	\E[|V_i(x)|] = \E[|V_{i-1}(x)|] m \lambda p_q = w (m \lambda p_q)^{i}.
\end{equation*}
Given $x, y \in X$ define $Z_i = V_i(x) \cap W_i(y)$. Define $p = \Pr[x \in Q, y \in U]$ where $(Q, U)$ is sampled from $\LSF$.
The expected number of collisions between $x$ and $y$ at level $i$ is then given by
\begin{equation*}
	\E[|Z_i|] = \E[|Z_{i-1}|] m \lambda p = w(m \lambda p)^{i}.  
\end{equation*}
To show correctness of the scheme we will use Chebyshev's inequality to show that with constant probability we have $|Z_i| > 0$ for points $x, y$ with $\dist(x, y) \leq r$. 
We proceed by upper bounding $\E[|Z_i|^2]$ in order to bound the variance $\Var[|Z_i|] = \E[|Z_i|^2] - \E[|Z_i|]^2$.
To ease the derivation we define $Y_{p,j} = \1\{ h_i(p \circ j) < \lambda \land (x,y) \in (Q_{i,j}, U_{i,j}) \}$ where we suppress the subscript $i$.
Without loss of generality we can assume that $p = p_1$ since $\dist(x, y) \leq r$.
\begin{align*}
\E[|Z_i|^2] &= \E\left[\left(\sum_{p \in Z_{i-1}} \sum_{j \in [m]} Y_{p,j} \right)^2 \right] \\
	        &= \E\left[ \sum_{p \neq p'} \left(\sum_{j \neq j'} Y_{p,j} Y_{p',j'} +  \sum_{j = j'} Y_{p,j} Y_{p',j'} \right)  \right] \\
	&\qquad +  \E\left[ \sum_{p = p'} \left(\sum_{j \neq j'} Y_{p,j} Y_{p',j'} +  \sum_{j = j'} Y_{p,j} Y_{p',j'} \right)  \right] \\	
	&= \E[|Z_{i-1}|^2 - |Z_{i-1}|]((m^2 - m) \lambda^2 p_1^{2} + m\lambda^2 p_1) \\
	&\qquad + \E[|Z_{i-1}|]((m^2 - m)\lambda^2 p_{1}^{2} + m\lambda p_1) \\
	&\leq \E[|Z_{i-1}|^2](1 + 1/m p_1)(m \lambda p_1)^2 + \E[|Z_{i-1}|]m \lambda p_1. 
\end{align*}
Since $i \leq k$ if we set $m \geq \ln(1 + \varepsilon) k / p_1$ we have that
\begin{equation*}
	\E[|Z_i|^2] \leq (1+\varepsilon)w^2 (m \lambda p_1)^{2i} + (1+\varepsilon) \E[|Z_i|] \sum_{s = 0}^{i-1} (m \lambda p_1)^{s}.
\end{equation*}
We will set the parameters in order to give a simple upper bound the worst-case performance of the data structure.
The constants can be improved.
\begin{align*}
	\varepsilon &= 1/10, \\
	k &= \lceil \log(n) / \log(p_q / p_2) \rceil, \\ 
	m &= \lceil \ln(1 + \varepsilon) k / p_1 \rceil, \\ 
	\lambda &= 1/(m \min(p_q, p_u)), \\
	w &= \lceil 10 k (\min(p_q, p_u) / p_1)^{k} \rceil.
\end{align*}
We can now bound the variance of $|Z_k|$ as follows:
\begin{align*}
	\Var[|Z_k|] &= \E[|Z_k|^2] - \E[|Z_k|]^2 \\
				&\leq  \E[|Z_k|]^{2}/10 + (1 + 1/10) \E[|Z_k|] \sum_{s = 0}^{k-1} (p_1 / \min(p_q, p_u))^{s} \\
				&\leq  \E[|Z_k|]^{2}/10 + (1 + 1/10) k \E[|Z_k|] 
\end{align*}
where we use the fact that $p_1 \leq \min(p_q, p_u)$.
By Chebyshev's inequality we have that 
\begin{align*}
	\Pr[|Z_{k}| < 0] &\leq \Var[|Z_k|] / \E[|Z_k|]^2 \\
					 &\leq 1/10 + (1 + 1/10) k / \E[|Z_k|] 
\end{align*}
By our parameter setting we have $\E[|Z_k|] = w(m \lambda p_1)^k \geq 10 k$ so $x, y$ collide with probability at least $7/10$ under $V, W$, ensuring correctness.

\subsection{Fast evaluation}
We will use a hashing trick to compute $V_k(x)$ in expected time $O(k \E[|V_k(x)|])$. 
This technique is only briefly mentioned in~\cite{christiani2017set}. 
Observe that for the correctness argument to hold, it suffices that the hash functions $h_1, \dots, h_k$ are sampled independently from a pairwise independent family~\cite{carter1979, wegman1981}.
At the $i$th step in the computation of $V_{k}(x)$ we wish to determine, for each $p \in V_{i-1}(x)$ the set of $j \in [m]$ satisfying $h_i(p \circ j) < \lambda$ and $x \in Q_{i,j}$. 
In order to answer this efficiently we will make use of the property that a pairwise independent hash function can be decomposed as
\begin{equation*}
	h_i(p \circ j) = g_i(p) \oplus f_i(j)
\end{equation*}
where $g_i, f_i$ are pairwise independent and $\oplus$ denotes addition in an abelian group.
For concreteness assume that $g_i, f_i$ map to $b$-bit strings and let $\oplus$ denote the exclusive-or operator.
If we view the $b$-bit output of $h_i(p \circ j)$ as an integer in the set $0, 1,  \dots, 2^{b} - 1$ using the standard base two representation, the original condition $h_i(p \circ j) < \lambda$ can be transformed into the condition $h_i(p \circ j) < \lambda M$ where $M = 2^b - 1$. 
By choosing $b = \Theta(\log n)$ we can with high probability determine whether the condition is satisfied without reading more than $b$ bits, so we can effectively treat the output of the hash function as a real number at the cost of a small increase in the failure probability of the data structure. 

Continuing with the new representation, in order for $h_i(p \circ j) < \lambda M$ we must have that the leading $\kappa = b - \lceil \log_{2}(\lambda M) \rceil - 1$ bits of the output of $g_i(p) \oplus f_i(j)$ is all zeroes. 
Given the leading $\kappa$ bits of $g_i(p)$ we can restrict our attention to $j \in [m]$ with the same value in the leading $\kappa$ bits of $f_i(j)$.
At the beginning of the query algorithm, for each $i \in k$ we determine the subset $J \subseteq [m]$ such that $x \in Q_{i,j}$
We then create a table with $2^\kappa$ linked lists and for each $j \in J$ we append $j$ to the linked list at the table entry given by the leading $\kappa$ bits of $f_i(j)$.
The running time and space usage of preparing these additional data structures is dominated by the complexity of evaluating and storing $O(k^2 / p_1)$ filters from $\LSF$.

Now, given $V_{i-1}(x)$ we can compute $V_{i}(x)$ in expected time $O(m \lambda p_q |V_{i-1}(x)|)$ by, for each $p \in V_i(x)$, looking up the relevant table entry (given by the leading $\kappa$ bits of $g_i(p)$) and verifying whether the elements of the linked list satisfy the hashing condition.
Every element of the linked list found in this way satisfies the hashing condition with constant probability by our setting of $\kappa$.
To implement $g_i$ and $f_i$ we can use simple tabulation hashing~\cite{zobrist1970new}.

One problem remains: long paths $p \in [w] \times [m]^i$ can take super constant time to hash. 
To prevent this we again use hashing to create $O(\log n)$-bit fingerprints of the paths that we work on instead.
A conservative upper bound on the expected time to compute $V_k(x)$ is $O(k \E[|V_k(x)|])$ since $\E[|V_i(x)|]$ is non-decreasing in $i$ and the expected time spent at level $i$ is upper bounded by $O(\E[|V_{i}(x)|])$.
We use the same approach to compute $W_k(x)$.

\subsection{Framework}
We are now ready to state the properties of the new framework.
\begin{theorem}\label{thm:lsf_new}
	Given a $(r, cr, p_q, p_u, p_1, p_2)$-sensitive family $\LSF$ we can construct a fully dynamic data structure that solves the $(r, cr)$-near neighbor problem. Define $k = \lceil \log(n) / \log(p_q / p_2) \rceil$, then:
\begin{itemize}
	\item The data structure uses $O(kn(p_u/p_1)^k)$ words of space in addition to the space required to store $n$ data points and $O( k^2 / p_1)$ filters from $\LSF$.
	\item The query operation uses $O(k^2 (p_q / p_1)^k))$ word-RAM operations, $O(k (p_q / p_1)^k)$ distance computations, and $O(k^2 / p_1)$ filter evaluations.
	\item The update operation uses $O(k^2 (p_u / p_1)^k))$ word-RAM operations and $O(k^2 / p_1)$ filter evaluations.
\end{itemize}
\end{theorem}
Compared to the usual formulation where the query time is stated as $n^{\rho_q + o(1)}$ Theorem \ref{thm:lsf_new} offers a more precise statement of the complexity and can be converted to the other formulation.
The lower order terms are now confined to the multiplicative factor $k$ which is a standard expression that also appears in the LSH framework as $k = \lceil \log(n) / \log(1/p_2) \rceil$ where $p_2$ is an upper bound on the collision probability between pairs of points $x, y$ with $\dist(x, y) > cr$. 
The analysis can be tightened further by not using $k$ as an upper bound for $\sum_{s = 0}^{k-1} (p_1 / \min(p_q, p_u))^{s}$ when bounding the variance, 
but removing the multiplicative dependence on $k$ entirely as in the improved LSH framework~\cite{christiani2017fast} is an interesting open problem.

\chapter{Set similarity search beyond MinHash} \label{ch:sparse}
\sectionquote{From the ashes, a fire shall be woken}
\noindent We consider the problem of approximate set similarity search under Braun-Blanquet similarity $B(x, y) = |x \cap y| / \max(|x|, |y|)$.
The $(b_1, b_2)$-approximate Braun-Blanquet similarity search problem is to preprocess a collection of sets $P$ such that, given a query set $q$, 
if there exists $x \in P$ with $B(q, x) \geq b_1$, then we can efficiently return $x' \in P$ with $B(q, x') > b_2$.

We present a simple data structure that solves this problem with space usage $O(n^{1+\rho}\log n + \sum_{x \in P}|x|)$ and query time $O(|q|n^{\rho} \log n)$ where $n = |P|$ and $\rho = \log(1/b_1)/\log(1/b_2)$.
Making use of existing lower bounds for locality-sensitive hashing by O'Donnell et al.~\cite{odonnell2014optimal} we show that this value of $\rho$ is tight across the parameter space, i.e., for every choice of constants $0 < b_2 < b_1 < 1$. 

In the case where all sets have the same size our solution strictly improves upon the value of $\rho$ that can be obtained through the use of state-of-the-art data-independent techniques in the Indyk-Motwani locality-sensitive hashing framework~\cite{indyk1998} such as Broder's MinHash~\cite{broder1997syntactic} for Jaccard similarity and Andoni et al.'s cross-polytope LSH~\cite{andoni2015practical} for cosine similarity.
Surprisingly, even though our solution is data-independent, for a large part of the parameter space we outperform the currently best data-\emph{dependent} method by Andoni and Razenshteyn~\cite{andoni2015optimal}.
\section{Introduction}
In this paper we consider the approximate set similarity problem or, 
equivalently, the problem of approximate Hamming near neighbor search in sparse vectors. 
Data that can be represented as sparse vectors is ubiquitous --- a typical example is the representation of text documents as \emph{term vectors}, 
where non-zero vector entries correspond to occurrences of words (or shingles).
In order to perform identification of near-identical text documents in web-scale collections, 
Broder et al.~\cite{bro97b,bro97a} designed and implemented \emph{MinHash} (a.k.a.~min-wise hashing), 
now understood as a locality-sensitive hash function~\cite{har-peled2012}.
This allowed approximate answers to similarity queries to be computed much faster than by other methods, 
and in particular made it possible to cluster the web pages of the AltaVista search engine (for the purpose of eliminating near-duplicate search results).
Almost two decades after it was first described, 
MinHash remains one of the most widely used locality-sensitive hashing methods as witnessed by thousands of citations of~\cite{bro97b,bro97a} 
as well as the ACM Paris Kanellakis Theory and Practice Award that Broder shared with Indyk and Charikar in 2012. 

A \emph{similarity measure} maps a pair of vectors to a similarity score in $[0,1]$.
It will often be convenient to interpret a vector $x\in\{0,1\}^d$ as the set $\{ i \; | \; x_i=1\}$.
With this convention the \emph{Jaccard similarity} of two vectors can be expressed as $J(x,y) = |x \cap y|/|x \cup y|$.
In \emph{approximate similarity search} we are interested the problem of searching a data set $P\subseteq \{0,1\}^d$ for a vector 
of similarity at least~$j_1$ with a query vector $q \in \{0,1\}^d$, but allow the search algorithm to return a vector of similarity  $j_2 < j_1$. 
To simplify the exposition we will assume throughout the introduction that all vectors are $t$-sparse, i.e., have the same Hamming weight $t$.

Recent theoretical advances in data structures for approximate \emph{near neighbor} search in Hamming space~\cite{andoni2015optimal} make it possible to beat the asymptotic performance of MinHash-based Jaccard similarity search (using the LSH framework of~\cite{har-peled2012}) in cases where the similarity threshold $j_2$ is not too small.
However, numerical computations suggest that MinHash is always better when $j_2<1/45$. 

In this paper we address the problem: Can similarity search using MinHash be improved \emph{in general}?
We give an affirmative answer in the case where all sets have the same size $t$ by introducing \textsc{Chosen Path}: a simple data-independent search method that strictly improves MinHash, 
and is always better than the data-dependent method of~\cite{andoni2015optimal} when $j_2<1/9$.
Similar to data-independent locality-sensitive filtering (LSF) methods~\cite{becker2016, laarhoven2015, christiani2017framework} our method works 
by mapping each data (or query) vector to a set of keys that must be stored (or looked up).
The name \textsc{Chosen Path} stems from the way the mapping is constructed: As paths in a layered random graph where the vertices at each layer is identified with the set $\{1,\dots,d\}$ of dimensions, 
and where a vector $x$ is only allowed to choose paths that stick to non-zero components $x_i$.
This is illustrated in Figure~\ref{fig:branchingprocess}.
\begin{figure}   
	\centering
	\includegraphics[width=0.5\textwidth]{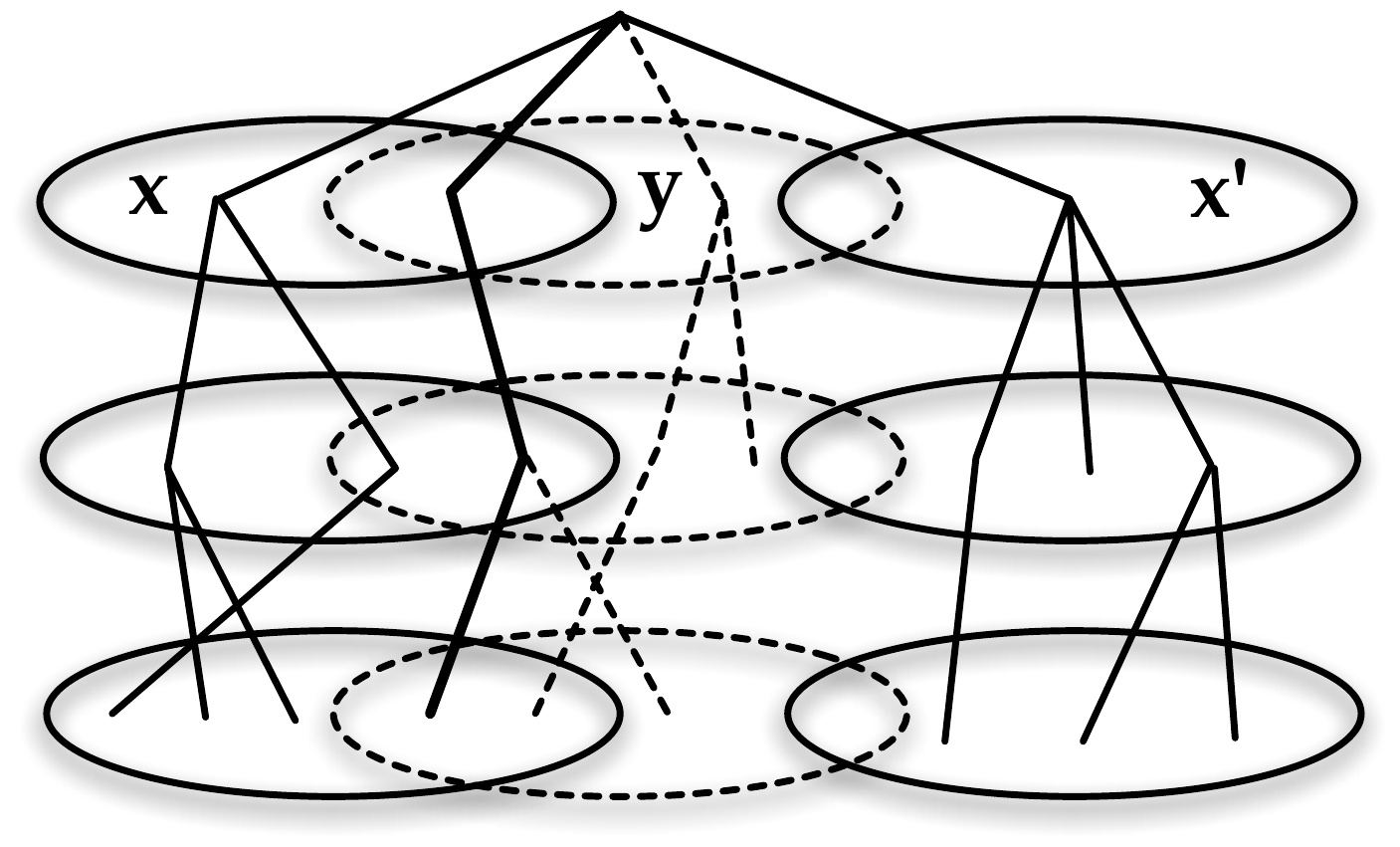}
	\caption{\textsc{ Chosen Path} uses a branching process to associate each vector $x\in \bitcube{d}$ with a set $M_k(x) \subseteq \{1,\dots,d\}^k$ of paths of lengtk~$k$ (in the picture $k=3$). 
	The paths associated with $x$ are limited to indices in the set $\{ i \; | \; x_i=1\}$, represented by an ellipsoid at each level in the illustration. 
	In the example the set sizes are: $|M_3(x)| = 4$ and $|M_3(y)| = |M_3(x')| = 3$. 
	Parameters are chosen such that a query~$y$ that is similar to $x \in P$ is likely to have a common path in $x \cap y$ (shown as a bold line), 
	whereas it shares few paths in expectation with vectors such as $x'$ that are not similar.}
	\label{fig:branchingprocess}
\end{figure}
\subsection{Related Work}
High-dimensional approximate similarity search methods can be characterized in terms of their \emph{$\rho$-value} which is the exponent for which queries can be answered in time $O(dn^\rho)$, 
where $n$ is the size of the set $P$ and $d$ denotes the dimensionality of the space.
Here we focus on the ``balanced'' case where we aim for space $O(n^{1+\rho} + dn)$, 
but note that there now exist techniques for obtaining other trade-offs between query time and space overhead~\cite{andoni2017optimal,christiani2017framework}.

\paragraph{Locality-sensitive hashing methods.}
We begin by describing results for Hamming space, which is a special case of similarity search on the unit sphere (many of the results cited apply to the more general case).
In Hamming space the focus has traditionally been on the $\rho$-value that can be obtained for solutions to the \emph{$(r, cr)$-approximate near neighbor} problem: 
Preprocess a set of points $P \subseteq \bitcube{d}$ such that, given a query point $q$, if there exists $x \in P$ with $\norm{x - q}_{1} \leq r$, then return $x' \in P$ with $\norm{x' - q}_1 < cr$.
In the literature this problem is often presented as the $c$-approximate near neighbor problem where bounds for the $\rho$-value are stated in terms of $c$ and, 
in the case of upper bounds, hold for every choice of $r$, while lower bounds may only hold for specific choices of $r$. 

O'Donnell et al.~\cite{odonnell2014optimal} have shown that the value $\rho=1/c$ for $c$-approximate near neighbor search in Hamming space, 
obtained in the seminal work of Indyk and Motwani~\cite{indyk1998}, is the best possible in terms of $c$ for schemes based on Locality-Sensitive Hashing (LSH). 
However, the lower bound only applies when the distances of interest, $r$ and~$cr$, are relatively small compared to $d$, and better upper bounds are known for large distances.
Notably, other LSH schemes for angular distance on the unit sphere such as cross-polytope LSH~\cite{andoni2015practical} give lower $\rho$-values for large distances.
Extensions of the lower bound of~\cite{odonnell2014optimal} to cover more of the parameter space were recently given in~\cite{andoni2017optimal,christiani2017framework}.
Until recently the best $\rho$-value known in terms of $c$ was $1/c$, 
but in a sequence of papers Andoni et al.~\cite{andoni2014beyond,andoni2015optimal} have shown how to use \emph{data-dependent} LSH techniques to achieve $\rho = 1/(2c-1) + o_n(1)$, 
bypassing the lower bound framework of~\cite{odonnell2014optimal} which assumes the LSH to be independent of data.
\medskip

\paragraph{Set similarity search.}
There exists a large number of different measures of set similarity with various applications for which it would be desirable to have efficient approximate similarity search algorithms~\cite{choi2010survey}.  
Given a measure of similarity assume that we have access to a family $\LSH$ of locality-sensitive hash functions (defined in Section \ref{sparse:sec:preliminaries})
such that for every pair of sets $x, y$ it holds that
\begin{equation*}
	\Pr[h(x) = h(y)] = \simil(x, y). 
\end{equation*}
when $h$ is sampled randomly from $\LSH$.
We will refer to a family of locality-sensitive hash functions with this specific property as a \emph{similarity-sensitive} family. 
Given a similarity-sensitive family we can use the LSH framework to construct a solution for the $(s_1, s_2)$-approximate similarity search problem with exponent $\rho = \log(1/s_1)/\log(1/s_2)$.

Regarding the existence of similarity-sensitive families it was shown by Charikar~\cite{charikar2002} that if the similarity measure $\simil(x, y)$ admits a similarity-sensitive LSH, 
then $1-\simil(x, y)$ must be a metric.
Recently, Chierichetti and Kumar \cite{chierichetti2015} showed that,
given a similarity $\sim$ that admits a similarity-sensitive LSH, 
the transformed similarity $f(\sim)$ will continue to admit an LSH if $f(\cdot)$ is a probability generating function.
The existence of an LSH that admits a similarity measure $\simil$ will therefore give rise to the existence of solutions to the approximate similarity search problem for the much larger class of similarities $f(\simil)$.
However, this still leaves open the problem of finding efficient explicit constructions, and as it turns out,
the property of similarity-sensitive families $\Pr[h(x) = h(y)] = \simil(x, y)$, while intuitively appealing and useful for similarity estimation,
does not necessarily imply that the LSH is optimal for solving the approximate search problem. 
In fact, it was recently shown~\cite{chierichetti2017distortion} that for Braun-Blanquet there does not exist a LSH scheme with $\Pr[h(x) = h(y)] = B(x, y) = |x \cap y| / \max(|x|, |y|)$.
Moreover, it was shown that MinHash achieves a two-approximation to Braun-Blanquet similarity and that this is optimal for LSH schemes. 

The problem of finding tight upper and lower bounds on the $\rho$-value that can be obtained through the LSH framework for data-independent $(s_1, s_2)$-approximate similarity search 
across the entire parameter space $(s_1, s_2)$ remains open for two of the most common measures of set similarity: 
Jaccard similarity $J(x, y) = |x \cap y| / |x \cup y|$ and cosine similarity $C(x, y) = |x \cap y|/\sqrt{|x||y|}$.

A random function from the MinHash family $\LSH_{\text{minhash}}$ hashes a set $x \subseteq \{1, \dots, d \}$
to the first element of $x$ in a random permutation of the set $\{1, \dots, d \}$.
For $h \sim \LSH_{\text{minhash}}$ we have that $\Pr[h(x) = h(y)] = J(x, y)$, yielding an LSH solution to the approximate Jaccard similarity search problem.
For cosine similarity the SimHash family $\LSH_{\text{simhash}}$, introduced by Charikar~\cite{charikar2002}, 
works by sampling a random hyperplane in $\real^{d}$ that passes through the origin and hashing $x$ according to what side of the hyperplane it lies on.
For $h \sim \LSH_{\text{simhash}}$ we have that $\Pr[h(x) = h(y)] = 1 - \arccos(C(x, y))/\pi$, which can be used to derive a solution for cosine similarity, 
although not the clean solution that we could have hoped for in the style of MinHash for Jaccard similarity.
There exists a number of different data-independent LSH approaches~\cite{terasawa2007spherical, andoni2014beyond, andoni2015practical} that improve upon the $\rho$-value of SimHash.
Perhaps surprisingly, it turns out that these approaches yield lower $\rho$-values for the $(j_1, j_2)$-approximate Jaccard similarity search problem compared to MinHash for certain combinations of $(j_1, j_2)$.  
Unfortunately, while asymptotically superior these techniques suffer from a non-trivial $o_{n}(1)$-term in the exponent that only decreases very slowly with $n$.
In comparison, both MinHash and SimHash are simple to describe and have closed expressions for their $\rho$-values.
Furthermore, MinHash and SimHash both have the advantage of being efficient in the sense that a hash function can be represented using space $O(d)$ and the time to compute $h(x)$ is $O(|x|)$. 
\begin{table*}
\caption{Overview of $\rho$-values for similarity search with Hamming vectors of equal weight $t$.}
\label{tab:comparison}
\renewcommand\arraystretch{2}
\scriptsize
\begin{center}
\begin{tabular}{|l|c|c|c|}
\hline
\textbf{{\diagbox{Ref.}{Measure}}} & 
\thead{Hamming\\ $r_1 < r_2$} &
\thead{Braun-Blanquet\\$b_1>b_2$} &  
\thead{Jaccard\\$j_1>j_2$} \\ 
\hline
\hline
Bit-sampling~\cite{indyk1998} & 
$r_1/r_2$ &
$\tfrac{1-b_1}{1-b_2}$ &
$\tfrac{1-j_1}{1+j_1}/\tfrac{1-j_2}{1+j_2}$ \\ 
\hline 
MinHash~\cite{bro97b} & 
$\log \tfrac{1 - r_{1}}{1 + r_{1}} / \log \tfrac{1 - r_{2}}{1 + r_{2}}$ & 
$\log\tfrac{b_1}{2-b_1} / \log\tfrac{b_2}{2-b_2}$ & 
$\log(j_1)/\log(j_2)$ \\
\hline
Cross-poly.~\cite{andoni2015practical} & 
$\frac{r_1}{r_2}\frac{1 - r_{2}/2}{1 - r_{1}/2}$ &
$\tfrac{1-b_1}{1+b_1}/\tfrac{1-b_2}{1+b_2}$ &
$\tfrac{1-j_1}{1+3j_1}/\tfrac{1-j_2}{1+3j_2}$ \\
\hline
Data-dep.~\cite{andoni2015optimal} & 
$\frac{r_1}{r_2}\frac{1}{2 - r_{1}/r_{2}}$ &
$\tfrac{1-b_1}{1 + b_1 - 2 b_2}$ & 
$\frac{(1-j_{1})(1+j_{2})}{1 - j_{1}j_{2} + 3(j_{1}-j_{2})}$ \\
\hline
\textbf{Theorem~\ref{thm:upper}} & 
$\log(1-r_{1}) / \log(1-r_{2})$ &
$\log(b_1)/\log(b_2)$ &
$\log\tfrac{2j_1}{1+j_1}/\log\tfrac{2j_2}{1+j_2}$ \\
\hline
\end{tabular}
\end{center}
\textbf{Notes:} While most results in the literature are stated for a single measure, 
the fixed weight restriction gives a 1-1 correspondence that makes it possible to express the results in terms of other similarity measures.
Hamming distances are normalized by a factor $2t$ to lie in $[0,1]$. 
Lower order terms of $\rho$-values are suppressed, and for bit-sampling LSH we assume that the Hamming distances 
are small relative to the dimensionality of the space, i.e., that $2r_{1}t/d = o(1)$.
\end{table*}
In Table \ref{tab:comparison} we show how the upper bounds for similarity search under different measures of set similarity relate to each other in the case where all sets are $t$-sparse.
In addition to Hamming distance and Jaccard similarity, we consider Braun-Blanquet similarity~\cite{braunblanquet1932} defined as
\begin{equation}\label{eq:B}
B(x, y) = |x \cap y| / \max(|x|, |y|), 
\end{equation}
which for $t$-sparse vectors is identical to cosine similarity.
When the query and the sets in $P$ can have different sizes the picture becomes muddled, 
and the question of which of the known algorithms is best for each measure of similarity is complicated and can depend on $(s_1, s_2)$.
In Section \ref{sec:equivalence} we treat the problem of different set sizes and provide a brief discussion for Jaccard similarity, 
specifically in relation to our upper bound for Braun-Blanquet similarity.

Similarity search under set similarity and the batched version often referred to as \emph{set similarity join}~\cite{arasu2006efficient,bayardo2007scaling}
have also been studied extensively in the information retrieval and database literature, but mostly without providing theoretical guarantees on performance.
Recently the notion of containment search, where the similarity measure is the (unnormalized) intersection size, was studied in the LSH framework~\cite{shrivastava2015asymmetric}.
This is a special case of \emph{maximum inner product} search~\cite{shrivastava2015asymmetric,ahle2015}.
However, these techniques do not give improvements in our setting.

\paragraph{Similarity estimation.}
Finally, we mention that another application of MinHash~\cite{bro97b,bro97a} is the (easier) problem of \emph{similarity estimation}, 
where the task is to condense each vector $x$ into a short signature $s(x)$ in such a way that the similarity $J(x,y)$ can be estimated from $s(x)$ and~$s(y)$.
A related similarity estimation technique was independently discovered by Cohen~\cite{cohen1997size}.
Thorup~\cite{thorup2013bottom} has shown how to perform similarity estimation using just a small amount of randomness in the definition of the function $s(\cdot)$.
In another direction, Mitzenmacher et al.~\cite{mitzenmacher2014} showed that it is possible to improve the performance of MinHash for similarity estimation when the Jaccard similarity is close to~1,
but for smaller similarities it is known that succinct encodings of MinHash such as the one in~\cite{li2011theory} comes within a constant factor of the optimal space for storing~$s(x)$~\cite{pagh2014min}.
Curiously, our improvement to MinHash in the context of similarity \emph{search} comes when the similarity is neither too large nor too small.
Our techniques do not seem to yield any improvement for the similarity \emph{estimation} problem.

\subsection{Contribution}
We show the following upper bound for approximate similarity search under Braun-Blanquet similarity: 
\begin{theorem}\label{thm:upper}
	For every choice of constants \mbox{$0 < b_{2} < b_{1} < 1$} we can solve the $(b_{1}, b_{2})$-approximate similarity search problem under Braun-Blanquet similarity 
	with query time $O(|q|n^{\rho} \log n)$ and space usage $O(n^{1+\rho}\log n + \sum_{x \in P}|x|)$ where $\rho = \log(1/b_1) / \log(1/b_2)$.
\end{theorem}
In the case where the sets are $t$-sparse our Theorem \ref{thm:upper} gives the first strict improvement on the $\rho$-value 
for approximate Jaccard similarity search compared to the data-independent LSH approaches of MinHash and Angular LSH.
Figure~\ref{fig:two-approx} shows an example of the improvement for a slice of the parameter space.
The improvement is based on a new locality-sensitive mapping that considers a specific random collection of length-$k$ paths on the vertex set $\{1,\dots,d\}$, 
and associates each vector $x$ with the paths in the collection that only visits vertices in $\{ i \; | \; x_i=1 \}$.
Our data structure exploits that similar vectors will be associated with a common path with constant probability, while vectors with low similarity have a negligible probability of sharing a path.
However, the collection of paths has size superlinear in~$n$, so an efficient method is required for locating the paths associated with a particular vector.
Our choice of the collection of paths balances two opposing constraints:
It is random enough to match the filtering performance of a truly random collection of sets, and at the same time it is structured enough to allow efficient search for sets matching a given vector.
The search procedure is comparable in simplicity to the classical techniques of bit sampling, MinHash, SimHash, and $p$-stable LSH, and we believe it might be practical.
This is in contrast to most theoretical advances in similarity search in the past ten years that suffer from $o(1)$ terms in the exponent of complexity bounds.

\begin{figure}   
	\centering
	\includegraphics[width=0.85\textwidth]{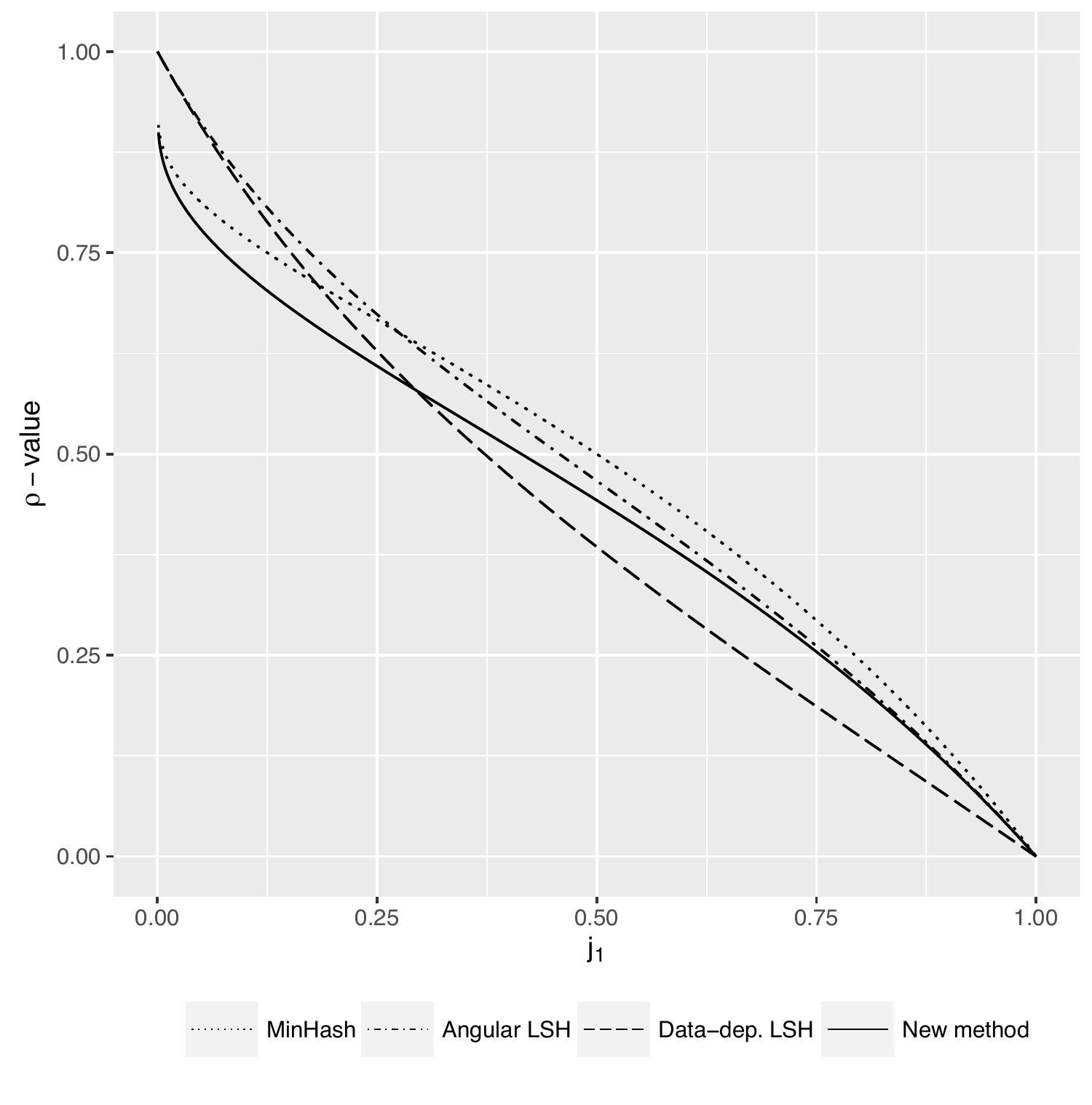}
	\caption{Exponent when searching for a vector with Jaccard similarity~$j_{1}$ with approximation factor~2 (i.e., guaranteed to return a vector with Jaccard similarity $j_{1}/2$) for various methods in the setting where all sets have the same size. 
	Our new method is the best data-independent method, and is better than data-dependent LSH up to about $j_1 \approx 0.3$.}
\label{fig:two-approx}
\end{figure}

\medskip

\paragraph{Intuition.} 
Recall that we will think of a vector $x \in \bitcube{d}$ also as a set, $\{ i \;|\; x_i=1\}$.
MinHash can be thought of as a way of sampling an element $i_{x}$ from $x$, namely, we let $i_{x} = \argmin_{i \in x} h(i)$ where $h$ is a random hash function.
For sets $x$ and $y$ the probability that $i_{x}=i_{y}$ equals their Jaccard similarity $J(x,y)$, which is much higher than if the samples had been picked independently.
Consider the case in which $|x|=|y|=t$, so $J(x,y) = \frac{|x \cap y|}{2t-|x \cap y|}$.
Another way of sampling is to compute $I_{x} = x \cap {b}$, where $\Pr[i\in {b}]=1/t$, independently for each $i\in [d]$.
The expected size of $I_{x}$ is 1, so this sample has the same expected ``cost'' as the MinHash-based sample.
But if the Jaccard similarity is small, the latter samples are more likely to overlap: 
$$\Pr[I_{x}\cap I_{y} \neq \emptyset] = 1-(1-1/t)^{|x \cap y|} \approx 1 - e^{-|x \cap y|/t} \approx |x \cap y|/t,$$ 
nearly a factor of 2 improvement.
In fact, whenever $|x \cap y| < 0.6\, t$ we have 
$\Pr[I_{x}\cap I_{y} \ne \emptyset] > \Pr[i_{x}=i_{y}]$.
So in a certain sense, MinHash is not the best way of collecting evidence for the similarity of two sets.
(This observation is not new, and has been made before e.g.~in~\cite{cohen2009leveraging}.)

\medskip

\paragraph{Locality-sensitive maps.}
The intersection of the samples $I_{x}$ does not correspond directly to hash collisions, so it is not clear how to turn this insight into an algorithm in the LSH framework.
Instead, we will consider a generalization of both the locality sensitive filtering (LSF) and LSH frameworks where we define a distribution $\LSM$ over maps $M \colon \bitcube{d} \to 2^{R}$.
The map $M$ performs the same task as the LSH data structure: It takes a vector $x$ and returns a set of memory locations $M(x) \subseteq \{1, \dots, R \}$.
A randomly sampled map $M \sim \LSM$ has the property that if a pair of points $x,y$ are close then $M(x) \cap M(y) \neq \emptyset$ with constant probability, 
while if $x, y$ are distant then the expected size $M(x) \cap M(y)$ is small (much smaller than $1$).
It is now straightforward to see that this distribution can be used to construct a data structure for similarity search by storing each data point $x \in P$ in the set of memory locations or buckets $M(x)$.
A query for a point $y$ is performed by computing the similarity between $y$ and every point $x$ contained in the set buckets $M(y)$, reporting the first sufficiently similar point found.

\medskip

\paragraph{\textsc{Chosen Path}.}
It turns out that to most efficiently filter out vectors of low similarity in the setting where all sets have equal size, 
we would like to map each data point $x \in \bitcube{d}$ to a collection $M(x)$ of random subsets of $\bitcube{d}$ that are contained in $x$.
Furthermore, to best distuinguish similar from dissimilar vectors when solving the approximate similarity search problem, we would like the random subsets of $\bitcube{d}$ to have size $\Theta(\log n)$.
This leads to another obstacle: The collection of subsets of $\bitcube{d}$ required to ensure that $M(x) \cap M(y) \neq \emptyset$ for similar points,
i.e., that $M$ maps to a subset contained in $x \cap y$, is very large.
The space usage and evaluation time of a locality-sensitive map $M$ to fully random subsets of $\bitcube{d}$ would far exceed $n$, rendering the solution useless. 
To overcome this we create the samples in a gradual, correlated way using a pairwise independent branching process that turns out to yield ``sufficiently random'' samples for the argument to go through.

\medskip

\paragraph{Lower bound.}
On the lower bound side we show that our solution for Braun-Blanquet similarity is best possible in terms of parameters $b_1$ and $b_2$ within the class of solutions that can be characterized as data-independent locality-sensitive maps.
The lower bound works by showing that a family of locality-sensitive maps for Braun-Blanquet similarity with a $\rho$-value below $\log(1/b_1)/\log(1/b_2)$ 
can be used to construct a locality-sensitive hash family for the $c$-approximate near neighbor problem in Hamming space with a $\rho$-value below $1/c$, 
thereby contradicting the LSH lower bound by O'Donnell et al.~\cite{odonnell2014optimal}.
We state the lower bound here in terms of locality-sensitive hashing, formally defined in Section \ref{sparse:sec:preliminaries}. 
\begin{theorem} \label{thm:lower}
For every choice of constants $0 < {b_2} < {b_1} < 1$ any $({b_1}, {b_2}, p_1, p_2)$-sensitive hash family $\LSH_{B}$ for $\bitcube{d}$ under Braun-Blanquet similarity must satisfy
\begin{equation*}
	\rho(\LSH_{B}) = \frac{\log(1/p_{1})}{\log(1/p_{2})} \geq \frac{\log(1/{b_1})}{\log(1/{b_2})} - O\left(\frac{\log(d/p_2)}{d}\right)^{1/3}. 
\end{equation*}
\end{theorem}
The details showing how this LSH lower bound implies a lower bound for locality-sensitive maps are given in Section \ref{sec:lower}. 
\section{Preliminaries}\label{sparse:sec:preliminaries}
As stated above we will view $x \in \bitcube{d}$ both as a vector and as a subset of $[d] = \{1,\dots,d\}$.
Define $x$ to be \emph{$t$-sparse} if $|x| = t$; 
we will be interested in the setting where $t \leq d/2$, and typically the sparse setting $t \ll d$.
Although many of the concepts we use hold for general spaces, for simplicity we state definitions in the same setting as our results: 
the boolean hypercube $\bitcube{d}$ under some measure of similarity $\simil \colon \bitcube{d} \times \bitcube{d} \rightarrow [0,1]$. 

\begin{definition} (Approximate similarity search)
	Let $P \subset \bitcube{d}$ be a set of $|P| = n$ data vectors, let $\simil \colon \bitcube{d} \times \bitcube{d} \rightarrow [0,1]$ be a similarity measure, and let $s_1, s_2 \in [0,1]$ such that $s_1 > s_2$.
	A solution to the \emph{$(s_1,s_2)$-similarity search problem} is a data structure that supports the following query operation: 
	on input $q \in \bitcube{d}$ for which there exists a vector $x \in P$ with $\simil(x,q) \geq s_1$, return $x' \in P$ with $\simil(x',q) > s_2$.
\end{definition}
Our data structures are randomized, and queries succeed with probability at least~$1/2$ (the probability can be made arbitrarily close to~$1$ by independent repetition).
Sometimes similarity search is formulated as searching for vectors that are near $q$ according to the distance measure $\dist(x, y) = 1 - \simil(x,y)$. 
For our purposes it is natural to phrase conditions in terms of similarity, but we  compare to solutions originally described as ``near neighbor'' methods.

Many of the best known solutions to approximate similarity search problems are based on a technique of randomized space partitioning. 
This technique has been formalized in the locality-sensitive hashing framework \cite{indyk1998} and the closely related locality-sensitive filtering framework~\cite{becker2016, christiani2017framework}.
\begin{definition} (Locality-sensitive hashing {\cite{indyk1998}})
	A $({s_1}, {s_2}, p_1, p_2)$-sensitive family of hash functions for a similarity measure 
	$\simil \colon \bitcube{d} \times \bitcube{d} \to [0,1]$ is a distribution $\LSH_{\simil}$ over functions $h \colon \bitcube{d} \to R$ 
	such that for all $x,y \in \bitcube{d}$ and random $h$ sampled according to $\LSH_{\simil}$: 
	\begin{itemize}
		\item If $\simil(x,y) \geq s_1$ then $\Pr[h(x) = h(y)] \geq p_1$.
		\item If $\simil(x,y) \leq s_2$ then $\Pr[h(x) = h(y)] \leq p_2$.
	\end{itemize}
\end{definition}
The range $R$ of the family will typically be fairly small such that an element of $R$ can be represented in a constant number of machine words.
In the following we assume for simplicity that the family of hash functions is \emph{efficient} such that $h(x)$ can be computed in time~$O(|x|)$.
Furthermore, we will assume that the time to compute the similarity $\sim(x, y)$ can be upper bounded by the time it takes to compute the size of the intersection of preprocessed sets, i.e.,~$O(\min(|x|, |y|))$.

Given a locality-sensitive family it is quite simple to obtain a solution to the approximate similarity search problem, 
essentially by hashing points to buckets such that close points end up in the same bucket while distant points are kept apart. 
\begin{lemma}[LSH framework {\cite{indyk1998, har-peled2012}}] 
	Given a $(s_1, s_2, p_1, p_2)$-sensitive family of hash functions it is possible to solve the $(s_1, s_2)$-similarity search problem
	with query time $O(|q| n^{\rho} \log n)$ and space usage $O(n^{1 + \rho} + \sum_{x \in P}|x|)$ where $\rho = \log(1/p_1) / \log(1/p_2)$.
\end{lemma}
The upper bound presented in this paper does not quite fit into the existing frameworks. 
However, we would like to apply existing LSH lower bound techniques to our algorithm.
Therefore we define a more general framework that captures solutions constructed using the LSH and LSF framework, as well as the upper bound presented in this paper.
\begin{definition}[Locality-sensitive map]
	A $(s_1, s_2, m_1, m_2)$-sensitive family of maps for a similarity measure $\simil \colon \bitcube{d} \times \bitcube{d} \to [0,1]$ is a distribution $\LSM_{\simil}$ over mappings $M \colon \bitcube{d} \to 2^{R}$ (where $2^{R}$ denotes the power set of $R$) such that for all $x, y \in \bitcube{d}$ and random $M \in \LSM_{\simil}$:
	\begin{enumerate}
		\item $\E[|M(x)|] \leq m_1$.	
		\item If $\simil(x, y) \leq s_2$ then $\E[|M(x) \cap M(y)|] \leq m_2$.
		\item If $\simil(x, y) \geq s_1$ then $\Pr[M(x) \cap M(y) \neq \emptyset] \geq 1/2$.
	\end{enumerate}
 \end{definition}
Once we have a family of locality-sensitive maps $\LSM$ we can use it to obtain a solution to the $(s_1, s_2)$-similarity search problem.
\begin{lemma}
	Given a $(s_1, s_2, m_1, m_2)$-sensitive family of maps $\LSM$ we can solve the $(s_1, s_2)$-similarity search problem 
	with query time $O(m_{1} + nm_{2}|q| + T_{M})$ and space usage $O(nm_{1} + \sum_{x \in P}|x|)$ where $T_{M}$ is the time to evaluate a map $M \in \LSM$.
\end{lemma}
\begin{proof}
	We construct the data structure by sampling a map $M$ from $\LSM$ and use it to place points in $P$ into buckets.
	To run a query for a point $q$ we proceed by evaluating $M(q)$ and computing the similarity between $q$ and the points in the buckets associated with $M(q)$.
	If a sufficiently similar point is found we return it. 
	We get rid of the expectation in the guarantees by independent repetitions and applying Markov's inequality.
\end{proof}

\paragraph{Model of computation.}
We assume the standard word RAM model \cite{hagerup1998} with word size $\Theta(\log n)$, where $n=|P|$.
In order to be able to draw random functions from a family of functions we augment the model with an instruction that generates a machine word uniformly at random in constant time.
\section{Upper Bound} \label{sec:upper}
We will describe a family of locality-sensitive maps $\LSM_{B}$ for solving the $(b_1, b_2)$-similarity search problem under Braun-Blanquet similarity~(\ref{eq:B}).
After describing $\LSM_{B}$ we will give an efficient implementation of $M \in \LSM_{B}$ and show how to set parameters to obtain our Theorem \ref{thm:upper}. 
\subsection{Chosen Path}
The \textsc{Chosen Path} family $\LSM_{B}$ is defined by $k$ random hash functions $h_1, \dots, h_k$ where $h_{i} \colon [w] \times [d]^{i} \to [0,1]$ and $w$ is a positive integer.
The evaluation of a map $M_{k} \in \LSM_{B}$ proceeds in a sequence of $k+1$ steps that can be analyzed as a Galton-Watson branching process, 
originally devised to investigate population growth under the assumption of identical and independent offspring distributions. 
In the first step $i = 0$ we create a population of $w$ starting points
\begin{equation}\label{eq:M0}
M_{0}(x) = [w].
\end{equation}
In subsequent steps, every path that has survived so far produces offspring according to a random process that depends on $h_i$ and the element $x \in \bitcube{d}$ being evaluated. 
We use $p \circ j$ to denote concatenation of a path $p$ with a vertex $j$. 
\begin{equation}\label{eq:Mi}
M_i(x) = \left\{ p \circ j \mid p \in M_{i-1}(x) \land h_i(p\circ j) < \frac{x_{j}}{b_{1}|x|} \right\}. 
\end{equation}
Observe that $h_i(p\circ j) < \frac{x_{j}}{b_{1}|x|}$ can only hold when $x_{j}=1$, so the paths in $M_i(x)$ are constrained to $j \in x$.
The set $M(x) = M_{k}(x)$ is given by the paths that survive to the $k$th step.
We will proceed by bounding the evaluation time of $M \in \LSM_{B}$ as well as showing the locality-sensitive properties of~$\LSM_{B}$.
In particular, for similar points $x, y \in \bitcube{d}$ with $B(x, y) \geq b_1$ we will show that with probability at least $1/2$ there will be a path that is chosen by both $x$ and $y$. 
\begin{lemma}[Properties of {\textsc{Chosen Path}}]\label{lem:cp}
	For all $x, y \in \bitcube{d}$, integer $i\geq 0$, and random $M \in \LSM_{B}$:
	\begin{enumerate}
		\item $\E[|M_{i}(x)|] \leq (1/b_{1})^{i}w$.	
		\item If $B(x, y) < b_2$ then $\E[|M_{i}(x) \cap M_{i}(y)|] \leq (b_{2}/b_{1})^{i}w$.
		\item If $B(x, y) \geq b_1$ then $\Pr[M_{i}(x) \cap M_{i}(y) \neq \emptyset] \geq w/(i + w)$.
	\end{enumerate}
\end{lemma}
\begin{proof}
We prove each property by induction on $i$. 
The base cases $i=0$ follow from (\ref{eq:M0}). 
Now consider the inductive step for property 1.
Let $\1\{\mathcal{P}\}$ denote the indicator function for predicate~$\mathcal{P}$. 
Using independence of the hash functions $h_i$ we get:
\begin{align*}
	\E[|M_{i}(x)|] 
	&= \E\left[ \sum_{p \in M_{i-1}(x)} \sum_{j \in [d]} \1\! \left\{ h_{i}(p\circ j) < \frac{x_{j}}{b_{1} |x|} \right\} \right] \\ 
	&= \E\left[ \sum_{p \in M_{i-1}(x)} 1 \right] \E\left[\sum_{j \in [d]} \1\! \left\{ h_{i}(p\circ j) < \frac{x_{j}}{b_{1} |x|} \right\} \right] \\ 
				    &\leq \E[|M_{i-1}(x)|] /b_{1} \\
					&\leq (1/b_{1})^{i}w \enspace .
\end{align*}
The last inequality uses the induction hypothesis.
We use the same approach for the second property where we let $X_{i}=M_{i}(x) \cap M_{i}(y)$.
\begin{align*}
	\E[|X_{i}|] 
	&= \E\left[ \sum_{p \in X_{i-1}} \sum_{j \in [d]} \1\! \left\{ h_{i}(p\circ j) < \frac{x_{j}}{b_{1} |x|} \land h_{i}(p\circ j) < \frac{y_{j}}{b_{1} |y|}  \right\} \right] \\ 
	&= \E\left[ \sum_{p \in X_{i-1}} 1 \right] \sum_{j \in [d]} \Pr\left[ h_{i}(p\circ j) < \frac{\min(x_{j},y_{j})}{b_{1} \max(|x|,|y|)} \right] \\ 
			    &\leq \E[|X_{i-1}|] (B(x, y)/b_{1}) \\
				&\leq (B(x, y)/b_{1})^{i}w \enspace . 	
\end{align*}

To prove the third property we bound the variance of $|X_{i}|$ and apply Chebyshev's inequality to bound the probability of $X_{i}=\emptyset$.
First consider the case where $|x| \leq 1/b_{1}$ and $|y| \leq 1/b_1$. 
Here it must hold that $X_{i}>0$ as intersecting paths exist ($b_1>0$) and always activate. 
In all other cases we have that 
$$\E[|X_{i}|] = (B(x, y)/b_{1})^{i}w \enspace .$$
Knowing the expected value we can apply Chebyshev's inequality once we have an upper bound for $\Var[|X_{i}|] = \E[|X_{i}|^2] - \E[|X_{i}|]^{2}$.
Specifically we show that $\E[|X_{i}|^2]\leq wi (B(x, y)/b_{1})^{2i}$, by induction on $i$.
To simplify notation we define the indicator variable
$$Y_{p,j} = \1\! \left\{ h_i(p\circ j) < \frac{x_{j}}{b_{1} |x|} \land h_i(p\circ j) < \frac{y_{j}}{b_{1} |y|} \right\}$$
where we suppress the subscript~$i$.
First observe that
$$\E[Y_{p,j}]= 1 / (b_1 \max(|x|,|y|))\enspace .$$
By (\ref{eq:Mi}) we see that $|X_i| = \sum_{p\in X_{i-1}} \sum_{j\in [d]} Y_{p,j}$,
which means:
\begin{align*}
	\E[|X_{i}|^{2}] &= \E\left[ \left( \sum_{p \in X_{i-1}} \sum_{j \in [d]} Y_{p,j}  \right)^{2} \right] \\ 
   				    & = \E\left[ \sum_{p \in X_{i-1}} \sum_{j\in [d]} Y_{p,j}^{2}\right] 
					 + \E\left[\sum_{p,p' \in X_{i-1}} \sum_{j,j'\in [d]} Y_{p,j} Y_{p',j'} \1\{(p,j)\ne (p',j')\} \right] \\
					&< \E[|X_{i-1}|] (B(x, y)/b_1) + \E[|X_{i-1}|^2] (B(x, y)/b_1)^2 \\
				    &\leq \sum_{s = 1}^{i}\E[|X_{i-s}|] (B(x, y)/b_1)^{2s-1} +   \E[|X_{0}|]^{2}(B(x, y)/b_1)^{2i} \\
					& = \E[|X_i|] \sum_{s = 0}^{i-1}(B(x, y)/b_{1})^{s} + \E[|X_i|]^2 \\
				    &\leq wi(B(x,y)/b_1)^{2i} +  \E[|X_i|]^2 \enspace .
\end{align*}
The third property now follows from a one-sided version of Chebychev's inequality applied to $|X_i|$.
\end{proof}
\subsection{Implementation details}
Lemma \ref{lem:cp} continues to hold when the hash functions $h_1, \dots, h_k$ are individually \emph{2-independent} (and mutually independent) 
since we only use bounds on the first and second moment of the hash values.
We can therefore use a simple and practical scheme such as Zobrist hashing \cite{zobrist1970new} 
that hashes strings of $\Theta(\log n)$ bits to strings of $\Theta(\log n)$ bits in $O(1)$ time using space, say, $O(n^{1/2})$.
It is not hard to see that the domain and range of $h_1, \dots, h_k$ can be compressed to $O(\log n)$ bits (causing a neglible increase in the failure probability of the data structure). 
We simply hash the paths $p \in M_{i}(x)$ to intermediate values of $O(\log n)$ bits, avoiding collisions with high probability, 
and in a similar vein, with high probability $O(\log n)$ bits of precision suffice to determine whether $h_i(p\circ j) < \frac{x_{j}}{b_{1}|x|}$. 

We now consider how to parameterize $\LSM_{B}$ to solve the $(b_{1}, b_{2})$-similarity problem for Braun-Blanquet similarity on a set $P$ of $|P| = n$ points for every choice of constant parameters $0 < b_2 < b_1 < 1$, independent of $n$.
Note that we exclude $b_{1} = 1$ (which would correspond to identical vectors that can be found in time $O(1)$ by resorting to standard hashing) and $b_{2} = 0$ (for which every data point would be a valid answer to a query).
We set parameters 
\begin{align*}
k &= \lceil \log(n) / \log (1/b_{2}) \rceil, \\
w &= 2k
\end{align*}
from which it follows that $\LSM_{B}$ is $(b_1, b_2, m_1, m_2)$-sensitive with $m_1 = n^{\rho}w/b_1$ and $m_2 = n^{\rho - 1}w$ where $\rho = \log(1/b_1) / \log(1/b_2)$.
To bound the expected evaluation time of $M_{k}$ we use Zobrist hashing as well as intermediate hashes for the paths as described above.
In the $i$th step in the branching process the expected number of hash function evaluations is bounded by $|q|$ times the number of paths alive at step $i - 1$. 
We can therefore bound the expected time to compute $M_{k}(q)$ by
\begin{equation} \label{eq:hashtime}
	\sum_{i=0}^{k-1}\E[|q||M_{i}(q)|] \leq \frac{b_{1}^{-k}-1}{b_{1}^{-1}-1}|q|w = O(|q|n^{\rho}w).
\end{equation}
This completes the proof of Theorem \ref{thm:upper}.\footnote{We know of a way of replacing the multiplicative factor $|q|$ in equation \eqref{eq:hashtime} by an additive term of $O(|q|k)$ 
by choosing the hash functions $h_i$ carefully, but do not discuss this improvement here since $|q|$ can be assumed to be polylogarithmic and our focus is on the exponent of $n$.}
\subsection{Comparison}\label{sec:comparison}
We will proceed by comparing our Theorem \ref{thm:upper} to results that can be achieved using existing techniques.
Again we focus on the setting where data points and query points are exactly $t$-sparse.
An overview of different techniques for three measures of similarity is shown in Table~\ref{tab:comparison}.
To summarize: The \textsc{Chosen Path} algorithm of Theorem \ref{thm:upper} improves upon all existing data-independent results over the entire $0 < b_2 < b_1 < 1$ parameter space.
Furthermore, we improve upon the best known \emph{data-dependent} techniques \cite{andoni2015optimal} for a large part of the parameter space (see Figure~\ref{fig:data}). 
The details of the comparisons are given in Appendix \ref{app:comparison}.

\paragraph{MinHash.} 
For $t$-sparse vectors there is a 1-1 mapping between Braun-Blanquet and Jaccard similarity. 
In this setting $J(x,y)=B(x,y)/(2-B(x,y))$.
Let $b_1 = 2j_1/(j_1+1)$ and $b_2 = 2j_2/(j_2+1)$ be the Braun-Blanquet similarities corresponding to Jaccard similarities $j_1$ and $j_2$.
The LSH framework using MinHash achieves $\rho_\text{minhash} = \log\left(\tfrac{b_1}{2-b_1}\right) / \log\left(\tfrac{b_2}{2-b_2}\right)$; 
this should be compared to $\rho = \log(b_1)/\log(b_2)$ achieved in Theorem~\ref{thm:upper}.
Since the function $f(z) = \log(\tfrac{z}{2-z})/\log z$ is monotonically increasing in $[0,1]$ we have that $\rho / \rho_\text{minhash} = f(b_2)/f(b_1) < 1$, i.e., $\rho$ is always smaller than $\rho_\text{minhash}$.
As an example, for $j_1 = 0.2$ and $j_2 = 0.1$ we get $\rho = 0.644...$ while $\rho_\text{minhash} = 0.698...$.
Figure~\ref{fig:minhash} shows the difference for the whole parameter space.
\begin{figure}
	\centering
	\includegraphics[width=\textwidth]{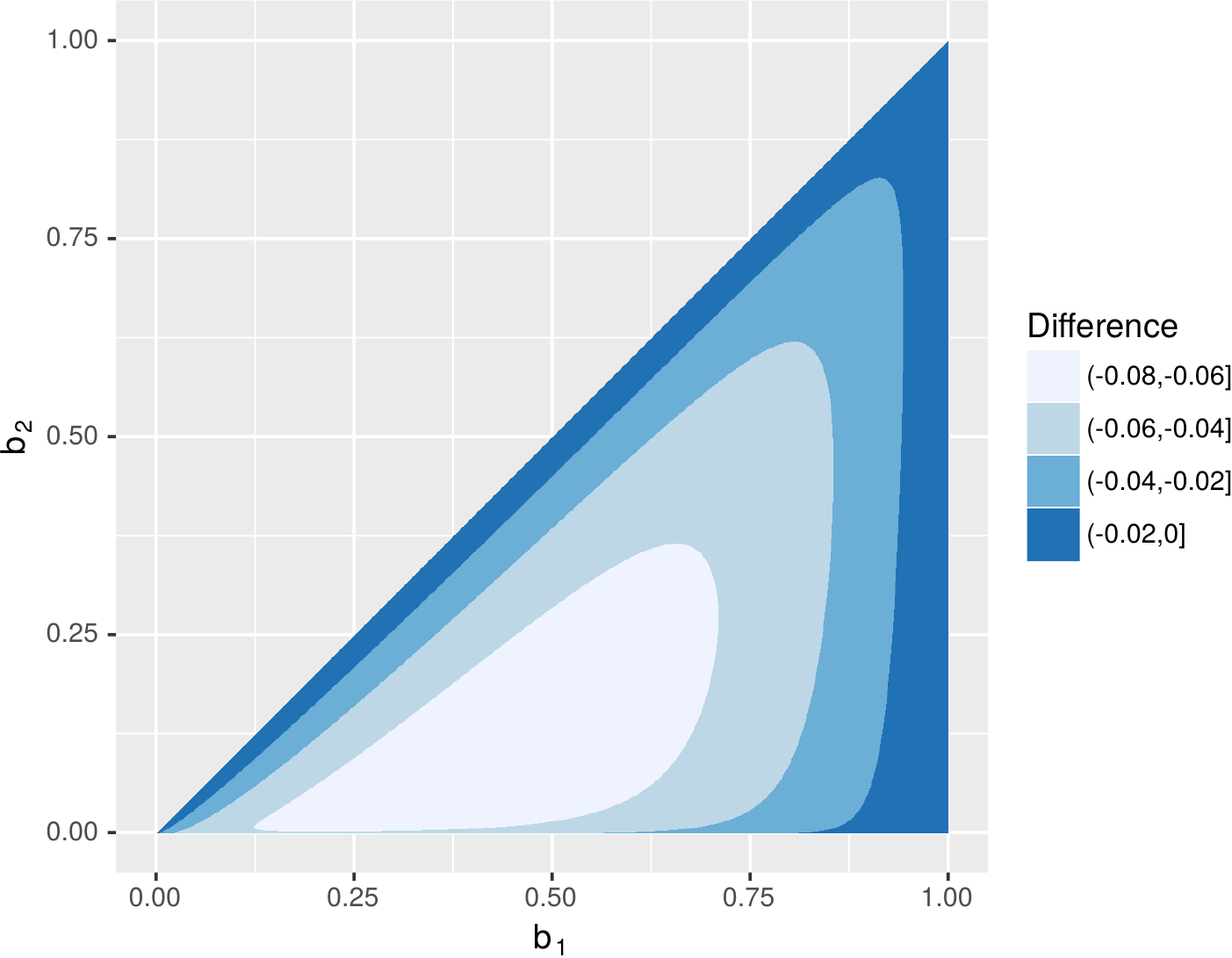}
	\caption{The difference $\rho - \rho_{\text{minhash}}$ comparing \textsc{Chosen Path} and MinHash in terms of Braun-Blanquet similarities $0 < b_2 < b_1 < 1$.}
	\label{fig:minhash}
\end{figure}

\paragraph{Angular LSH.}
Since our vectors are exactly $t$-sparse Braun-Blanquet similarities correspond directly to dot products (which in turn correspond to angles).
Thus we can apply angular LSH such as SimHash~\cite{charikar2002} or cross-polytope LSH~\cite{andoni2015practical}.
As observed in~\cite{christiani2017framework} one can express the $\rho$-value of cross-polytope LSH in terms of dot products as $\rho_\text{angular} = \tfrac{1-b_1}{1+b_1}/\tfrac{1-b_2}{1+b_2}$.
Since the function $f'(z) = (1 + z) \log(z)/(1 - z)$ is negative and monotonically increasing in $[0,1]$ we have that $\rho / \rho_\text{angular} = f'(b_1)/f'(b_2) < 1$, 
i.e., $\rho$ is always smaller than $\rho_\text{angular}$.
In the above example, for $j_1=0.2$ and $j_2=0.1$ we have $\rho_\text{angular} = 0.722...$ which is about $0.078$ more than \textsc{ Chosen Path}.
See Figure~\ref{fig:angular} for a visualization of the difference for the whole parameter space.
\begin{figure}
	\centering
	\includegraphics[width=\textwidth]{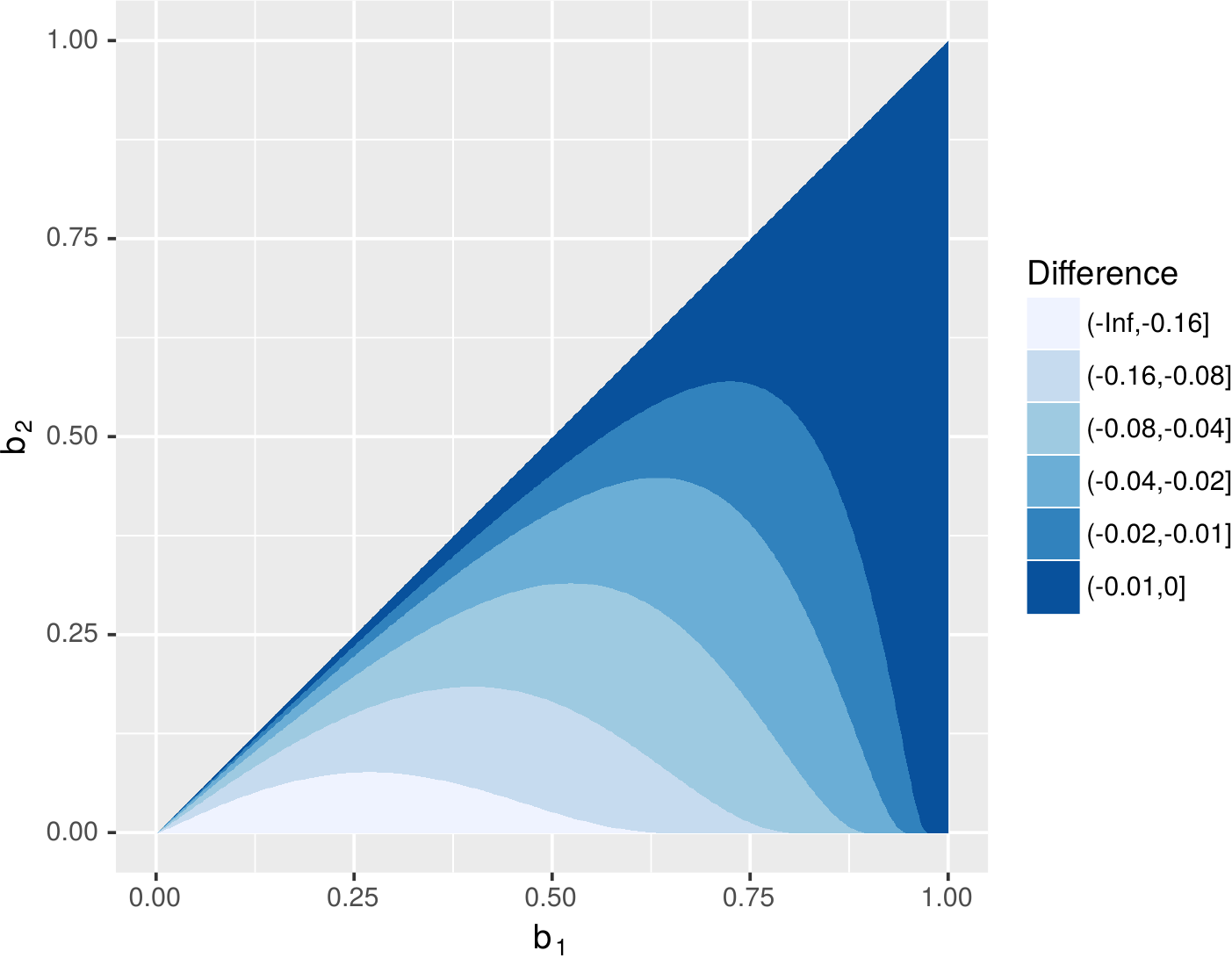}
	\caption{The difference $\rho - \rho_{\text{angular}}$ comparing \textsc{Chosen Path} and angular LSH in terms of Braun-Blanquet similarities $0 < b_2 < b_1 < 1$.}
	\label{fig:angular}
\end{figure}

\paragraph{Data-dependent Hamming LSH.}
The Hamming distance between two $t$-sparse vectors with Braun-Blanquet similarity $b$ is $2t(1-b)$, since the intersection of the vectors has size $tb$.
This means that $(b_1,b_2)$-similarity search under Braun-Blanquet similarity can be reduced to Hamming similarity search with approximation factor $c = (2t(1-b_1))/(2t(1-b_2)) = (1-b_1)/(1-b_2)$.
As mentioned above, the \emph{data dependent} LSH technique of~\cite{andoni2015optimal} achieves $\rho = 1/(2c-1)$ ignoring $o_{n}(1)$ terms. 
In terms of $b_1$ and $b_2$ this is $\rho_\text{datadep} = \frac{1-b_1}{1+b_1-2b_2}$, 
which in incomparable to the $\rho$ of Theorem \ref{thm:upper}.
In Appendix \ref{app:comparison} we show that $\rho < \rho_{\text{datadep}}$ whenever $b_2 \leq 1/5$, or equivalently, whenever $j_2 \leq 1/9$.
Revisiting the above example, for $j_1 = 0.2$ and $j_2 = 0.1$ we have $\rho_\text{datadep} = 0.6875$ which is about $0.043$ more than \textsc{Chosen Path}.
Figure~\ref{fig:data} gives a comparison covering the whole parameter space.
\begin{figure}
	\centering
	\includegraphics[width=\textwidth]{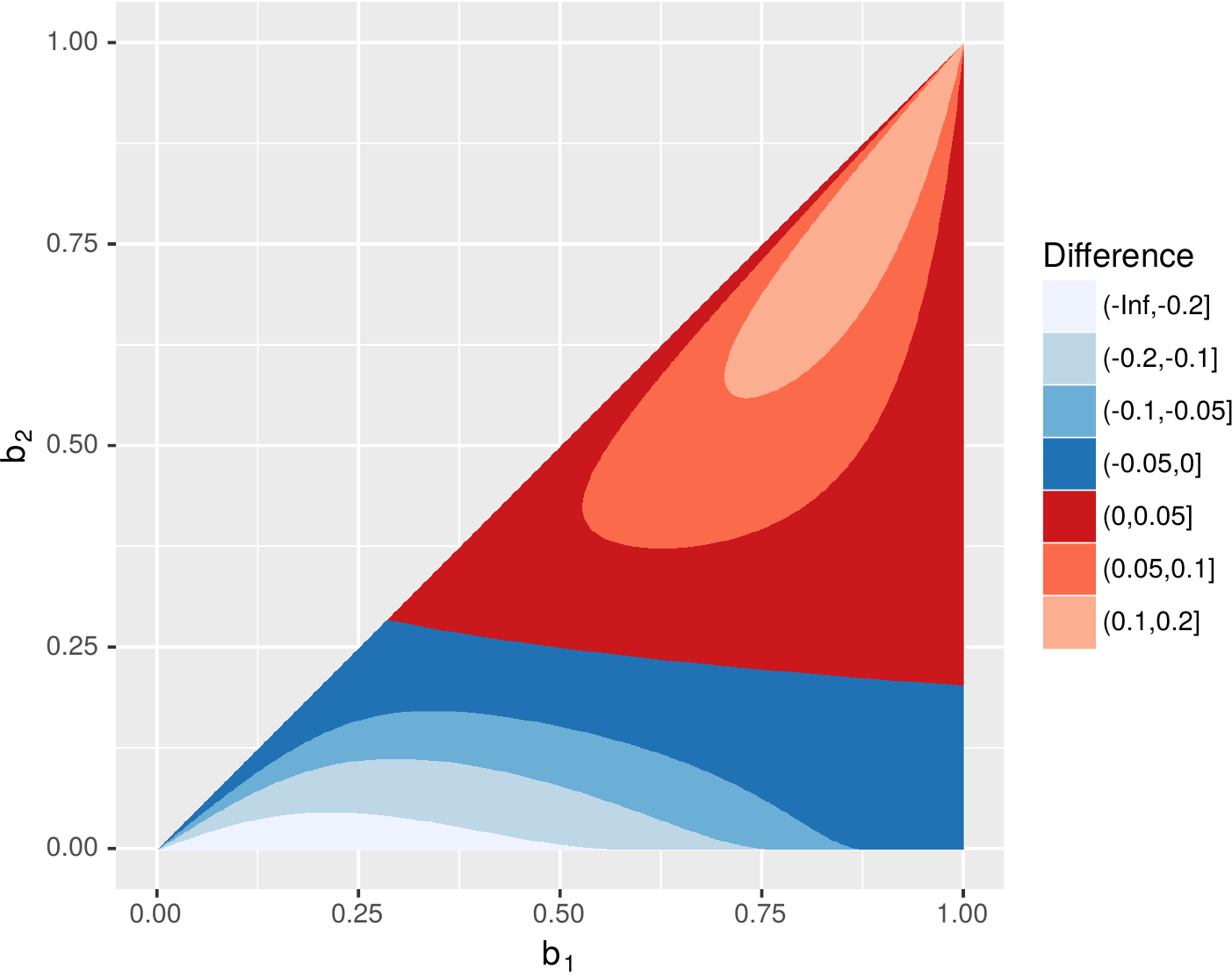}
	\caption{The difference $\rho - \rho_{\text{datadep}}$ comparing \textsc{Chosen Path} and data-dependent LSH in terms of Braun-Blanquet similarities $0 < b_2 < b_1 < 1$. 
	In the area of the parameter space that is colored blue we have that $\rho \leq \rho_{\text{datadep}}$ while for the red area it holds that $\rho > \rho_{\text{datadep}}$.}
	\label{fig:data}
\end{figure}
\section{Lower bound} \label{sec:lower}
In this section we will show a locality-sensitive hashing lower bound for $\bitcube{d}$ under Braun-Blanquet similarity. 
We will first show that LSH lower bounds apply to the class of solutions to the approximate similarity search problem that are based on locality-sensitive maps, thereby including our own upper bound.
Next we will introduce some relevant tools from the literature, in particular the LSH lower bounds for Hamming space by O'Donnell et al. \cite{odonnell2014optimal} which we use, 
through a reduction, to show LSH lower bounds under Braun-Blanquet similarity.

\paragraph{Lower bounds for locality-sensitive maps.}
Because our upper bound is based on a locality-sensitive map $\LSM_{B}$ and not LSH-based we first show that LSH lower bounds apply to LSM-based solutions. 
This is not too surprising as both the LSH and LSF frameworks produce LSM-based solutions.
We note that the idea of showing lower bounds for a more general class of algorithms that encompasses both LSH and LSF was used by Andoni et al.~\cite{andoni2017optimal} 
in their list-of-points data structure lower bound for the space-time tradeoff of solutions to the approximate near neighbor problem in the random data regime.
We use the approach of Christiani \cite{christiani2017framework} to convert an LSM family into an LSH family using MinHash. 
\begin{lemma}\label{lem:lsmtolsh}
	Suppose we have a $(s_1, s_2, m_1, m_2)$-sensitive family of maps $\LSM$. 
	Then we can construct a $(s_1, s_2, p_1, p_2)$-sensitive family of hash functions~$\LSH$ with $p_1 = 1/8m$ and $p_2 = m_{2}/m$ where $m = \lceil 8 m_{1} \rceil$.
\end{lemma}
\begin{proof}
We sample a function $h$ from $\LSH$ by sampling a function $M$ from $\LSM$, modify $M$ to output a set of fixed size, and apply MinHash to the resulting set. 
For $M \in \LSM$ we define the function $\tilde{M}$ where we ensure that the size of the output set is $m$. 
We note that the purpose of this step is to be able to simultaneously lower bound $p_1$ and upper bound $p_2$ for $\LSH$ when we apply MinHash to the resulting sets.
\begin{equation*}
\tilde{M}(x) =
\begin{cases}
\{ (x, 1), \dots, (x, m) \} & \text{if } |M(x)| \geq m, \\
\{ (x, 1), \dots, (x, m - |M(x)|) \} \cup M(x) & \text{otherwise.} 
\end{cases}
\end{equation*}
We proceed by applying MinHash to the set $\tilde{M}(x)$. Let $\pi$ denote a random permutation of the range of $\tilde{M}$ and define  
\begin{equation*}
	h(x) = \argmin_{z \in \tilde{M}(x)} \pi(z).
\end{equation*}
We then have 
\begin{equation*}
\Pr[h(x) = h(y)] = \sum_{\xi} \Pr[ J(\tilde{M}(x), \tilde{M}(y)) = \xi] \cdot \xi
\end{equation*}
summing over the finite set of all possible Jaccard similarities $\xi = a/b$ with $a, b \in \{0, 1, \dots, 2m \}$.
It is now fairly simple to lower bound $p_1$ and upper bound $p_2$. 
Assume that $x, y$ satisfy that $\simil(x, y) \geq s_1$. 
To lower bound $p_1$ we use a union bound together with Markov's inequality to bound the following probability:
\begin{align*}
	&\Pr[\tilde{M}(x) \cap \tilde{M}(y) = \emptyset] \\ 
	&\qquad \leq \Pr[M(x) \cap M(y) = \emptyset \lor |M(x)| \geq m \lor |M(y)| \geq m] \\
&\qquad \leq \Pr[M(x) \cap M(y) = \emptyset] + \Pr[|M(x)| \geq m] + \Pr[|M(y)| \geq m] \\
&\qquad \leq 1/2 + 1/8 + 1/8 
\end{align*}
We therefore have that $\Pr[\tilde{M}(x) \cap \tilde{M}(y) \neq \emptyset] \geq 1/4$. 
In the event of a nonempty intersection the probability of collision is given by $J(\tilde{M}(x) \cap \tilde{M}(y)) \geq 1/2m$ allowing us to conclude that $p_1 \geq 1/8m$.

Bounding the collision probability for distant pairs of points $x, y$ with $\sim(x, y) \leq s_2$ we get
\begin{equation*}
\sum_{\xi} \Pr[J(\tilde{M}(x), \tilde{M}(y)) = \xi] \cdot \xi \leq (1/m) \sum_{i = 1}^{\infty}\Pr[|\tilde{M}(x) \cap \tilde{M}(y)|] \cdot i = \frac{m_{2}}{m}.
\end{equation*}
\end{proof}
We are now ready to justify the statement that LSH lower bounds apply to LSM, allowing us to restrict our attention to proving LSH lower bounds for Braun-Blanquet similarity.
\begin{corollary}\label{cor:lsmtolsh}
Suppose that we have an LSM-based solution to the $(s_1, s_2)$-similarity search problem with query time $O(n^{\rho})$.
Then there exists a family $\LSH$ of locality-sensitive hash functions with $\rho(\LSH) = \rho + O(1/\log n)$.  
\end{corollary}
\begin{proof}
	The existence of the LSM-based solution implies that for every $n$ there exists a $(s_1, s_2, m_1, m_2)$-sensitive family of maps $\LSM$ with $m_1 = O(n^{\rho})$ and $nm_2 = O(n^{\rho})$.
	The upper bound on $\rho$ follows from applying Lemma \ref{lem:lsmtolsh}.  
\end{proof}

\paragraph{LSH lower bounds for Hamming space.}
There exist a number of powerful results that lower bound the $\rho$-value that is attainable by locality-sensitive hashing and related approaches in various settings 
\cite{motwani2007, panigrahy2010lower, odonnell2014optimal, andoni2016tight, christiani2017framework, andoni2017optimal}.
O'Donnell et al. \cite{odonnell2014optimal} showed an LSH lower bound of $\rho = \log(1/p_1) / \log(1/p_2) \geq 1/c - o_{d}(1)$
for $d$-dimensional Hamming space under the assumption that $p_2$ is not too small compared to $d$, i.e.,~$\log(1/p_2) = o(d)$.
The lower bound by O'Donnell et al.\ holds for $(r, cr, p_1, p_2)$-sensitive families for a particular choice of $r$ that depends on $d$, $p_2$, and $c$, 
and where $r$ is small compared to $d$ (for instance, we have that $r = \tilde{\Theta}(d^{2/3})$ when $c$ and $p_2$ are constant).

\medskip

We state a simplified version of the lower bound due to O'Donnell et al.~where $r = \sqrt{d}$ that we will use as a tool to prove our lower bound for Braun-Blanquet similarity. 
The full proof of Lemma \ref{lem:owzsimple} is given in Appendix \ref{sparse:app:lower}.
\begin{lemma} \label{lem:owzsimple}
	For every $d \in \mathbb{N}$, $1/d \leq p_2 \leq 1 - 1/d$, and $1 \leq c \leq d^{1/8}$ 
	every $(\sqrt{d}, c\sqrt{d}, p_1, p_2)$-sensitive hash family $\LSH$ for $\bitcube{d}$ under Hamming distance must have 
	\begin{equation}
		\rho(\LSH) = \frac{\log(1/p_{1})}{\log(1/p_{2})} \geq \frac{1}{c} - O(d^{-1/4}).
	\end{equation}
\end{lemma}
In general, good lower bounds for the entire parameter space $(r, cr)$ are not known, 
although the techniques by O'Donnell et al. appear to yield a bound of $\rho \gtrsim \log(1-2r/d)/\log(1-2cr/d)$.
This is far from tight as can be seen by comparing it to the bit-sampling \cite{indyk1998} upper bound of $\rho = \log(1-r/d)/\log(1-cr/d)$. 
Existing lower bounds are tight in two different settings. 
First, in the setting where $cr \approx d/2$ (random data), lower bounds \cite{motwani2007, dubiner2010bucketing, andoni2016tight} 
match various instantiations of angular LSH \cite{terasawa2007spherical, andoni2014beyond, andoni2015practical}. 
Second, in the setting where $r \ll d$, the lower bound by O'Donnell et al.~\cite{odonnell2014optimal} becomes $\rho \gtrsim \log(1-2r/d)/\log(1-2cr/d) \approx 1/c$, 
matching bit-sampling LSH~\cite{indyk1998} as well as Angular LSH.
\subsection{Braun-Blanquet LSH lower bound}
We are now ready to prove the LSH lower bound from Theorem \ref{thm:lower}.
The lower bound together with Corollary \ref{cor:lsmtolsh} shows that the $\rho$-value of Theorem \ref{thm:upper} 
is best possible up to $o_{d}(1)$ terms within the class of data-independent locality-sensitive maps for Braun-Blanquet similarity.
Furthermore, the lower bound also applies to angular distance on the unit sphere where it comes close to matching the best known upper bounds 
for much of the parameter space as can be seen from Figure~\ref{fig:angular}. 

\paragraph{Proof sketch.}
The proof works by assuming the existence of a $(b_1, b_2, p_1, p_2)$-sensitive family $\LSH_{B}$ 
for $\bitcube{d}$ under Braun-Blanquet similarity with $\rho = \log(1/b_1)/\log(1/b_2) - \gamma$ for some $\gamma > 0$.
We use a transformation $T$ from Hamming space to Braun-Blanquet similarity to show that the existence of $\LSH_{B}$ implies the existence of a $(r, cr, p_{1}', p_{2}')$-sensitive 
family $\LSH_{H}$ for $D$-dimensional Hamming space that will contradict the lower bound of O'Donnell et al.~\cite{odonnell2014optimal} 
as stated in Lemma \ref{lem:owzsimple} for some appropriate choice of $\gamma = \gamma(d, p_{2})$.

We proceed by giving an informal description of a simple ``tensoring'' technique for converting a similarity search problem in Hamming space 
into a Braun-Blanquet set similarity problem for target similarity thresholds $b_1, b_2$.  
For $x \in \bitcube{d}$ define 
$$\tilde{x} = \{(i, x_{i}) \mid i \in [d] \}$$ 
and for a positive integer $\tau$ define $x^{\otimes \tau} = \{ (v_1, \dots, v_\tau) \mid v_i \in \tilde{x} \}$.
We have that $|x^{\otimes \tau}| = |\tilde{x}|^{\tau} = d^\tau$ 
and 
$$B(x^{\otimes \tau}, y^{\otimes \tau}) = |\tilde{x} \cap \tilde{y}|^{\tau} / |\tilde{x}|^\tau = (1 - r/d)^\tau$$ 
where $r = \norm{x - y}_{1}$.
For every choice of constants $0 < b_2 < b_1 < 1$ we can choose $d$, $\tau$, $r$, and $c \geq 1$ such that $(1 - r/d)^\tau \approx b_1$ and $(1 - cr/d)^\tau \approx b_2$.
Now, given an LSH family for Braun-Blanquet with $\rho < \log(1/b_1)/\log(1/b_2)$ we would be able to obtain an LSH family for Hamming space with 
\begin{equation*}
	\rho < \frac{\log(1/b_1)}{\log(1/b_2)} = \frac{\log(1/(1 - r/d))}{\log(1/(1-cr/d))} \leq 1/c.
\end{equation*}
For appropriate choices of parameters this would contradict the O'Donnell et al.~LSH lower bound of $\rho \gtrsim 1/c$ for Hamming space. 
The proof itself is mostly an exercise in setting parameters and applying the right bounds and approximations to make everything fit together with the intuition above.
Importantly, we use sampling in order to map to a dimension that is much lower than the $d^\tau$ from the proof sketch in order to make the proof hold for small values of $p_2$ in relation to $d$. 

\paragraph{Hamming distance to Braun-Blanquet similarity.}
Let $d \in \mathbb{N}$ and let $0 < b_2 < b_1 < 1$ be constant as in Theorem \ref{thm:lower}.
Let $\varepsilon \geq 1/d$ be a parameter to be determined.
We want to show how to use a transformation $T \colon \bitcube{D} \to \bitcube{d}$ from Hamming distance to Braun-Blanquet similarity
together with our family $\LSH_{B}$ to construct a $(r, cr, p_{1}', p_{2}')$-sensitive family $\LSH_{H}$ for $D$-dimensional Hamming space with parameters
\begin{align*}
D &= 2^d \\
r &= \sqrt{D} \\
c &= \frac{\ln(1/(b_2 - \varepsilon))}{\ln(1/(b_1 + \varepsilon))}
\end{align*}
where $p_{1}'$ and $p_{2}'$ remain to be determined.

The function $T$ takes as parameters positive integers $t$, $l$, and~$\tau$.
The output of $T$ consists of $t$ concatenated $l$-bit strings, each of of Hamming weight one.
Each of the $t$ strings is constructed independently at random according to the following process:
Sample a vector of indices ${i} = (i_1, i_2, \dots, i_{\tau})$ uniformly at random from $[D]^{\tau}$ and define 
$x_{{i}} \in \bitcube{\tau}$ as $x_{{i}} = x_{i_{1}} \circ x_{i_{2}} \circ \dots \circ x_{i_{\tau}}$.
Let $z(x) \in \bitcube{2^{\tau}}$ be indexed by $j \in \bitcube{\tau}$ and set the bits of $z(x)$ as follows: 
\begin{equation*}
	z(x)_{j} =
	\begin{cases}
		1 & \text{if } x_{{i}} = j, \\
		0 & \text{otherwise.} 
	\end{cases}
\end{equation*}
Next we apply a random function $g \colon \bitcube{\tau} \to [l]$ in order to map $z(x)$ down to an $l$-bit string ${r}(z(x))$ 
of Hamming weight one while approximately preserving Braun-Blanquet similarity. For $i \in [l]$ we set
\begin{equation*}
	{r}(z(x))_{i} = \bigvee_{j : g(j) = i} z(x)_{j}.
\end{equation*}
Finally we set 
\begin{equation*}
	T(x) = {r}_{1}(z_{1}(x)) \circ {r}_{2}(z_{2}(x)) \circ \dots \circ {r}_{t}(z_{t}(x))
\end{equation*}
where each ${r}_{i}(z_{i}(x))$ is constructed independently at random.

We state the properties of $T$ for the following parameter setting:
\begin{align*}
\tau &= \lfloor \sqrt{D} \ln(1/(b_1 + \varepsilon)) \rfloor \\
l &= \lceil 8/\varepsilon \rceil \\
t &= \lfloor d/l \rfloor.
\end{align*}
\begin{lemma}
	For every $d \in \mathbb{N}$ and $D = 2^{d}$ there exists a distribution over functions of the form $T \colon \bitcube{D} \to \bitcube{d}$ such that for all $x, y \in \bitcube{D}$ and random $T$:
\begin{enumerate}
	\item $|T(x)| = t$.
	\item If $\norm{x-y}_1 \leq r$ then $B(T(x), T(y)) \geq {b_1}$ with probability at least $1 - e^{-t\varepsilon^{2}/2}$.
	\item If $\norm{x-y}_1 > cr$ then $B(T(x), T(y)) < {b_2}$ with probability at least $1 - 2e^{-t\varepsilon^{2}/32}$.
\end{enumerate}
\end{lemma}
\begin{proof}
The first property is trivial. 
For the second property we consider $x, y$ with $\norm{x - y}_1 \leq r$ where we would like to lower bound  
\begin{equation*}
B(T(x), T(y)) = \frac{|T(x) \cap T(y)|}{\max(|T(x)|, |T(y)|)}.
\end{equation*}
We know that $|T(x)| = |T(y)| = t$ so it remains to lower bound the size of the intersection $|T(x) \cap T(y)|$.
Consider the expectation
\begin{equation*}
	\E[|T(x) \cap T(y)|] = t\Pr[z(x) = z(y)].
\end{equation*}
We have that $z(x) = z(y)$ if $x$ and $y$ take on the same value in the $\tau$ underlying bit-positions that are sampled to construct $z$.
Under the assumption that $\varepsilon \geq 1/d$, 
then for $d$ greater than some sufficiently large constant we can use a standard approximation to the exponential function (detailed in Lemma \ref{lem:exp} in Appendix \ref{sparse:app:lower}) to show that
\begin{align*}
	\Pr[z(x) = z(y)] &\geq (1 - r/D)^\tau \\ 
						 &\geq (1 - 1/\sqrt{D})^{\sqrt{D} \ln(1/(b_1 + \varepsilon))} \\
						 &\geq e^{\ln(b_1 + \varepsilon)}(1 - (\ln(b_1 + \varepsilon))^2 / \sqrt{D}) \\
						 &\geq b_1 + \varepsilon/2.
\end{align*}
Seeing as $|T(x) \cap T(y)|$ is the sum of $t$ independent Bernoulli trials we can apply Hoeffding's inequality to yield the following bound:
\begin{equation*}
	\Pr[|T(x) \cap T(y)| \leq b_1 t] \leq e^{-t \varepsilon^{2}/2}. 
\end{equation*}
This proves the second property of $T$.

For the third property we consider the Braun-Blanquet similarity of distant pairs of points $x, y$ with $\norm{x - y}_{1} > cr$.
Again, under our assumption that $\varepsilon \geq 1/d$ and for $d$ greater than some constant we have
\begin{align*}
	\Pr[z(x) = z(y)] &\leq (1 - cr/D)^\tau \\ 
																	  &\leq \frac{\left(1 - \frac{\ln(1/(b_2 - \varepsilon))}{\sqrt{D} \ln(1/(b_1 + \varepsilon))}\right)^{\sqrt{D} \ln(1/(b_1 + \varepsilon))}} {1 - c/\sqrt{D} } \\
						  &\leq (1 + 2c/\sqrt{D})(b_2 - \varepsilon) \\
						  &\leq b_2 - \varepsilon/2.
\end{align*}
There are two things that can cause the event $B(T(x), T(y)) < {b_2}$ to fail.
First, the sum of the $t$ independent Bernoulli trials for the event $z(x) = z(x')$ can deviate too much from its expected value.
Second, the mapping down to $l$-bit strings that takes place from $z(x)$ to ${r}(z(x))$ can lead to an additional increase in the similarity due to collisions.
Let $Z$ denote the sum of the~$t$ Bernoulli trials for the events $z(x) = z(x')$ associated with $T$. 
We again apply a standard Hoeffding bound to show that
\begin{equation*}
	\Pr[Z \geq (b_2 - \varepsilon/4)t] \leq e^{-t\varepsilon^{2}/8}.
\end{equation*}
Let $X$ denote the number of collisions when performing the universe reduction to $l$-bit strings.
By our choice of $l$ we have that $E[X] \leq (\varepsilon/8)t$. Another application of Hoeffding's inequality shows that
\begin{equation*}
	\Pr[X \geq (\varepsilon/4)t] \leq e^{-t\varepsilon^{2}/32}.
\end{equation*}
We therefore get that
\begin{equation*}
	\Pr[|T(x) \cap T(x')| \geq b_2 t] \leq 2e^{-t \varepsilon^{2}/32}. 
\end{equation*}
This proves the third property of $T$.
\end{proof}
\paragraph{Contradiction.}
To summarize, using the random map $T$ together with the LSH family $\LSH_{B}$ we can obtain an $(r, cr, {p_{1}}', {p_{2}}')$-sensitive family $\LSH_{H}$ for $D$-dimensional Hamming space 
with ${p_{1}}' = {p_{1}} - \delta$ and ${p_{2}}' = {p_{2}} + \delta$ for $\delta = 2e^{-t \varepsilon^{2}/32}$. 
For our choice of $c = \frac{\ln(1/(b_2 - \varepsilon))}{\ln(1/(b_1 + \varepsilon))}$ we plug the family $\LSH_{H}$ into the lower bound of Lemma \ref{lem:owzsimple} 
and use that $O(D^{-1/4}) = O(\varepsilon)$ which follows from our constraint that $\varepsilon \geq 1/d$.
\begin{align*}
	\rho(\LSH_{H}) &\geq 1/c - O(D^{-1/4}) \\
				   &= \frac{\ln(1/(1 + \varepsilon/b_{1})) + \ln(1/b_{1})}{\ln(1/(1-\varepsilon/b_{2})) + \ln(1/b_2)} - O(\varepsilon)\\
				   &\geq \frac{\ln(1/b_{1}) - \varepsilon/b_{1}}{\ln(1/b_2) + 2\varepsilon/b_{2}} - O(\varepsilon)\\
				   &= \frac{\ln(1/b_{1})}{\ln(1/b_2)} - O(\varepsilon) 
\end{align*}
Under our assumed properties of $\LSH_{B}$, we can upper bound the value of $\rho$ for $\LSH_{H}$. 
For simplicity we temporarily define $\lambda = 2\delta/{p_{2}}$ and assume that $\lambda / \ln(1/{p_{2}}) \leq 1/2$ and $\ln(1/{p_{2}}) \geq 1$.
The latter property holds without loss of generality through use of the standard LSH powering technique \cite{indyk1998, har-peled2012, odonnell2014optimal} 
that allows us to transform an LSH family with ${p_{2}} < 1$ to a family that has ${p_{2}} \leq 1/e$ without changing its associated $\rho$-value. 
\begin{align*}
	\rho(\LSH_{H}) &= \frac{\ln(1/{p_{1}}')}{\ln(1/{p_{2}}')} = \frac{ \ln(1/{p_{1}}) + \ln(1/(1-\delta/{p_{1}}))}{\ln(1/{p_{2}}) + \ln(1/(1 + \delta/{p_{2}}))} \\
				   &\leq \frac{\ln(1/{p_{1}}) + \lambda}{\ln(1/{p_{2}}) - \lambda} = \frac{\ln(1/{p_{1}}) + \lambda}{(\ln 1/{p_{2}} )(1 - \lambda / (\ln 1/{p_{2}}))} \\
				   &\leq \frac{\ln(1/{p_{1}}) + \lambda}{\ln(1/{p_{2}})}(1 +  2 \lambda / (\ln 1/{p_{2}})) = \frac{\ln(1/{p_{1}})}{\ln(1/{p_{2}})} + O(\delta / {p_{2}}) \\
				   &\leq \frac{\ln(1/b_{1})}{\ln(1/b_{2})} - \gamma  + O(\delta / {p_{2}}). 
\end{align*}
We get a contradiction between our upper bound and lower bound for $\rho(\LSH_{H})$ whenever $\gamma$ violates the following relation that summarizes the bounds:
\begin{equation*}
\frac{\ln(1/b_{1})}{\ln(1/b_2)} - O(\varepsilon) \leq  \rho(\LSH_{H}) \leq \frac{\ln(1/b_{1})}{\ln(1/b_{2})} - \gamma  + O(\delta / {p_{2}}).
\end{equation*}
In order for a contradiction to occur, the value of $\gamma$ has to satisfy
\begin{equation*}
\gamma > O(\varepsilon) + O(\delta / {p_{2}}).
\end{equation*}
By our setting of $t = \lfloor d/l \rfloor$ and $l = \lceil 8/\varepsilon \rceil$ we have that $\delta = e^{-\Omega(d\varepsilon^{3})}$.
We can cause a contradiction for a setting of $\varepsilon^{3} = K\frac{\ln(d/{p_{2}})}{d}$ where $K$ is some constant and where we assume that $d$ is greater than some constant.
The value of $\gamma$ for which the lower bound holds can be upper bounded by
\begin{equation*}
	\gamma = O\left(\frac{\ln(d/{p_{2}})}{d}\right)^{1/3}.
\end{equation*}
This completes the proof of Theorem \ref{thm:lower}.

\section{Equivalent set similarity problems}\label{sec:equivalence}
In this section we consider how to use our data structure for Braun-Blanquet similarity search to support other similarity measures such as Jaccard similarity.
We already observed in the introduction that a direct translation exists between several similarity measures whenever the size of every sets is fixed to $t$.
Call an $(s_1,s_2)$-similarity search problem \emph{($t$,$t'$)-regular} if $P$ is restricted to vectors of weight $t$ and queries are restricted to vectors of weight~$t'$.
Obviously, a $(t,t')$-regular similarity search problem is no harder than the general similarity search problem, but it also cannot be too much easier when expressed as a function of the thresholds $(s_1,s_2)$:
For every pair $(t,t') \in \{0,\dots,d\}^2$ we can construct a ($t$,$t'$)-regular data structure (such that each point $x \in P$ is represented in the $d+1$ data structures with $t=|x|$), 
and answer a query for $q\in\{0,1\}^d$ by querying all data structures with $t'=|q|$.
Thus, the time and space for the general $(s_1,s_2)$-similarity search problem is at most $d+1$ times larger than the time and space of the most expensive ($t$,$t'$)-regular data structure.
This does \emph{not} mean that we cannot get better bounds in terms of other parameters, and in particular we expect that the difficulty of $(t,t')$-regular similarity search problems depends on parameters $t$ and $t'$.

\paragraph{Dimension reduction.}
If the dimension is large a factor of $d$ may be significant.
However, for most natural similarity measures a $(s_1,s_2)$-similarity problem in $d \gg (\log n)^3$ dimensions can be reduced to a logarithmic number of $(s'_1,s'_2)$-similarity problems 
on $P'\subseteq \{0,1\}^{d'}$ in $d'=(\log n)^3$ dimensions with $s'_1 = s_1-O(1/\log n)$ and $s'_2 = s_2+O(1/\log n)$.
Since the similarity gap is close to the one in the original problem, $s'_1 - s'_2 = s_1 - s_2 - O(1/\log n)$, where $s_1$ and $s_2$ are assumed to be independent of $n$, the difficulty ($\rho$-value) remains essentially the same.
First, split $P$ into $\log d$ size classes $P_i$ such that vectors in class $i$ have size in $[2^{i}, 2^{i+1})$.
For each size class the reduction is done independently and works by a standard technique: 
sample a sequence of random sets $I_j\subseteq \{1,\dots,d\}$, $i=1,\dots,d'$, and set $x'_j = \vee_{\ell\in I_j} x_\ell$.
The size of each set $I_j$ is chosen such that $Pr[x'_j = 1] \approx 1/\log(n)$ when $|x|=2^{i+1}$.
By Chernoff bounds this mapping preserves the relative weight of vectors up to size $2^i \log n$ up to an additive $O(1/\log n)$ term with high probability.
Assume now that the similarity measure is such that for vectors in $P_i$ we only need to consider $|q|$ in the range from $2^i/\log n$ to $2^i \log n$ (since if the size difference is larger, the similarity is negligible).
The we can apply Chernoff bounds to the relative weights of the dimension-reduced vectors $x'$, $q'$ and the intersection $x' \cap q'$.
In particular, we get that the Jaccard similarity of a pair of vectors is preserved up to an additive error of $O(1/\log n)$ with high probability.
The class of similarity measures for which dimension reduction to $(\log n)^{O(1)}$ dimensions is possible is large, and we do not attempt to characterize it here.
Instead, we just note that for such similarity measures we can determine the complexity of similarity search up to a factor $(\log n)^{O(1)}$ by only considering regular search problems.

\paragraph{Equivalence of regular similarity search problems.}
We call a set similarity measure on $\{0,1\}^d$ \emph{symmetric} if it can be written in the form $S(q,x) = f_{d,|q|,|x|}(|q\cap x|)$, 
where each function $f_{d,|q|,|x|} \colon \mathbb{N} \rightarrow [0,1]$ is nondecreasing.
All 59 set similarity measures listed in the survey~\cite{choi2010survey}, normalized to yield similarities in $[0,1]$, are symmetric.
In particular this is the case for Jaccard similarity (where $J(q,x) = |q\cap x| /(|q|+|x|-|q\cap x|)$) and for Braun-Blanquet similarity.
For a symmetric similarity measure, the predicate $\simil(q,x)\geq s_1$ is equivalent to the predicate $|q\cap x|\geq i_1$, where $i_1 = \min \{ i \; | \; f_{d,t',t}(i)\geq s_1 \}$, 
and $\simil(q,x) > s_2$ is equivalent to the predicate $|q\cap x|\geq i_2$, where $i_2 = \min \{ i \; | \; f_{d,t',t}(i) > s_2 \}$.
This means that every ($t$,$t'$)-regular $(s_1,s_2)$-similarity search problem on $P\subseteq \{0,1\}^d$ is equivalent to an $(i_1/d,i_2/d)$-similarity search problem on $P$, where $\simil(q,x)=|x\cap q|/d$.
In other words, all symmetric similarity search problems can be translated to each other, and it suffices to study a single one, such as Braun-Blanquet similarity.

\paragraph{Jaccard similarity.}
We briefly discuss Jaccard similarity since it is the most widely used measure of set similarity.
If we consider the problem of $(j_1, j_2)$-approximate Jaccard similarity search in the 
$(t,t')$-regular case with $t\ne t'$
then our Theorem \ref{thm:upper} is no longer guaranteed to yield the lowest value of $\rho$ among competing data-independent approaches such as MinHash and Angular LSH.
To simplify the comparision between different measures we introduce parameters $\beta$ and $b$ defined by $|y| = \beta |x|$ and $b = |x \cap y|/|x|$ (note that $0 \leq b \leq \beta \leq 1$).
The three primary measures of set similarity considered in this paper can then be written as follows:
\begin{align*}
	B(x, y) &= b \\
	J(x, y) &= \frac{b}{1 + \beta - b} \\
	C(x, y) &= \frac{b}{\sqrt{\beta}}
\end{align*}
As shown in Figure \ref{fig:jaccard} among angular LSH, MinHash, and \textsc{Chosen Path}, the technique with the lowest $\rho$-value is different depending on the parameters $(j_1, j_2)$ and asymmetry $\beta$. 
\begin{figure*}
\centering
\subfloat[$\beta = 0.25$]{\includegraphics[width = 0.5\columnwidth]{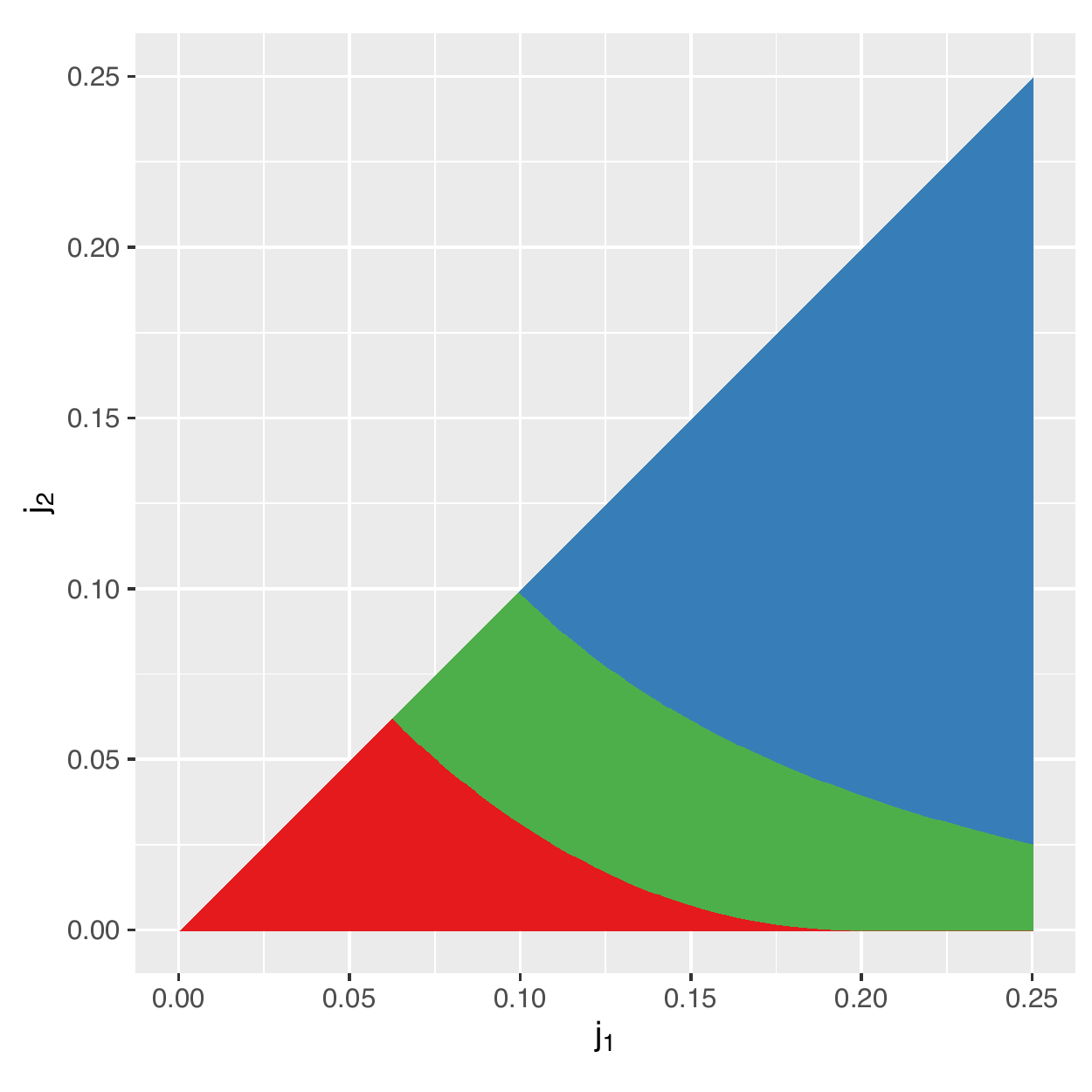}} 
\subfloat[$\beta = 0.5$]{\includegraphics[width = 0.5\columnwidth]{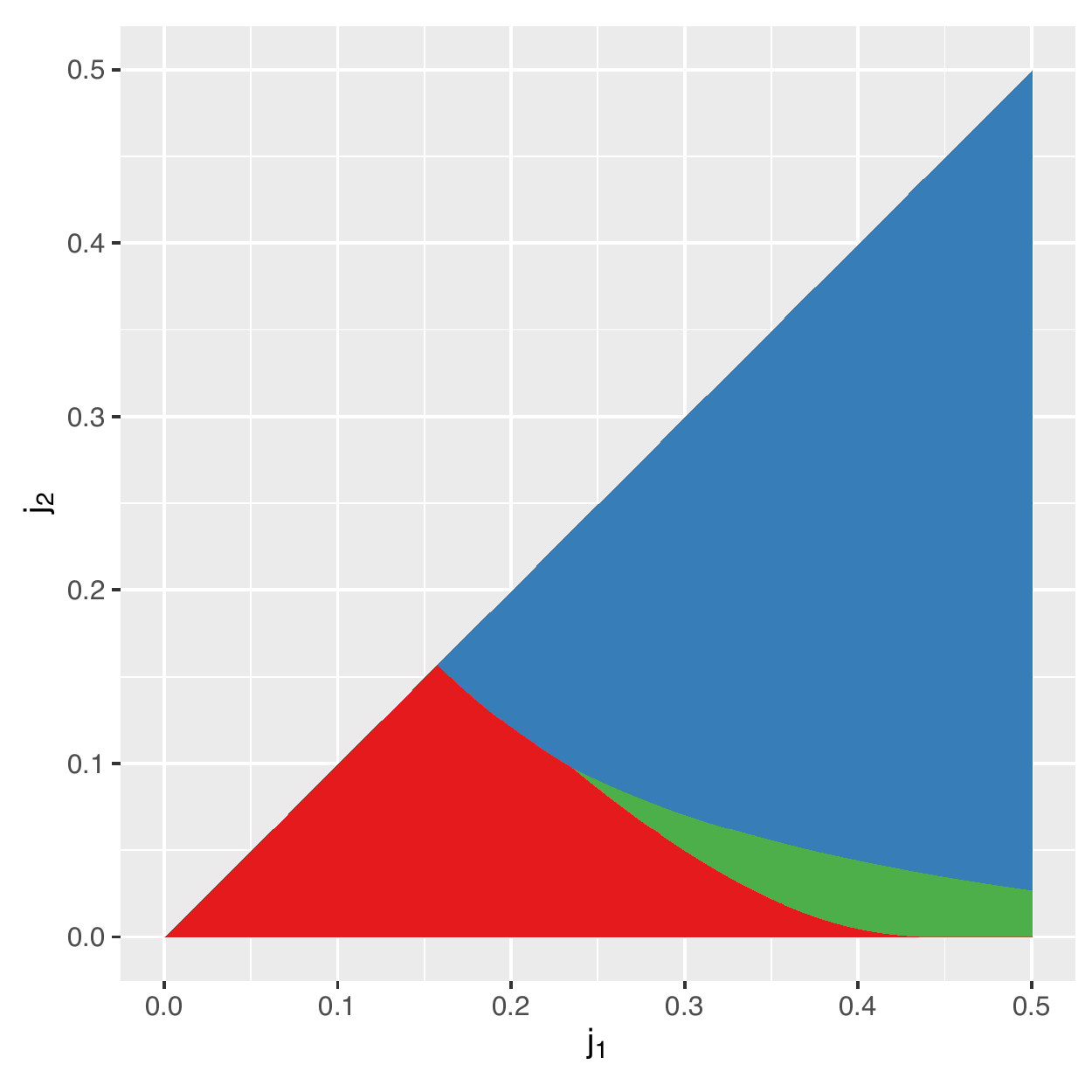}}   \\ 
\subfloat[$\beta = 0.75$]{\includegraphics[width = 0.5\columnwidth]{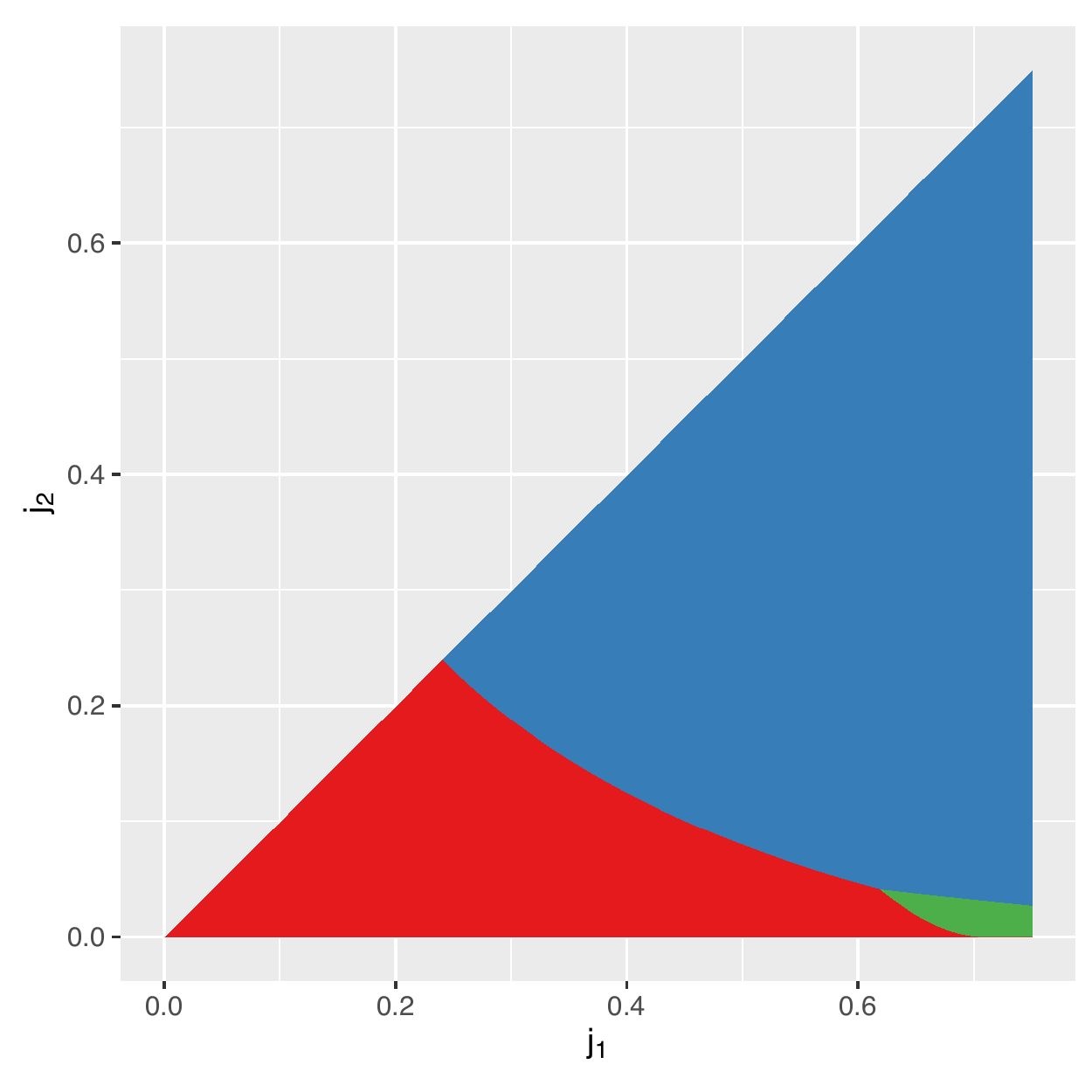}}    
\caption{Solution with lowest $\rho$-value for the $(j_1, j_2)$-approximate Jaccard similarity search problem for different values of $\beta$. 
Blue is angular LSH. Green is MinHash. Red is \textsc{Chosen Path}. 
Note the difference in the axes for different values of $\beta$ as it must hold that $0 \leq j_2 \leq j_1 \leq \beta$.}
\label{fig:jaccard}
\end{figure*}
We know that \textsc{Chosen Path} is optimal and strictly better than the competing data-independent techniques across the entire parameter space $(j_1, j_2)$ when $\beta = 1$, 
but it remains open to find tight upper and lower bounds in the case where $\beta \neq 1$.
\section{Conclusion and open problems}
We have seen that, perhaps surprisingly, there exists a relatively simple way of strictly improving the $\rho$-value for data-independent set similarity search in the case where all sets have the same size.
To implement the required locality-sensitive map efficiently we introduce a new technique based on branching processes 
that could possibly lead to more efficient solutions in other settings.

It remains an open problem to find tight upper and lower bounds on the $\rho$-value for Jaccard and cosine similarity search that hold for the entire parameter space in the general setting with arbitrary set sizes. 
Perhaps a modified version of the \textsc{Chosen Path} algorithm can yield an improved solution to Jaccard similarity search in general.
One approach is to generalize the condition $h_i(p \circ j) < x_{j} / b_{1}|x|$ to use different thresholds for queries and updates.
This yields different space-time tradeoffs when applying the \textsc{Chosen Path} algorithm to Jaccard similarity search.

Another interesting question is if the improvement shown for sparse vectors can be achieved in general for inner product similarity.
A similar, but possibly easier, direction would be to consider \emph{weighted} Jaccard similarity.
\section*{Acknowledgment}
We thank Thomas Dybdahl Ahle for comments on a previous version of this manuscript.
\section{Appendix: Details behind the lower bound}\label{sparse:app:lower}
\subsection{Tools}
For clarity we state some standard technical lemmas that we use to derive LSH lower bounds. 
\begin{lemma}[{Hoeffding \cite[Theorem 1]{hoeffding1963}}] \label{sparse:lem:hoeffding}
	Let $X_1, X_2, \dots, X_n$ be independent random variables satisfying $0 \leq X_i \leq 1$ for $i \in [n]$.
	Define $X = X_1 + X_2 + \dots + X_n$, $Z = X/n$, and $\mu = \E[Z]$, then:
	\begin{itemize}
		\item[-] For $\hat{\mu} \geq \mu$ and $0 < \varepsilon < 1 - \hat{\mu}$ we have that $\Pr[Z - \hat{\mu} \geq \varepsilon] \leq e^{-2n\varepsilon^{2}}$.
		\item[-] For $\hat{\mu} \leq \mu$ and $0 < \varepsilon < \hat{\mu}$ we have that $\Pr[Z - \hat{\mu} \leq - \varepsilon] \leq e^{-2n\varepsilon^{2}}$.
	\end{itemize}
\end{lemma}
\begin{lemma}[{Chernoff \cite[Thm.~4.4 and~4.5]{mitzenmacher2005}}] \label{sparse:lem:chernoff} 
 	Let $X_1, \dots, X_n$ be independent Poisson trials and define $X = \sum_{i=1}^{n}X_i$ and $\mu = \E[X]$.
 	Then, for $0 < \varepsilon < 1$ we have 
 	\begin{itemize}
 		\item[-] $\Pr[X \geq (1+\varepsilon)\mu] \leq e^{-\varepsilon^2 \mu / 3}$.
 		\item[-] $\Pr[X \leq (1-\varepsilon)\mu] \leq e^{-\varepsilon^2 \mu / 2}$.
 	\end{itemize}
\end{lemma}
\begin{lemma}[Bounding the logarithm {\cite{topsoe2007}}] \label{lem:lnbounds}
 	For $x > -1$ we have that $\tfrac{x}{1+x} \leq \ln(1 + x) \leq x$. 
\end{lemma}
\begin{lemma}[Approximating the exponential function {\cite[Prop.~B.3]{motwani2010randomized}}]\label{lem:exp} 
 	For all $t, n \in \real$ with $|t| \leq n$ we have that $e^{t}(1 - \tfrac{t^2}{n}) \leq (1 + \tfrac{t}{n})^n \leq e^t$.
\end{lemma}
\subsection{Proof of Lemma \ref{lem:owzsimple}} 
\paragraph{Preliminaries.}
We will reuse the notation of Section 3.~from O'Donnell et al.~\cite{odonnell2014optimal}.
\begin{definition}
	For $0 \leq \lambda < 1$ we say that $(x, y)$ are $(1-\lambda)$-correlated if $x$ is chosen uniformly at random from $\bitcube{d}$ 
and $y$ is constructed by rerandomizing each bit from $x$ independently at random with probability $\lambda$.
\end{definition}
Let $(x, y)$ be $e^{-t}$-correlated and let $\LSH$ be a family of hash functions on $\bitcube{d}$, then we define
\begin{equation*}
	\K_{\LSH}(t) = \Pr_{\substack{h \sim \LSH \\ (x, y)\, e^{-t}\text{- corr'd}}}[h(x) = h(y)].
\end{equation*}
We have that $\K_{\LSH}(t)$ is a log-convex function which implies the following property that underlies the lower bound: 
\begin{lemma}\label{lem:logconvexity}
	For every family of hash functions $\LSH$ on $\bitcube{d}$, every $t \geq 0$, and $c \geq 1$ we have
	\begin{equation}
		\frac{\ln(1/\K_{\LSH}(t))}{\ln(1/\K_{\LSH}(ct))} \geq \frac{1}{c}.
	\end{equation}
\end{lemma}
The idea behind the proof is to tie $p_1$ to $\sen{t}$ and $p_2$ to $\sen{ct}$ through Chernoff bounds and then apply Lemma \ref{lem:logconvexity} to show that $\rho \gtrsim 1/c$.

\paragraph{Proof.}
Begin by assuming that we have a family $\LSH$ that satisfies the conditions of Lemma \ref{lem:owzsimple}.
Note that the expected Hamming distance betwee $(1-\lambda)$-correlated points $x$ and $y$ is given by $(\lambda/2)d$.
We set $\lambda_{p_{1}}/2 = d^{-1/2} - d^{-5/8}$ and $\lambda_{p_{2}}/2 = cd^{-1/2} + 2cd^{-5/8}$ and let $(x, y)$ denote $(1 - \lambda_{p_{1}})$-correlated random strings 
and $(x, x')$ denote $(1 - \lambda_{p_{2}}q$)-correlated random strings.
By standard Chernoff bounds we get the following guarantees:
\begin{align*}
	\Pr[\norm{x-y}_{1} \geq r] &\leq e^{-\Omega(d^{1/4})}, \\
	\Pr[\norm{x-x'}_1 \leq cr] &\leq e^{-\Omega(d^{1/4})}.
\end{align*}
We will establish a relationship between $\K_{\LSH}(t_{p_{1}})$ and ${p_{1}}$ on the one hand, and $\K_{\LSH}(t_{p_{2}})$ and ${p_{2}}$ on the other hand, 
for the following choice of parameters $t_{p_{1}}$ and $t_{p_{2}}$:
\begin{align*}
	t_{p_{1}} &= -\ln(1 - 2(d^{-1/2} - d^{-5/8})) \\
	t_{p_{2}} &= -\ln(1 - 2c(d^{-1/2} + 2d^{-5/8})).
\end{align*}
By the properties of $\LSH$ and from the definition of $\K_{\LSH}$ we have that
\begin{align*}
	\K_{\LSH}(t_{p_{1}}) &\geq {p_{1}}(1 - \Pr[\norm{x - y}_1 > r]) \geq {p_{1}} - \Pr[\norm{x - y}_1 \geq r] \\
	\K_{\LSH}(t_{p_{2}}) &\leq {p_{2}}(1 - \Pr[\norm{x - x'}_1 \leq cr]) + \Pr[\norm{x - x'}_1 \leq cr] \\
						 &\leq {p_{2}} + \Pr[\norm{x - x'}_1 \leq cr].
\end{align*}
Let $\delta = \max\{\Pr[\norm{x - y}_1 \geq r], \Pr[\norm{x - x'}_1 \leq cr]\} = e^{-\Omega(d^{1/4})}$.
By Lemma \ref{lem:logconvexity} and our setting of $t_{p_{1}}$ and $t_{p_{2}}$ we can use the bounds on the natural logarithm from Lemma \ref{lem:lnbounds} to show the following: 
\begin{align*}
	\frac{\ln(1/\K_{\LSH}(t_{p_{1}}))}{\ln(1/\K_{\LSH}(t_{p_{2}}))} &\geq \frac{t_{p_{1}}}{t_{p_{2}}} = \frac{\ln(1 - 2(d^{-1/2} - d^{-5/8}))}{\ln(1 - 2c(d^{-1/2} + 2d^{-5/8}))} \\
	&\geq \frac{2(d^{-1/2} - d^{-5/8})}{2c(d^{-1/2} + 2d^{-5/8})} -  2(d^{-1/2} - d^{-5/8}) \\
	&\geq \frac{1 - d^{-1/4}}{c + 2d^{-1/4}} - 2(d^{-1/2} - d^{-5/8}) \\
	&= \frac{1}{c} - O(d^{-1/4}).
\end{align*}
We proceed by lower bounding $\rho$ where we make use of the inequalities derived above.
\begin{equation*}
\sen{t_{p_{2}}} - \delta \leq {p_{2}} < {p_{1}} \leq \sen{t_{p_{1}}} + \delta.
\end{equation*}
By Lemma \ref{lem:logconvexity} combined with the restrictions on our parameters,
for $d$ greater than some constant we have that $\sen{t_{p_{2}}} \geq \sen{t_{p_{1}}}^{2c} \geq ({p_{1}}/2)^{2c} \geq (2d)^{-2c} \geq (2d)^{-2d^{1/8}}$.
Furthermore, we lower bound $\ln(1/\sen{t_{p_{2}}})$ by using that $\sen{t_{p_{2}}} \leq {p_{2}} + \delta$ together with the restriction that ${p_{2}} \geq 1 - 1/d$ and the properties of $\delta$.
For $d$ greater than some constant it therefore holds that $\sen{t_{p_{2}}} \leq 1 - 1/2d$ from which it follows that $\ln(1/\sen{t_{p_{2}}}) \geq 1/2d$.   
\begin{align*}
	\frac{\ln(1/{p_{1}})}{\ln(1/{p_{2}})} &\geq \frac{\ln(1/(\sen{t_{p_{1}}} + \delta))}{\ln(1/(\sen{t_{p_{2}}} - \delta))} \\
							  &= \frac{\ln(1/\sen{t_{p_{1}}}) - \ln(1 + \delta/\sen{t_{p_{1}}})}{\ln(1/\sen{t_{p_{2}}}) + \ln(1/(1 - \delta/\sen{t_{p_{2}}}))} \\
							  &\geq \frac{\ln(1/\sen{t_{p_{1}}}) - \delta/\sen{t_{p_{1}}}}{\ln(1/\sen{t_{p_{2}}}) + 2\delta/\sen{t_{p_{2}}}} \\
							  &\geq \frac{\ln(1/\sen{t_{p_{1}}})}{\ln(1/\sen{t_{p_{2}}})} - \frac{3 \delta}{\sen{t_{p_{2}}} \ln(1/\sen{t_{p_{2}}})}.
\end{align*}
By the arguments above we have that 
\begin{equation*}
	\frac{3 \delta}{\sen{t_{p_{2}}} \ln(1/\sen{t_{p_{2}}})} = e^{-\Omega(d^{1/4})} = O(d^{-1/4}).
\end{equation*}
Inserting the lower bound for $\frac{\ln(1/\sen{t_{p_{1}}})}{\ln(1/\sen{t_{p_{2}}})}$ results in the lemma. 

\section{Appendix: Comparisons} \label{app:comparison}
For completeness we state the proofs behind the comparisons between the $\rho$-values obtained by the \textsc{Chosen Path} algorithm and other LSH techniques.
\subsection{MinHash}
For data sets with fixed sparsity and Braun-Blanquet similarities $0 < b_2 < b_1 < 1$ we have that $\rho/\rho_{\text{minhash}} = f(b_{2})/f(b_{1})$ where $f(x) = \log(x/(2-x)) / \log(x)$.
If $f(x)$ is monotone increasing in $(0,1)$ then $\rho/\rho_{\text{minhash}} < 1$.
For $x \in (0,1)$ we have that $\sign(f'(x)) = \sign(g(x))$ where $g(x) = \ln(x) + (2-x) \ln(2-x)$.
The function $g(x)$ equals zero at $x = 1$ and has the derivative $g'(x) = \ln(x) - \ln(2-x)$ which is negative for values of $x \in (0,1)$.
We can thefore see that $f'(x)$ is positive in the interval and it follows that $\rho < \rho_{\text{minhash}}$ for every choice of $0 < b_2 < b_1 < 1$. 

\subsection{Angular LSH}
We have that $\rho/\rho_{\text{angular}} < 1$ if $f(x) = \ln(x) \frac{1+x}{1-x}$ is a monotone increasing function for $x \in (0,1)$.
For $x \in (0,1)$ we have that $\sign(f'(x)) = \sign(g(x))$ where $g(x) = (1-x^2)/2 + x \ln x$.
We note that $g(1) = 0$ and $g'(x) = 1 - x + \ln x$. 
Therefore, if $g'(x) < 0$ for $x \in (0,1)$ it holds that $g(x) > 0$ and $f(x)$ is monotone increasing in the same interval.
We have that $g'(1) = 0$ and $g''(x) = -1 + 1/x > 0$ implying that $g'(x) < 0$ in the interval. 

\subsection{Data-dependent LSH}
\begin{lemma}\label{lem:fixedrho}
Let $0 < b_2 < b_1 < 1$ and fix $\rho = 1/2$ such that $b_1 = \sqrt{b_2}$.
Then we have that $\rho < \rho_{\text{datadep}}$ for every value of $b_2 < 1/4$. 
\end{lemma}
\begin{proof}
We will compare $\rho = \log(b_{1})/\log(b_{2})$ and $\rho_{\text{datadep}} = \frac{1 - b_{1}}{1 + b_{1} - 2b_{2}}$ when $\rho$ is fixed at $\rho = 1/2$, or equivalently, $b_{1} = \sqrt{b_{2}}$.
We can solve the quadratic equation $1/2 = \frac{1 - \sqrt{b_2}}{1 + \sqrt{b_2} - 2b_{2}}$ to see that for $0 < b_2 < 1$ we have that $\rho = \rho_{\text{datadep}}$ only when $b_2 = 1/4$. 
The derivative of $\rho_{\text{datadep}}$ with respect to $b_2$ is negative when $b_{1} = \sqrt{b_2}$. 
Under this restriction we therefore have that $\rho < \rho_{\text{datadep}}$ for $b_2 < 1/4$ which is equivalent to $j_2 < 1/7$ in the fixed-weight setting.
\end{proof}

To compare $\rho$-values over the full parameter space we use the following two lemmas. 
\begin{lemma} \label{lem:partial}
For every choice of fixed $0 < \rho < 1$ let $b_{2} = b_{1}^{1/\rho}$. 
Then $\rho_{\text{datadep}} = \frac{1 - b_{1}}{1 + b_{1} - 2b_{2}}$ is decreasing in $b_1$ for $b_1 \in (0,1)$.
\end{lemma}
\begin{proof}
The sign of the derivative of $\rho_{\text{datadep}}$ with respect to $b_1$ is equal to the sign of the function $g(x) = -\rho x^{-1/\rho} + \rho - 1 + x^{-1}$ for $x \in (0,1)$.
We have that $g(1) = 0$ and $g'(x) = x{-1/p - 1} - x^{-2} > 0$ for $x \in (0,1)$ which shows that $g(x) < 0$ in the interval.
\end{proof}

\begin{lemma} \label{lem:onefifth}
For $1/5 = b_2 < b_1 < 1$ we have that $\rho < \rho_{\text{datadep}}$.
\end{lemma}
\begin{proof}
For fixed $b_2 = 1/5$ consider $f(b_1) = \rho - \rho_{\text{datadep}}$ as a function of $b_1$ in the interval $[1/5, 1]$.
We want to show that $f(b_1) < 0$ for $b_1 \in (1/5,1)$. 
In the endpoints the function takes the value $0$. 
Between the endpoints we find that $f'(b_1) = \frac{1}{\ln(5)b_1} + \frac{8/5}{(3/5 + b_1)^2}$ and that $f'(b_1) = 0$ is a quadratic form with only one solution $b_{1}^{*}$ in $[1/5,1]$.
By Lemma \ref{lem:fixedrho} we know that that for $b_2 = 1/5$ and $b_1 = 1/\sqrt{5}$ it holds that $f(b_1) < 0$.
Since $f(1/5) = f(1) = 0$, $f'(b_1) = 0$ only in a single point in $[1/5,1]$, and $f(1/\sqrt{5}) < 0$ we can conclude that the lemma holds.
\end{proof}

\begin{corollary}
For every choice of $b_1, b_2$ satisfying $0 < b_2 \leq 1/5$ and $b_2 < b_1 < 1$ we have that $\rho < \rho_{\text{datadep}}$.
\end{corollary}
\begin{proof}
If $b_2 = 1/5$ the property holds by Lemma \ref{lem:onefifth}.
If $b_2 < 1/5$ we define new variables $\hat{b}_2, \hat{b}_2$, setting $\hat{b}_1 = \hat{b}_{1}^{\rho(b_1, b_2)}$ and initially consider $\hat{b}_{2} = 1/5$.
In this setting we again have that $\rho(\hat{b}_{1}, \hat{b}_{2}) < \rho_{\text{datadep}}(\hat{b}_{1},\hat{b}_{2})$.
According to Lemma \ref{lem:partial} it holds that $\rho_{\text{datadep}}$ is decreasing in $b_2$ for fixed $\rho$.
Therefore, as $\hat{b}_{2}$ decreases to $\hat{b}_{2} = b_2$ where $\hat{b}_{1} = b_1$ we have that $\rho(\hat{b}_{1}, \hat{b}_{2}) = \rho$ remains constant while $\rho_{\text{datadep}}$ increases. 
Since it held that $\rho < \rho_{\text{datadep}}$ at the initial values of $\hat{b}_{1}, \hat{b}_{2}$ it must also hold for $b_1, b_2$.
\end{proof}

\paragraph{Numerical comparison of MinHash and Data-dep.~LSH.}
Comparing $\rho_{\text{minhash}}$ to $\rho_{\text{datadep}}$ we can verify numerically that even for $b_2$ fixed as low as $b_2 = 1/23$ 
we can find values of $b_1$ (for example $b_1 = 0.995$ such that $\rho_{\text{minhash}} > \rho_{\text{datadep}}$.
\chapter{Adaptive similarity join}
\sectionquote{A light from the shadows shall spring}
\noindent Set similarity join is a fundamental and well-studied database operator.
It is usually studied in the \emph{exact} setting where the goal is to compute all pairs of sets that exceed a given similarity threshold (measured e.g.~as Jaccard similarity).
But set similarity join is often used in settings where 100\% recall may not be important --- indeed, where the exact set similarity join is itself only an approximation of the desired result set.

We present a new randomized algorithm for set similarity join that can achieve any desired recall up to 100\%, 
and show theoretically and empirically that it significantly improves on existing methods.
The present state-of-the-art exact methods are based on prefix-filtering, the performance of which depends on the prevalence of rare elements in the sets.
Our method is robust against the absence of such structure in the data. 
At 90\% recall our algorithm is often more than an order of magnitude faster than state-of-the-art exact methods, depending on how well a data set lends itself to prefix filtering.
Our experiments on benchmark data sets also show that the method is several times faster than comparable approximate methods.
Our algorithm makes use of recent theoretical advances in high-dimensional sketching and indexing.
\section{Introduction}
It is increasingly important for data processing and analysis systems to be able to work with data that is imprecise, incomplete, or noisy.
\emph{Similarity join} has emerged as a fundamental primitive in data cleaning and entity resolution over the last decade~\cite{augsten2013similarity, chaudhuri2006primitive, sarawagi2004efficient}.
In this paper we focus on \emph{set similarity join}:
Given collections $R$ and $S$ of sets the task is to compute 
$$ R \simjoin S = \{(x,y)\in R\times S \; | \; \simil(x,y)\geq \lambda\}$$
where $\simil(\cdot, \cdot)$ is a similarity measure and $\lambda$ is a threshold parameter.
We deal with sets $x,y \subseteq \{1,\dots,d\}$, where the number $d$ of distinct tokens can be naturally thought of as the dimensionality of the data.

Many measures of set similarity exist~\cite{choi2010survey}, but perhaps the most well-known such measure is the \emph{Jaccard similarity},
$$J(x,y) = |x\cap y|/|x\cup y| \enspace .$$
For example, the sets $x = \{$\texttt{IT, University, Copenhagen}$\}$ and $y = \{$\texttt{University, Copenhagen, Denmark}$\}$  have Jaccard similarity $J(x,y) = 1/2$ which could suggest that they both correspond to the same entity.
In the context of entity resolution we want to find a set~$T$ that contains $(x,y)\in R\times S$ if and only if $x$ and $y$ correspond to the same entity.
The quality of the result can be measured in terms of \emph{precision} $|(R \simjoin S) \cap T|/|T|$ and \emph{recall} $|(R \simjoin S) \cap T|/|R \simjoin S|$, both of which should be as high as possible.
We will be interested in methods that achieve 100\% precision, but that might not have 100\% recall.
We refer to methods with 100\% recall as exact, and others as approximate.
\subsection{Our Contributions}
We present a new approximate set similarity join algorithm: Chosen Path Similarity Join (\textsc{CPSJoin}).
We cover its theoretical underpinnings, and show experimentally that it achieves high recall with a substantial speedup compared to state-of-the-art exact techniques.
The key ideas behind \cpsj are:
\begin{itemize}
	\item A new recursive filtering technique inspired by the recently proposed {\textsc ChosenPath} index for set similarity search~\cite{christiani2017set}, adding new ideas to make the method parameter-free, near-linear space, and adaptive to a given data set.
	\item Apply efficient sketches for estimating set similarity~\cite{li2011theory} that take advantage of modern hardware. 
\end{itemize}

We compare \cpsj to the exact set similarity join algorithms in the comprehensive empirical evaluation of Mann et al.~\cite{mann2016}, using the same data sets, and to other approximate set similarity join methods suggested in the literature.
We find that \cpsj outperforms other approximate methods and scales better than exact methods when the sets are relatively large (100 tokens or more) and the similarity threshold is low (e.g.~Jaccard similarity~0.5) where we see speedups of more than an order of magnitude at 90\% recall.
Our experiments on benchmark datasets show that exact methods are faster in the case of high similarity thresholds, when the average set size is small, and when sets have many rare elements, 
whereas approximate methods are faster in the case of low similarity thresholds and when sets are large.
This finding is consistent with theory and is further corroborated by experiments on synthetic datasets.
\subsection{Related Work}\label{sec:related}
For space reasons we present just a sample of the most related previous work, and refer to the book of Augsten and B{\"o}hlen~\cite{augsten2013similarity} for a survey of algorithms for exact similarity join in relational databases, covering set similarity joins as well as joins based on string similarity.

\paragraph{Exact Similarity Join.}
Early work on similarity join focused on the important special case of detecting \emph{near-duplicates} with similarity close to~1, see e.g.~\cite{broder2000identifying,sarawagi2004efficient}.
A sequence of results starting with the seminal paper of Bayardo et al.~\cite{bayardo2007scaling} studied the range of thresholds that could be handled.
Recently, Mann et al.~\cite{mann2016} conducted a comprehensive study of 7 state-of-the-art algorithms for exact set similarity join for Jaccard similarity threshold $\lambda \in \{0.5,0.6,0.7,0.8,0.9\}$.
These algorithms all use the idea of \emph{prefix filtering}~\cite{chaudhuri2006primitive}, which generates a sequence of candidate pairs of sets that includes all pairs with similarity above the threshold.
The methods differ in how much additional filtering is carried out.
For example,~\cite{xiao2011efficient} applies additional \emph{length} and \emph{suffix} filters to prune the candidate pairs.

Prefix filtering uses an inverted index that for each element stores a list of the sets in the collection containing that element.
Given a set $x$, assume that we wish to find all sets $y$ such that~$|x \cup y| > t$.
A valid result set $y$ must be contained in at least one of the inverted lists associated with any subset of $|x| - t$ elements of $x$, or we would have $|x \cup y| \leq t$. 
In particular, to speed up the search, prefix filtering looks at the elements of $x$ that have the shortest inverted lists.

The main finding by Mann et al.~is that while more advanced filtering techniques do yield speedups on some data sets, an optimized version of the basic prefix filtering method (referred to as ``ALL'') is always competitive within a factor 2.16, and most often the fastest of the algorithms.
For this reason we will be comparing our results against ALL.

\paragraph{Locality-sensitive hashing.}
Locality-sensitive hashing (LSH) is a theoretically well-founded randomized method for generating candidate pairs~\cite{gionis1999}.
A family of locality-sensitive hash functions is a distribution over functions with the property that the probability that similar points (or sets in our case) are more likely to have the same function value.
We know only of a few papers using LSH techniques to solve similarity join.
Cohen et al.~\cite{cohen2001finding} used LSH techniques for set similarity join in a knowledge discovery context before the advent of prefix filtering.
They sketch a way of choosing parameters suitable for a given data set, but we are not aware of existing implementations of this approach.
Chakrabarti et al.~\cite{chakrabarti2015bayesian} improved plain LSH with an adaptive similarity estimation technique, \emph{BayesLSH}, that reduces the cost of checking candidate pairs and typically improves upon an implementation of the basic prefix filtering method by $2$--$20\times$.
Our experiments include a comparison against both methods~\cite{cohen2001finding,chakrabarti2015bayesian}.
We refer to the survey paper~\cite{pagh2015large} for an overview of newer theoretical developments on LSH-based similarity joins, but point out that these developments have not matured sufficiently to yield practical improvements.

\paragraph{Distance estimation.}
Similar to BayesLSH~\cite{chakrabarti2015bayesian} we make use of algorithms for similarity \emph{estimation}, but in contrast to BayesLSH we use algorithms that make use of bit-level parallelism.
This approach works when there exists a way of picking a random hash function $h$ such that
\begin{equation}\label{eq:lsh-able}
\Pr[h(x)=h(y)] = \simil(x,y)
\end{equation}
for every choice of sets $x$ and $y$.
Broder et al.~\cite{broder1997syntactic} presented such a hash function for Jaccard similarity, now known as ``minhash'' or ``minwise hashing''.
In the context of distance estimation, 1-bit minwise hashing of Li and K{\"o}nig~\cite{li2011theory} maps minhash values to a compact sketch, often using just 1 or 2 machine words.
Still, this is sufficient information to be able to estimate the Jaccard similarity of two sets $x$ and $y$ just based on the Hamming distance of their sketches. 

\paragraph{Locality-sensitive mappings.}
Several recent theoretical advances in high-dimensional indexing~\cite{andoni2017optimal,christiani2017framework,christiani2017set} have used an approach that can be seen as a generalization of LSH.
We refer to this approach as locality-sensitive \emph{mappings} (also known as locality-sensitive \emph{filters} in certain settings).
The idea is to construct a function $F$, mapping a set $x$ into a set of machine words, such that:
\begin{itemize}
	\item If $\simil(x,y)\geq \lambda$ then $F(x)\cap F(y)$ is nonempty with some fixed probability~$\varphi > 0$.
	\item If $\simil(x,y) < \lambda$, then the expected intersection size $\E[|F(x)\cap F(y)|]$ is ``small''.
\end{itemize}
Here the exact meaning of ``small'' depends on the difference~\mbox{$\lambda - \simil(x,y)$}, 
but in a nutshell, if it is the case that almost all pairs have similarity significantly below $\lambda$ then we can expect \mbox{$|F(x)\cap F(y)| = 0$} for almost all pairs.
Performing the similarity join amounts to identifying all candidate pairs $x,y$ for which $F(x)\cap F(y) \ne \varnothing$ (for example by creating an inverted index), 
and computing the similarity of each candidate pair.
To our knowledge these indexing methods have not been tried out in practice, probably because they are rather complicated.
An exception is the recent paper~\cite{christiani2017set}, which is relatively simple, and indeed our join algorithm is inspired by the index described in that paper.
\section{Preliminaries}
The \textsc{CPSJoin} algorithm solves the $(\lambda,\varphi)$-similarity join problem with a probabilistic guarantee on recall, formalized as follows:
\begin{definition}
	\label{def:simjoin}
	An algorithm solves the $(\lambda,\varphi)$-similarity join problem with threshold $\lambda\in (0,1)$ and recall probability $\varphi \in (0,1)$ if for every $(x, y) \in S \bowtie_{\lambda} R$ the output $L \subseteq S \bowtie_{\lambda} R$ of the algorithm satisfies $\Pr[(x, y) \in L] \geq \varphi$. 
\end{definition}

It is important to note that the probability is over the random choices made by the algorithm, and \emph{not} over a random choice of $(x,y)$.
This means that for any $(x,y)\in S \bowtie_{\lambda} R$ the probability that the pair is \emph{not} reported in $r$ independent repetitions of the algorithm is bounded by $(1-\varphi)^r$.
For example if $\varphi = 0.9$ it takes just $r=3$ repetitions to bound the recall to at least $99.9\%$.
\subsection{Similarity Measures}\label{sec:reduction}
Our algorithm can be used with a broad range of similarity measures through randomized \emph{embeddings}.
This allows it to be used with, for example, Jaccard and cosine similarity thresholds.

Embeddings map data from one space to another while approximately preserving distances, with accuracy that can be tuned.
In our case we are interested in embeddings that map data to sets of tokens.
We can transform any so-called \emph{LSHable} similarity measure $\simil$, where we can choose $h$ to make (\ref{eq:lsh-able}) hold, into a set similarity measure by the following randomized embedding:
For a parameter $t$ pick hash functions $h_1,\dots,h_t$ independently from a family satisfying~(\ref{eq:lsh-able}).
The embedding of $x$ is the following set of size~$t$:
$$ f(x) = \{ (i,h_i(x)) \; | \; i=1,\dots,t \} \enspace .$$
It follows from~(\ref{eq:lsh-able}) that the expected size of the intersection $f(x)\cap f(y)$ is $t \cdot \simil(x,y)$.
Furthermore, it follows from standard concentration inequalities that the size of the intersection will be close to the expectation with high probability.
For our experiments with Jaccard similarity thresholds $\geq0.5$,  we found that $t=64$ gave sufficient precision for $>90\%$ recall.

In summary we can perform the similarity join $R\simjoin S$ for any LSHable similarity measure by creating two corresponding relations $R'=\{f(x) \; | \; x\in R\}$ and $S'=\{f(y) \; | \; y\in S\}$, 
and computing $R'\simjoin S'$ with respect to the similarity measure
\begin{equation}\label{eq:bb}
B(f(x),f(y)) = |f(x)\cap f(y)|/t \enspace .
\end{equation}
This measure is the special case of \emph{Braun-Blanquet} similarity where the sets are known to have size~$t$~\cite{choi2010survey}.
Our implementation will take advantage of the set size $t$ being fixed, though it is easy to extend to general Braun-Blanquet similarity.

The class of LSHable similarity measures is large, as discussed in~\cite{chierichetti2015}.
If approximation errors are tolerable, even \emph{edit distance} can be handled by our algorithm~\cite{chakraborty2016streaming,zhang2017embedjoin}.
\subsection{Notation}
We are interested in sets $S$ where an element, $x\in S$ is a set with elements from some universe $[d]=\{1,2,3,\cdots,d\}$.
To avoid confusion we sometimes use ``record'' for $x\in S$ and ``token'' for the elements of $x$.
Throughout this paper we will think of a record $x$ both as a set of tokens from $[d]$, as well as a vector from $\{0,1\}^d$, where:
\[
x_i=
\begin{cases}
1\text { if } i\in x\\
0\text { if } i\notin x
\end{cases}
\]
It is clear that 
these representations are equivalent.
The set $\{1,4,5\}$ is equivalent to $(1,0,0,1,1,0,\cdots,0)$, $\{1,d\}$ is equivalent to $(1,0,\cdots,0,1)$, etc.
\section{Overview of approach}\label{sec:overview}
%
Our high-level approach is recursive and works as follows.
To compute $R\simjoin S$ we consider each $x\in R$ and either:
\begin{enumerate}
	\item Compare $x$ to each record in $S$ (referred to as ``brute forcing'' $x$), or
	\item create several subproblems $S_i \simjoin R_i$ with $x\in R_i \subseteq R$, $S_i \subseteq S$, and solve them recursively.
\end{enumerate}
The approach of~\cite{christiani2017set} corresponds to choosing option 2 until reaching a certain level $k$ of the recursion, where we finish the recursion by choosing option~1.
This makes sense for certain worst-case data sets, but we propose an improved parameter-free method that is better at adapting to the given data distribution.
In our method the decision on which option to choose depends on the size of $S$ and the average similarity of $x$ to the records of $S$.
We choose option~1 if $S$ has size below some (constant) threshold, or if the average Braun-Blanquet similarity of $x$ and $S$, $\tfrac{1}{|S|}\sum_{y\in S} B(x,y)$, is close to the threshold~$\lambda$.
In the former case it is cheap to finish the recursion.
In the latter case many records $y\in S$ will have $B(x,y)$ larger than or close to $\lambda$, so we do not expect to be able to produce output pairs with $x$ in sublinear time in $|S|$.

If neither of these pruning conditions apply we choose option~2 and include $x$ in recursive subproblems as described below.
But first we note that the decision of which option to use can be made efficiently for each $x$, since the average similarity of pairs from $R\times S$ can be computed from token frequencies in time $O(t|R|+ t|S|)$.
Pseudocode for a self-join version of \cpsj is provided in Algorithm~\ref{alg:cpsjoin} and~\ref{alg:bruteforce}.
\subsection{Recursion} 
We would like to ensure that for each pair $(x,y)\in R\simjoin S$ the pair is computed in one of the recursive subproblems, i.e., that $(x,y)\in R_i\simjoin S_i$ for some $i$.
In particular, we want the expected number of subproblems containing $(x,y)$ to be at least~1, i.e.,
\begin{equation}\label{eq:expect}
\E[|\{i \;|\; (x,y)\in R_i\simjoin S_i\}|] \geq 1.
\end{equation}
%
To achieve (\ref{eq:expect}) for every pair $(x,y)\in R\simjoin S$ we proceed as follows: for each $i\in \{1,\dots,d\}$ we recurse with probability $1/(\lambda t)$ on the subproblem $R_i\simjoin S_i$ with sets
\begin{align*}
	R_i &= \{ x\in R \mid i\in x \}\\
	S_i &= \{ y\in S \mid i\in y \}
\end{align*}
where $t$ denotes the size of records in $R$ and $S$.
It is not hard to check that (\ref{eq:expect}) is satisfied for every pair $(x,y)$ with $B(x,y)\geq\lambda$.
Of course, expecting one subproblem to contain $(x,y)$ does not \emph{directly} imply a good probability that $(x,y)$ is contained in at least one subproblem.
But it turns out that we can use results from the theory of branching processes to show such a bound;
details are provided in section~\ref{sec:algorithm}.

\section{Chosen Path Similarity Join} \label{sec:algorithm}
The \textsc{CPSJoin} algorithm solves the $(\lambda,\varphi)$-set similarity join~(Definition~\ref{def:simjoin}) for every choice of $\lambda \in (0,1)$ and with a guarantee on $\varphi$ that we will lower bound in the analysis. 

To simplify the exposition we focus on a self-join version where we are given a set $S$ of $n$ subsets of $[d]$ and we wish to report $L \subseteq S \simjoin S$.
Handling a general join $S \simjoin R$ follows the overview in section~\ref{sec:overview} and requires no new ideas: Essentially consider a self-join on $S\cup R$ but make sure to consider only pairs in $S\times R$ for output.
We also make the simplifying assumption that all sets in $S$ have a fixed size~$t$.
As argued in section~\ref{sec:reduction} the general case can be reduced to this one by embedding.

\subsection{Description}
The \text{CPSJoin} algorithm (see Algorithm \ref{alg:cpsjoin} for pseudocode) works by recursively splitting the data set on elements of~$[d]$ that are selected according to a random process, 
forming a recursion tree with $S$ at the root and subsets of $S$ that are non-increasing in size as we get further down the tree.
The randomized splitting has the property that the probability of a pair of sets $(x,y)$ being in a random subproblem is increasing as a function of $|x\cap y|$.

Before each recursive splitting step we run the \textsc{BruteForce} subprocedure (see Algorithm \ref{alg:bruteforce} for pseudocode) that identifies subproblems that are best solved by brute force.
It has two parts:

1. If $S$ is below some constant size, controlled by the parameter \texttt{limit}, we report $S\bowtie_{\lambda}S$ exactly using a simple loop with $O(|S|^2)$ distance computations (\textsc{BruteForcePairs}) and exit the recursion.
In our experiments we have set \texttt{limit} to $250$, with the precise choice seemingly not having a large effect as shown experimentally in Section~\ref{sec:parameters}. 

2. If $S$ is larger than \texttt{limit} the second part activates:
for every $x \in S$ we check whether the expected number of distance computations involving $x$ is going to decrease by continuing the recursion.
If this is not the case, we immediately compare $x$ against every point in $S$ (\textsc{BruteForcePoint}), reporting close pairs, and proceed by removing $x$ from $S$.
The \textsc{BruteForce} procedure is then run again on the reduced set. 

This procedure where we choose to handle some points by brute force crucially separates our algorithm from many other approximate similarity join methods in the literature that typically are LSH-based~\cite{pagh2015simjoin, cohen2001finding}.
By efficiently being able to remove points at the ``right'' time, before they generate too many expensive comparisons further down the tree,
we are able to beat the performance of other approximate similarity join techniques in both theory and practice.
Another benefit of this approach is that it reduces the number of parameters compared to the usual LSH setting where the depth of the tree has to be selected by the user.
\begin{algorithm}
\DontPrintSemicolon
\emph{For $j \in [d]$ initialize $S_j \leftarrow \varnothing$.}\;
$S \leftarrow \textsc{BruteForce}(S, \lambda)$\;
$r \leftarrow \textsc{SeedHashFunction}()$\; 
\For{$x \in S$} {
	\For{$j \in x$} {

		\lIf{$r(j) < \frac{1}{\lambda |x|}$}{$S_{j} \leftarrow S_{j} \cup \{ x \} $}
	}
}
\lFor{$S_{j} \neq \varnothing$} {$\textsc{CPSJoin}(S_{j}, \lambda)$}
\caption{\textsc{CPSJoin}$(S, \lambda)$} \label{alg:cpsjoin}
\end{algorithm}
\begin{algorithm}
\SetKwData{Limit}{limit}
\SetKwArray{Count}{count}
\SetKwInOut{Global}{Global parameters}
\DontPrintSemicolon
\Global{\Limit $\geq 1$, $\varepsilon \geq 0$.}
\emph{Initialize empty map \Count{\,} with default value $0$.}\;
\If{$|S| \leq$ \Limit} {
	$\textsc{BruteForcePairs}(S, \lambda)$\;
	\Return $\varnothing$\;
}
\For{$x \in S$} {
	\For{$j \in x$} {
		\Count{j} $\leftarrow$ \Count{j} + 1\;
	}
}
\For{$x \in S$} {
	\If{$\frac{1}{|S| - 1}\sum_{j \in x}($\Count{j}$ - 1)/t > (1-\varepsilon)\lambda$} {
		$\textsc{BruteForcePoint}(S, x, \lambda)$\;
		\Return $\textsc{BruteForce}(S \setminus \{ x \}, \lambda)$\;
	} 
}
\Return $S$\;
	
\caption{\textsc{BruteForce}$(S, \lambda)$} \label{alg:bruteforce}
\end{algorithm}
\subsection{Comparison to Chosen Path}
\label{sec:comp-chos-path}
The \textsc{CPSJoin} algorithm is inspired by the \textsc{Chosen Path} algorithm~\cite{christiani2017set} for the approximate near neighbor problem 
and uses the same underlying random splitting tree that we will refer to as the Chosen Path Tree.
In the approximate near neighbor problem, the task is to construct a data structure that takes a query point and correctly reports an approximate near neighbor, if such a point exists in the data set.
Using the \textsc{Chosen Path} data structure directly to solve the $(\lambda,\varphi)$-set similarity join problem has several drawbacks that we avoid in the \textsc{CPSJoin} algorithm.
First, the \textsc{Chosen Path} data structure is parameterized in a non-adaptive way to provide guarantees for worst-case data, 
vastly increasing the amount of work done compared to the optimal parameterization when data is not worst-case.
Our recursion rule avoids this and instead continuously adapts to the distribution of distances as we traverse down the tree. 
Secondly, the data structure uses space $O(n^{1+\rho})$ where $\rho > 0$, storing the Chosen Path Tree of size $O(n^\rho)$ for every data point. 
The \textsc{CPSJoin} algorithm, instead of storing the whole tree, essentially performs a depth-first traversal, using only near-linear space in $n$ in addition to the space required to store the output.
Finally, the \textsc{Chosen Path} data structure only has to report a single point that is approximately similar to a query point, and can report points with similarity $< \lambda$.
To solve the approximate similarity join problem the \textsc{CPSJoin} algorithm has to satisfy reporting guarantees for \emph{every} pair of points $(x, y)$ in the exact join.
\subsection{Analysis} \label{sec:analysis}
The Chosen Path Tree for a set $x \subseteq [d]$ is defined by a random process: 
at each node, starting from the root, we sample a random hash function $r \colon [d] \to [0,1]$ and construct children for every element $j \in x$ such that $r(j) < \frac{1}{\lambda |x|}$.
Nodes at depth $k$ in the tree are identified by their path $p = (j_1, \dots, j_k)$. 
Formally, the set of nodes at depth $k > 0$ in the Chosen Path Tree for $x$ is given by
\begin{equation*}
	F_{k}(x) = \left\{ p \circ j \mid p \in F_{k-1}(x) \land r_{p}(j) < \frac{x_j}{\lambda |x|} \right\}
\end{equation*}
where $p \circ j$ denotes vector concatenation and $F_{0}(x) = \varnothing$. 
The subset of the data set $S$ that survives to a node with path $p = (j_1, \dots, j_{k})$ is given by
\begin{equation*}
	S_{p} = \{ x \in S \mid x_{j_1} = 1 \land \dots \land x_{j_{k}} = 1 \}.
\end{equation*}
The random process underlying the Chosen Path Tree belongs to the well studied class of Galton-Watson branching processes~\cite{harris2002theory}.
Originally these where devised to answer questions about the growth and decline of family names in a model of population growth assuming i.i.d.\ offspring for every member of the population across generations~\cite{watson1875}.
In order to make statements about the properties of the \textsc{CPSJoin} algorithm we study in turn the branching processes 
of the Chosen Path Tree associated with a point $x$, a pair of points $(x, y)$, and a set of points $S$.
Note that we use the same random hash functions for different points in $S$.

\paragraph{Brute forcing.}
The \textsc{BruteForce} subprocedure described by Algorithm \ref{alg:bruteforce} takes two global parameters: $\mathtt{limit} \geq 1$ and $\varepsilon \geq 0$.
The parameter $\mathtt{limit}$ controls the minimum size of $S$ before we discard the \cpsj algorithm for a simple exact similarity join by brute force pairwise distance computations.
The second parameter, $\varepsilon > 0$, controls the sensitivity of the \textsc{BruteForce} step to the expected number of comparisons that a point $x \in S$ will generate if allowed to continue in the branching process.
The larger $\varepsilon$ the more aggressively we will resort to the brute force procedure.
In practice we typically think of $\varepsilon$ as a small constant, say $\varepsilon = 0.05$, 
but for some of our theoretical results we will need a sub-constant setting of $\varepsilon \approx 1/\log(n)$ to show certain running time guarantees. 
The \textsc{BruteForce} step removes a point $x$ from the Chosen Path branching process, 
instead opting to compare it against every other point $y \in S$, if it satisfies the condition
\begin{equation*}
	\frac{1}{|S| - 1}\sum_{y \in S \setminus \{ x \}}|x \cap y|/t > (1-\varepsilon)\lambda. \label{eq:bruteforce}
\end{equation*}
In the pseudocode of Algorithm \ref{alg:bruteforce} we let \texttt{count} denote a hash table that keeps track of the number of times each element $j \in [d]$ appears in $S$.
This allows us to evaluate the condition in equation \eqref{eq:bruteforce} for an element $x \in S$ in time $O(|x|)$ by rewriting it as
\begin{equation*}
\frac{1}{|S|-1}\sum_{j \in x} (\mathtt{count}[j] - 1)/t > (1-\varepsilon)\lambda.
\end{equation*}
We claim that this condition minimizes the expected number of comparisons performed by the algorithm:
Consider a node in the Chosen Path Tree associated with a set of points $S$ while running the \cpsj algorithm.
For a point ${x\in S}$, we can either remove it from $S$ immediately at a cost of $|S|-1$ comparisons, or we can choose to let continue in the branching process (possibly into several nodes) and remove it later.
The expected number of comparisons if we let it continue $k$ levels before removing it from every node that it is contained in, is given by
\begin{equation*}
	\sum_{y \in S \setminus \{ x \}} \left(\frac{1}{\lambda}\frac{|x \cap y|}{t}\right)^{k}.
\end{equation*}
This expression is convex and increasing in the similarity \mbox{$|x \cap y|/t$} between $x$ and other points $y \in S$, allowing us to state the following remark: 
\begin{remark}[Recursion]\label{obs:bruteforce}
	Let $\varepsilon = 0$ and consider a set $S$ containing a point $x \in S$ such that $x$ satisfies the recursion condition in equation~\eqref{eq:bruteforce}.
	Then the expected number of comparisons involving $x$ if we continue branching exceeds $|S|-1$ at every depth $k \geq 1$.
        If $x$ does not satisfy the condition, the opposite is observed.
\end{remark}
\paragraph{Tree depth.}
We proceed by bounding the maximal depth of the set of paths in the Chosen Path Tree that are explored by the \cpsj algorithm.
Having this information will allow us to bound the space usage of the algorithm and will also form part of the argument for the correctness guarantee.
Assume that the parameter \texttt{limit} in the \textsc{BruteForce} step is set to some constant value, say $\mathtt{limit} = 100$.
Consider a point $x \in S$ and let $S' = \{ y \in S \mid |x \cap y|/ t \leq (1-\varepsilon)\lambda \}$ be the subset of points in $S$ that are not too similar to $x$.
For every $y \in S'$ the expected number of vertices in the Chosen Path Tree at depth $k$ that contain both $x$ and $y$ is upper bounded by
\begin{equation*}
	\E[|F_{k}(x \cap y)|] = \left(\frac{1}{\lambda}\frac{|x \cap y|}{t}\right)^{k} \leq (1-\varepsilon)^{k} \leq e^{-\varepsilon k}.
\end{equation*}
Since $|S'| \leq n$ we use Markov's inequality to show the following bound: 
\begin{lemma}
	Let $x, y \in S$ satisfy that $|x \cap y|/ t \leq (1-\varepsilon)\lambda$ then the probability that there exists a vertex at depth $k$ 
	in the Chosen Path Tree that contains $x$ and $y$ is at most $e^{-\varepsilon k}$.
\end{lemma}
If $x$ does not share any paths with points that have similarity that falls below the threshold for brute forcing, then the only points that remain are ones that will cause $x$ to be brute forced.
This observation leads to the following probabilistic bound on the tree depth:
\begin{lemma}\label{lem:depth}
	With high probability the maximal depth of paths explored by the \cpsj algorithm is $O(\log(n) / \varepsilon)$.
\end{lemma}

\paragraph{Correctness.}
Let $x$ and $y$ be two sets of equal size $t$ such that $B(x, y) = |x \cap y|/t \geq \lambda$. 
We are interested in lower bounding the probability that there exists a path of length~$k$ in the Chosen Path Tree that has been chosen by both $x$ and~$y$, 
i.e. $\Pr\left[F_{k}(x \cap y)\neq\varnothing\right]$.
Agresti~\cite{agresti1974} showed an upper bound on the probability that a branching process becomes extinct after at most $k$ steps. 
We use it to show the following lower bound on the probability of a close pair of points colliding at depth $k$ in the Chosen Path Tree.
\begin{lemma}[{Agresti~\cite{agresti1974}}]
	If $\simil(x, y) \geq \lambda$ then for every $k > 0$ we have that \mbox{$\Pr[F_{k}(x \cap y) \neq \varnothing] \geq \frac{1}{k+1}$}.
\end{lemma}
The bound on the depth of the Chosen Path Tree for $x$ explored by the \cpsj algorithm in Lemma~\ref{lem:depth} then implies a lower bound on $\varphi$.


\begin{lemma}\label{lem:correctness}
	Let $0 < \lambda < 1$ be constant. Then for every set $S$ of $|S| = n$ points the \cpsj algorithm solves the set similarity join problem with $\varphi = \Omega(\varepsilon / \log(n))$.  
\end{lemma}

\begin{remark}
  This analysis is very conservative: if either $x$ or $y$ is removed by the \textsc{BruteForce} step prior to reaching the maximum depth then it only increases the probability of collision.
  We note that similar guarantees can be obtained when using fast pseudorandom hash functions as shown in the paper introducing the \cp algorithm~\cite{christiani2017set}.
\end{remark}

\paragraph{Space usage.}
We can obtain a trivial bound on the space usage of the \cpsj algorithm by combining Lemma \ref{lem:depth} with the observation that every call to $\cpsj$ on the stack uses additional space at most $O(n)$. 
The result is stated in terms of working space: the total space usage when not accounting for the space required to store the data set itself 
(our algorithms use references to data points and only reads the data when performing comparisons) as well as disregarding the space used to write down the list of results. 
\begin{lemma} \label{lem:space}
	With high probability the working space of the \cpsj algorithm is at most $O(n \log (n) /\varepsilon)$.
\end{lemma}
\begin{remark}
We conjecture that the expected working space is $O(n)$ due to the size of $S$ being geometrically decreasing in expectation as we proceed down the Chosen Path Tree.
\end{remark}
\paragraph{Running time.}
We will bound the running time of a solution to the general set similarity self-join problem that uses several calls to the 
\cpsj algorithm in order to piece together a list of results $L \subseteq S \simjoin S$.
In most of the previous related work, inspired by Locality-Sensitive Hashing, the fine-grainedness of the randomized partition of space, 
here represented by the Chosen Path Tree in the \cpsj algorithm, has been controlled by a single global parameter~$k$~\cite{gionis1999, pagh2015simjoin}.
In the Chosen Path setting this rule would imply that we run the splitting step without performing any brute force comparison until reaching depth $k$ 
where we proceed by comparing $x$ against every other point in nodes containing $x$, reporting close pairs. 
In recent work by Ahle et al.~\cite{ahle2017parameter} it was shown how to obtain additional performance improvements by setting an individual depth $k_{x}$ for every $x \in S$.
We refer to these stopping strategies as global and individual, respectively.
Together with our recursion strategy, this gives rise to the following stopping criteria for when to compare a point $x$ against everything else contained in a node: 
\begin{itemize}
	\item Global: Fix a single depth $k$ for every $x \in S$.
	\item Individual: For every $x \in S$ fix a depth $k_{x}$. 
	\item Adaptive: Remove $x$ when the expected number of comparisons is non-decreasing in the tree-depth.
\end{itemize}
Let $T$ denote the running time of our similarity join algorithm. 
We aim to show the following relation between the running time between the different stopping criteria when applied to the Chosen Path Tree:
\begin{equation*}
	\E[T_{\text{Adaptive}}] \leq \E[T_{\text{Individual}}] \leq \E[T_{\text{Global}}]. 
\end{equation*}
First consider the global strategy. 
We set $k$ to balance the contribution to the running time from the expected number of vertices containing a point, given by $(1/\lambda)^{k}$,
and the expected number of comparisons between pairs of points at depth $k$, resulting in the following expected running time for the global strategy:
\begin{equation*}
	O\left(\min_{k} n (1/\lambda)^{k} +  \sum_{\substack{x,y \in S \\ x \neq y }}  (\simil(x, y) / \lambda)^{k}  \right).
\end{equation*}
The global strategy is a special case of the individual case, and it must therefore hold that $\E[T_{\text{Individual}}] \leq \E[T_{\text{Global}}]$.
The expected running time for the individual strategy is upper bounded by:
\begin{equation*}
	O\left(\sum_{x \in S}\min_{k_{x}} \left( (1/\lambda)^{k_{x}} + \sum_{y \in S \setminus \{ x \}} (\simil(x, y) / \lambda)^{k_{x}} \right) \right).
\end{equation*}
We will now argue that the expected running time of the \cpsj algorithm under the adaptive stopping criteria is no more than a constant factor greater than 
$\E[T_{\text{Individual}}]$ when we set the global parameters of the \textsc{BruteForce} subroutine as follows:
\begin{align*}
	\mathtt{limit} = \Theta(1), \\
	\varepsilon = \frac{\log(1/\lambda)}{\log n}.
\end{align*}
Let $x \in S$ and consider a path $p$ where $x$ is removed in from $S_{p}$ by the \textsc{BruteForce} step. 
Let $k_{x}'$ denote the depth of the node (length of $p$) at which $x$ is removed.
Compared to the individual strategy that removes $x$ at depth $k_{x}$ we are in one of three cases, also displayed in Figure \ref{fig:paths}.
\begin{enumerate}
	\item The point $x$ is removed from $p$ at depth $k_{x}' = k_{x}$.
	\item The point $x$ is removed from $p$ at depth $k_{x}' < k_{x}$.
	\item The point $x$ is removed from $p$ at depth $k_{x}' > k_{x}$.
\end{enumerate}
\begin{figure}[htpb]
  \centering
  \includegraphics[width=0.85\textwidth]{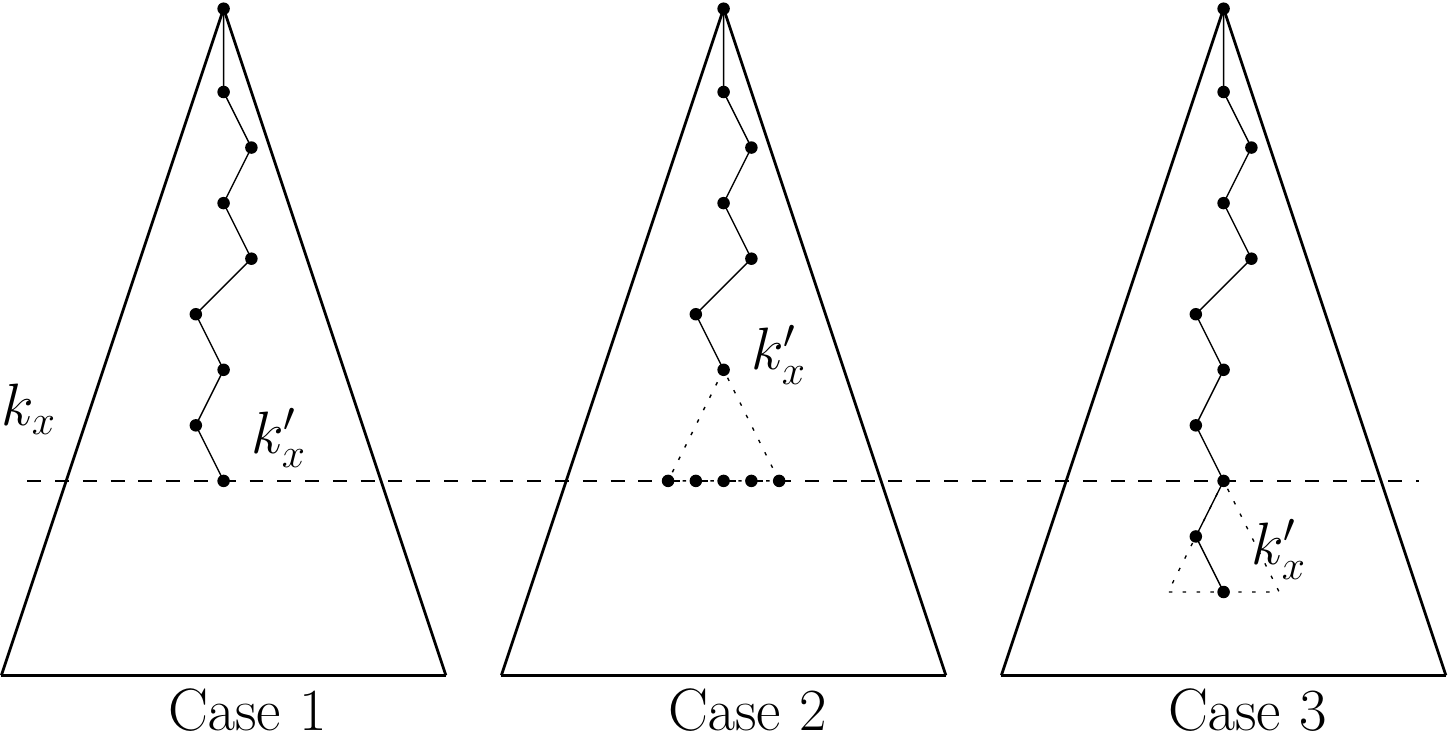}
  \caption{Path termination depth in the Chosen Path Tree}
    \label{fig:paths}
\end{figure}
The underlying random process behind the Chosen Path Tree is not affected by our choice of termination strategy.
In the first case we therefore have that the expected running time is upper bounded by the same (conservative) expression as the one used by the individual strategy.
In the second case we remove $x$ earlier than we would have under the individual strategy. 
For every $x \in S$ we have that $k_{x} \leq 1/\varepsilon$ since for larger values of $k_{x}$ the expected number of nodes containing $x$ exceeds $n$.
We therefore have that $k_{x} - k_{x}' \leq 1/\varepsilon$.
Let $S'$ denote the set of points in the node where $x$ was removed by the \textsc{BruteForce} subprocedure.
There are two rules that could have triggered the removal of $x$: Either $|S'| = O(1)$ or the condition in equation \eqref{eq:bruteforce} was satisfied.
In the first case, the expected cost of following the individual strategy would have been $\Omega(1)$ simply from the $1/\lambda$ children containing $x$ in the next step.
This is no more than a constant factor smaller than the adaptive strategy.
In the second case, when the condition in equation \eqref{eq:bruteforce} is activated we have that the expected number of comparisons 
involving $x$ resulting from $S'$ if we had continued under the individual strategy is at least
\begin{equation*}
	(1-\varepsilon)^{1/\varepsilon}|S'| = \Omega(|S'|)
\end{equation*}
which is no better than what we get with the adaptive strategy.
In the third case where we terminate at depth $k_{x}' > k_{x}$, if we retrace the path to depth $k_{x}$ we know that $x$ was not removed in this node, 
implying that the expected number of comparisons when continuing the branching process on~$x$ is decreasing compared to removing $x$ at depth $k_{x}$.
We have shown that the expected running time of the adaptive strategy is no greater than a constant times the expected running time of the individual strategy.

We are now ready to state our main theoretical contribution, stated below as Theorem \ref{thm:main}. 
The theorem combines the above argument that compares the adaptive strategy against the individual strategy together with Lemma \ref{lem:depth} and Lemma \ref{lem:correctness}, 
and uses $O(\log^{2} n)$ runs of the \cpsj algorithm to solve the set similarity join problem for every choice of constant parameters $\lambda,\varphi$.
\begin{theorem}\label{thm:main}
	For every LSHable similarity measure and every choice of constant threshold $\lambda \in (0,1)$ and probability of recall $\varphi \in (0,1)$ 
	we can solve the $(\lambda,\varphi)$-set similarity join problem on every set $S$ of $n$ points using working space $\tilde{O}(n)$ and with expected running time
	\begin{equation*}
		\tilde{O}\left(\sum_{x \in S}\min_{k_{x}} \left( \sum_{y \in S \setminus \{ x \}} (\simil(x, y) / \lambda)^{k_{x}} + (1/\lambda)^{k_{x}} \right) \right).
	\end{equation*}
\end{theorem}
\section{Implementation} \label{sec:implementation}
We implement an optimized version of the \cpsj algorithm for solving the Jaccard similarity self-join problem.
In our experiments (described in Section \ref{join:sec:experiments}) we compare the \cpsj algorithm against 
the approximate methods of MinHash LSH~\cite{gionis1999,broder1997syntactic} and BayesLSH~\cite{chakrabarti2015bayesian},
as well as the AllPairs~\cite{bayardo2007scaling} exact similarity join algorithm.
The code for our experiments is written in C\texttt{++} and uses the benchmarking framework and data sets of the recent experimental survey on exact similarity join algorithms by Mann et al.~\cite{mann2016}.  
For our implementation we assume that each set $x$ is represented as a list of 32-bit unsigned integers. 
We proceed by describing the details of each implementation in turn. 

\subsection{Chosen Path Similarity Join} \label{sec:implementation_cpsj}
The implementation of the \cpsj algorithm follows the structure of the pseudocode in Algorithm \ref{alg:cpsjoin} and Algorithm \ref{alg:bruteforce},
but makes use of a few heuristics, primarily sampling and sketching, in order to speed things up. 
The parameter setting is discussed and investigated experimentally in section \ref{sec:parameters}.

\paragraph{Preprocessing.}
Before running the algorithm we use the embedding described in section~\ref{sec:reduction}.
Specifically $t$ independent MinHash functions $h_1, \dots, h_t$ are used to map each set $x \in S$ to a list of $t$ hash values $(h_1(x), \dots, h_t(x))$.
The MinHash function is implemented using Zobrist hashing~\cite{zobrist1970new} from 32 bits to 64 bits with 8-bit characters.
We sample a MinHash function $h$ by sampling a random Zobrist hash function $g$ and let $h(x) = \arg\!\min_{j \in x} g(j)$. 
Zobrist hashing (also known as simple tabulation hashing) has been shown theoretically to have strong MinHash properties and is very fast in practice~\cite{patrascu2012power,thorup2017fast}.  
We set $t = 128$ in our experiments, see discussion later. 

During preprocessing we also prepare sketches using the 1-bit minwise hashing scheme of Li and K{\"o}nig~\cite{li2011theory}. 
Let $\ell$ denote the length in 64-bit words of a sketch $\hat{x}$ of a set $x \in S$. 
We construct sketches for a data set $S$ by independently sampling $64 \times \ell$ MinHash functions $h_i$ and Zobrist hash functions $g_i$ that map from 32 bits to 1 bit. 
The $i$th bit of the sketch $\hat{x}$ is then given by $g_i(h_i(x))$.
In the experiments we set $\ell = 8$.

\paragraph{Similarity estimation using sketches.}
We use 1-bit minwise hashing sketches for fast similarity estimation in the \textsc{BruteForcePairs} and \textsc{BruteForcePoint} subroutines of the \textsc{BruteForce} step of the \cpsj algorithm.
Given two sketches, $\hat{x}$ and $\hat{y}$, we compute the number of bits in which they differ by going through the sketches word for word, computing the popcount of their XOR using the \texttt{gcc} builtin \texttt{\_mm\_popcnt\_u64} that translates into a single instruction on modern hardware. 
Let $\hat{J}(x, y)$ denote the estimated similarity of a pair of sets $(x, y)$. If $\hat{J}(x, y)$ is below a threshold $\hat{\lambda} \approx \lambda$, we exclude the pair from further consideration. If the estimated similarity is greater than $\hat{\lambda}$ we compute the exact similarity and report the pair if $J(x, y) \geq \lambda$. 

The speedup from using sketches comes at the cost of introducing false negatives:
A pair of sets $(x, y)$ with $J(x, y) \geq \lambda$ may have an estimated similarity less than $\hat{\lambda}$, causing us to miss it. 
We let $\delta$ denote a parameter for controlling the false negative probability of our sketches and set $\hat{\lambda}$ such that for sets $(x, y)$ with $J(x, y) \geq \lambda$ we have that $\Pr[\hat{J}(x, y) < \hat{\lambda}] < \delta$. 
In our experiments we set the sketch false negative probability to be $\delta = 0.05$.

\paragraph{Splitting step.}
In the recursive step of the \cpsj algorithm~(Algorithm \ref{alg:cpsjoin}) the set $S$ is split into buckets $S_j$ using the following heuristic: 
Instead of sampling a random hash function and evaluating it on each element $j \in x$, 
we sample an expected $1/\lambda$ elements from $[t]$ and split $S$ according to the corresponding minhash values from the preprocessing step.
This saves the linear overhead in the size of our sets $t$, reducing the time spent placing each set into buckets to $O(1)$.
Internally, a collection of sets $S$ is represented as a C\texttt{++} \texttt{std::vector<uint32\_t>} of set ids.
The collection of buckets $S_j$ is implemented using Google's \texttt{dense\_hash} hash map implementation from the \texttt{sparse\_hash} package~\cite{sparsehash}.

\paragraph{BruteForce step.}
Having reduced the overhead for each set $x \in S$ to $O(1)$ in the splitting step, we wish to do the same for the \textsc{BruteForce} step (described in Algorithm \ref{alg:bruteforce}), 
at least in the case where we do not call the \textsc{BruteForcePairs} or \textsc{BruteForcePoint} subroutines.
The main problem is that we spend time $O(t)$ for each set when constructing the \texttt{count} hash map and estimating the average similarity of $x$ to sets in $S \setminus \{x\}$.
To get around this we construct a 1-bit minwise hashing sketch $\hat{s}$ of length $64 \times \ell$ for the set $S$ using sampling and our precomputed 1-bit minwise hashing sketches.
The sketch $\hat{s}$ is constructed as follows: Randomly sample $64 \times \ell$ elements of $S$ and set the $i$th bit of $\hat{s}$ to be the $i$th bit of the $i$th sample from $S$.
This allows us to estimate the average similarity of a set $x$ to sets in $S$ in time $O(\ell)$ using word-level parallelism. 
A set $x$ is removed from $S$ if its estimated average similarity is greater than $(1 - \varepsilon)\lambda$. 
To further speed up the running time we only call the \textsc{BruteForce} subroutine once for each call to \cpsj, 
calling \textsc{BruteForcePoint} on all points that pass the check rather than recomputing $\hat{s}$ each time a point is removed.
Pairs of sets that pass the sketching check are verified using the same verification procedure as the \all implementation by Mann et al.~\cite{mann2016}.
In our experiments we set the parameter $\varepsilon = 0.1$.
Duplicates are removed by sorting and performing a single linear scan.

\paragraph{Repetitions.}
In theory, for any constant desired recall $\varphi \in (0,1)$ it suffices with $O(\log^2 n)$ independent repetitions of the \cpsj algorithm. 
In practice this number of repetitions is prohibitively large and we therefore set the number of independent repetitions used in our experiments to be fixed at ten. 
With this setting we were able to achieve more than $90\%$ recall across all datasets and similarity thresholds.

\subsection{MinHash LSH} \label{sec:minhash}
We implement a locality-sensitive hashing similarity join using MinHash according to the pseudocode in Algorithm~\ref{alg:minhash}.
A single run of the \mh algorithm can be divided into two steps: 
First we split the sets into buckets according to the hash values of $k$ concatenated MinHash functions $h(x) = (h_1(x), \dots, h_k(x))$.
Next we iterate over all non-empty buckets and run \textsc{BruteForcePairs} to report all pairs of points with similarity above the threshold $\lambda$.
The \textsc{BruteForcePairs} subroutine is shared between the \mh and \cpsj implementation.
\mh therefore uses 1-bit minwise sketches for similarity estimation in the same way as in the implementation of the \cpsj algorithm described above. 

The parameter $k$ can be set for each dataset and similarity threshold $\lambda$ to minimize the combined cost of lookups and similarity estimations performed by algorithm.
This approach was mentioned by Cohen et al.~\cite{cohen2001finding} but we were unable to find an existing implementation.
In practice we set $k$ to the value that results in the minimum estimated running time when running the first part (splitting step) of the algorithm for values of $k$ in the range $\{2, 3, \dots, 10\}$ and estimating the running time by looking at the number of buckets and their sizes. 
Once $k$ is fixed we know that each repetition of the algorithm has probability at least $\lambda^k$ of reporting a pair $(x, y)$ with $J(x, y) \geq \lambda$. 
For a desired recall $\varphi$ we can therefore set $L = \lceil \ln(1/(1-\varphi)) / \lambda^k \rceil$.
In our experiments we report the actual number of repetitions required to obtain a desired recall rather than using the setting of $L$ required for worst-case guarantees.
\begin{algorithm}
\SetKwInOut{Params}{Parameters}
\SetKwArray{Buckets}{buckets}
\DontPrintSemicolon
\Params{$k \geq 1, L \geq 1$.}
\For{$i \leftarrow 1$ \KwTo $L$} {
\emph{Initialize hash map \Buckets{\,}.}\;
\emph{Sample $k$ MinHash fcts.} $h \leftarrow (h_1, \dots, h_k)$\;
\For{$x \in S$} {
	\Buckets{$h(x)$} $\leftarrow$ \Buckets{$h(x)$} $\cup$ $\{x\}$\; 
}
\For{$S' \in$ \Buckets} {
	$\textsc{BruteForcePairs}(S', \lambda)$\;
}
}
\caption{\textsc{MinHash}$(S, \lambda)$} \label{alg:minhash}
\end{algorithm}

\subsection{AllPairs}
To compare our approximate methods against a state-of-the-art exact similarity join we use Bayardo et al.'s \all algorithm~\cite{bayardo2007scaling} as recently implemented in the set similarity join study by Mann et al.~\cite{mann2016}. 
The study by Mann et al. compares implementations of several different exact similarity join methods and finds that the simple \all algorithm is most often the fastest choice. 
Furthermore, for Jaccard similarity, the \all algorithm was at most $2.16$ times slower than the best out of six different competing algorithm across all the data sets and similarity thresholds used, 
and for most runs \all is at most $11\%$ slower than the best exact algorithm (see Table 7 in Mann et al.~\cite{mann2016}). 
Since our experiments run in the same framework and using the same datasets and with the same thresholds as Mann et al.'s study, 
we consider their \all implementation to be a good representative of exact similarity join methods for Jaccard similarity.  

\subsection{BayesLSH}
For a comparison against previous experimental work on approximate similarity joins we use an implementation of \blsh in C as provided by the \blsh authors~\cite{chakrabarti2015bayesian, bayeslsh}.
The BayesLSH package features a choice between \all and LSH as candidate generation method. 
For the verification step there is a choice between \blsh and \blsh-lite.
Both verification methods use sketching to estimate similarities between candidate pairs.
The difference between BayesLSH and BayesLSH-lite is that the former uses sketching to estimate the similarity of pairs that pass the sketching check, 
whereas the latter uses an exact similarity computation if a pair passes the sketching check.
Since the approximate methods in our \cpsj and \mh implementations correspond to the approach of BayesLSH-lite we restrict our experiments to this choice of verification algorithm.
In our experiments we will use \blsh to represent the fastest of the two candidate generation methods, combined with BayesLSH-lite for the verification step.
\section{Experiments} \label{join:sec:experiments}
We run experiments using the implementations of \cpsj, \mh, \blsh, and \all described in the previous section.
In the experiments we perform self-joins under Jaccard similarity for similarity thresholds $\lambda \in \{0.5, 0.6, 0.7, 0.8, 0.9 \}$.
We are primarily interested in measuring the join time of the algorithms, but we also look at the number of candidate pairs $(x,y)$ considered by the algorithms during the join as a measure of performance.  
Note that the preprocessing step of the approximate methods only has to be performed once for each set and similarity measure, 
and can be re-used for different similarity joins, we therefore do not count it towards our reported join times.
In practice the preprocessing time is at most a few minutes for the largest data sets.

\paragraph{Data sets.}
The performance is measured across $10$ real~world data sets along with $4$ synthetic data sets described in Table \ref{tab:datasets}. 
All datasets except for the TOKENS datasets were provided by the authors of~\cite{mann2016} where descriptions and sources for each data set can also be found. 
Note that we have excluded a synthetic ZIPF dataset used in the study by Mann et al.\cite{mann2016} due to it having no results for our similarity thresholds of interest.
Experiments are run on versions of the datasets where duplicate records are removed and any records containing only a single token are ignored.
\begin{table}
	\centering
	\caption{Dataset size, average set size, and average number of sets that a token is contained in.}
	\label{tab:datasets}
	\footnotesize
	\begin{tabular}{lrrr} \toprule
		Dataset & \# sets / $10^6$ & avg. set size & sets / tokens \\\midrule
		AOL        & $7.35$ &   $3.8$ & $18.9$ \\
		BMS-POS    & $0.32$ &   $9.3$ & $1797.9$ \\
		DBLP       & $0.10$ &  $82.7$ & $1204.4$ \\
		ENRON      & $0.25$ & $135.3$ & $29.8$ \\
		FLICKR     & $1.14$ &  $10.8$ & $16.3$ \\
		LIVEJ      & $0.30$ &  $37.5$ & $15.0$ \\
		KOSARAK    & $0.59$ &  $12.2$ & $176.3$ \\
		NETFLIX    & $0.48$ & $209.8$ & $5654.4$ \\
		ORKUT      & $2.68$ & $122.2$ & $37.5$ \\
		SPOTIFY    & $0.36$ &  $15.3$ & $7.4$ \\
		UNIFORM    & $0.10$ &  $10.0$ & $4783.7$ \\
		TOKENS10K  & $0.03$ & $339.4$ & $10000.0$ \\
		TOKENS15K  & $0.04$ & $337.5$ & $15000.0$ \\
		TOKENS20K  & $0.06$ & $335.7$ & $20000.0$ \\ \bottomrule
	\end{tabular}

\end{table}

In addition to the datasets from the study of Mann et al. we add three synthetic datasets TOKENS10K, TOKENS15K, and TOKENS20K, designed to showcase the robustness of the approximate methods.
These datasets have relatively few unique tokens, but each token appears in many sets. 
Each of the TOKENS datasets were generated from a universe of $1000$ tokens ($d = 1000$) and each token is contained in respectively, $10,000$, $15,000$, and $20,000$ different sets as denoted by the name.
The sets in the TOKENS datasets were generated by sampling a random subset of the set of possible tokens,
rejecting tokens that had already been used in more than the maximum number of sets ($10,000$ for TOKENS10K).
To sample sets with expected Jaccard similarity $\lambda'$ the size of our sampled sets should be set to $(2\lambda'/(1+\lambda'))d$. 
For $\lambda' \in \{0.95, 0.85, 0.75, 0.65, 0.55\}$ the TOKENS datasets each have $100$ random sets planted with expected Jaccard similarity $\lambda'$.
This ensures an increasing number of results for our experiments where we use thresholds $\lambda \in \{0.5, 0.6, 0.7, 0.8, 0.9 \}$. 
The remaining sets have expected Jaccard similarity $0.2$.
We believe that the TOKENS datasets give a good indication of the performance on real-world data that has the property that most tokens appear in a large number of sets.

\paragraph{Recall.}
In our experiments we aim for a recall of at least $90\%$ for the approximate methods. 
In order to achieve this for the \cpsj and \mh algorithms we perform a number of repetitions after the preprocessing step, stopping when the desired recall has been achieved.
This is done by measuring the recall against the recall of \all and stopping when reaching $90\%$.
In practice this approach is not feasible as the size of the true result set is not known.
However, it can be efficiently estimated using sampling if it is not too small.
Another approach is to perform the number of repetitions required to obtain the theoretical guarantees on recall as described for \cpsj in Section \ref{sec:analysis} and for \mh in Section \ref{sec:minhash}.
Unfortunately, with our current analysis of the \cpsj algorithm the number of repetitions required to guarantee theoretically a recall of $90\%$ far exceeds the number required in practice as observed in our experiments where ten independent repetitions always suffice. 
For \blsh using LSH as the candidate generation method, the recall probability with the default parameter setting is $95\%$, although we experience a recall closer to $90\%$ in our experiments.

\paragraph{Hardware.}
All experiments were run on an Intel Xeon E5-2690v4 CPU at 2.60GHz with $35$MB L$3$,$256$kB L$2$ and $32$kB L$1$ cache and $512$GB of RAM.
Since a single experiment is always confined to a single CPU core we ran several experiments in parallel~\cite{tange2011gnu} to better utilize our hardware.
\subsection{Results}
\paragraph{Join time.}
Table \ref{tab:jointimes} shows the average join time in seconds over five independent runs, when approximate methods are required to have at least $90\%$ recall.
We have omitted timings for \blsh since it was always slower than all other methods, and in most cases it timed out after 20 minutes when using LSH as candidate generation method.
The join time for \mh is always greater than the corresponding join time for \cpsj except in a single setting: the dataset KOSARAK with threshold $\lambda = 0.5$.
Since \cpsj is typically $2-4\times$ faster than \mh we can restrict our attention to comparing \all and \cpsj where the picture becomes more interesting.
\begin{table*}
	\caption{Join time in seconds for \all (ALL) and \cpsj (CP) with recall $\ge90\%$.}
	\label{tab:jointimes}
\resizebox{\textwidth}{!}{%
	\footnotesize
\begin{filecontents*}{results_table.csv}
File,AA,ACP,AMH,BA,BCP,BMH,CA,CCP,CMH,DA,DCP,DMH,EA,ECP,EMH
AOL, 483.5,\textbf{ 362.1},1329.9, 117.8,\textbf{ 113.4}, 444.2,\textbf{  13.7},  42.2, 152.9,\textbf{   4.2},  34.6, 100.6,\textbf{   1.6},  21.0,  43.8
BMS-POS,  62.5,\textbf{  27.0},  40.0,  20.9,\textbf{   7.1},  13.7,   5.6,\textbf{   2.7},   5.6,\textbf{   1.3},   2.0,   3.9,\textbf{   0.2},   0.9,   1.4
DBLP, 127.9,\textbf{   9.2},  22.1,  63.8,\textbf{   2.5},  10.1,  27.4,\textbf{   1.1},   3.7,   7.8,\textbf{   0.6},   1.8,   0.8,\textbf{   0.3},   0.7
ENRON,  78.0,\textbf{   6.9},  16.4,  23.2,\textbf{   4.4},   9.9,   6.0,\textbf{   2.4},   6.3,\textbf{   1.6},   1.6,   2.7,\textbf{   0.4},   0.7,   1.7
FLICKR,\textbf{  17.2},  48.6,  68.0,\textbf{   6.0},  30.9,  37.2,\textbf{   2.5},  13.8,  21.3,\textbf{   1.0},   6.3,  11.3,\textbf{   0.3},   3.4,   5.2
KOSARAK,\textbf{  73.1}, 377.9, 311.1,\textbf{  14.4},  62.7,  89.2,\textbf{   1.6},   7.2,  16.1,\textbf{   0.5},   3.9,   9.9,\textbf{   0.1},   1.2,   2.6
LIVEJ, 571.7,\textbf{ 131.3}, 279.4, 145.3,\textbf{  48.7}, 129.6,  30.6,\textbf{  28.2},  52.9,\textbf{   7.1},  16.2,  41.0,\textbf{   1.5},   9.2,  12.6
NETFLIX,1354.7,\textbf{  25.3}, 121.8, 520.4,\textbf{   8.2},  60.0, 177.3,\textbf{   4.8},  22.6,  46.2,\textbf{   2.4},  14.1,   5.4,\textbf{   1.6},   5.8
ORKUT, 359.7,\textbf{  26.5}, 115.7, 106.4,\textbf{  15.4},  60.1,  36.3,\textbf{   8.0},  25.1,  12.2,\textbf{   7.4},  19.7,\textbf{   3.7},   4.8,  13.3
SPOTIFY,\textbf{   0.5},   2.5,   9.3,\textbf{   0.3},   1.5,   3.4,\textbf{   0.2},   1.0,   2.6,\textbf{   0.1},   1.0,   1.9,\textbf{   0.1},   0.5,   0.6
TOKENS10K, 312.1,\textbf{   3.4},   4.8, 236.8,\textbf{   2.9},   3.9, 164.0,\textbf{   1.5},   1.7, 114.9,\textbf{   0.6},   1.2,  63.2,\textbf{   0.2},   0.4
TOKENS15K, 688.4,\textbf{   4.4},   6.2, 535.3,\textbf{   4.0},   7.1, 390.4,\textbf{   1.8},   3.7, 258.2,\textbf{   0.7},   1.7, 140.0,\textbf{   0.2},   0.7
TOKENS20K,1264.1,\textbf{   5.7},  12.0, 927.0,\textbf{   4.0},  11.4, 698.4,\textbf{   2.1},   4.5, 494.3,\textbf{   0.8},   2.2, 273.4,\textbf{   0.3},   0.8
UNIFORM005,  54.1,\textbf{   3.9},   6.6,  27.6,\textbf{   1.6},   3.0,  10.5,\textbf{   0.9},   1.4,   3.6,\textbf{   0.5},   1.0,   0.4,\textbf{   0.1},   0.3
\end{filecontents*}
\csvreader[tabular=l|rr|rr|rr|rr|rr
, table head=\toprule&\multicolumn{2}{c}{Threshold $0.5$}&\multicolumn{2}{c}{Threshold $0.6$}&\multicolumn{2}{c}{Threshold $0.7$}&\multicolumn{2}{c}{Threshold $0.8$}&\multicolumn{2}{c}{Threshold $0.9$}\\
	\midrule Dataset & CP  & ALL & CP & ALL & CP & ALL & CP & ALL & CP & ALL \\ \midrule
, table foot=, head to column names
, late after line=\\
, late after last line =\\\bottomrule]{results_table.csv}{}
{\File& \AA &  \ACP & \BA & \BCP & \CA & \CCP &\DA & \DCP &\EA & \ECP}
}
\end{table*}

Figure \ref{fig:speed} shows the join time speedup that \cpsj achieves over \all. 
We achieve speedups of between $2-50\times$ for most of the datasets, with greater speedups at low similarity thresholds.
For a number of the datasets the \cpsj algorithm is slower than \all for the thresholds considered here.
Comparing with Table \ref{tab:datasets} it seems that \cpsj generally performs well on most data sets where tokens are contained in a large number of sets on average (NETFLIX, UNIFORM, DBLP),
but is beaten by \all on datasets that have a lot of ``rare'' tokens (SPOTIFY, FLICKR, AOL).
This is logical because rare tokens are exploited by the sorted prefix-filtering in \all.
Without rare tokens \all will be reading long inverted indexes.
This is a common theme among all the current state-of-the-art exact methods examined in \cite{mann2016}, including \pp.
\cpsj is robust in the sense that it does not depend on the presence of rare tokens in the data to perform well.
This difference is showcased with the synthetic TOKEN data sets.

\begin{figure}
  \centering
  \includegraphics[width=0.70\textwidth]{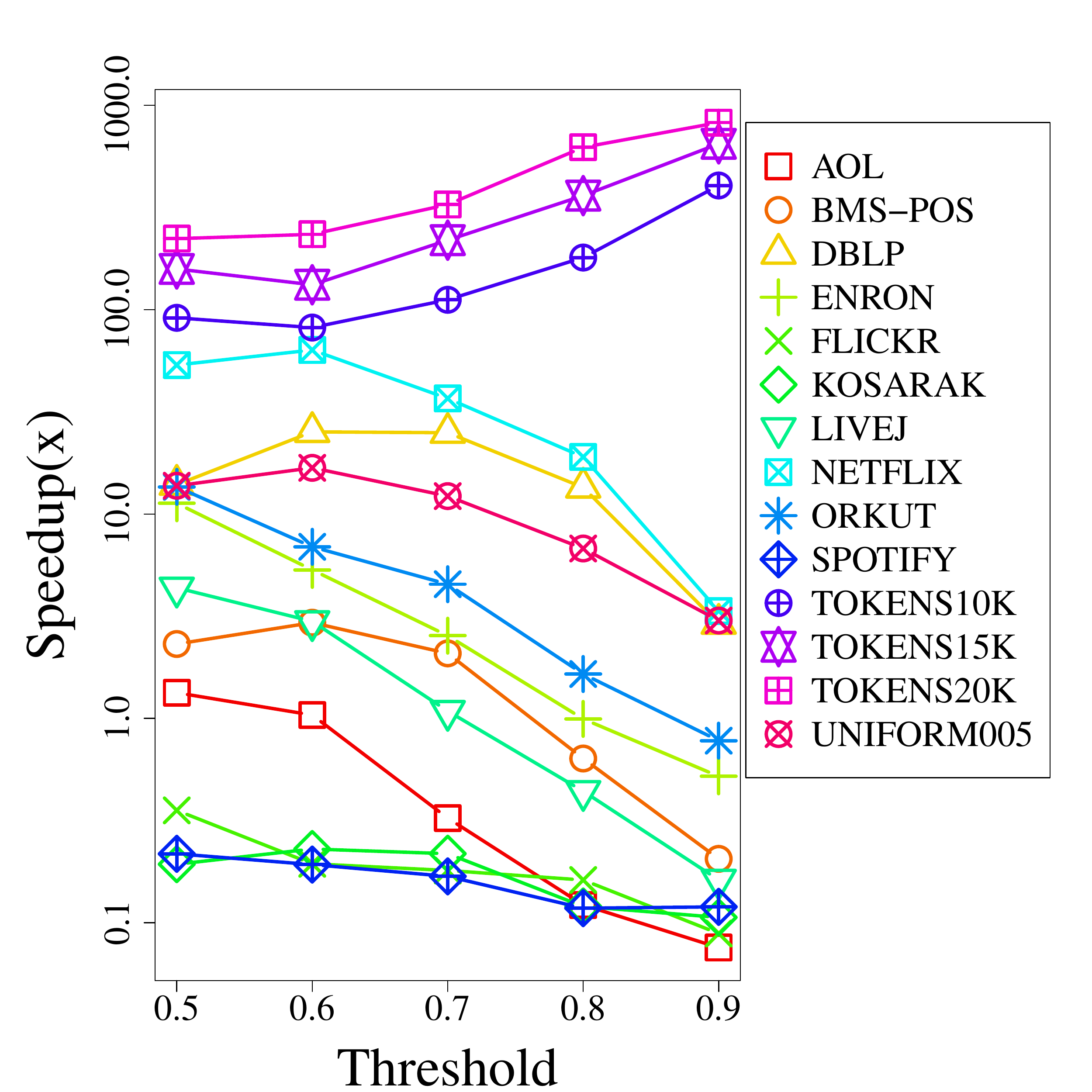}
  \caption{Join time of \cpsj with at least $90\%$ recall relative to \all.}
  \label{fig:speed}
\end{figure}

\paragraph{BayesLSH.}
The poor performance of \blsh compared to the other algorithms (\blsh was always slower) can most likely be tracked down to differences in the implementation of the candidate generation methods of \blsh.
The \blsh implementation uses an older implementation of \all compared to the implementation by Mann et al.~\cite{mann2016} which was shown to yield performance improvements by using a more efficient verification procedure. 
The LSH candidate generation method used by \blsh corresponds to the \mh splitting step, but with $k$ (the number of hash functions) fixed to one. 
Our technique for choosing $k$ in the \mh algorithm, aimed at minimizing the total join time, typically selects $k \in \{3,4,5,6\}$ in the experiments. 
It is therefore likely that \blsh can be competitive with the other techniques by combining it with other candidate generation procedures.
Further experiments to compare the performance of BayesLSH sketching to 1-bit minwise sketching for different parameter settings and similarity thresholds would also be instructive.

\paragraph{TOKEN datasets.}
The TOKENS datasets clearly favor the approximate join algorithms where \cpsj is two to three orders of magnitude faster than \all.
By increasing the number of times each token appears in a set we can make the speedup of \cpsj compared to \all arbitrarily large as shown by the progression from TOKENS10 to TOKENS20.
The \all algorithm generates candidates by searching through the lists of sets that contain a particular token, starting with rare tokens.
Since every token appears in a large number of sets every list will be long.

Interestingly, the speedup of \cpsj is even greater for higher similarity thresholds.
We believe that this is due to an increase in the gap between the similarity of sets to be reported and the remaining sets that have an average Jaccard similarity of $0.2$.
This is in line with our theoretical analysis of \cpsj and most theoretical work on approximate similarity search 
where the running time guarantees usually depend on the approximation factor.

\paragraph{Candidates and verification.}
Table \ref{tab:candidates} compares the number of pre-candidates, candidates, and results generated by the \all and \cpsj algorithms where the desired recall for \cpsj is set to be greater than $90\%$.
For \all the number of pre-candidates denotes all pairs $(x, y)$ investigated by the algorithm that pass checks on their size so that it is possible that $J(x, y) \geq \lambda$.
The number of candidates is simply the number of unique pre-candidates as duplicate pairs are removed explicitly by the \all algorithm.

For \cpsj we define the number of pre-candidates to be all pairs $(x, y)$ considered by the \textsc{BruteForcePairs} and \textsc{BruteForcePoint} subroutines of Algorithm \ref{alg:bruteforce}.
The number of candidates are pre-candidate pairs that pass size checks (similar to \all) and the 1-bit minwise sketching check as described in Section \ref{sec:implementation_cpsj}.
Note that for \cpsj the number of candidates may still contain duplicates as this is inherent to the approximate method for candidate generation. 
Removing duplicates though the use of a hash table would drastically increase the space usage of \cpsj.
For both \all and \cpsj the number of candidates denotes the number of points that are passed to the exact similarity verification step of the \all implementation of Mann et al.~\cite{mann2016}.

Table \ref{tab:candidates} shows that for \all there is not a great difference between the number of pre-candidates and number of candidates, 
while for \cpsj the number of candidates is usually reduced by one or two orders of magnitude for datasets where \cpsj performs well.
For datasets where \cpsj performs poorly such as AOL, FLICKR, and KOSARAK there is less of a decrease when going from pre-candidates to candidates.
It would appear that this is due to many duplicate pairs from the candidate generation step and not a failure of the sketching technique.
\subsection{Parameters} \label{sec:parameters}
To investigate how parameter settings affect the performance of the \cpsj algorithm we run experiments where we vary the brute force parameter \texttt{limit},
the brute force aggressiveness parameter $\varepsilon$, and the sketch length in words $\ell$.
Table \ref{tab:parameters} shows the parameter settings used during theses experiments and the final setting used for our join time experiments.
\begin{table}
\caption{Parameters of the \cpsj algorithm, their setting during parameter experiments, and their setting for the final join time experiments}
\label{tab:parameters}
\footnotesize
\centering
	\begin{tabular}{llrr} \toprule
		Parameter & Description & Test & Final \\ \midrule
		\texttt{limit} & Brute force limit & $100$ & $250$ \\
		$\ell$ & Sketch word length & $4$ & $8$ \\
		$t$ & Size of MinHash set & $128$ & $128$ \\
		$\varepsilon$ & Brute force aggressiveness & $0.0$ & $0.1$ \\ 
		$\delta$ & Sketch false negative prob. & $0.1$ & $0.05$ \\ \bottomrule 
	\end{tabular}
\end{table}
\begin{figure*}
\centering
\subfloat[\texttt{limit} $\in \{10, 50, 100, 250, 500 \}$]{\includegraphics[width = 0.5\textwidth]{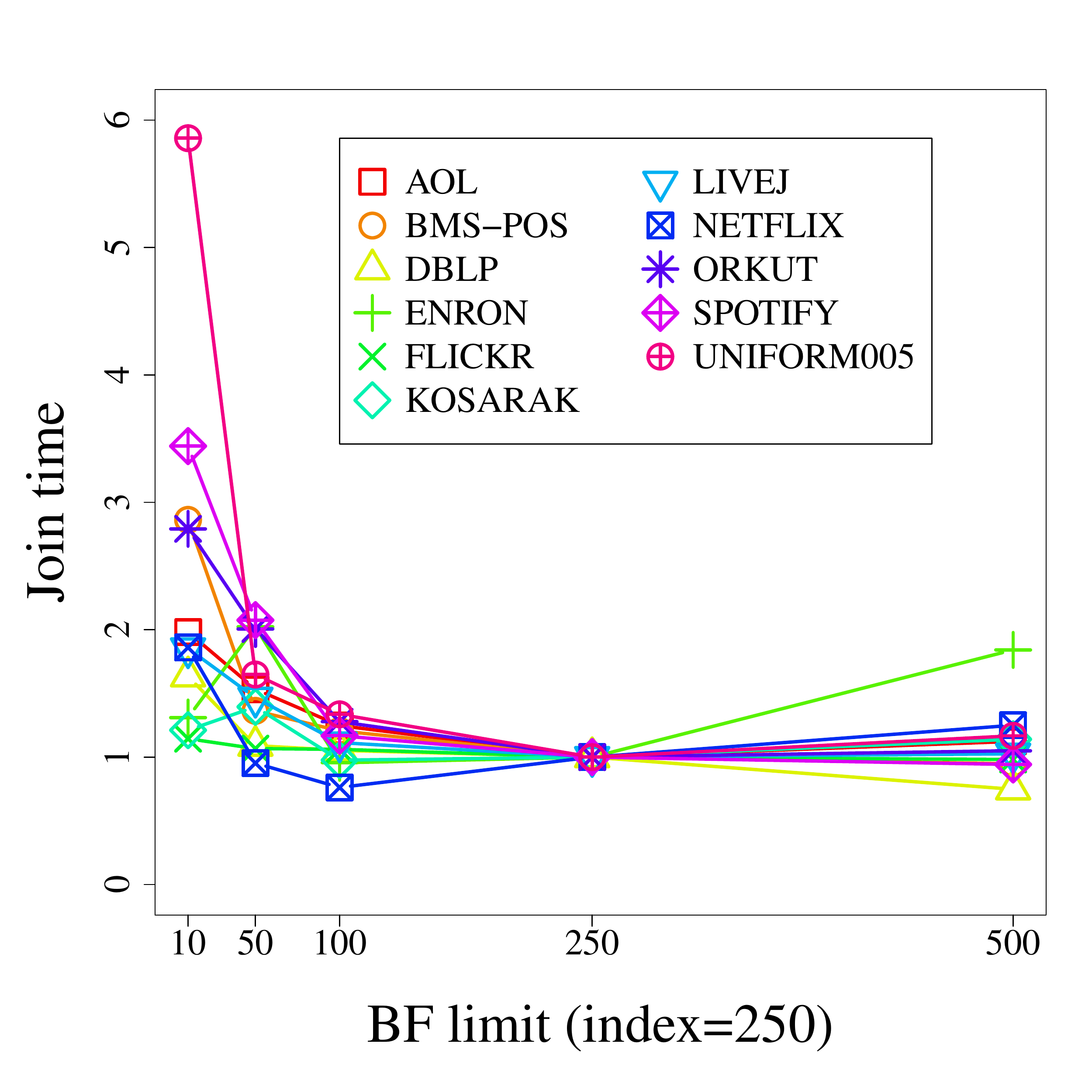}} 
\subfloat[$\varepsilon \in \{0.0, 0.1, 0.2, 0.3, 0.4, 0.5 \}$]{\includegraphics[width = 0.5\textwidth]{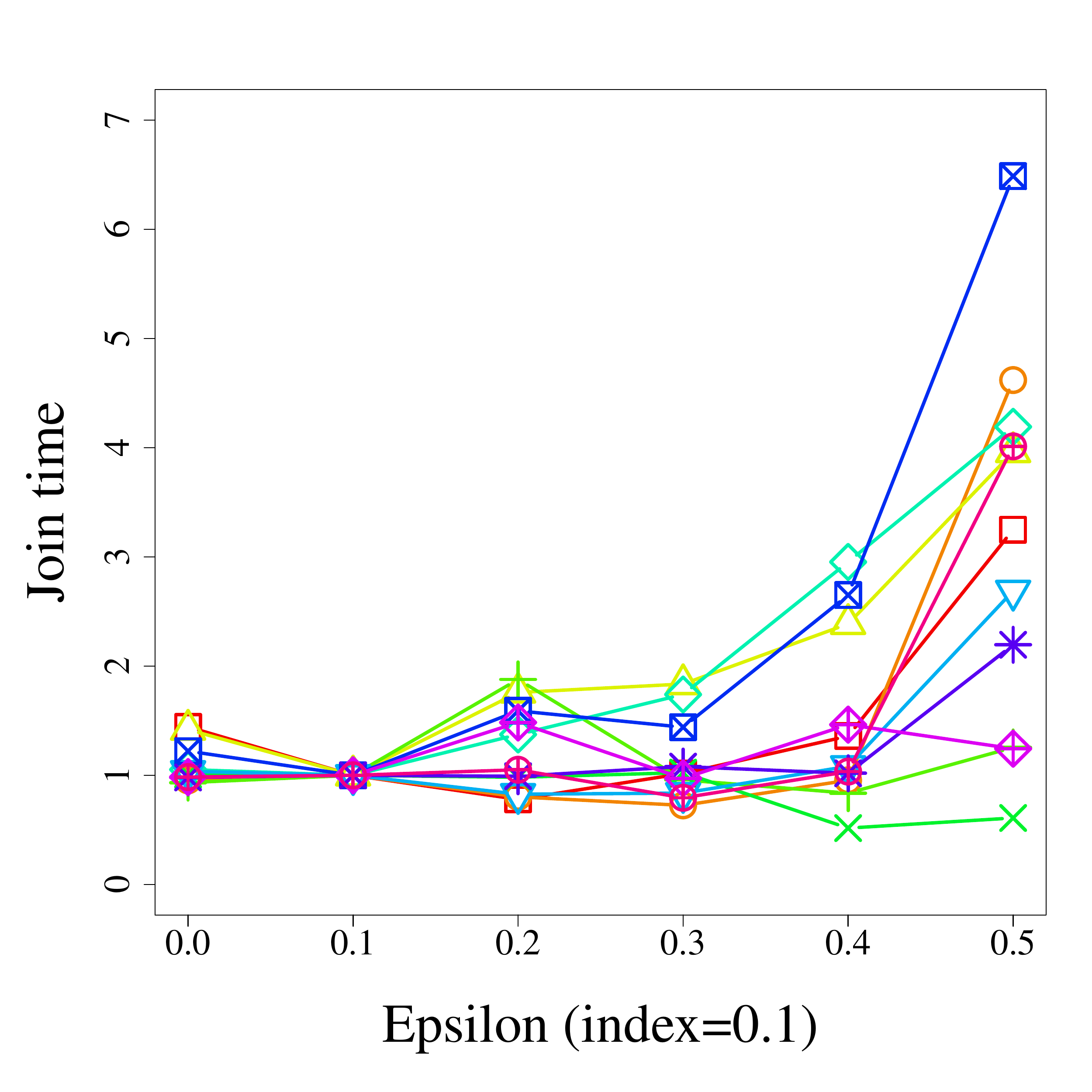}} \\
\subfloat[$w \in \{1, 2, 4, 8, 16 \}$]{\includegraphics[width = 0.55\textwidth]{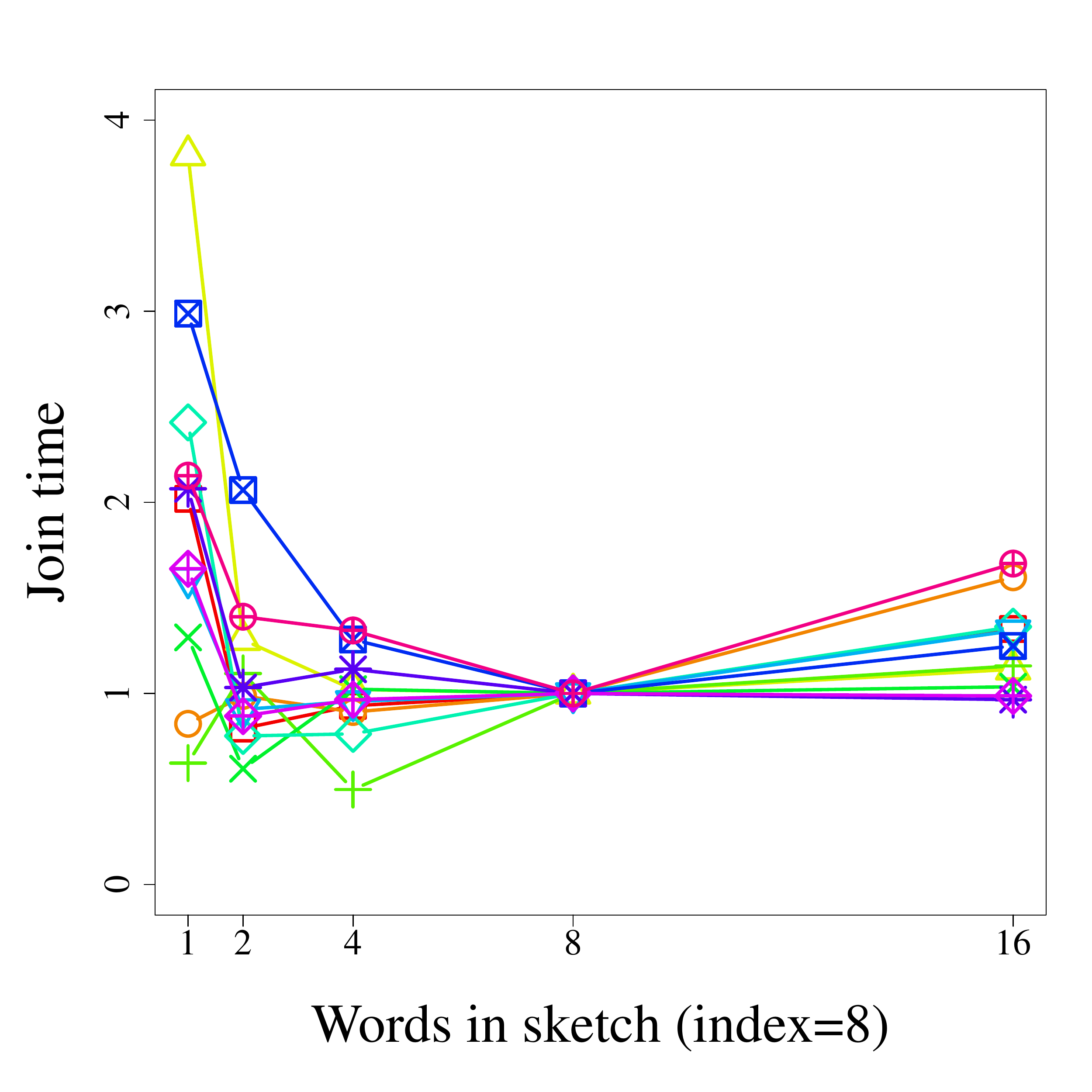}}    
\caption{Relative join time for \cpsj with at least $80\%$ recall and similarity threshold $\lambda = 0.5$ for different parameter settings of \texttt{limit}, $\varepsilon$, and $w$.} 
\label{fig:parameters}
\end{figure*}
Figure \ref{fig:parameters} shows the \cpsj join time for different settings of the parameters.
By picking one parameter at a time we are obviously ignoring possible interactions between the parameters, but the stability of the join times lead us to believe that these interactions have limited effect.

Figure \ref{fig:parameters} (a) shows the effect of the brute force limit on the join time. 
Lowering \texttt{limit} causes the join time to increase due to a combination of spending more time splitting sets into buckets and due to the lower probability of recall from splitting at a deeper level. The join time is relatively stable for \texttt{limit} $\in \{100, 250, 500\}$.  

Figure \ref{fig:parameters} (b) shows the effect of brute force aggressiveness on the join time. 
As we increase $\varepsilon$, sets that are close to the other elements in their buckets are more likely to be removed by brute force comparing them to all other points.
The tradeoff here is between the loss of probability of recall by letting a point continue in the $\cp$ branching process versus the cost of brute forcing the point.
The join time is generally increasing with $\varepsilon$, but it turns out that $\varepsilon = 0.1$ is a slightly better setting than $\varepsilon = 0.0$ for almost all data sets.

Figure \ref{fig:parameters} (c) shows the effect of sketch length on the join time.
There is a trade-off between the sketch similarity estimation time and the precision of the estimate, leading to fewer false positives.
For a similarity threshold of $\lambda = 0.5$ using only a single word negatively impacts the performance on most datasets compared to using two or more words.
The cost of using longer sketches seems neglible as it is only a few extra instructions per similarity estimation so we opted to use $\ell = 8$ words in our sketches.
\begin{table}
\caption{Number of pre-candidates, candidates and results for ALL and CP with at least $90\%$ recall.}
\label{tab:candidates}
\scriptsize
\centering
\renewcommand*{\arraystretch}{1.05}
\begin{tabular}{l|rr|rr}
\toprule
Dataset    &   \multicolumn{2}{c}{Threshold $0.5$}&\multicolumn{2}{c}{Threshold $0.7$} \\ 
		   &\multicolumn{1}{c}{ALL}&\multicolumn{1}{c}{CP} &	\multicolumn{1}{c}{ALL}&\multicolumn{1}{c}{CP} 	  \\ \midrule
           & 8.5E+09 & 7.4E+09 & 6.2E+08 & 2.9E+09 \\
AOL        & 8.5E+09 & 1.4E+09 & 6.2E+08 & 3.1E+07 \\
           & 1.3E+08 & 1.2E+08 & 1.6E+06 & 1.5E+06 \\ \hline
           & 2.0E+09 & 9.2E+08 & 2.7E+08 & 3.3E+08 \\
BMS-POS    & 1.8E+09 & 1.7E+08 & 2.6E+08 & 4.9E+06 \\
           & 1.1E+07 & 1.0E+07 & 2.0E+05 & 1.8E+05 \\ \hline
           & 6.6E+09 & 4.6E+08 & 1.2E+09 & 1.3E+08 \\
DBLP       & 1.9E+09 & 4.6E+07 & 7.2E+08 & 4.3E+05 \\
           & 1.7E+06 & 1.6E+06 & 9.1E+03 & 8.5E+03 \\ \hline
           & 2.8E+09 & 3.7E+08 & 2.0E+08 & 1.5E+08 \\
ENRON      & 1.8E+09 & 6.7E+07 & 1.3E+08 & 2.1E+07 \\
           & 3.1E+06 & 2.9E+06 & 1.2E+06 & 1.2E+06 \\ \hline
           & 5.7E+08 & 2.1E+09 & 9.3E+07 & 9.0E+08 \\
FLICKR     & 4.1E+08 & 1.1E+09 & 6.3E+07 & 3.8E+08 \\
           & 6.6E+07 & 6.1E+07 & 2.5E+07 & 2.3E+07 \\ \hline
           & 2.6E+09 & 4.7E+09 & 7.4E+07 & 4.2E+08 \\
KOSARAK    & 2.5E+09 & 2.1E+09 & 6.8E+07 & 2.1E+07 \\
           & 2.3E+08 & 2.1E+08 & 4.4E+05 & 4.1E+05 \\  \hline
           & 9.0E+09 & 2.8E+09 & 5.8E+08 & 1.2E+09 \\
LIVEJ      & 8.3E+09 & 3.6E+08 & 5.6E+08 & 1.8E+07 \\
           & 2.4E+07 & 2.2E+07 & 8.1E+05 & 7.6E+05 \\ \hline
           & 8.6E+10 & 1.3E+09 & 1.0E+10 & 4.3E+08 \\
NETFLIX    & 1.3E+10 & 3.1E+07 & 3.4E+09 & 6.4E+05 \\
           & 1.0E+06 & 9.5E+05 & 2.4E+04 & 2.2E+04 \\ \hline
           & 5.1E+09 & 1.1E+09 & 3.0E+08 & 7.2E+08 \\
ORKUT      & 3.9E+09 & 1.3E+06 & 2.6E+08 & 8.1E+04 \\
           & 9.0E+04 & 8.4E+04 & 5.6E+03 & 5.3E+03 \\ \hline
           & 5.0E+06 & 1.2E+08 & 4.7E+05 & 8.5E+07 \\
SPOTIFY    & 4.8E+06 & 3.1E+05 & 4.6E+05 & 2.7E+03 \\
           & 2.0E+04 & 1.8E+04 & 2.0E+02 & 1.9E+02 \\ \hline
           & 1.5E+10 & 1.7E+08 & 8.1E+09 & 4.9E+07 \\
TOKENS10K  & 4.1E+08 & 5.7E+06 & 4.1E+08 & 1.9E+06 \\
           & 1.3E+05 & 1.3E+05 & 7.4E+04 & 6.9E+04 \\ \hline
           & 3.6E+10 & 3.0E+08 & 1.9E+10 & 8.1E+07 \\
TOKENS15K  & 9.6E+08 & 7.2E+06 & 9.6E+08 & 1.9E+06 \\
           & 1.4E+05 & 1.3E+05 & 7.5E+04 & 6.9E+04 \\ \hline
           & 6.4E+10 & 4.4E+08 & 3.4E+10 & 1.0E+08 \\
TOKENS20K  & 1.7E+09 & 8.8E+06 & 1.7E+09 & 1.9E+06 \\
           & 1.4E+05 & 1.4E+05 & 7.9E+04 & 7.4E+04 \\ \hline
           & 2.5E+09 & 3.7E+08 & 6.5E+08 & 1.3E+08 \\
UNIFORM005 & 2.0E+09 & 9.5E+06 & 6.1E+08 & 3.9E+04 \\
           & 2.6E+05 & 2.4E+05 & 1.4E+03 & 1.3E+03 \\ \bottomrule 
\end{tabular}
\end{table}

\section{Conclusion}
We provided experimental and theoretical results on a new randomized set similarity join algorithm, \cpsj, and compared it experimentally to state-of-the-art exact and approximate set similarity join algorithms.
\cpsj is typically $2-4$ times faster than previous approximate methods.
Compared to exact methods it obtains speedups of more than an order of magnitude on real-world datasets, while keeping the recall above $90\%$.
Among the datasets used in these experiments we note that NETFLIX and FLICKR represents two archetypes.
On average a token in the NETFLIX dataset appears in more than $5000$ sets while on average a token in the FLICKR dataset appears in less than $20$ sets. 
Our experiments indicate that \cpsj brings large speedups to the NETFLIX type datasets, while it is hard to improve upon the perfomance of \all on the FLICKR type.

A direction for future work could be to tighten and simplify the theoretical analysis.
We conjecture that the running time of the algorithm can be bounded by a simpler function of the sum of similarities between pairs of points in $S$.
We note that recursive methods such as ours lend themselves well to parallel and distributed implementations since most of the computation happens in independent, recursive calls. Further investigating this is an interesting possibility.

\section*{Acknowledgment}
The authors would like to thank Willi Mann for making the source code and data sets of~\cite{mann2016} available, Aniket Chakrabarti for information about the implementation of BayesLSH, and the anonymous reviewers for useful suggestions.
\chapter{Lower bounds for asymmetric locality-sensitive hashing} \label{ch:asymmetric}
\sectionquote{Deep roots are not reached by the frost}
\noindent A locality-sensitive hashing (LSH) scheme~\cite{indyk1998} for a set of objects $X$ equipped with a pairwise measure of similarity $\simil \colon X \times X \to \real$ is a distribution $\LSH$ over functions $h \colon X \to R$ such that the probability that a pair of points in $X$ collide under a randomly sampled $h \sim \LSH$ is a function of their similarity. 
That is, there exists a collision probability function (CPF) $f \colon \real \to [0,1]$ such that for every $x, y \in X$ we have $\Pr_{h \sim \LSH}[h(x) = h(y)] = f(\simil(x, y))$. 

To apply locality-sensitive hashing to solve the problem of approximate $(s_1, s_2)$-similarity search we are interested in finding the distribution $\LSH$ with a CPF $f$ that magnifies the gap between the collision probability $f(s_1)$ of points with similarity at least $s_1$ and the collision probability $f(s_2)$ of points with similarity at most $s_2 < s_1$ where we assume that $f$ is non-decreasing. 
The performance of the solution will be governed by the parameter $\rho$ defined by $f(s_1) = f(s_2)^\rho$.  

Asymmetric locality-sensitive hashing~\cite{shrivastava2014asymmetric} extends locality-sensitive hashing by using a distribution $\ALSH$ over pairs of functions $h, g \colon X \to R$ such that $\Pr_{(h,g) \sim \ALSH}[h(x) = g(y)] = f(\simil(x, y))$. 
The use of asymmetry allows us to construct schemes where the probability of collision between identical points is less than $1$.
Applications of asymmetric locality-sensitive hashing include maximal inner product search~\cite{shrivastava2014asymmetric}, orthogonal vector search~\cite{vijayanarasimhan2014hashing}, annulus queries~\cite{aumuller2017distance}, privacy-preserving similarity estimation~\cite{aumuller2017distance}, and spherical range reporting~\cite{aumuller2017distance}.
This chapter presents the lower bounds for asymmetric locality-sensitive hashing from~\cite{aumuller2017distance} and extends them slightly.
\paragraph{Contribution.}
We show lower bounds for asymmetric locality-sensitive hashing schemes for the Boolean hypercube $\cube{d}$ under cosine similarity $\simil(x, y) = \ip{x}{y}/d$.
Specifically, we will both upper and lower bound the probability of randomly $\alpha$-correlated points $x, y \in \cube{d}$ colliding under $(h,g) \sim \ALSH$ as a function of $\alpha$ and the probability that randomly $0$-correlated points collide under $(h, g) \sim \ALSH$. 
Note that $(x, y$) being randomly $0$-correlated simply corresponds to $x$ and $y$ being independent and uniformly distributed in $\cube{d}$.
\begin{theorem} \label{thm:asymmetric_lower}
	Let $\ALSH$ be a distribution over function pairs $h, g \colon \cube{d} \to R$, 
	and define $\hat{f} \colon [-1,1] \to [0,1]$ as $\hat{f}(\alpha) = \Pr[h(x) = g(y)]$ where $x, y$ are randomly $\alpha$-correlated and $(h,g) \sim \ALSH$.  
	Then, for every $\alpha \in (-1, 1)$ we have 
	\begin{equation*}
		\hat{f}(0)^{\frac{1+|\alpha|}{1-|\alpha|}} \leq \hat{f}(\alpha) \leq \hat{f}(0)^{\frac{1-|\alpha|}{1+|\alpha|}}. 
	\end{equation*}
\end{theorem}
As $d$ increases the empirical correlation of a pair of randomly $\alpha$-correlated points will be concentrated around $\alpha$. 
Theorem \ref{thm:asymmetric_lower} can therefore be used to show that for $\alpha \geq 0$,
asymmetric locality-sensitive hashing for $(\alpha, 0)$-cosine similarity search must have $\rho = \log(f(\alpha)) / \log(f(0)) \geq \frac{1-\alpha}{1+\alpha}$ up to lower order terms.
This lower bound on the $\rho$-value matches an existing lower bound for standard locality-sensitive hashing \cite{motwani2007, andoni2016tight} and matches the upper bound of cross-polytope LSH~\cite{andoni2015practical}, showing that asymmetry cannot help provide better hashing schemes for this problem.
Implicitly, this extension of the standard LSH lower bound to asymmetric LSH for $(\alpha, 0)$-cosine similarity search already follows from the space-time tradeoff lower bounds for similarity search shown independently by Andoni et al.~\cite{andoni2017optimal} and Christiani~\cite{christiani2017framework}. 

Let $\alpha \geq 0$ and suppose we are interested in solving the $(0, -\alpha)$-similarity search problem under negative cosine similarity, i.e.\ we want to report points with cosine similarity at most $0$ while avoiding points with cosine similarity greater than $\alpha$. 
Theorem \ref{thm:asymmetric_lower} lower bounds the $\rho$-value of asymmetric locality-sensitive hashing schemes for this problem by $\rho = \log(f(0)) / \log(f(\alpha)) \geq \frac{1-\alpha}{1+\alpha}$ up to lower order terms.
This can again be matched by a simple asymmetric version of cross-polytope LSH: 
Let $\LSH$ denote the cross-polytope LSH, then we sample $(h, g) \sim \ALSH$ by sampling $h' \sim \LSH$ and setting $h(x) = h'(x)$ and $g(x) = h'(-x)$.  
For more details and related work see~\cite{aumuller2017distance}.
\paragraph{Techniques.}
The proof of Theorem \ref{thm:asymmetric_lower} combines the small-set expansion theorems by O'Donnell~\cite{odonnell2014analysis} with techniques inspired by the LSH lower bound of Motwani et al.~\cite{motwani2007}.
The (reverse) small-set expansion theorem (lower) upper bounds the probability that randomly $\alpha$-correlated points $(x, y)$ 
end up in a pair of subsets $A, B$ of the Boolean hypercube, as a function of the size of these subsets.  
We are able to extend these bounds for pairs of subsets of the Boolean hypercube to distributions over pairs of functions through a sequence of applications of primarily Jensen's inequality.
\section{Preliminaries}
We begin by introducing the required tools from~\cite{odonnell2014analysis}, starting with the definition of randomly $\alpha$-correlated pairs of points.
\begin{definition}\label{def:correlation}
	For $-1 \leq \alpha \leq 1$ and $x, y \in \cube{d}$ we say that $(x, y)$ is randomly $\alpha$-correlated 
	if the pairs $(x_i, y_i)$ are i.i.d.\ with $x_i$ uniform in $\cube{d}$ and 
	\begin{equation*}
		y_i =
		\begin{cases}
			x_i & \text{with probability } \frac{1 + \alpha}{2}, \\
			-x_i & \text{with probability } \frac{1 - \alpha}{2}. 
		\end{cases}
	\end{equation*}
\end{definition}
Next we define a probabilistic version of the CPF for cosine similarity on the Boolean hypercube.
Theorem \ref{thm:asymmetric_lower} provides upper and lower bounds on the probabilistic CPF. 
\begin{definition}[Probabilistic CPF]
	Let $\ALSH$ be a distribution over pairs $h, g \colon \cube{d} \to R$. 
    We define the probabilistic CPF $\hat{f} \colon [-1,1] \to [0,1]$ by 
	\begin{equation*}
		\hat{f}(\alpha) = \Pr_{\substack{(h,g) \sim \ALSH \\ \text{$(x, y)$ $\alpha$-corr.} }}[h(x) = g(y)].
	\end{equation*}
\end{definition}
The proof of the lower bound will make use of the following technical inequality that follows from two applications of Jensen's inequality.
\begin{lemma} \label{lem:inequality} 
	Let $p, q$ denote discrete probability distributions, then for every $c \geq 1$ we have that 
	\begin{equation*}
		\sum_{i} (p_i q_i)^c \geq \left(\sum_{i} p_i q_i \right)^{2c-1} 
	\end{equation*}
	with reverse inequality for $c \leq 1$.
\end{lemma}
\begin{proof}
	Assume $c \geq 1$.
	By Jensen's inequality, using the fact that $x \mapsto x^c$ and $x \mapsto x^{2 - 1/c}$ are convex we have that 
	\begin{equation*}
	\sum_{i} (p_i q_i)^c = \sum_{i} p_i(p_i^{1 - 1/c} q_i)^c 
						 \geq \left( \sum_i p_i^{2 - 1/c} q_i \right)^c 
						 \geq \left(\sum_{i} p_i q_i \right)^{2c-1}. 
	\end{equation*}
	For $c \leq 1$ we have that $x \mapsto x^c$ and $x \mapsto x^{2 - 1/c}$ are concave and the inequality is reversed.
\end{proof}
\section{Lower bounding the collision probability}
The reverse small-set expansion theorem lower bounds the probability that random $\alpha$-correlated points $(x, y)$ 
end up in a pair of subsets $A, B$ of the Hamming cube, as a function of the size of these subsets.  
In the following, for $A \subset \cube{d}$ we refer to the quantity $|A|/2^d$ as the \emph{volume} of $A$.
\begin{theorem}[Rev.\ Small-Set Expansion~\cite{odonnell2014analysis}]\label{thm:expansionreverse}
Let $0 \leq \alpha \leq 1$. Let $A, B \subseteq \cube{d}$ have volumes $\exp(-a^{2}/2)$, $\exp(-b^{2}/2)$, respectively, where $a, b \geq 0$. 
Then we have that
\begin{equation*}
	\Pr_{\substack{(x, y) \\ \alpha \text{-corr.}}} [x \in A, y \in B]\geq \exp\left(-\frac{1}{2}\frac{a^2 + 2\alpha a b + b^2}{1 - \alpha^2}\right).
\end{equation*}
\end{theorem}
In the following lemma we convert the lower bound in Theorem \ref{thm:expansionreverse} into a lower bound on the probabilistic CPF.
\begin{lemma} \label{lem:revexphash}
	For every distribution $\ALSH$ over function pairs $h, g \colon \cube{d} \to R$ and $\alpha \in [0,1)$ we have that $\hat{f}(\alpha) \geq \hat{f}(0)^{\frac{1+\alpha}{1-\alpha}}$.
\end{lemma}
\begin{proof}
For a function $h \colon \cube{d} \to R$ define its inverse image $h^{-1} \colon R \to 2^{\cube{d}}$ by $h^{-1}(i) = \{ x \in \cube{d} \mid h(x) = i \}$.
For a pair of functions $(h, g) \in \DSH$ and $i \in R$ we define $a_{h,i}, b_{g,i} \geq 0$ such that
$|h^{-1}(i)|/2^{d} = \exp(-a_{h,i}^{2}/2)$ and $|g^{-1}(i)|/2^{d} = \exp(-b_{g,i}^{2}/2)$.
For fixed $(h, g)$ define $\hat{f}_{h,g}(\alpha) = \Pr_{\text{$(x, y)$ $\alpha$-corr.}}[h(x) = g(y)]$. We obtain a lower bound on $\hat{f}(\alpha)$ as follows:
\begin{align*}
	\hat{f}(\alpha) &= \mathop{\E}_{(h, g) \sim \ALSH} \left[ \sum_{i \in R} \Pr_{\text{$(x, y)$ $\alpha$-corr.}}[h(x) = g(y) = i] \right] \\
                 &\stackrel{(1)}{\geq} \mathop{\E}_{(h, g) \sim \ALSH} \left[ \sum_{i \in R} \exp\left(-\frac{1}{2}\frac{a_{h,i}^{2} + 2\alpha a_{h,i}b_{g,i} + b_{g,i}^{2}}{1 - \alpha^2}\right) \right]  \\
                 &\stackrel{(2)}{\geq} \mathop{\E}_{(h, g) \sim \ALSH} \left[ \sum_{i \in R} \exp\left(-\frac{1}{2}\frac{a_{h,i}^{2} + b_{g,i}^{2}}{1 - \alpha}\right) \right] \\
                 &\stackrel{(3)}{\geq} \mathop{\E}_{(h, g) \sim \ALSH} \hat{f}_{h,g}(0)^\frac{1 + \alpha}{1 - \alpha} 
                 \stackrel{(4)}{\geq} \left( \mathop{\E}_{(h, g) \sim \ALSH} \hat{f}_{h,g}(0) \right)^\frac{1 + \alpha}{1 - \alpha} \\ 
				 &= \hat{f}(0)^\frac{1 + \alpha}{1 - \alpha}. 
\end{align*}
Here, (1) is due to Theorem \ref{thm:expansionreverse}, 
(2) follows from the fact that $(1+\alpha)(a^2 + b^2) \geq a^2 + 2\alpha a b + b^2$, 
(3) follows from Lemma~\ref{lem:inequality} with $c = 1/(1-\alpha)$, 
and (4) follows from a standard application of Jensen's Inequality.
\end{proof}
\section{Upper bounding the collision probability}
We can re-apply the techniques behind Lemma \ref{lem:revexphash} to upper bound the probabilistic CPF. 
This is similar to the well-studied problem of constructing LSH lower bounds and our results match known LSH bounds~\cite{motwani2007,andoni2016tight}, 
showing that the asymmetry afforded by $\ALSH$ does not help us when we wish to maximize the gap in the collision probability between random and $\alpha$-correlated points as measured by the $\rho$-value.
As with Lemma~\ref{lem:revexphash}, the following theorem by O'Donnell~\cite{odonnell2014analysis} is the foundation of our upper bound on the probilistic CPF.
\begin{theorem}[Gen.\ Small-Set Expansion]\label{thm:expansion}
Let $0 \leq \alpha \leq 1$. Let $A, B \subseteq \cube{d}$ have volumes $\exp(-a^{2}/2)$, $\exp(-b^{2}/2)$ and assume $0 \leq \alpha b \leq a \leq b$. Then,
\begin{equation*}
	\Pr_{\substack{(x, y) \\ \alpha \text{-corr.}}}[x \in A, y \in B] \leq \exp\left(-\frac{1}{2}\frac{a^2 - 2\alpha a b + b^2}{1 - \alpha^2}\right).
\end{equation*}
\end{theorem}
In the proof of the upper bound on the probabilistic CPF we can without loss of generality assume that the family $\ALSH$ satisfies the volume restrictions that $0 \leq \alpha b \leq a \leq b$ from Theorem \ref{thm:expansion} on the parts of $(h, g) \sim \ALSH$. 
The reason is that for $a = \alpha b$ the upper bound becomes $\Pr_{x \in \cube{d}}[x \in B]$ which is an upper bound for every $0 \leq a < \alpha b$ as well.
For completeness we include the proof of the upper bound.
\begin{lemma} \label{lem:exphash}
	For every distribution $\ALSH$ over function pairs $h, g \colon \cube{d} \to R$ and $\alpha \in [0,1)$ we have that $\hat{f}(\alpha) \leq \hat{f}(0)^{\frac{1-\alpha}{1+\alpha}}$.
\end{lemma}
\begin{proof}
Using the same notation as in the proof of Lemma \ref{lem:revexphash} we derive the upper bound as follows:
\begin{align*}
	\hat{f}(\alpha) &= \mathop{\E}_{(h, g) \sim \ALSH} \left[ \sum_{i \in R} \Pr_{\text{$(x, y)$ $\alpha$-corr.}}[h(x) = g(y) = i] \right] \\
                 &\stackrel{(1)}{\leq} \mathop{\E}_{(h, g) \sim \ALSH} \left[ \sum_{i \in R} \exp\left(-\frac{1}{2}\frac{a_{h,i}^{2} - 2\alpha a_{h,i}b_{g,i} + b_{g,i}^{2}}{1 - \alpha^2}\right) \right]  \\
                 &\stackrel{(2)}{\leq} \mathop{\E}_{(h, g) \sim \ALSH} \left[ \sum_{i \in R} \exp\left(-\frac{1}{2}\frac{a_{h,i}^{2} + b_{g,i}^{2}}{1 + \alpha}\right) \right] \\
                 &\stackrel{(3)}{\leq} \mathop{\E}_{(h, g) \sim \ALSH} \hat{f}_{h,g}(0)^\frac{1 + \alpha}{1 - \alpha} 
                 \stackrel{(4)}{\leq} \left( \mathop{\E}_{(h, g) \sim \ALSH} \hat{f}_{h,g}(0) \right)^\frac{1 - \alpha}{1 + \alpha} \\ 
				 &= \hat{f}(0)^\frac{1 - \alpha}{1 + \alpha}. 
\end{align*}
Here, (1) is due to Theorem \ref{thm:expansion}, 
(2) follows from the fact that $(1 - \alpha)(a^2 + b^2) \leq a^2 -2\alpha a b + b^2$, 
(3) follows from Lemma~\ref{lem:inequality} with $c = 1/(1+\alpha)$, 
and (4) follows from a standard application of Jensen's Inequality.
\end{proof}
\section{Extension to negative correlation}
We can extend the bounds from Lemma \ref{lem:revexphash} and Lemma \ref{lem:exphash} to negative values of $\alpha$, essentially mirroring the bounds around $\alpha = 0$, 
by showing that an asymmetric LSH scheme $\ALSH$ with a value of $\hat{f}(\alpha)$ that is too low or too high for negative values of $\alpha$ will contradict the bounds we have for positive values of $\alpha$.
The extension of the bounds to negative correlation uses the following lemma to show contradictions:
\begin{lemma} \label{lem:correlation_symmetry}
	Suppose we have an asymmetric LSH $\ALSH'$ for $\cube{d}$ with probabilistic CPF $\hat{f'}$,
	then the we can construct an asymmetric LSH $\ALSH$ for $\cube{d}$ with probabilistic CPF $\hat{f}(\alpha) = \hat{f}'(-\alpha)$.
\end{lemma}
\begin{proof}
	We sample $(h, g) \sim \ALSH$ by sampling $(h', g') \sim \ALSH'$ and setting $h(x) = h'(x)$ and $g(y) = g'(-y)$. 
	The property that $\hat{f}(\alpha) = \hat{f}'(-\alpha)$ follows from the fact that if $(x, y)$ are randomly $\alpha$-correlated then $(x, -y)$ are randomly $-\alpha$-correlated.
\end{proof}
We can now extend the bounds on the probabilistic CPF to prove Theorem \ref{thm:asymmetric_lower}.
For the upper bound, suppose for $\alpha \in [0, 1)$ that there exists a family $\ALSH'$ with $\hat{f}'(-\alpha) > \hat{f}'(0)^{\frac{1-\alpha}{1 + \alpha}}$.
According to Lemma \ref{lem:correlation_symmetry} this implies the existence of a family $\ALSH$ with $\hat{f}(\alpha) > \hat{f}(0)^{\frac{1-\alpha}{1 + \alpha}}$, contradicting Lemma \ref{lem:exphash}.
The argument for extending the lower bound is identical, completing the proof of Theorem \ref{thm:asymmetric_lower}.
\section{Conclusion and open problems}
We have shown upper and lower bounds on the collision probability of asymmetric locality-sensitive hashing schemes.
The bounds are tight up to lower order terms and match an existing symmetric LSH scheme on the unit sphere and its natural asymmetric extension.
An interesting consequence of the upper bound is that asymmetry cannot improve the $\rho$-value for the standard similarity search problem for random instances in the Boolean hypercube. 

The application of asymmetric locality-sensitive hashing to orthogonal vector search seeks a family $\ALSH$ with a CPF $f$ that peaks at $f(0)$ and is decreasing in $|\alpha|$~\cite{vijayanarasimhan2014hashing}.
Currently the best known upper bound on the $\rho$-value for orthogonal vector search is $\rho = \log(f(0))/\log(\max(f(-\alpha), f(\alpha))) = \frac{1 - \alpha^2}{1 + \alpha^2}$ which can be achieved by tensoring the input (squaring the cosine similarity between pairs of points) and applying the standard cross-polytope LSH family, or by combining standard cross-polytope LSH and its negated asymmetric version~\cite{aumuller2017distance} using the same technique as in Lemma \ref{lem:correlation_symmetry}. 
We conjecture that the current upper bound is tight implying that the following bound must hold:
\begin{equation*}
\max(\hat{f}(-\alpha), \hat{f}(\alpha)) \geq \hat{f}(0)^{\frac{1 + \alpha^2}{1 - \alpha^2}}.
\end{equation*}
This bound can be characterized as a two-sided bound whereas Theorem \ref{thm:asymmetric_lower} only provides one-sided upper and lower bounds.
The problem of obtaining a two-sided lower bound is related to the open symmetric Gaussian problem~\cite{odonnell2012open}.
\chapter{Optimal Boolean locality-sensitive hashing} \label{ch:bool}
\sectionquote{The crownless again shall be king}
\begin{theorem} \label{thm:bool_main}
For $0 \leq \beta < \alpha < 1$ the distribution $\LSH$ over Boolean functions $h \colon \cube{d} \to \{-1, 1\}$ that minimizes the expression 
\begin{equation*} 
	\rho_{\alpha, \beta} = \frac{\log(1/\Pr_{\substack{h \sim \LSH \\ (x, y) \text{ $\alpha$-corr.}}}[h(x) = h(y)])}{\log(1/\Pr_{\substack{h \sim \LSH \\ (x, y) \text{ $\beta$-corr.}}}[h(x) = h(y)])}
\end{equation*}
assigns nonzero probability only to members of the set of dictator functions $h(x) = \pm x_i$. 
\end{theorem}
\section{Introduction}
We will be studying Boolean functions, i.e., functions that for a positive integer $d$ can be written in the form
\begin{equation*}
	h \colon \{-1, 1\}^d \to \{-1, 1\}.
\end{equation*}
We are concerned with the behavior of such Boolean functions on input pairs $x, y \in \cube{d}$ that are randomly generated.
\begin{definition}\label{bool:def:correlation}
	For $-1 \leq \alpha \leq 1$ and $x \in \cube{d}$ we let $N_{\alpha}(x)$ denote the distribution over $\cube{d}$ where each component of $y \sim N_{\alpha}(x)$ is independently distributed according to 
	\begin{equation*}
		y_i =
		\begin{cases}
			x_i & \text{with probability } \frac{1 + \alpha}{2}, \\
			-x_i & \text{with probability } \frac{1 - \alpha}{2}. 
		\end{cases}
	\end{equation*}
	We say that $(x, y)$ is randomly $\alpha$-correlated if $x$ is uniformly distributed over $\cube{d}$ and $y \sim N_{\alpha}(x)$.
\end{definition}
Let $\LSH$ denote a distribution over functions $h \colon \{-1, 1\}^d \to R$ where $R$ is a finite set and define 
\begin{equation*}
	p_\alpha = \Pr_{\substack{h \sim \LSH \\ (x, y) \text{ $\alpha$-corr.}}}[h(x) = h(y)].
\end{equation*}
For $0 \leq \beta < \alpha \leq 1$ we wish to characerize the distributions that minimize the expression
\begin{equation} \label{eq:rho}
	\rho_{\alpha, \beta} = \frac{\log(1/p_\alpha)}{\log(1/p_\beta)} 
\end{equation}
when we restrict $\LSH$ to be a distribution over Boolean functions $h \colon \cube{d} \to \{-1, 1\}$.
The expression for $\rho_{\alpha, \beta}$ in equation \eqref{eq:rho} is a well-known quantity in the study of approximate near neighbor search governing the query time and space usage of solutions based on locality-sensitive hashing~\cite{indyk1998}.
\section{Related work}
Indyk and Motwani~\cite{indyk1998} introduced the uniform distribution over the set of dictator functions as a family of locality-sensitive hash functions for the Boolean hypercube.
O'Donnell et al.~\cite{odonnell2014optimal} showed that for general families $\LSH$ it must hold that $\rho_{\alpha, \beta} \geq \log(1/\alpha) / \log(1/\beta)$. 
This matches the upper bound of Indyk and Motwani~\cite{indyk1998} when $\alpha, \beta$ approach $1$. 
Another line of work\cite{panigrahy2008geometric, andoni2016tight} using hypercontractive inequalities showed that $\rho_{\alpha, 0} \geq (1-\alpha)/(1+\alpha)$, 
matching the upper bound of Andoni et al.~\cite{andoni2015optimal}.

The question of finding lower bounds for $\rho_{\alpha, \beta}$ for every choice of $0 \leq \beta < \alpha \leq 1$ is still open. 
In this note we answer the question for distributions over \emph{Boolean} functions, showing that the upper bound of Indyk and Motwani is optimal.
The resulting $\rho$-value is given by
\begin{equation*} 
	\rho_{\alpha, \beta} = \frac{\log((1 + \alpha)/2)}{\log((1 + \beta)/2)}. 
\end{equation*}
\section{Preliminaries}
We will be using tools from the Fourier analysis of Boolean functions to find the minimum of $\rho_{\alpha, \beta}$.
For a more detailed overview we refer to the book by O'Donnell~\cite{odonnell2014analysis}.
We will be using the fact that Boolean functions can be uniquely expressed as multilinear polynomials:
\begin{theorem}
	Every function $f \colon \cube{d} \to \mathbb{R}$ can be uniquely expressed as a multilinear polynomial
	\begin{equation*}
		f(x) = \sum_{S \subseteq [d]}\hat{f}(S)x^S
	\end{equation*}
	where $\hat{f}(S) \in \mathbb{R}$ and $x^S = \prod_{i \in S}x_i$. 
\end{theorem}
For $S \subseteq [d]$ we refer to $\hat{f}(S)$ as the Fourier coefficient of $f$ on $S$. 
The two following Theorems define an inner product between Boolean function and shows how it relates to their Fourier coefficents.  
\begin{theorem}[Plancherel's Theorem]
	For any $f, g \colon \cube{d} \to \mathbb{R}$
	\begin{equation*}
		\ip{f}{g} = \E_{x \sim \cube{d}}[f(x)g(x)] = \sum_{S \subseteq [d]}\hat{f}(S)\hat{g}(S).
	\end{equation*}
\end{theorem}
The concept of Fourier weight will be useful when characterizing the how Boolean functions behave on noisy inputs:
\begin{definition} \label{def:fourier_weight}
	For $f \colon \cube{d} \to R$ define the Fourier weight of $f$ at degree $k \in [d]$ by
	\begin{equation*}
		W^{k}[f] = \sum_{\substack{S \subseteq [d] \\ |S| = k}}\hat{f}(S)^2.
	\end{equation*}
\end{definition}
Consider Plancherel's Theorem with $f = g$ and where $f$ is Boolean-valued.
In this case we get that the sum of the squared Fourier coefficients of $f$ equals 1. 
This result is known as Parseval's Theorem and we will make use of it to determine where to place to Fourier weight of $f$ in order to minimize $\rho$.  
\begin{theorem}[Parseval's Theorem]
	For any $f \colon \cube{d} \to \{-1, 1\}$ 
	\begin{equation*}
		\ip{f}{f} = \E_{x \sim \cube{d}}[f(x)^2] = \sum_{S \subseteq [d]}\hat{f}(S)^2 = \sum_{i = 0}^{d}W^{i}[f] = 1.
	\end{equation*}
\end{theorem}
In order to study the behavior of Boolean functions under noise we introduce the noise operator $T_\alpha$.
\begin{definition}
	For $\alpha \in [-1,1]$ the noise operator with parameter $\alpha$ is the linear operator $T_{\alpha}$ on functions $f \colon \cube{d} \to \mathbb{R}$ defined by 
	\begin{equation*}
		T_{\alpha}f(x) = \E_{y \sim N_{\alpha}(x)}[f(y)].
	\end{equation*}
\end{definition}
The Fourier expansion of $T_{\alpha}f(x)$ is given by $\sum_{S \subseteq [d]}\alpha^{|S|}\hat{f}(S)x^S$ and it follows from Plancherel's Theorem that
\begin{align} 	
	\ip{f}{T_\alpha g} &= \E_{x \sim \cube{d}}[f(x) \E_{y \sim N_{\alpha}(x)}[g(y)]] \nonumber \\
	&= \E_{\substack{(x, y) \text{ $\alpha$-corr.}}}[f(x)g(y)] = \sum_{S \subseteq [d]}\alpha^{|S|}\hat{f}(S)\hat{g}(S). \label{eq:noise_plancherel}
\end{align}
In the analysis of our problem the following inequality will be used several times.
For the remainder of this Chapter we will use $\log x$ to denote the natural logarithm of $x$.
\begin{lemma}\label{lem:log_inequality}
	For $x > 0$ we have $\log x \leq x - 1$ with equality if and only if $x = 1$.
\end{lemma}
\section{Bit-sampling is optimal}
Our approach will be to minimize $\rho_{\alpha, \beta}$ subject to the constraint that members of $\LSH$ are Boolean functions $h \colon \cube{d} \to \{-1, 1\}$.
We begin by making some observations to simplify the problem.
For $h \sim \LSH$ we can directly relate the noise-sensitivity under random $\alpha$-correlated inputs to the collision probability.
\begin{align*}
	\E_{\substack{h \sim \LSH \\ (x, y) \text{ $\alpha$-corr.}}}[h(x)h(y)] 
	&= \Pr_{\substack{h \sim \LSH \\ (x, y) \text{ $\alpha$-corr.}}}[h(x) = h(y)] - \Pr_{\substack{h \sim \LSH \\ (x, y) \text{ $\alpha$-corr.}}}[h(x) \neq h(y)] \\
	&= p_\alpha - (1 - p_\alpha) \\
	&= 2p_\alpha - 1.
\end{align*}
Using Equation \eqref{eq:noise_plancherel} we can write $p_\alpha$ as follows:
\begin{equation*}
p_\alpha = (1 + \E_{\substack{h \sim \LSH \\ (x, y) \text{ $\alpha$-corr.}}}[h(x)h(y)])/2 = (1 + \sum_{i = 0}^{d}\alpha^{i} w_i)/2 
\end{equation*}
where we use $w_i$ to denote the expected Fourier weight of $h \sim \LSH$ at degree $i$ defined by $w_i = \E_{h \sim \LSH} \sum_{i = 0}^{d}W^{i}[h]$.
From Plancherel's Theorem we have that $\sum_{i = 0}^{d}w_i = 1$. 
We will now consider how to set $w_{0}, w_{1}, \dots, w_d$ to minimize the expression
\begin{equation*}
	\rho_{\alpha, \beta} = \frac{\log((1 + \sum_{i = 0}^{d}\alpha^{i} w_i)/2)}{\log((1 + \sum_{i = 0}^{d}\beta^{i} w_i)/2)}. 
\end{equation*}
An optimal solution $w^{*}_{0}, \dots, w^{*}_d$ for this problem will yield an optimal solution to the original problem,
provided there actually exists a Boolean-valued function satisfying the weight assignment.
We will show that the assignment $w^{*}_{1} = 1$ and $w^{*}_i = 0$ for $i \neq 1$ minimizes $\rho_{\alpha, \beta}$.
The distribution $\LSH$ therefore only assigns positive probability to functions $h$ that have all their Fourier weight concentrated at degree $1$. 
It turns out that a Boolean function satisfies this weight assignment if and only if it is a dictator function. 
Lemma \ref{lem:dictator} is well-known and is the answer to exercise 1.19 in~\cite{odonnell2014analysis}. 
We include the proof for completeness.
\begin{lemma} \label{lem:dictator}
	Let $f \colon \{-1,1\}^d \to \{-1, 1\}$ and suppose that $W^{1}[f] = 1$, then $f(x) = \pm x_i$.
\end{lemma}
\begin{proof}
	From Parseval's Theorem we know that $\sum_i W^{i}[f] = 1$ and it follows that $\hat{f}(S) = 0$ for $|S| \neq 1$. 
	The function $f$ can therefore be written in the form $f(x) = \sum_{i = 1}^{d}\hat{f}_i x_i$ where $\hat{f}_i = \hat{f}(S)$ for $S = \{i\}$.
	By the condition $W^{1}[f] = 1$ there exists $j \in [d]$ such that $\hat{f}_j \neq 0$. 
	Fix the $d-1$ components $x_{i \neq j}$ of $x$ and note that since $f$ maps to $\{-1, 1\}$ the sum $f(x) = \hat{f}_j x_j + \sum_{i \neq j}\hat{f}_i x_i$ must satisfy $f(x) = \pm1$ when $x_j = \pm1$.
	For $\hat{f}_j \neq 0$ this is only possible when $\hat{f}_j = \pm 1$ which implies that $\hat{f}_{i} = 0$ for $i \neq j$.
	It follows that $f$ must be one of the $2d$ functions of the form $f(x) = \pm x_i$.
\end{proof}
\subsection{Optimal Fourier weight at degree zero}
We begin by arguing that we can restrict our attention to showing that dictator functions are optimal in the case where $0 < \beta < \alpha < 1$.
If $\alpha = 1$ then for $w_{1} = 1$ we have that $\rho = 0$ which is the best we can hope for (but this could also be achieved by other weight assignments, hence the statement of the main theorem is for $\alpha < 1$.).
For $\beta = 0$ the following Lemma showing that $w_{0}^{*} = 0$ combined with the fact that for this setting we maximize $p_\alpha$ by setting $w_{1} = 1$ shows that the dictator functions are optimal.
We will now show that an optimal solution has no Fourier weight at degree zero.
\begin{lemma} \label{lem:w_0}
	$w_{0}^{*} = 0$.
\end{lemma}
\begin{proof}
	If $w_{0} = 1$ we have $\rho = 1$ and it is clear that $\rho < 1$ if we set $w_{1} = 1$. 
	Suppose that $0 < w_{0}^{*} < 1$. 
	We will show that in this case we can move some weight from $w_{0}$ to $w_{1}$ and decrease the value of $\rho$.
	For a given weight assignment define $s(\alpha) = \sum_{i} \alpha^i w_i$ and write $w_{1}$ as $w_{1} = 1 - \sum_{j \neq i}w_j$. 
	The partial derivative of $\rho = \log((1 + s(\alpha))/2)/\log((1 + s(\beta))/2)$ with respect to $w_{0}$ is given by
	\begin{equation*}
		\frac{\partial \rho}{\partial w_{0}} = \frac{\frac{\partial s(\alpha) / \partial w_{0}}{1 + s(\alpha)}\log \frac{1 + s(\beta)}{2} 
		- \frac{\partial s(\beta) / \partial w_{0}}{1 + s(\beta)}\log \frac{1 + s(\alpha)}{2}}
		{\log^2 \frac{1 + s(\beta)}{2}}.
	\end{equation*}
	By rearranging and using that $\partial s(\alpha) / \partial w_{0} = 1 - \alpha$ we find that $\frac{\partial \rho}{\partial w_{0}} > 0$ is equivalent to  
	\begin{equation*} 
		\frac{1 + s(\beta)}{1 - \beta} \log \frac{1 + s(\beta)}{2} > \frac{1 + s(\alpha)}{1 - \alpha} \log \frac{1 + s(\alpha)}{2}.
	\end{equation*}
	It suffices to show that the function $g(x) = \frac{1 + s(x)}{1 - x} \log \frac{1 + s(x)}{2}$ is decreasing for $0 < x < 1$. 
	\begin{equation} \label{eq:condition}
		\frac{\partial g}{\partial x} = \frac{s'(x)(1 - x) + (1 + s(x))}{(1-x)^2} \log \frac{1 + s(x)}{2} + \frac{s'(x)}{1 - x}.
	\end{equation}
	Rewriting, this is equivalent to showing that 
	\begin{equation*}
		(s'(x)(1-x) + 1 + s(x)) \log \frac{1 + s(x)}{2} + (1-x)s'(x) < 0.
	\end{equation*}
	By the assumption that $0 < w_0 < 1$ we have that $0 < s(x) < 1$ and using Lemma \ref{lem:log_inequality} we get that $\log((1 + s(x))/2) < (s(x) - 1)/2$.
	The condition in equation \eqref{eq:condition} then simplifies to showing that $s'(x)(1-x) + s(x) \leq 1$.
	The function $s(x) = \sum_i w_i x^i$ is a weighted sum of simple monomials where the weights sum to one.
	It therefore suffices to show that the inequality holds for every monomial $s_k(x) = x^k$ where $k = \{0, 1, \dots, d \}$.
	For $k = 0$ and $k = 1$ we have $s_{k}'(x)(1-x) + s_{k}(x) = 1$ satisfying the desired inequality.
	For $k \geq 2$ we have $s_{k}'(x)(1-x) + s_{k}(x) = kx^{k-1} + (k-1)x^k$.
	We see that $s_{k}(0) = 0$ and $s_{k}(1) = 1$ and by inspecting the derivative of $s_{k}(x)$ we see that it is increasing for $x \in (0,1)$. 
	It follows that the inequality is satisfied, completing the proof.
\end{proof}
\subsection{A continuous optimization problem}
In order to simplify the problem of minimizing $\rho$ we will optimize over a larger space.
In particular we will let $W$ denote a collection of pairs $(w, \kappa)$ such that $\sum_{w \in W} w = 1$ where we restrict $\kappa \in \mathbb{R}$ to satisfy $\kappa \geq 1$.
We define $s(x) = \sum_{(w, \kappa) \in W} w x^\kappa$ and we will now attempt to specify the function $s$ that minimizes  
\begin{equation*}
	\rho_{\alpha, \beta} = \frac{\log \frac{1 + s(\alpha)}{2}}{\log \frac{1 + s(\beta)}{2}}
\end{equation*}
subject to the constraint that $s(\beta) = b \leq \beta$ is fixed. 
The constraint that $s(\beta) \leq \beta$ follows from the restrictions on $s$.
We can therefore write $b = \beta^\gamma$ for some $\gamma \geq 1$.
For fixed $s(\beta)$ it is clear that we minimize $\rho$ by maximizing $s(\alpha)$. 
\begin{lemma} \label{lem:gamma}
	For fixed $s(\beta) = \beta^\gamma$ we maximize $s(\alpha)$ by setting $s(x) = x^\gamma$.
\end{lemma}
\begin{proof}
Let $w$ denote the weight on the exponent $\gamma$ in the specification $W$ of $s$. 
We will prove that if $w < 1$ then we can increase $s(\alpha)$ by rearranging the weights of $s$ to put more weight onto $(w, \gamma)$.
Note that if $w < 1$ and we have a valid configuration of weights (in the sense that $s(\beta) = \beta^\gamma$) there must exist exponents $\gamma_{0} < \gamma < \gamma_{1}$ such that there is positive weight on $\gamma_{0}$ and $\gamma_{1}$. 
If all the remaining weight was concentrated to either side of $\gamma$ the condition $s(\beta) = \beta^\gamma$ would be violated.
We will now move $\varepsilon_{0}$ weight from $w_{0}$ to $w$ and $\varepsilon_{1}$ weight from $w_{1}$ to $w$ where we set $\varepsilon_{0}, \varepsilon_{1}$ to ensure that $s(\beta) = \beta^\gamma$ after the move.
It turns out that this condition is satisfied for the following ratio
\begin{equation*}
	\varphi(\beta) = \varepsilon_{1} / \varepsilon_{0} = \frac{\beta^{\gamma_{0}} - \beta^\gamma}{\beta^\gamma - \beta^{\gamma_{1}}} > 0.
\end{equation*}
The change in $s(\alpha)$ due to the rearrangement of weights can be shown to be positive if $\varphi(\alpha) < \varphi(\beta)$.
Therefore, it suffices to show that $\varphi(x)$ is decreasing for $0 < x < 1$ when $\gamma_{0} < \gamma < \gamma_{1}$.
To simplify further, we define $\lambda_{0} = \gamma_{0} - \gamma$ and $\lambda_{1} = \gamma_{1} - \gamma$ which satisfy $\lambda_{0} < 0 < \lambda_{1}$.
Rewriting $\varphi(x) = -(1 - x^\lambda_{0})/(1 - x^\lambda_{1})$ and differentiating we get
\begin{align*}
	\frac{\partial \varphi}{\partial x} &= \lambda_{0} x^{\lambda_{0}}(1 - x^{\lambda_{1}} - \lambda_{1} x^{\lambda_{1}}(1-x^{\lambda_{0}}) < 0 \\
										&\iff \frac{\lambda_{0} x^{\lambda_{0}}}{1 - x^{\lambda_{0}}} > \frac{\lambda_{1} x^{\lambda_{1}}}{1 - x^{\lambda_{1}}}. 
\end{align*}
It suffices to show that $\psi(x) = \frac{xa^x}{1- a^x}$ is decreasing in $x$ for $a \in (0,1)$.
We have that $\psi'(x) = a^{x}(1 - a^{x}) + a^x \log a^x$. 
Define $z = a^x$ and note that $z > 0$ and $z \neq 1$. 
We have that $z(1-z) + z \log z < 0 \iff (1 - z) + \log z < 0$ and by Lemma \ref{lem:log_inequality} we see that $\log z < z - 1$, completing the proof. 
\end{proof}
\subsection{Univariate analysis}
According to Lemma \ref{lem:gamma} we can now restrict our attention to the problem of finding $\gamma \geq 1$ that minimizes the function
\begin{equation*}
	\rho(\gamma) = \frac{\log \frac{1 + \alpha^\gamma}{2}}{\log \frac{1 + \beta^\gamma}{2}} \,.
\end{equation*}
We will show the derivative of $\rho$ is positive, implying that it is minimized when $\gamma = 1$.
\begin{lemma} \label{lem:derivative_gamma}
	$\rho'(\gamma) > 0$.
\end{lemma}
\begin{proof}
From inspecting the derivative of $\rho$ with respect to $\gamma$ we see that
\begin{align*}
&\frac{\partial \rho}{\partial \gamma} > 0 \\
		&\iff \frac{\alpha^{\gamma} \log \alpha}{1 + \alpha^{\gamma}} \log \frac{1 + \beta^\gamma}{2} - \frac{\beta^{\gamma} \log \beta}{1 + \beta^{\gamma}} \log \frac{1 + \alpha^\gamma}{2} > 0 \\ 
		&\iff \frac{1 + \beta^{\gamma}}{\beta^{\gamma} \log \beta} \log \frac{1 + \beta^\gamma}{2} >  \frac{1 + \alpha^{\gamma}}{\alpha^{\gamma} \log \alpha} \log \frac{1 + \alpha^\gamma}{2}.  
\end{align*}
Therefore it suffices to show that the function $g(x) = \frac{1 + x^{\gamma}}{x^{\gamma} \log x} \log \frac{1 + x^\gamma}{2}$ is decreasing for $0 < x < 1$ and $\gamma \geq 1$.
From inspecting $g'(x)$ we see that the condition that $g'(x) < 0$ is equivalent to
\begin{equation*}
-(1 + x^\gamma + \gamma \log x) \log \frac{1 + x^\gamma}{2} + \gamma x^{\gamma} \log x < 0 
\end{equation*}
If $1 + x^\gamma + \gamma \log x \geq 0$ then the condition is satisfied and we are done.
Otherwise we can use the fact that $-(1 + x^\gamma + \gamma \log x) > 0$ together with Lemma \ref{lem:log_inequality} to produce following derivation: 
\begin{align*}
&-(1 + x^\gamma + \gamma \log x) \log \frac{1 + x^\gamma}{2} + \gamma x^{\gamma} \log x \\
&\qquad < -(1 + x^\gamma + \log x^\gamma)\frac{x^\gamma - 1}{2} + x^\gamma \log x^\gamma \\
&\qquad = 1 - x^{2\gamma} + (1 + x^\gamma) \log x^\gamma 
\end{align*}
Reapplying Lemma \ref{lem:log_inequality} we see that $(1 + x^\gamma) \log x^\gamma < (1 + x^\gamma)(x^\gamma - 1) = -(1 - x^{2\gamma})$ completing the proof. 
\end{proof}
\subsection{Stating the result}
We will now summarize how the results from the previous subsections yield Theorem \ref{thm:bool_main}. 
To find the the distribution over Boolean functions minimizing $\rho$ we first considered the optimal weight assignment in the expression $s(x) = \sum_i w_i \alpha^i$ subject to the constraint that $\sum_i w_i = 1$. 
Finding an optimal assignment does not guarantee that we have solved the problem, because there may not exist a Boolean function with a given weight assignment, 
but if one or more Boolean functions that satisfy the optimal assignment exists we will have solved the problem. 
In Lemma \ref{lem:w_0} we showed that an optimal solution $w_0^*, w_1^*, \dots w_d^*$ must have $w_0^* = 0$.
Therefore the optimal solution can only have non-zero weight on exponents $k \geq 1$.
Next, in Lemma \ref{lem:gamma}, we argued that if we allow continuous exponents $k \in \mathbb{R}$ with $k \geq 1$ in $s(x)$ then the problem of minimizing $\rho$ becomes the problem of selecting $\gamma \geq 1$ where $s(x) = x^\gamma$.
Lemma \ref{lem:derivative_gamma} showed that $\rho(\gamma)$ is increasing, so to minimize $\rho$ we want to set $\gamma = 1$.
The conclusion from these optimization problems is that we minimize $\rho$ by setting $w_1^* = 1$.
Finally Lemma \ref{lem:dictator} shows that the subset of the Boolean functions with $w_1 = 1$ is exactly the set of dictator functions $f(x) = \pm x_i$. 
Together with the fact that $w_1^* = 1$ is a unique minimum of $\rho$ in the weight assignment problem we get Theorem \ref{thm:bool_main}
\section{Open problems}
\paragraph{Orthogonal search.}
It appears that the same techniques can be used to show that pairs of functions of the form $f(x) = x_i x_j$, $g(y) = -x_i x_j$ minimize the function
\begin{equation*}
	\frac{\log(1/\min(p_{\alpha},p_{-\alpha}))}{\log(1/\max(p_{\beta}, p_{-\beta}))}.
\end{equation*}

\paragraph{Extension to negative correlation.}
It seems likely that the dictator functions or bit-sampling minimizes $\rho$ for the entire interval $-1 \leq \beta < \alpha \leq 1$. 
Unfortunately the current proof breaks down in places.

\paragraph{General hash functions.}
Showing tight bounds for hash function with an arbitrary range is an interesting open problem.
For orthogonal search this is an open problem even in the case of $\rho_{\alpha, 0}$.
For more information see the symmetric Gaussian problem in \cite{odonnell2012open}. 

Investigating what the implications of the results in this paper for functions with an arbitrary range through the use of $1$-bit hashing is an interesting problem.

\paragraph{Simpler proof.}
The current proof seems needlessly complicated and inelegant. Perhaps some properties of the ratio of logarithms could be established and a simple proof would follow.

\paragraph{Exact relation to standard LSH.}
Find out when an LSH for Hamming space in the standard formulation violates the lower bound in this paper.

\paragraph{Asymmetric LSH and set similarity search.}
Investigate whether the same techniques can be used to show optimal 1-bit schemes for a boolean anti-LSH and for set similarity search.

\part{Pseudorandomness}
\chapter{Generating $k$-independent random variables in constant time} \label{ch:rng}
\sectionquote{The old that is strong does not wither}
\noindent The generation of pseudorandom elements over finite fields is fundamental to the time, space and randomness complexity of randomized algorithms and data structures. 
We consider the problem of generating $k$-independent random values over a finite field $\GF$ in a word RAM model equipped with constant time addition and multiplication in $\GF$, and present the first nontrivial construction of a generator that outputs each value in \emph{constant time}, not dependent on~$k$.
Our generator has period length $|\GF|\poly \log k$ and uses $k \poly(\log k) \log |\GF|$ bits of space, which is optimal up to a $\poly \log k$ factor.
We are able to bypass Siegel's lower bound on the time-space tradeoff for \mbox{$k$-independent} functions by a restriction to sequential evaluation.  
\section{Introduction}
Pseudorandom generators transform a short random seed into a longer output sequence. 
The output sequence has the property that it is indistinguishable from a truly random sequence by algorithms with limited computational resources. 
Pseudorandom generators can be classified according to the algorithms (distinguishers) that they are able to fool.
An algorithm from a class of algorithms that is fooled by a generator can have its randomness replaced by the output of the generator,
while maintaining the performance guarantees from the analysis based on the assumption of full randomness.
When truly random bits are costly to generate or supplying them in advance requires too much space, a pseudorandom generator can reduce the time, space and randomness complexity of an algorithm.

This paper presents an explicit construction of a pseudorandom generator that outputs a $k$-independent sequence of values in \emph{constant time} per value, not dependent on $k$, on a word RAM~\cite{hagerup1998}. 
The generator works over an arbitrary finite field that allows constant time addition and multiplication over $\GF$ on the word RAM.

Previously, the most efficient methods for generating $k$-independent sequences were either based on multipoint evaluation of degree $k-1$ polynomials, or on direct evaluation of constant time hash functions.
Multipoint evaluation has a time complexity of $O(\log^{2} k \log \log k)$ field operations per value while hash functions with constant evaluation time use excessive space for non-constant $k$ by Siegel's lower bound~\cite{siegel2004}.
We are able to get the best of both worlds: constant time generation and near-optimal seed length and space usage. 

\paragraph{Significance.}
In the analysis of randomized algorithms and in the hashing literature in particular, $k$-independence has been the dominant framework for limited randomness. 
Sums of $k$-independent variables have their $j$th moment identical to fully random variables for $j \leq k$ which preserves many properties of full randomness.
For output length $n$, $\Theta(\log n)$-independence yields Chernoff-Hoeffding bounds~\cite{schmidt1995} and random graph properties~\cite{alon2008},
while $\Theta(\poly \log n)$-independence suffices to fool $AC^{0}$ circuits~\cite{braverman2010}.

Our generator is particularly well suited for randomized algorithms with time complexity $O(n)$ that use a sequence of $k$-independent variables of length $n$, for non-constant $k$.
For such algorithms, the generation of $k$-independent variables in constant time by evaluating a hash function over its domain requires space $O(n^{\epsilon})$ for some constant $\epsilon > 0$. 
In contrast, our generator uses space $O(k \poly \log k)$ to support constant time generation. 
Algorithms for randomized load balancing such as the simple process of randomly throwing $n$ balls into $n$ bins fit the above description and presents an application of our generator.
Using the bounds by Schmidt et al.~\mbox{\cite[Theorem 2]{schmidt1995}} it is easy to show that $\Theta(\log n / \log \log n)$-independence suffices to obtain a maximal load of any bin of $O(\log n / \log \log n)$ with high probability.
This guarantee on the maximal load is asymptotically the same as under full randomness. 
Using our generator, we can allocate each ball in constant time using space $O(\log n \poly \log \log n)$ compared to the lower bound of $O(n^{\epsilon})$ of hashing-based approaches to generating $k$-independence. 
In Section \ref{sec:loadbalancing} we show how our generator improves upon existing solutions to a dynamic load balancing problem.

The generation of pseudorandomness for Monte Carlo experiments presents another application.
Limited independence between Monte Carlo experiments can be shown to yield Chernoff-like bounds on the deviation of an estimator from its expected value. 
Consider a randomized algorithm~$\A(Y)$ that takes $m$ random elements from $\GF$ encoded as a string $Y$ and returns a value in the interval $[0,1]$.    
Let $\mu_{\A} > 0$ denote the expectation of the value returned by~$\A(Y)$ under the assumption that $Y$ encodes a truly random input.
Define the estimator
\begin{equation*}
\hat{\mu}_{\A} = \frac{1}{t}\sum_{i=1}^{t}\A(Y_{i}).
\end{equation*}
Due to a result by Schmidt et al. \cite[Theorem 5]{schmidt1995}, for every choice of constants $\epsilon, \alpha > 0$,
it suffices that $Y_{1}, \dots, Y_{t}$ encodes a sequence of $\Theta(m \log t)$-independent variables over $\GF$ to yield the following high probability bound on the deviation of $\hat{\mu}_{\A}$ from $\mu_{\A}$.
\begin{equation*}
Pr[|\hat{\mu}_{\A} - \mu_{\A}| \geq \epsilon\mu_{\A}] \leq O(t^{-\alpha}). 
\end{equation*}
We hope that our generator can be a useful tool to replace heuristic methods for generating pseudorandomness in applications where theoretical guarantees are important. 
In order to demonstrate the practicality of our techniques, we present experimental results on a variant of our generator in Section \ref{sec:experiments}. 
Our experiments show that $k$-independent values can be generated nearly as fast as output from heuristic pseudorandom generators, even for large $k$.

\paragraph{Methods.}
Our construction is a surprisingly simple combination of bipartite unique neigbor expanders with multipoint polynomial evaluation.
The basic, probabilistic construction of our generator proceeds in two steps: 
First we use multipoint evaluation to fill a table with \mbox{$\Theta(k)$-independent} values from a finite field, using an average of $\poly \log k$ operations per table entry.
Next we apply a bipartite unique neighbor expander with constant outdegree and with right side nodes corresponding to entries in the table and a left side that is $\poly \log k$ times larger than the right side.
For each node in the left side of the expander we generate a $k$-independent value by returning the sum of its neighboring table entries.
Our main result stated in Theorem 1 uses the same idea, but instead of relying on a single randomly constructed expander graph, 
we employ a cascade of explicit constant degree expanders and show that this is sufficient for constant time generation.   

\paragraph{Relation to the literature.}
Though the necessary ingredients have been known for around 10 years, we believe that a constant time generator has evaded discovery by residing in a blind spot between the fields of hashing and pseudorandom generators. 
The construction of constant time $k$-independent \emph{hash functions} has proven to be a difficult task, and a fundamental result by Siegel~\cite{siegel2004} shows a time-space tradeoff that require hashing-based generators with sequence length $n$ to use $O(n^{\epsilon})$ space for some constant $\epsilon > 0$. 
On the other hand, from the point of view of pseudorandom generators, a generator of $k$-independent variables, for super-constant~$k$, can not be used as an efficient method of derandomization: 
A lower bound by Chor et al.~\cite{chor1985} shows that the sample space of such generators must be superpolynomial in their output length.
Specifically, a $k$-independent generator that outputs $n$ bits must have a seed length of $\Omega(k \log n)$ bits.
Consequently, research shifted towards generators that produce other types of outputs such as biased sequences or almost $k$-independent variables \cite{alon1992, naor1993, goldreich2010}. 

It is relevant to ask whether there already exist constructions of constant time pseudorandom generators on the word RAM that can be used instead of generators that output $k$-independent variables. 
For example, Nisan's pseudorandom generator~\cite{nisan1992} uses constant time to generate a pseudorandom word and has remarkably strong properties: 
Every algorithm running in $\textsc{space}(s)$ that uses $n$ random words can have its random input replaced by the output of a constant time generator with seed length $O(s \log n)$. 
The probability that the outcome of the algorithm differs when using pseudorandomness as opposed to statistical randomness is exponentially decreasing in the seed length. 

In spite of this strong result, there are many natural applications where the restrictions on Nisan's model means that we cannot use his generator directly to replace the use of a $k$-generator. 
An example is the analysis that uses a union bound over all subsets of $k$ words of a randomly generated structure described by $n$ words.
Algorithms shown to be derandomized by Nisan's generator are restricted to one-way access to the output of the generator. 
Therefore the output of Nisan's generator can not be used to derandomize an algorithm that tests for the events of the union bound without using excessive space.
In this case, \mbox{$k$-independence} can directly replace the use of full randomness without changing the analysis.
\subsection{Our contribution}
We present three improved constructions of \emph{$k$-generators}, formally defined in Section~\ref{rng:sec:preliminaries}, that are able to generate a sequence of $k$-independent values over a finite field $\GF$.
Our results are stated in a word RAM model equipped with constant time addition and multiplication in $\GF$.     
Our main result is a fully explicit generator:
\begin{theorem}\label{thm:explicit}
For every finite field $\GF$ with constant time arithmetic there exists a data structure that for every choice of $k$ is an explicit constant time $k$-generator with
range $\GF$, period $|\GF|\poly \log k$, and seed length, space usage and initialization time $k \poly \log k$.
\end{theorem}
We further investigate how the space usage and seed length may be reduced by employing a probabilistic construction that has a certain probability of error:
\begin{theorem} \label{thm:existence}
For every finite field $\GF$ with constant time arithmetic and every choice of constants $\varepsilon$, $\delta > 0$ there exists a data structure 
that for every choice of $k$ is a constant time $k$-generator with failure probability $\delta$, range~$\GF$, period $|\GF|$, seed length $O(k)$, 
space usage $O(k \log^{2+\varepsilon}k)$, and initialization time $O(k \poly \log k)$.
\end{theorem}
Finally, we improve existing $k$-generators with optimal space complexity:
\begin{theorem} \label{thm:fastmultipoint}
For every finite field $\GF$ that supports computing the discrete Fourier transform of length $k$ in $O(k \log k)$ operations, 
there exists a data structure that, for every choice of $k$ and given a primitive element $\omega$, 
is an explicit $O(\log k)$~time $k$-generator with range $\GF$, period $|\GF|$, seed length $k$, space usage $O(k)$, and initialization time $O(k \log k)$.
\end{theorem}

Table~\ref{rng:tab:results} summarizes our results along with previous methods of generating sequences of $k$-independent values over $\GF$. 
The length of the output is at least~$|\GF|$ for all the different methods.

\begin{table}[htpb]
  \centering
  \small
    \begin{tabular}{llll}
    \toprule
	\textbf{Construction}                     & \textbf{Time}            & \textbf{Space}                      & \textbf{Comment}  \\ \midrule
    Polynomials \cite{joffe1974,wegman1981}   & $O(k)$                   & $O(k)$                              & \\
	Multipoint \cite{gathen2013}              & $O(\log^2 k \log\log k)$ & $O(k \log k)$                   	   & \\
	Multipoint \cite{bostan2005}              & $O(\log k \log\log k)$   & $O(k)$                          	   & Requires $\omega$. \\ 
	Siegel \cite{siegel2004}                  & $O(1)$                   & $O(|\GF|^{\varepsilon})$            & Probabilistic. \\
	Theorem \ref{thm:explicit}                & $O(1)$                   & $k \poly \log k$                    & Explicit.  \\
	Theorem \ref{thm:existence}               & $O(1)$                   & $O(k \log^{2 + \varepsilon}k)$      & Probabilistic. \\
	Theorem \ref{thm:fastmultipoint}          & $O(\log k)$              & $O(k)$                              & Requires $\omega_{k}$, FFT. \\
   \bottomrule
    \end{tabular}
\caption{
Overview of generators that produce a $k$-independent sequence over a finite field $\GF$. 
We use $\varepsilon$ to denote an arbitrary positive constant and $\omega$ and $\omega_{k}$ to denote, respectively, a primitive element and a $k$-th root of unity of $\GF$.
The unit for space is the number of elements of $\GF$ that need to be stored, i.e., a factor $\log_2 |\GF|$ from the number of bits. 
Probabilistic constructions rely on random generation of objects for which no explicit construction is known, and may fail with some probability.
}
\label{rng:tab:results}
\end{table}
\paragraph{Overview of paper}
In Section \ref{rng:sec:preliminaries} we define \mbox{$k$-generators} and related concepts and review results that lead up to our main results.
Section \ref{sec:explicit} presents the details of our explicit construction of constant time generators. 
In Section \ref{sec:probabilistic} we apply the same techniques with a probabilistic expander construction to obtain generators with improved space and randomness complexity. 
Section \ref{sec:faster} presents an algorithm for evaluating a polynomial over all elements of $\GF$ that improves existing generators with optimal space.
Section \ref{sec:wordRAM} shows how arithmetic over $\GF_{p}$ can be implemented in constant time on a standard word RAM with integer multiplication and also reviews algorithms and the state of hardware support for $\GF_{2^{w}}$.
Section \ref{sec:loadbalancing} applies our generator to improve the time-space tradeoff of previous solutions to a load balancing problem.
Section \ref{sec:experiments} presents experimental results on the generation time of different $k$-generators for a range of values of $k$. 
\section{Preliminaries} \label{rng:sec:preliminaries}
We begin by defining two fundamental concepts:
\begin{definition}
A sequence $(X_{1}, X_{2}, \dots, X_{n})$ of $n$ random variables with finite range $R$ is an \emph{$(n,k)$-sequence} if the variables at every set of $k$ positions in the sequence are independent and uniformly distributed over $R$. 
\end{definition}
\begin{definition}
A family of functions ${\mathcal{F} \subseteq \{f \mid f \colon U \to R \}}$ is \emph{$k$-independent} if for every set of $k$ distinct inputs $x_{1}, x_{2}, \dots, x_{k}$ 
it holds that $f(x_{1}), f(x_{2}), \dots, f(x_{k})$ are independent and uniformly distributed over $R$ when $f$ is selected uniformly at random from $\mathcal{F}$. 
We say that a function $f$ selected uniformly at random from $\mathcal{F}$ is a \emph{$k$-independent function}.
\end{definition}
We now give a formal definition of the generator data structure. 
\begin{definition}
A \emph{$k$-generator} with range $R$, period $n$ and failure probability $\delta$ is a data structure with the following properties:
\begin{itemize}
\item It supports an initialization operation that takes a random seed $s$ as input.
\item After initialization it supports an \texttt{emit()} operation that returns a value from $R$. 
\item There exists a set $B$ such that $\Pr[s \in B] \leq \delta$ and conditioned on $s \not\in B$ the sequence $(X_{1}, X_{2}, \dots, X_{n})$ of values returned by \texttt{emit()} is an $(n,k)$-sequence. 
\end{itemize}
A $k$-generator is \emph{explicit} if the initialization and emit operation has time complexity $\poly k$ and the probability of failure is zero. 
We refer to a $k$-generator as a constant time $k$-generator if the \texttt{emit()} operation has time complexity $O(1)$, not dependent on $k$.
\end{definition}

A $k$-generator differs from a data structure for representing a $k$-independent hash function by only allowing sequential access to the underlying $(n,k)$-sequence. 
It is this restriction on generators that allows us to obtain a better time-space tradeoff for the problem of generating $k$-independent variables than is possible by using a $k$-independent hash function directly as a generator.  
We are interested in the following parameters of $k$-generators: seed length, period, probability of failure, space needed by the data structure, the time complexity of the initialization operation and the time complexity of a single \texttt{emit()} operation.

\paragraph{Model of computation.}
Our results are stated in the word RAM model of computation with word length ${w = \Theta(\log|\GF|)}$ bits. 
In addition to the standard bit manipulation and integer arithmetic instructions, we also assume the ability to perform arithmetic operations $(+, -, \times)$ over $\GF$ in constant time. 
In the context of our results that use abelian groups $(A, +)$ we assume that an element of $A$ can be stored in a constant number of words and that addition can be performed in constant time. 

Let $\GF_{q}$ denote a field of cardinality $q = p^{z}$ for $p$ prime and $z$ a positive integer.
Constant time arithmetic in $\GF_{p}$ is supported on a standard word RAM with integer multiplication \cite{granlund1994}. 
Section \ref{sec:wordRAM} presents additional details about the algorithms required to implement finite field arithmetic over $\GF_p$ and $\GF_{2^{w}}$ and how they relate to a standard word RAM with integer multiplication.
\subsection{$k$-independent functions from the literature} \label{sec:hashing}
We now review the literature on $k$-independent functions and how they can be used to construct $k$-generators.
We distinguish between a $k$-independent function $f : U \to R$ and a $k$-independent hash function by letting the latter refer to a data structure 
that after initialization supports random access to the $(n,k)$-sequence defined by evaluating $f$ over~$U$.  
There exists an extensive literature that focuses on how to construct $k$-independent hash functions that offer a favorable tradeoff between representation space and evaluation time \cite{dietzfelbinger2012}. 
We note that a family of $k$-independent hash functions can be used to construct a $k$-generator by setting the seed to a random function in the family.

\paragraph{Constant time $k$-independent hash functions.}
A fundamental cell probe lower bound by Siegel \cite{siegel2004} shows that a data structure to support constant time evaluation of $f$ on every input in $U$ 
cannot use less than $\Omega(|U|^{\epsilon})$ space for some constant $\epsilon > 0$. 
This bound holds even for amortized constant evaluation time over functions in the family and elements in the domain.
From Siegel's lower bound, it is clear that we cannot use $k$-independent hash functions directly to obtain a constant time $k$-generator that uses only $O(k \poly \log k)$ words of space.

Known constructions of $k$-independent hash functions with constant evaluation time are based on expander graphs.
Siegel \cite{siegel2004} gave a probabilistic construction of a family of \mbox{$k$-independent} hash functions in the word RAM model based on an iterated product of bipartite expander graphs. 
Thorup \cite{thorup2013} showed that a simple tabulation hash function with high probability yields the type of expander graphs required by Siegel's construction.
Unfortunately only randomized constructions of the expanders required by these hash functions are known, introducing a positive probability of error in \mbox{$k$-generators} based on them.

\paragraph{Polynomials.}
Here we briefly review the classic construction of $k$-independent functions based on polynomials over finite fields.  
\begin{lemma}[Joffe \cite{joffe1974}, Carter and Wegman \cite{wegman1981}] \label{lem:kpoly}
For every choice of finite field $\GF$ and every $k \leq |\GF|$, let $\mathcal{H}_{k} \subset \GF[X]$ be the family of polynomials of degree at most $k-1$ over $\GF$.
${\mathcal{H}_{k} \subset \{ f \mid f \colon \GF \to \GF \}}$ is a family of $k$-independent functions.
\end{lemma}

An advantage of using families of polynomials as hash functions is that they use near optimal randomness, allow any choice of $k \leq |\GF|$, and have no probability of failure. 
It can also be noted that in the case where $k = O(\log |\GF|)$ and we are restricted to linear space $O(k)$, 
polynomial hash functions evaluated using Horner's scheme are optimal \mbox{$k$-independent} hash functions \cite{larsen2012, siegel2004}.   

Using slightly more space and for sufficiently large $k$, a data structure by Kedlaya and Umans \cite{kedlaya2008} supports evaluation of a polynomial of degree $k$ over $\GF$.
The space usage and preprocessing time of their data structure is $k^{1 + \epsilon}\log^{1 + o(1)}|\GF|$ for constant $\epsilon > 0$.
After preprocessing a polynomial $f$, the data structure can evaluate $f$ in an arbitrary point of $\GF$ using time $\poly(\log k)\log^{1 + o(1)}|\GF|$.         

\paragraph{Multipoint evaluation.}
Using algorithms for multipoint evaluation of polynomials we are able to obtain a \mbox{$k$-generator} with $\poly \log k$ generation time and space usage that is linear in $k$. 
Multipoint evaluation of a polynomial~${f \in \GF[X]}$ of degree at most $k-1$ in $k$ arbitrary points of~$\GF$ has a time complexity of $O(k \log^{2} k \log \log k)$ in the word RAM model that supports field operations~\mbox{\cite[Corollary 10.8]{gathen2013}}. 
Bostan and Schost \cite{bostan2005} mention an algorithm for multipoint evaluation of $f$ over a geometric progression of $k$ elements with running time $O(k \log k \log \log k)$. 
In order to use this method to construct a $k$-generator with period $|\GF|$ 
it is necessary to know a primitive element $\omega$ of $\GF_{q}$ so we can perform multipoint evaluation over $\GF^{*} = \{\omega^{0}, \omega^{1}, \dots, \omega^{q-2} \}$. 
Given the prime factorization of $q - 1$ there exists a Las Vegas algorithm for finding $\omega$ with expected running time $O(\log^{4} q)$~\mbox{\cite[Chapter 11]{shoup2009}}. 
In the following lemma we summarize the properties of $k$-generators based on multipoint evaluation of polynomials over finite fields.         

\begin{lemma}[{Gathen and Gerhard \cite[Corollary 10.8]{gathen2013}, Bostan and Schost \cite{bostan2005}}] \label{lem:multipoint}
For every finite field $\GF$ there exists for every $k \leq |\GF|$ and bijection $\pi : [|\GF|] \to \GF$ an explicit $k$-generator with period $|\GF|$ and seed length $k$.
The space required by the generator and the initialization and generation time depends on the choice of $\pi$ and multipoint evaluation algorithm.
\begin{itemize}
\item For arbitrary choice of $\pi$ there exists a $k$-generator with generation time $O(\log^{2} k \log \log k)$, intialization time $O(k \log^{2} k \log \log k)$ and space usage $O(k \log k)$. 
\item Given a primitive element $\omega$ of $\GF$ and a bijection $\pi(i) = \omega^{i}$ there exists a generator with generation time $O(\log k \log \log k)$, initialization time $O(k \log k \log \log k)$ and space usage $O(k)$.
\end{itemize}
\end{lemma}

\paragraph{Space lower bounds.}
Since randomness can be viewed as a resource like time and space, we are naturally interested in generators that can output long $k$-independent sequences using as few random bits as possible. 
Families of \mbox{$k$-independent functions} $f : U \rightarrow R$ with $U = R$ and $k \leq |U|$ will trivially have to use at least $k \log |U|$ random bits --- a bound matched by polynomial hash functions. 
We are often interested in generators with $|U| \gg |R|$, for example if we wish to use a generator for randomized load balancing in the heavily loaded case. 
A lower bound by Chor et al.~\cite{chor1985} shows that even in this case the minimal seed length required for $k$-independence is $\Omega(k \log |U|)$ for every $|R| \leq |U|$.
\subsection{Expander graphs}
All graphs in this paper are bipartite with $cm$ vertices on the left side, $m$ vertices on the right side and left outdegree $d$.
Graphs are specified by their edge function $\Gamma : [cm] \times [d] \to [m]$ where the notation $[n]$ is used to denote the set $\{0,1,\dots,n-1\}$.
Let $S$ be a subset of left side vertices. 
For convenience we use $\Gamma(S)$ to denote the neighbors of $S$. 
\begin{definition}
The bipartite graph $\Gamma : [cm] \times [d] \to [m]$ is \mbox{\emph{$(c,m,d,k)$-unique}} ($k$-unique) 
if for every $S \subseteq [cm]$ with $|S| \leq k$ there exists $y \in \Gamma(S)$ such that $y$ has a unique neighbor in $S$. 
An expander graph is \emph{explicit} if it has a deterministic description and $\Gamma$ is computable in time polynomial in $\log cm + \log d$. 
\end{definition}
The performance of our generator constructions are directly tied to the parameters of such expanders. 
In particular, we would like explicit expanders that simultanously have a low outdegree $d$, are highly unbalanced and are $k$-unique for $k$ as close to $m$ as possible.
A direct application of a result by Capalbo et al. \cite[Theorem 7.1]{capalbo2002} together with an equivalence relation between different types of expander graphs from Ta-Shma et al. \cite[Theorem 8.1]{tashma2007} yields explicit constructions of unbalanced unique neighbor expanders.\footnote{We state the results here without the restriction from \cite{capalbo2002} that $c$ and $m$ are powers of two. We do this to simplify notation and it only affects constant factors in our results.}  
\begin{lemma}[Capalbo et al. {\cite[Theorem 7.1]{capalbo2002}}] \label{lem:explicit}
For every choice of $c$ and $m$ there exists a $(c,m,d,k)$-unique expander with $d = \poly \log c$ and $k = \Omega(m/d)$. For constant $c$ the expander is explicit. 
\end{lemma}
We note the following simple technique for constructing a larger $k$-unique expander from a smaller $k$-unique expander.
\begin{lemma} \label{lem:stacking}
Let $\Gamma$ be a $(c,m,d,k)$-unique expander with $cm \times m$ adjacency matrix $M$.
For any positive integer $b$ define $\Gamma^{(b)}$ as the bipartite graph with block diagonal adjacency matrix $M^{(b)} = \diag(M, \dots, M)$ with $b$ blocks in the diagonal.
Then $\Gamma^{(b)}$ is a $(c, bm, d, k)$-unique expander.
\end{lemma}

\paragraph{From expanders to independence.}
By associating each right vertex in a $(c,m,d,k)$-unique expander with a position in a $(m,dk)$-sequence over an abelian group $(A,+)$, we can generate a $(cm,k)$-sequence over $A$.
This approach was pioneered by Siegel and has been used in different constructions of families of $k$-independent hash functions~\cite{siegel2004, thorup2013}.    
\begin{lemma}[Siegel {\cite[Lemma 2.6, Corollary 2.11]{siegel2004}}] \label{rng:lem:expanderhashing}
Let $\Gamma : [cm] \times [d] \to [m]$ be a $k$-unique expander and let $h : [m] \rightarrow A$ be a $dk$-independent function with range an abelian group. 
Let $g : [cm] \rightarrow A$ be defined as  
\begin{equation*}
g(x) = \sum_{y \in \Gamma(\{x\})}h(y).
\end{equation*}
Then $g$ is a $k$-independent function.
\end{lemma}
\section{Explicit constant time generators} \label{sec:explicit}
In this section we show how to obtain a constant time \mbox{$k$-generator} by combining an explicit $ \poly k$-generator with a cascading composition of unbalanced unique neighbor expanders. 
Our technique works by generating a small number of highly independent elements in an abelian group and then successively applying constant degree expanders to produce a greater number of less independent elements. 
We continue this process up until the point where the final number of elements is large enough to match the cost of generating the smaller batch of highly independent elements.    

The generator has two components.
The first component is an explicit $m$-generator $g_{0} : [n] \to A$ with period $n$ and range an abelian group $A$.
The second component is an explicit sequence $\left(\Gamma_{i}\right)^{t}_{i = 1}$ of unbalanced unique neighbor expanders.
The expanders are constructed such that the left side of the $i$th expander matches the right side of the $(i+1)$th expander.
By Lemma \ref{lem:explicit}, for every choice of imbalance $c$, target independence $k$ and length of the expander sequence $t$ there exists a sequence of expanders with the property that 
\begin{equation}
\Gamma_{i} \text{ is } (c, c^{i-1}m, d, d^{t-i}k)\text{-unique}, \label{eq:expandersequence}
\end{equation}
for $m = O(d^{t}k)$ and $d = \poly \log c$.
For constant $c$ each expander in the sequence is explicit.

We now combine the explicit $m$-generator $g_{0}$ and the sequence of expanders $\left(\Gamma_{i}\right)^{t}_{i = 1}$ to define the $k$-independent function $g_{t}$.
Let $b = m/n$ and assume for simplicity that $m$ divides $n$. 
For each $\Gamma_{i}$ we use the technique from Lemma \ref{lem:stacking} to construct a $(c, c^{i-1}n, d, d^{t-i}k)$-unique expander $\Gamma_{i}^{(b)}$.
Let $x_{i}$ denote a number in $[c^{i}n]$ corresponding to a vertex in the right side of $\Gamma_{i}^{(b)}$. 
We are now ready to give a recursive definition of $g_{i} : [c^{i}n] \to A$.
\begin{equation}
g_{i}(x_{i}) = 
\sum \limits_{x_{i-1} \in \Gamma_{i}^{(b)}(\{x_{i}\})} g_{i-1}(x_{i-1}), \quad 1 \leq i \leq t.
\label{eq:gexplicit}
\end{equation}
\begin{lemma}
$g_{i}$ is $d^{t-i}k$-independent. 
\end{lemma}
\begin{proof}
We proceed by induction on $i$. 
By definition, $g_{0} : [n] \to A$ is $d^{t}k$-independent.
Assume by induction that $g_{i} : [c^{i}n] \to A$ is $d^{t-i}k$-independent.
By definition $\Gamma_{i+1}^{(b)}$ is a $(c, c^{i}n, d, d^{t-(i+1)}k)$-unique expander.
Applying Lemma \ref{rng:lem:expanderhashing} we have that $g_{i+1} : [c^{i+1}n] \to A$ is $d^{t-(i+1)}k$-independent.
\end{proof}

We will now show that $g_{t}$ supports fast sequential evaluation and prove that we can use $g_{t}$ to construct an explicit constant time $k$-generator from any explicit \mbox{$m$-generator}, for an appropriate choice of $m$.
Divide the domain of each $g_{i}$ evenly into $b = n/m$ batches of size $c^{i}m$ corresponding to each block of the adjacency matrix of $\Gamma_{i}$ used to construct $\Gamma_{i}^{(b)}$ and index the batches by $j \in [b]$. 
In order to evaluate $g_{i+1}$ over batch number $j$ it suffices to know $\Gamma_{i+1}$ and the values of $g_{i}$ over batch number $j$.
Fast sequential evaluation of $g_{t}$ is achieved in the following steps.
First we tabulate the sequence of expanders $\left(\Gamma_{i}\right)^{t}_{i = 1}$ such that $\Gamma_{i}(\{x_{i}\})$ can be read in $d$ operations.
Secondly, to evaluate $g_{t}$ over batch $j$, we begin by tabulating the output of $g_{0}$ over batch $j$ and then successively apply our tabulated expanders to produce tables for the output of $g_{1}, g_{2}, \dots, g_{t}$ over batch $j$.

Given tables for the sequence of expanders and assuming that the generator underlying $g_{0}$ has been initialized, we now consider the average number of operations used per output when performing batch-evaluation of $g_{t}$.
The number of values output is $c^{t}m$.
The cost of emitting $m$ values from $g_{0}$ is by definition at most $\poly(m)$.
The cost of producing tables for the output of $g_{1}, g_{2}, \dots, g_{t}$ for the current batch is given by $\sum_{i=1}^{t}dc^{i}m = O(dc^{t}m)$ for $c > 1$.
The average number of operations used per output when performing batch-evaluation of $g_{t}$ is therefore bounded from above by
\begin{equation}
\frac{O(dc^{t}m) + \poly m}{c^{t}m} = O(d) + \frac{\poly m}{c^{t}}. \label{eq:averagetime}
\end{equation}
The following lemma states that we can obtain a constant time $k$-generator from every explicit $m$-generator by setting $t = O(\log k)$ and choosing $c$ to be an appropriately large constant.  
\begin{lemma} \label{lem:general}
Let $A$ be an abelian group with constant time addition. 
Suppose there exists an explicit $m$-generator with range $A$, period $n$ and space usage $\poly m$.
Then there exists a positive constant $\epsilon$ such that for every $k \leq m^{\epsilon}$ 
there exists an explicit constant time $k$-generator with range $A$, period $n$, and seed length, space usage and initialization time $\poly k$.
\end{lemma}
\begin{proof}
The sequence of expanders $\left(\Gamma_{i}\right)^{t}_{i = 1}$ with the properties given in \eqref{eq:expandersequence} exists for $m = O(d^{t}k)$ and $d = \poly \log c$ and is explicit for $c$ constant.
By inserting $m = O(d^{t}k)$ into equation \eqref{eq:averagetime} it can be seen that the average number of operations is constant for $c = O(1)$ and $t = O(\log k)$ with constants that depend on the parameters of the $m$-generator. 
The $k$-generator is initialized by initializing the $m$-generator, finding and tabulating the sequence of expanders and producing the first batch of values, all of which can be done in $\poly k$ time and space.
After initialization, each call to \texttt{emit()} will return a value from the current batch and use a constant number of operations for the task of preparing the next batch of outputs.    
\end{proof}

We now show our main theorem about explicit constant time $k$-generators over finite fields. 
The construction uses an explicit $m$-generator based on multipoint evaluation. 
Combined with the approach of Lemma \ref{lem:general} this yields a near-optimal time-space tradeoff for $k$-generation.

\begin{customthm}{\ref{thm:explicit}}[Repeated]
For every finite field $\GF$ with constant time arithmetic there exists a data structure that for every choice of $k$ is an explicit constant time $k$-generator with
range $\GF$, period $|\GF|\poly \log k$, and seed length, space usage and initialization time $k \poly \log k$.
\end{customthm}
\begin{proof}
Fix the choice of finite field $\GF$. 
By Lemma~\ref{lem:multipoint} there exists an explicit $m$-generator in $\GF$ for $m \leq |\GF|$ with period $|\GF|$ that uses time $O(m \log^{3} m)$ to emit $m$ values.
Fix some constant $c > 1$ and let $\left(\Gamma_{i}\right)^{t}_{i = 1}$ denote an explicit sequence of constant degree expanders with the properties given by \eqref{eq:expandersequence}.
The average number of operations per $k$-independent value output by $g_{t}$ when performing batch evaluation is given by
\begin{equation}
\frac{O(dc^{t}m) + O(m \log^{3} m)}{c^{t}m} = O(d) + \frac{O(\log^{3} d^{t}k)}{c^{t}}. \label{eq:averagetimefield}
\end{equation}
Setting $t = O(\log \log k)$ and following the approach of Lemma \ref{lem:general} we obtain a $k$-generator with the stated properties. 
\end{proof}
Based on the discussion in a paper by Capalbo \cite{capalbo2005} that introduces unbalanced unique neighbor expanders for concrete values of $c$ and $d$, it appears likely that the constants hidden in Theorem \ref{thm:explicit} for the current best explicit constructions make our explicit generators unsuited for practical use since $c$ is close to $1$ when $d$ is reasonably small. 
The next section explores how randomly generated unique neighbor expanders can be used to show stronger existence results and yield $k$-generators with tractable constants.  
%
\section{Constant time generators with optimal seed length} \label{sec:probabilistic}
Randomly constructed expanders of the type used in this paper have stronger properties than known explicit constructions, and can be generated with an overwhelming probability of success.
There is no known efficient algorithm for verifying whether a given graph is a unique neighbor expander.
Therefore randomly generated expanders cannot be used to replace explicit constructions without some probability of failure.

In this section we apply the probabilistic method to show the existence of $k$-generators with better performance characteristics than those based on known explicit constructions of expanders. 
We are able to show the existence of constant time generators with optimal seed length that use $O(k\log^{2+\varepsilon}k)$ words of space for any constant $\varepsilon > 0$. 
Furthermore, such generators can be constructed for any choice of constant failure probability $\delta > 0$. 
The generators we consider in this section use only a single expander graph but are otherwise identical to the generators described in Section \ref{sec:explicit}.
Using a single expander graph suffices for constant time generation because the probabilistic constructions are powerful enough to support an imbalance of $c = \poly \log k$ while maintaining constant degree.
This imbalance is enough to amortize the cost of multipoint evaluation in a single expansion step as opposed to the sequence of explicit expanders employed in Theorem \ref{thm:explicit}. 
Our arguments are a straightforward application of the probabilistic method, but we include them for completeness and because we are interested in somewhat nonstandard parameters.

We consider the following randomized construction of a $(c,m,d,k)$-unique expander $\Gamma$. 
For each vertex $x$ in $[cm]$, we add an edge between $x$ and each distinct node of $d$ nodes selected uniformly at random from $[m]$.  
By a standard argument, the graph can only fail to be unique neighbor expander if there exists a subset $S$ of left hand side vertices with $|S| \leq k$ such that $|\Gamma(S)| \leq \lfloor d|S|/2 \rfloor$ \cite[Lemma 2.8]{siegel2004}.
In the following we assume that $kd \leq m$. 
\begin{align}
&\Pr[\Gamma \text{ is not a unique neighbor expander}] \notag \\
&\leq \Pr[\exists S \subseteq [cm], |S| \leq k : |\Gamma(S)| \leq \lfloor d|S|/2 \rfloor] \notag \\
&\leq \sum_{\substack{S \subseteq [cm]\\ |S| \leq k}} \Pr[|\Gamma(S)| \leq \lfloor d|S|/2 \rfloor] \notag \\
&\leq \sum_{i = 1}^{k} \binom{cm}{i} \binom{m}{\lfloor id/2 \rfloor} \left(\frac{\lfloor id/2 \rfloor}{m}\right)^{id} \notag \\ 
&\leq \sum_{i=1}^{k} \left(\frac{cme}{i}\right)^{i} \left(\frac{me}{id/2}\right)^{id/2} \left(\frac{id/2}{m}\right)^{id} \notag \\
&= \sum_{i=1}^{k} \left( ec (i/m)^{d/2 - 1} (d e/2)^{d/2} \right)^{i} \label{eq:probexpander}
\end{align}
If the expression in the outer parentheses in \eqref{eq:probexpander} can be bounded from above by $1/2$ for $i = 1,2,\dots,k$, then the expander exists.
We also note that the randomized expander construction can be performed using $dk$-independent variables without changing the result in~\eqref{eq:probexpander}. 
Let $\gamma > 1$ be a number that may depend on $k$ and let $\delta$ denote an upper bound on the probability that the randomized construction fails. 
By setting $m = O(dk\gamma)$ we are able to obtain the following expression for the relation between $\delta$, the imbalance $c$ and the left outdegree bound $d$.
\begin{equation}
\delta = \frac{cd}{\gamma^{d/2 - 1}} \label{eq:expanderparameters}
\end{equation}
Equation \eqref{eq:expanderparameters} reveals tradeoffs for the parameters of the randomly constructed $k$-unique expander graphs.
For example, increasing $\gamma$ makes it possible to make the graph more unbalanced while maintaining the same upper bound on the probability of failure $\delta$. 
The increased imbalance comes at the cost of an increase in $m$, the size of the right side of the graph. 
Similarly it can be seen how increasing $d$ can be used to reduce the probability of error.
Setting the parameters to minimize the space occupied by the expander while maintaining constant outdegree and by extension constant generation time, we obtain Theorem~\ref{thm:existence}.  
\begin{customthm}{\ref{thm:existence}}[Repeated]
For every finite field $\GF$ with constant time arithmetic and every choice of constants $\varepsilon$, $\delta > 0$ there exists a data structure 
that for every choice of $k$ is a constant time $k$-generator with failure probability $\delta$, range~$\GF$, period $|\GF|$, seed length $O(k)$, 
space usage $O(k \log^{2+\varepsilon}k)$, and initialization time $O(k \poly \log k)$.
\end{customthm}
\begin{proof}
Let $\tilde{\varepsilon} < \varepsilon$ be a constant and set $\gamma = \log^{\tilde{\varepsilon}}k$.
Choosing $d$ to be a sufficiently large constant (dependent on $\tilde{\varepsilon}$), equation \eqref{eq:expanderparameters} shows that 
for every $\delta > 0$ there exists a $(c, m, d, k)$-unique expander $\Gamma$ with $c = \Omega(\log^{2+\varepsilon}k)$ and $m = O(k\gamma)$.
Using multipoint evaluation, the right side vertices of $\Gamma$ can be associated with $\Theta(k)$-independent variables over $\GF$ using $O(k \log^{2 + \varepsilon}k)$ operations.
By the properties of $\Gamma$ and applying Lemma \ref{rng:lem:expanderhashing} we are able to generate batches of $k$-independent variables of size $\Omega(k\log^{2+\varepsilon}k)$ using $O(k \log^{2+\varepsilon}k)$ operations.
The seed length of $O(k)$ holds by the observation that randomized construction of the expander only requires $O(k)$-independence.
The $O(k \poly \log k)$ initialization time is obtained by using multipoint evaluation to construct a table for $\Gamma$.
\end{proof}
\section{Faster multipoint evaluation for $k$-generators} \label{sec:faster}
This section presents an improved generator based directly on multipoint evaluation of a polynomial hash function $h \in \mathcal{H}_{k}$ over a finite field.
For our purpose of generating an $(n,k)$-sequence from $h$, we are free to choose the order of elements of $\GF$ in which to evaluate $h$.  
We present an algorithm for the systematic evaluation of $h$ over disjoint size $k$ subsets of $\GF$ using Fast Fourier Transform (FFT) algorithms.
Our technique yields a $k$-generator over $\GF$ with generation time $O(\log k)$, and space usage and seed length that is optimal up to constant factors.
The algorithm depends upon the structure of $\GF$, similarly to other FFT algorithms over finite fields \cite{bhattacharya2004}. 

The nonzero elements of $\GF$ form a multiplicative cyclic group $\GF^{*}$ of order $q-1$. 
The multiplicative group has a primitive element $\omega$ which generates $\GF^{*}$.
\begin{equation*}
\GF^{*} = \{ \omega^{0}, \omega^{1}, \omega^{2}, \dots, \omega^{q-2} \}.
\end{equation*}
For $k$ that divides $q-1$, we can construct a multiplicative subgroup $S_{k,0}^{*}$ of order $k$ with $\omega_{k} = \omega^{(q-1)/k}$ as the generating element. 
$S_{k,0}^{*}$ contains $k$ distinct elements of $\GF$. 
Define for $j = 0,1, \dots, (q-1)/k - 1$, 
\begin{equation*}
S_{k,j}^{*} = \omega^{j}S_{k,0} = \{ \omega^{j}\omega_{k}^{0}, \omega^{j}\omega_{k}^{1}, \dots, \omega^{j}\omega_{k}^{k-1} \}.  
\end{equation*}
Viewed as subsets of $\GF^{*}$ the sets $S_{k,j}^{*}$ form an exact cover of $\GF^{*}$. 
We now consider how to evaluate a degree $k-1$ polynomial $h(x) \in \GF[X]$ in the points of $S_{k,j}^{*}$. The polynomial takes the form
\begin{equation*}
h(x) = a_{0}x^{0} + a_{1}x^{1} + \dots + a_{k-1}x^{k-1}.
\end{equation*}
Rewriting the polynomial evaluation over $S_{k,j}^{*}$ in matrix notation:
\begin{equation*}
\colvec{ h(\omega^{j}\omega_{k}^{0}) \\  h(\omega^{j}\omega_{k}^{1}) \\ h(\omega^{j}\omega_{k}^{2}) \\ \vdots \\ h(\omega^{j}\omega_{k}^{k-1})} = 
\colvec{
	\omega_{k}^{0\cdot0} & \omega_{k}^{0\cdot1} & \dots & \omega_{k}^{0\cdot(k-1)} \\
	\omega_{k}^{1\cdot0} & \omega_{k}^{1\cdot1} & \dots & \omega_{k}^{1\cdot(k-1)} \\
	\omega_{k}^{2\cdot0} & \omega_{k}^{2\cdot1} & \dots & \omega_{k}^{2\cdot(k-1)} \\
	\vdots & \vdots &  & \vdots \\
	\omega_{k}^{(k-1)\cdot0} & \omega_{k}^{(k-1)\cdot1} & \dots & \omega_{k}^{(k-1)\cdot(k-1)}
}
\colvec{\omega^{j \cdot 0}a_{0} \\ \omega^{j \cdot 1}a_{1} \\ \omega^{j \cdot 2}a_{2} \\ \vdots \\ \omega^{j \cdot (k-1)}a_{k-1}}
\end{equation*}
We assume that the coefficients of $h$ and $\omega^{j}$ are given and consider algorithms for efficient evaluation of the matrix-vector product.
The coefficients $\tilde{a}_{j,i} = \omega^{j \cdot i}a_{i}$ for $i = 0,1,\dots,k-1$ can be found in $O(k)$ operations and define a polynomial $\tilde{h}_{j}(x) = \sum_{i=0}^{k-1}\tilde{a}_{i,j}x^{i}$. 
Evaluating $\tilde{h}_{0}(x)$ over $S_{k,0}^{*}$ corresponds to computing the Discrete Fourier Transform over a finite field.
\begin{customthm}{\ref{thm:fastmultipoint}}[Repeated]
For every finite field $\GF$ that supports computing the discrete Fourier transform of length $k$ in $O(k \log k)$ operations, 
there exists a data structure that, for every choice of $k$ and given a primitive element $\omega$, 
is an explicit $O(\log k)$~time $k$-generator with range $\GF$, period $|\GF|$, seed length $k$, space usage $O(k)$, and initialization time $O(k \log k)$.
\end{customthm}
\begin{proof}
Evaluation of $\tilde{h}_{j}(x)$ over $S_{k,j}^{*}$ takes $O(k \log k)$ operations by assumption. 
For every batch $j$ starting at $j = 0$, the value of $\omega^{j}$ is stored and used to compute the coefficients of $\tilde{h}_{j+1}(x)$ in $O(k)$ operations.
\end{proof}
We now discuss the validity of the assumption that we are able to compute the DFT over a finite field in $O(k \log k)$ operations.
Assume that $k \mid (q - 1)$ and that $\omega_{k}$ is known.
If $k$ is highly composite there exist Fast Fourier Transforms for computing the DFT in $O(k \log k)$ field operations~\cite{duhamel1990}.
If $k$ is not highly composite there exists an algorithm for computing the DFT in $O(kz \log kz )$ operations for fields of cardinality $q = p^{z}$ in our model of computation~\cite{preparata1977}. 
For $q = p^{O(1)}$ this reduces to the desired $O(k \log k)$ operations.
\section{Finite field arithmetic on the word RAM} \label{sec:wordRAM}
Throughout the paper we have used as our model of computation a modified word RAM with constant time arithmetic $(+,-, \times)$ over a finite field $\GF$.
In this section we show how our model relates to the more standard $\emph{multiplication model}$ defined as a word RAM with constant time arithmetic $(+, -, \times)$ over the integers $[2^{w}]$ for $w$-bit words \cite{hagerup1998}.

Arithmetic over $\GF_{p}$ for prime $p$ is integer arithmetic modulo $p$. 
We now argue that arithmetic operations over $\GF_{p}$ can be performed in $O(1)$ operations in the multiplication model.
Every integer $x$ can be written on the form $x = qp + r$ for non-negative integers $q, r$ with $r < p$.
Assume that $x$ can be represented in a constant number of words.
The problem of computing $r = x \bmod p$ can be solved by an integer division and $O(1)$ operations in the multiplication model due to the identity $r = x -  \lfloor x/p \rfloor p$.
An algorithm by Granlund and Montgomery \cite{granlund1994} computes $\lfloor x/p \rfloor$ for any constant $p$ using $O(1)$ operations in the multiplication model which gives the desired result.

Another finite field of interest is $\GF_{2^{w}}$ due to the correspondence between field elements and bit vectors of length $w$.
We will argue that a word RAM model that supports constant time multiplication over $\GF_{2^{w}}$ is not unrealistic considering current hardware.
Addition in $\GF_{2^{w}}$ has direct support in standard CPU instruction sets through the XOR operation.
A multiplication of two elements $x$ and $y$ in $\GF_{2^{w}}$ can be viewed as a two-step process.
First, we perform a carryless multiplication $z = x \cdot y$ of the representation of $x$ and $y$ as polynomials in $F_{2}[X]$. 
Second, we use a modular reduction to bring the product $x \cdot y$ back into $\GF_{2^{w}}$, similarly to modular arithmetic over $\GF_{p}$. 
Recently, hardware manufacturers have included partial support for multiplication in $\GF_{2^{w}}$ with the CLMUL instruction for carryless multiplication \cite{gueron2014}. 
The modular reduction step is performed by dividing $x \cdot y$ by an irreducible polynomial $g$ and returning the remainder.
Irreducible polynomials $g$ that can be represented as sparse binary vectors with constant weight results in a constant time algorithm for modular reduction as presented by Gueron and Kounavis \cite{gueron2014}.
We briefly introduce the computation underlying the algorithm to show that its complexity depends on the number of \texttt{1}s in the binary representation of $g$.
Let $L^{w}$ and $M^{w}$ be functions that return the $w$ least, respectively most, significant bits of their argument as represented in $\GF_{2^{2w}}$.  
The complexity of Gueron and Kounavis' algorithm for modular reduction of $z = x \cdot y$ is determined by the complexity of evaluating the expression
\begin{equation}
L^{w}(L^{w}(g)\cdot M^{w}(M^{w}(z) \cdot g)). \label{eq:clmul}
\end{equation}
Evaluating $L^{w}$ and $M^{w}$ is standard bit manipulation. 
For $g$ of constant weight, the carryless multiplications denoted by $\cdot$ in equation $\eqref{eq:clmul}$ can be implemented as a constant number of bit shifts and XORs. 
For every $w \leq 10000$ an irreducible trinomial or pentanomial ($g$ of weight at most 5) has been found~\cite{seroussi1998}.
Together with the hardware support for convolutions this allows us to implement fast multiplication over fields of practical interest. 
\section{A load balancing application} \label{sec:loadbalancing}
We next consider how our new generator yields stronger guarantees for load balancing.
Our setting is motivated by applications such as splitting a set of tasks of unknown duration among a set of $m$ machines, in order to keep the load as balanced as possible.
Once a task is assigned to a machine, it cannot be reassigned, i.e., we do not allow \emph{migration}.
For simplicity we consider the \emph{unweighted} case where we strive to keep the \emph{number} of tasks on each machine low, and we assume that $m$ divides $|\GF|$ for some field $\GF$ with constant time operations on a word RAM.
Suppose that each machine has capacity (e.g.~memory enough) to handle $b$ tasks at once, and that we are given a sequence of $t$ tasks $T_1,\dots,T_t$, where we identify each task with its duration (an interval in $\R$).
Now let $k = mb$ and suppose that we use our constant time $k$-generator to determine for each $i=1,\dots,t$ which machine should handle $T_i$.
(We emphasize that this is done without knowledge of $T_i$, and without coordination with the machines.)
Compared to using a fully random choice this has the advantage of requiring only $k\poly\log k$ words of random bits, which in turn may make the algorithm faster if random number generation is a bottleneck.
Yet, we are able to get essentially the same guarantee on load balancing as in the fully random case.
To see this let $L(x) = \{ i \; | \; x\in T_i\}$ be the set of tasks active at time $x$, and let $L_q(x)$ be the subset of $L(x)$ assigned to machine $q$ using our generator.
We have:
\begin{lemma}\label{lem:error}
For $\varepsilon > 0$, if $|L(x)| (1+\varepsilon) < mb$ then
$\Pr[\max_q |L_q(x)| > b] < m \exp(-\varepsilon^2 b / 3)$.
\end{lemma}
\begin{proof}
Since $|L(x)| < mb = k$ we have that the assignment of tasks in $L(x)$ to machines is uniformly random and independent.
This means that the number of tasks assigned to each machine follows a binomial distribution with mean $b/(1+\varepsilon)$, and we can apply a Chernoff bound of $\exp(-\varepsilon^2 b / 3)$ on the probability that more than $b$ tasks are assigned to a particular machine.
A union bound over all $m$ machines yields the result.
\end{proof}

Lemma~\ref{lem:error} allows us to give a strong guarantee on the probability of exceeding the capacity $b$ of a machine at any time, assuming that the average load is bounded by $b/(1+\varepsilon)$.
In particular, let $S\subseteq\R$ be a set of size at most $2t$ such that every workload $L(y)$ is equal to $L(x)$ for some $x\in S$.
The existence of $S$ is guaranteed since the $t$ tasks are intervals, and they have at most $2t$ end points.
This means that
$$\sup_{x\in\R} \max_q |L_q(x)| = \max_{x\in S} \max_q |L_q(x)|,$$
so a union bound over $x\in S$ gives
$$\Pr[\sup_{x\in\R} \max_q |L_q(x)| > b] < 2tm \exp(-\varepsilon^2 b / 3) \enspace .$$

For constant $\varepsilon$ and whenever $b = \omega(\log k)$ and $tm = 2^{o(b)}$ we get an error probability that is exponentially small in $b$.
Such a strong error guarantee can not be achieved with known constant time hashing methods~\cite{siegel2004,pagh2008,dietzfelbinger2003,thorup2013} in reasonable space, since they all have an error probability that decreases polynomially with space usage.
Even if explicit constructions for the expanders needed in Siegel's hash functions were found, the resulting space usage would be polynomially higher than with our $k$-generator.
\section{Experiments} \label{sec:experiments}
This section contains experimental results of an implementation of a $k$-generator over $\GF_{2^{64}}$. 
There are two main components to the generator: an algorithm for filling a table of size $m$ with $dk$-independent variables and a bipartite unbalanced expander graph.

For the first component, we use an implementation of Gao-Mateer's additive FFT \cite[Algorithm 2.]{gao2010}.
Utilizing the Gao-Mateer algorithm we can generate a batch of $k$ elements of an $(|\GF|, k)$-sequence using space $O(k)$ and $O(k \log^{2} k)$ operations on a word RAM that supports arithmetic over $\GF$.
The additive complexity of the FFT algorithm is $O(k \log^{2} k)$ while the multiplicative complexity is $O(k \log k)$.
Addition in $\GF_{2^{64}}$ is implemented as an XOR-operation on 64-bit words. 
Multiplication is implemented using the PCLMUL instruction along with the techniques for modular reduction by Gueron et
al. \cite{gueron2014} outlined in Section \ref{sec:wordRAM}.

For the second component we introduce a slightly different type of expander graphs that only work in the special case of
fields of characteristic two. 
Let $\GF_{2^{w}}$ be a field of characteristic two and let $M$ be a $cm \times m$ adjacency matrix of a graph $\Gamma$ where each entry of $M$ is viewed as an element of $\GF_{2^{w}}$.    
By a similar argument to the one used in Lemma \ref{rng:lem:expanderhashing} the linear system $Mx$ defines a $(cm, k)$-sequence if $x$ is a vector of \mbox{$dk$-independent} variables over $\GF_{2^{w}}$ and $M$ has row rank at least $k$.
We consider randomized constructions of $M$ over $\GF_{2}$ with at most $d$ \texttt{1}s in each row and row rank at least $k$.
It is easy to see that a matrix $M$ over $\GF_{2}$ with these properties also defines a matrix with the same properties over $\GF_{2^{w}}$.
Since $k$-uniqueness of $\Gamma$ implies that $M$ has row rank $k$, but not the other way around, 
we are able to obtain better performance characteristics of generators over $\GF_ {2^{w}}$ by focusing on randomized constructions of $M$.

The matrix $M$ is constructed in the following way.
Independently, for each $i \in [cm]$ sample $d$ integers uniformly with replacement from $[m]$ and define the $i$th row of $M$ as the vector constructed by taking the zero vector and adding~\texttt{1}s in the $d$ positions sampled for row $i$.
Observe that if $M$ does not have row rank at least $k$ then some non-empty subset of at most $k$ rows of $M$ sum to the zero vector.
In order for a non-empty set of vectors over $\GF_{2}^{m}$ to sum to the zero vector, the bit-parity must be even in each of the $m$ positions of the sum. 
The sum of any $i$ rows of $M$ corresponds to a balls and bins process that distributes $id$ balls into $m$ bins, independently and uniformly at random.
Let $id$ be an even number. Then there are $(id - 1)!!$ ways of ordering the balls into pairs and the probability that the outcome is equal to any particular pairing is $(1/m)^{id/2}$. 
This yields the following upper bound on the probability that a subset of $i$ rows sums to zero: 
\begin{equation}
\beta_{pair}(i, d, m) = (id-1)!!\! \left(\frac{1}{m}\right)^{id/2}. \label{eq:libound} 
\end{equation}
A comparison between this bound and the bound for $k$-uniqueness from equation \eqref{eq:probexpander} shows that, for each term in the sum, 
the multiplicative factor applied to the binomial coefficient $\binom{cm}{i}$ is exponentially smaller in $id$ for the bound in \eqref{eq:libound}.

The pair-based approach which yields the bound $\beta_{pair}$ overestimates the probability of failure on subsets of size $i$, increasingly as $id$ grows large compared to $m$. 
We therefore introduce a different bound based on the Poisson approximation to the binomial distribution: 
the number of balls in each of the $m$ positions can approximately be modelled as independent Poisson distributed variables \cite[Ch. 5.4]{mitzenmacher2005}. 
The probability that that the parity is even in each of the $m$ positions in a sum of $i$ rows is bounded by
\begin{equation}
\beta_{poisson}(i, d, m) = e\sqrt{id}\left(\frac{1+e^{-2\frac{id}{m}}}{2}\right)^{m}, \label{eq:poibound} 
\end{equation}
where we use the same approach as Mitzenmacher et al.~\cite{mitzenmacher2014}.
For any given subset of rows of $M$, we are free to choose between the two bounds. 
The probability that a randomly constructed matrix $M$ fails to have rank at least $k$ can be bounded from above using a union bound over subsets of rows of $M$.
\begin{equation}
\delta \leq \sum_{i=1}^{k}\binom{cm}{i}\min(\beta_{pair}(i, d, m), \beta_{poisson}(i, d, m)). \label{eq:combinedbound}
\end{equation}

We now consider the generation time of our implementation. 
Let $FFT_{dk}$ denote the time taken by the FFT algorithm to generate a $dk$-independent value and let $RA_{d,m}$ denote the time it takes to perform $d$ random accesses in a table of size $m$.  
The time taken to generate a value by the implementation of our generator is then given by
\begin{equation}
T = \frac{FFT_{dk}}{c} + RA_{d,m}. \label{eq:generationtime}
\end{equation}
In our experiments, the choice of parameters for the expander graphs were based on a search for the fastest generation time over every combination of imbalance $c \in \{16, 32, 64 \}$ and outdegree $d \in \{ 4, 8, 16 \}$.
Given choices of $d$, $c$ and independence $k$, the size of the right side of the expander $m$ was increased until existence could be guaranteed by the bound in \eqref{eq:combinedbound}.
The generator in the experiments had the restriction that $m \leq 2^{26}$ and we have measured $RA_{d,m}$ assuming that
the expander is read sequentially from RAM. 
The experiments were run on a machine with an Intel Core i5-4570 processor with 6MB cache and 8GB of RAM.

Table \ref{tab:experimentalresults} shows the generation time in nanoseconds per 64-bit output using Horner's scheme, Gao-Mateer's FFT and the implementation of our generator (FFT + Expander).
For the implementation of the generator, we also show the parameters of the randomly generated expander that yielded the fastest generation time among expanders in the search space.

The generation time for Horner's scheme is approximately linear in $k$ and logarithmic in $k$ for the FFT, as predicted by theory. 
The FFT is faster than using Horner's scheme already at $k = 64$ and orders of magnitude faster for large $k$.
Using our implementation of Gao-Mateer's FFT algorithm we are able to evaluate a polynomial of degree $2^{20}-1$ in $2^{20}$ points in less than a second. 
The same task takes over an hour when using Horner's rule, even with both algorithms using the same underlying implementation of algebraic operations in the field.

For small values of $k$, our generator is an order of magnitude faster than the FFT and comes close to the performance of the 64-bit C++11 implementation of the popular Mersenne Twister.
Our generator uses 25 nanoseconds to output a 1024-independent value. This is equivalent to an output of over 300MB/s.
The Mersenne Twister uses around 4 nanoseconds to generate a 64-bit value.

In practice, the memory hierarchy appears to be the primary obstacle to maintaining a constant generation time as $k$ increases.
Our generator reads the expander graphs sequentially and performs random lookups into the table of $dk$-independent values.
As $k$ grows large, the table can no longer fit into cache and for large imbalance $c$, the expander can no longer be stored in main memory.
Searching a wider range of expander parameters could easily yield a faster generation time, potentially at the cost of a larger imbalance $c$ or higher probability of failure $\delta$.

\begin{table}[htpb]
	\centering
	\begin{tabular}{rrr|rrrrr}
	
	\toprule
    
\multirow{2}{*}{$k$} & \multicolumn{1}{c}{Horner} & \multicolumn{1}{c|}{FFT}  & \multicolumn{5}{c}{FFT + Expander}       \\ 
           
& \multicolumn{1}{c}{ns} &   \multicolumn{1}{c|}{ns}  & \multicolumn{1}{c}{$c$} &  \multicolumn{1}{c}{$m$}    & \multicolumn{1}{c}{$d$} & \multicolumn{1}{c}{$\delta$}    & \multicolumn{1}{c}{ns}    \\ \midrule 
$2^{5}$    &       177 &  243  & 64 & $2^{13}$ &  8 & $10^{-7}$   &  15   \\    
$2^{6}$    &       361 &  294  & 64 & $2^{14}$ &  8 & $10^{-8}$   &  16   \\ 
$2^{7}$    &       730 &  338  & 64 & $2^{15}$ &  8 & $10^{-9}$   &  19   \\ 
$2^{8}$    &      1470 &  375  & 64 & $2^{16}$ &  8 & $10^{-10}$  &  23   \\ 
$2^{9}$    &      2950 &  412  & 64 & $2^{17}$ &  8 & $10^{-11}$  &  24   \\ 
$2^{10}$   &      5902 &  449  & 64 & $2^{18}$ &  8 & $10^{-12}$  &  25   \\ 
$2^{11}$   &     11808 &  487  & 32 & $2^{18}$ &  8 & $10^{-12}$  &  35   \\ 
$2^{12}$   &     23627 &  523  & 64 & $2^{18}$ & 16 & $10^{-29}$  &  43   \\ 
$2^{13}$   &     47183 &  561  & 32 & $2^{18}$ & 16 & $10^{-29}$  &  54   \\ 
$2^{14}$   &     94429 &  599  & 64 & $2^{22}$ &  8 & $10^{-15}$  &  68   \\ 
$2^{15}$   &    188258 &  638  & 64 & $2^{23}$ &  8 & $10^{-16}$  &  69   \\ 
$2^{16}$   &    376143 &  678  & 64 & $2^{24}$ &  8 & $10^{-17}$  &  77   \\ 
$2^{17}$   &    751781 &  719  & 64 & $2^{25}$ &  8 & $10^{-18}$  &  85   \\ 
$2^{18}$   &   1505016 &  765  & 64 & $2^{26}$ &  8 & $10^{-19}$  &  93   \\ 
$2^{19}$   &   3015969 &  808  & 32 & $2^{26}$ &  8 & $10^{-19}$  & 110   \\ 
$2^{20}$   &   6082313 &  864  & 64 & $2^{26}$ & 16 & $10^{-46}$  & 175   \\ \bottomrule

	\end{tabular}
\caption{Generation time in nanoseconds per 64-bit value using Horner's scheme, Gao-Mateer's FFT and an implementation of our constant-time generator}
\label{tab:experimentalresults}
\end{table}

\section*{Acknowledgment}
We are grateful to Martin Dietzfelbinger who gave feedback on an early version of the paper, allowing us to significantly enhance the presentation.
\chapter{Near-optimal $k$-independent hashing} \label{ch:expanders}
\sectionquote{Not all those who wander are lost}
\noindent We consider the following fundamental problems:
	\begin{itemize}
		\item Constructing $k$-independent hash functions with a space-time tradeoff close to Siegel's cell-probe lower bound (SICOMP 2004).
		\item Constructing representations of unbalanced expander graphs having small size and allowing fast computation of the neighbor function.
	\end{itemize}
	It is not hard to show that these problems are intimately connected in the sense that a good solution to one of them leads to a good solution to the other one.
	In this paper we exploit this connection to present efficient, recursive constructions of $k$-independent hash functions (and hence expanders with a small representation).
	While the previously most efficient construction (Thorup, FOCS 2013) 
	needed time quasipolynomial in Siegel's lower bound, our time bound is
	just a logarithmic factor from
	the lower bound.

\section{Introduction}
The problem of designing explicit unbalanced expander graphs with near-optimal parameters is of major importance in theoretical computer science.
In this paper we consider bipartite graphs with edge set $E \subset U \times V$ where ${|U| \gg |V|}$.
Vertices in $U$ have degree $d$ and expansion is desired for subsets $S \subset U$ with $|S| \leq k$ for some parameter~$k$.
Such expanders have numerous applications (e.g.~hashing~\cite{siegel2004}, routing~\cite{broder1999}, sparse recovery~\cite{indyk2010}, membership~\cite{buhrman2002}), 
yet coming up with explicit constructions that have close to optimal parameters has proved elusive.
At the same time it is easy to show that choosing $E$ at random will give a graph with essentially optimal parameters.
This means that we can efficiently and with a low probability of error produce a description of an optimal unbalanced expander that takes space proportional to $|U|$.
Storing a complete description is excessive for most applications that, provided access to an explicit construction, would use space proportional to $|V|$.    
On the other hand, explicit constructions can be represented using constant space, 
but the current best explicit constructions have parameters $d$ and $|V|$ that are polynomial in the optimal parameters of the probabilistic constructions~\cite{guruswami2009}.
Furthermore, existing explicit constructions have primarily aimed at optimizing the parameters of the expander, with the evaluation time of the neighbor function being of secondary interest, as long as it can be bounded by $\poly \log u$.
This evaluation time is excessive in applications that, provided access to the neighbor function of an optimal expander, would use time proportional to $d$, where $d$ is typically constant or at most logarithmic in $|U|$. 

In this paper we focus on optimizing the parameters of the expander while minimizing the space usage of the representation and the evaluation time of the neighbor function.
We present randomized constructions of unbalanced expanders in the standard word RAM model.
Our constructions have near-optimal parameters, use space close to $|V|$, and support computing the $d$ neighbors of a vertex in time close to~$d$.

\paragraph{Hash functions and expander graphs.} There is a close connection between $k$-independent hash functions and expanders. 
A $k$-independent function with appropriate parameters will, with some probability of failure, represent the neighbor function of a graph that expands on subsets of size $k$.
We refer to this as going from independence to expansion, and the fact follows from the standard union bound analysis of probabilistic constructions of expanders.
Going in the other direction, from expansion to independence, was first used by Siegel~\cite{siegel2004} as a technique for showing the existence of $k$-independent hash functions with evaluation time that does not depend on $k$.
We follow in Siegel's footsteps and a long line of work (see e.g.~\cite{dietzfelbinger2012} for an overview) that focuses on the space-time tradeoff of $k$-independent hash functions over a universe of size $u = |U|$.

Ideally, we would like to construct a data structure in the word RAM model that takes as input parameters $u$, $k$, and~$t$, and returns a $k$-independent hash function over $U$.
The hash function should use space $k(u/k)^{1/t}$ and have evaluation time $O(t)$, matching up to constant factors the space-time tradeoff of Siegel's cell probe lower bound for $k$-independent hashing~\cite{siegel2004}.
We present the first construction that comes close to matching the space-time tradeoff of the cell probe lower bound.

\paragraph{Method.} Our work is inspired by Siegel's graph powering approach \cite{siegel2004} and by recent advances in tabulation hashing~\cite{thorup2013}, 
showing that it is possible to efficiently describe expanders in space much smaller than $u$.
Our main insight is that it is possible to make simple, recursive expander constructions by alternating between strong unbalanced expanders and highly random hash functions.
Similarly to previous work, we follow the procedure of letting a $k$-independent function represent a bipartite graph $\Gamma$ that expands on subsets of size $k$. 
We then apply a graph product to $\Gamma$ in order to increase the size of the universe covered by the graph while retaining expander properties.
At each step of the recursion we return to $k$-independence by combining the graph product with a table of random bits, leaving us with a new $k$-independent function that covers a larger universe.
By combining the technique of alternating between expansion and independence with a new and more efficient graph product, we can improve upon existing randomized constructions of unbalanced expanders.    

\subsection{Our contribution}
Table~\ref{tab:results} compares previous upper and lower bounds on $k$-independent hashing with our results, as presented in Corollaries 1, 2, and 3.
As can be seen, most results present a trade-off between time and space controlled by a parameter $t$.
Tight lower and upper bounds have been known only in the cell probe model,
but our new construction nearly matches the cell probe lower bound by Siegel~\cite{siegel2004}.

The time bound for the construction using explicit expanders~\cite{guruswami2009} uses the degree of the expander as a conservative lower bound, 
based on the possibility that the neighbor function in their construction can be evaluated in constant time in the word RAM model.
The time bound that follows directly from their work is $\poly \log u$. 
While the constant factors in the exponent of the space usage of~\cite{siegel2004,thorup2013} have likely not been optimized, their techniques do not seem to be able to yield space close to the cell probe lower bound. 

As can be seen our construction polynomially improves either space or time compared to each of the previously best trade-offs.
We also find our construction easier to describe and analyze than the results of~\cite{guruswami2009,kedlaya2008,thorup2013}, with simplicity comparable to that of Siegel's influential paper~\cite{siegel2004}. 

\begin{table}[t]
	\renewcommand\arraystretch{1.5}
	\caption{Space-time tradeoffs for $k$-independent hash functions}
	\footnotesize
    \begin{tabular}{|l|l|l|} \hline
	\textbf{ Reference}  &  \textbf{ Space}   &  \textbf{ Time}  \\ \hline\hline
	Polynomials~\cite{joffe1974,wegman1981}  & $k$ & $O(k)$ \\ \hline
	Preprocessed polynomials~\cite{kedlaya2008} & $k^{1 + \varepsilon}(\log u)^{1 + o(1)}$ & $(\poly \log k) (\log u)^{1 + o(1)}$ \\ \hline
	Expanders~\cite{guruswami2009} + \cite{siegel2004} & $k^{1+ \varepsilon}d^{2}$ &  $d = O(\log(u)\log(k))^{1 + 1/\varepsilon}$ \\ \hline
    Expander powering \cite{siegel2004} & $k^{(1 - \varepsilon)t}u^{\varepsilon} + u^{1/t}$ & $O(1/\varepsilon)^{t}$ \\ \hline
    Double tabulation \cite{thorup2013} & $k^{5t} + u^{1/t}$ & $O(t)$ \\ \hline
    Recursive tabulation \cite{thorup2013} & $\poly k + u^{1/t}$ & $O(t^{\log t})$  \\ \hline 
    Corollary \ref{cor:simple}                              & $ku^{1/t}t^{3}$ & $O(t^{2} + t^{3}\log (k) / \log(u))$  \\ \hline 
    Corollary \ref{cor:majorityrecursion}                              & $k^{2}u^{1/t}t^{2}$ & $O(t \log t + t^{2} \log (k) /  \log (u))$  \\ \hline
    Corollary \ref{cor:linear}    & $ku^{1/t}t$ & $O(t \log t)$  \\ 
	\hline\hline
	    Cell probe lower bound~\cite{siegel2004} & $k(u/k)^{1/t}$            & $t < k$ \text{ probes}      \\ \hline
	    Cell probe upper bound~\cite{siegel2004}  &  $k(u/k)^{1/t}t$            & $O(t)$ \text{ probes}   \\
	\hline
    \end{tabular}

	\vspace{1em}
	
\scriptsize
Table notes: Space-time tradeoffs for $k$-independent hash functions from a domain of size $u$, with the trade-off controlled by a parameter $t$. 
Time bounds in the last two rows are number of cell probes, and remaining rows refer to the word RAM model with word size $\Theta(\log u)$.
Leading constants in the space bounds are omitted.
We use $t$ to denote an arbitrary positive integer parameter that controls the trade-off, and 
We use $\varepsilon$ to denote an arbitrary positive constant.
*Corollary 3 relies on the assumption $k = u^{O(1/t)}$.
\label{tab:results}
\end{table}
Like all other randomized constructions our data structures come with an error probability, 
but this error probability is \emph{universal} in the sense that if the construction works then it provides independent hash values on \emph{every subset} of at most $k$ elements from $U$.
This is in contrast to other known constructions~\cite{dietzfelbinger2003,pagh2008} that give independence with high probability on each \emph{particular} set of at most $k$ elements, 
but will fail almost surely if independence for a superpolynomial number of subsets is needed.

\paragraph{Applications.} 
Efficient constructions of highly random functions is of fundamental interest with many applications in computer science.
A $k$-independent function can, without changing the analysis, replace a fully random function in applications that only rely on $k$-subsets of inputs mapping to random values.
We can therefore view $k$-independent functions as space and randomness efficient alternatives to fully random functions, capable of providing compact representations of complex structures such as expander graphs over very large domains.
Apart from the construction of expander graphs with a small description, as an example application, $k$-independent functions with a universal error probability can be used to construct ``real-time'' dictionaries that are able to handle extremely long (in expectation) sequences of insertion and deletion operations in constant time per operation before failing.

Let $\tau > 1$ be a constant parameter.
We use a $k$-independent hash function with $k=w^{O(\tau)}$ to split a set of $n$ machine words of $w$ bits into $O(n)$ subsets such that each subset has size at most $k$, with probability at least $1 - 2^{-w^\tau}$.
Handling each subset with Thorup's recent construction of dictionaries for sets of size $w^{O(\tau)}$ using time $O(\tau)$ per operation~\cite{thorup2014} we get a dynamic dictionary in which, with high probability, every operation in a sequence of length $\ell < 2^{O(w^\tau)}$ takes constant time.
In comparison the hash functions of~\cite{DM90,dietzfelbinger2003,pagh2008} can only guarantee that sequences of length $\ell <\text{poly}(n)$ operations, where $n < 2^w$, succeed with high probability.
The splitting hash function needs space $u^{\Omega(1)}$, which might exceed the space usage of an individual dictionary, but this can be seen as a shared resource that is used for many dictionaries (in which case we bound the total number of operations before failure).

\section{Background and overview}
In the analysis of randomized algorithms we often assume access to a fully random function of the form~${f : [u] \to [r]}$ where $[n]$ denotes the set $\{0,1, \dots, n-1\}$.
To represent such a function we need a table with $u$ entries of $\log r$ bits.
This is impractical in applications such as hashing based dictionaries where we typically have that $u \gg r$ and the goal is to use space $O(r)$ to store $r$ elements of $[u]$.
Fortunately, the analysis that establishes the performance guarantees of a randomized algorithm can often be modified to work even in the case where the function $f$ has weaker randomness properties.

One such concept of limited randomness is $k$-independence, first introduced to computer science in the 1970s through the work of Carter and Wegman on universal hashing \cite{carter1977}.  
A family of functions from $[u]$ to $[r]$ is $k$-independent if, for every subset of $[u]$ of cardinality at most $k$, 
the output of a random function from the family evaluated on the subset is independent and uniformly distributed in $[r]$.
Trivially, the family of all functions from $[u]$ to $[r]$ is $k$-independent, but representing a random function from this family uses too much space.
It was shown in \cite{joffe1974} that for every finite field $\GF$ the family of functions that consist of all polynomials over $\GF$ of degree at most $k-1$ is $k$-independent.
A function from this family can be represented using near-optimal space \cite{chor1985} by storing the $k$ coefficients of the polynomial.
The mapping defined by a function $f$ from a $k$-independent polynomial family over $\GF = \{x_{1}, x_{2}, \dots, x_{u}\}$ takes the form 
\begin{equation}
\colvec{f(x_{1}) \\ f(x_{2}) \\ \vdots \\ f(x_{u})} = 
\colvec{ 
	x_{1}^{0} & x_{1}^{1} & \dots & x_{1}^{k-1} \\
	x_{2}^{0} & x_{2}^{1} & \dots & x_{2}^{k-1} \\
	\vdots    & \vdots    & \vdots & \vdots \\
	x_{u}^{0} & x_{u}^{1} & \dots & x_{u}^{k-1} \\
}
\colvec{a_{0} \\ a_{1} \\ \vdots \\ a_{k-1}}. \label{eq:vandermonde}
\end{equation}
The $k$-independence of the polynomial family follows from properties of the Vandermonde matrix: every subset of $k$ rows is linearly independent. 
The problem with this construction is that the Vandermonde matrix is dense, resulting in an evaluation time of $\Omega(k)$ if we simply store the coefficients of the polynomial.
The lower bounds by Siegel \cite{siegel2004}, and later Larsen \cite{larsen2012}, as presented in Table 1, show that a data structure for evaluating a polynomial of degree $k-1$ using time $t < k$ must use space at least $k(u/k)^{1/t}$.   
Kedlaya and Umans~\cite{kedlaya2008} give an elegant upper bound for faster polynomial evaluation when allowing preprocessing,
but their space-time tradeoff is still far from the lower bound for $k$-independent functions.

The quest for $k$-independent families of functions with evaluation time $t < k$ can be viewed as attempts to construct compact representations of sparse matrices that fill the same role as the Vandermonde matrix.
We are interested in compact representations that support fast computation of the sparse row associated with an element $x \in [u]$. 
An example of a sparse matrix with these properties is the adjacency matrix of a bipartite expander graph with sufficiently strong expansion properties.
For the purposes of constructing $k$-independent hash functions we are primarily interested in expanders that are highly unbalanced.

\paragraph{Expander hashing.} Prior constructions of fast and highly random hash functions has followed Siegel's approach of combining expander graphs with tables of random words.
If $\Gamma$ is a $k$-unique expander graph (see Definition \ref{def:k-unique}) then we can construct a $k$-independent function by composing it with a simple tabulation function $h$.
This approach would yield optimal $k$-independent hash functions if we had access to explicit expanders with optimal parameters that could be evaluated in time proportional to the left outdegree.
Unfortunately, no explicit construction of a $k$-unique expander with optimal parameters is known.

Siegel~\cite{siegel2004} addresses this problem by storing a smaller randomly generated $k$-unique expander, say, one that covers a universe of size $u^{1/t}$.
By the $k$-independent hashing lower bound, if an expander with $|U| = u^{1/t}$ has degree $d$, then in order for it to be $k$-unique it must have a right hand side of size $|V| \geq k(u^{1/t}/k)^{1/d}$.
To give a space efficient construction of a $k$-unique expander that covers a universe of size $u$, Siegel repeatedly applies the Cartesian product to the graph.
Applying the Cartesian product $t$ times to a $k$-unique expander results in a graph that remains $k$-unique but with the left degree and size of the left and right vertex sets raised to the power $t$.
Using space $u^{1/t}$ to store an expander with degree $t$, it follows from the lower bound that the expander resulting from repeatedly applying the Cartesian product must have
\begin{equation*}
|V'| \geq (k(u^{1/t}/k)^{1/d})^{t} = k^{(1 - 1/d)t}u^{1/d}.
\end{equation*}
Setting $d = 1/\varepsilon$, the randomly generated $k$-unique expander that forms the basis of the construction has degree $O(1/\epsilon)$, leading to the expression in Table~1.
Since we need to store $|V'|$ random words in a table in order to create a $k$-independent hash function, 
Siegel's graph powering approach offers a space-time tradeoff that is far from the lower bound from our perspective where both $u$, $k$, and $t$ are parameters to the hash function.

Thorup~\cite{thorup2013} shows that, for the right choice of parameters, a simple tabulation hash function is likely to form a compact representation of a $k$-unique expander.
A simple tabulation function takes a string $x = (x_{1}, x_{2}, \dots, x_{c})$ of $c$ characters from some input alphabet $\bit{n} = \{0,1\}^{n}$, and returns a string of $d$ characters from some output alphabet $\bit{m} = \{0,1\}^{m}$.
The simple tabulation function $h : \bit{n}^{c} \to \bit{m}^{d}$ is evaluated by taking the exclusive-or of $c$ table-lookups 
\begin{equation*}
h(x) = h_{1}(x_{1}) \oplus h_{2}(x_{2}) \oplus \dots \oplus h_{c}(x_{c})
\end{equation*}
where $h_{i} : \bit{n} \to \bit{m}^{d}$ are random functions for $1 \leq i \leq c$.
The advantage of a simple tabulation function compared to a fully random function is that we only need to store the random character tables $h_{1}, h_{2}, \dots, h_{c}$.
Thorup is able to show that for~${d \geq 6c}$ a simple tabulation function is $k$-unique with a low probability of failure when $k \leq (2^m)^{1/5c}$. 
Setting $n = m$ and composing the $k$-unique expander resulting from a single application of simple tabulation with another simple tabulation function, 
Thorup first constructs a hash function with space usage~$u^{1/c}$, independence~$u^{\Omega(1/c^2)}$, and evaluation time~$O(c)$. 
He then presents a second trade-off with space~$u^{1/c}$, independence~$u^{\Omega(1/c)}$, and time~$O(c^{\log c})$ that comes from applying simple tabulation recursively to the output of a simple tabulation function. 
Similar to Siegel's upper bound, the space usage of Thorup's upper bounds with respect to $k$ is much larger than the lower bound as can be seen from Table~\ref{tab:results} where the space-time tradeoff of his results have been parameterized in terms of the independence $k$.\footnote{ 
It should be noted that Thorup's analysis is not tuned to optimize the polynomial dependence on $k$, and that he gives stronger concrete bounds for some realistic parameter settings.}

\paragraph{Explicit constructions.}
The literature on explicit constructions has mostly focused on optimizing the parameters of the expander, 
with the evaluation time of the neighbor function being of secondary interest, as long as it is bounded by $\poly \log u$.
As can be seen from Siegel's cell probe lower and upper bounds, optimal constructions of $k$-independent hash functions have evaluation time in the range $t = 1$ to $t = \log u$.
Therefore, an explicit construction, even if we had one with optimal parameters, would without further guarantees on the running time not be enough to solve our problem of constructing efficient expanders. 
Here we briefly review the construction given by Guruswami et al. \cite{guruswami2009}.
It is, to our knowledge, currently the best explicit construction of unbalanced bipartite expanders in terms of the parameters of the graph.
Their construction and its analysis is, similarly to the polynomial hash function in equation \eqref{eq:vandermonde}, algebraic in nature and inspired by techniques from coding theory, 
in particular Parvaresh-Vardy codes and related list-decoding algorithms~\cite{parvaresh2005}.
In their construction, a vertex $x$ is identified with its Reed-Solomon message polynomial over a finite field $\GF$.
The $i$th neighbor of $x$ is found by taking a sequence of powers of the message polynomial over an extension field, evaluating each of the resulting polynomials in the $i$th element of $\GF$, and concatenating the output.
In contrast, the constructions presented in this paper only use the subset of standard word RAM instructions that can be implemented in $AC^{0}$.
In Table 1 we have assumed that we can evaluate their neighbor function in constant time as a conservative lower bound on the performance of their construction in the word RAM model. 
Other highly unbalanced explicit constructions given in \cite{capalbo2002, tashma2007} offer a tradeoff where either one of $d$ or $|V|$ is quasipolynomial in the lower bound.
In comparison, the construction by Guruswami et al.~is polynomial in both of these parameters. 
\section{Our constructions}
In this section we present three randomized constructions of efficient expanders in the word RAM model.
Each construction offers a different tradeoff between space, time, and the probability of failure.
We present our constructions as data structures, with the randomness generated by the model during an initialization phase.
The initialization time of our data structures is always bounded by their space usage, and to simplify the exposition we therefore only state the latter.
Alternatively, our constructions could be viewed directly as randomized algorithms, taking as input a list of parameters, a random seed, and a vertex $x \in [u]$ and returning the list of neighbors of $x$.
The hashing corollaries presented in Table~1 follow directly from our three main theorems using Siegel's expander hashing technique.
\subsection{Model of computation}
The algorithms presented in this paper are analyzed in the standard word RAM model with word size $w$ as defined by Hagerup~\cite{hagerup1998}, 
modeling what can be implemented in a standard programming language like \texttt{C}~\cite{kernighan1988}.
In order to show how our algorithms benefit from word-level parallelism we use~$w$ as a parameter in the analysis.
To simplify the exposition we impose the natural restriction that, for a given choice of parameters to a data structure, the word size is large enough to address the space used by the data structure.
In other words, our results are stated with $w$ as an unrestricted parameter, but are only valid when we actually have random access in constant time.

The data structures we present require access to a source of randomness in order to initialize the character tables of simple tabulation functions.
To accomodate this we augment the model with an instruction that uses constant time to generate a uniformly random and independent integer in $[r]$ where $r \leq 2^{w}$.
We note that our constructions use only the subset of arithmetic instructions required for evaluating a simple tabulation function, i.e, standard bit manipulation instructions, integer addition, and subtraction.
Our results therefore hold in a version of the word RAM model that only uses instructions that can be implemented in $AC^{0}$, known in the literature as the restricted model~\cite{hagerup1998} or the Practical RAM~\cite{miltersen1996}. 
\subsection{Notation and definitions}
Let $\bit{n} = \{0,1\}^{n}$ denote the alphabet of $n$-bit strings, and let $x = (x_{1}, x_{2}, \dots, x_{c}) \in \bit{n}^{c}$ denote a string of $n$-bit characters of length $c$.
We define a concatenation operator $\conc$ that takes as input two characters $x \in \bit{n}$ and $y \in \bit{m}$, and concatenates them to form $x \conc y \in \bit{n + m}$.
The concatenation operator can also be applied to strings of equal length where it performs component-wise concatenation.
Given strings $x \in \bit{n}^{c}$ and $y \in \bit{m}^{c}$ the concatenation $x \conc y$ is an element of $\bit{n+m}^{c}$ with the $i$th component of $x \conc y$ defined by $(x \conc y)_{i} = x_{i} \conc y_{i}$.
We also define a prefix operator. 
Given $x \in \bit{n}$ and a positive integer $m$, in the case where $m \leq n$ we use $x[m] \in \bit{m}$ to denote the $m$-bit prefix of $x$.
In the case where $m > n$ we pad the prefix such that $x[m] \in \bit{m}$ denotes $x[n] \conc 0^{m-n}$ where $0^{m-n}$ is the string of $m-n$ bits all set to $0$.

We will present word RAM data structures that represent functions of the form~${\Gamma : \bit{n}^{c} \to \bit{m}^{d}}$.
The function $\Gamma$ defines a $d$-regular bipartite graph with input set $\bit{n}^{c}$ and output set $\{1, 2, \dots, d\} \times \bit{m}$.
For $S \subseteq \bit{n}^{c}$ we overload $\Gamma$ and define $\Gamma(S) = \{ (i, \Gamma(x)_{i}) \mid x \in S \}$, i.e., $\Gamma(S)$ is the set of outputs of $S$.

We are interested in constructing functions where every subset $S$ of inputs of size at most $k$ contains an input that has many unique neighbors, formally:
\begin{definition} \label{def:k-unique}
Let $\Gamma : \bit{n}^{c} \to \bit{m}^{d}$ be a function satisfying the following property:   
\begin{equation*}
	\forall S \subseteq \bit{n}^{c}, |S| \leq k, \exists x \in S : |\Gamma(\{x\}) {\setminus} \Gamma(S {\setminus} \{ x \})| > l.
\end{equation*}
Then, for $l = 0$ we say that $\Gamma$ is \emph{$k$-unique}. 
If further $l = d/2$ we say that $\Gamma$ is \emph{$k$-majority-unique}.
\end{definition}

For completeness we define the concept of $k$-independence:
\begin{definition} \label{def:kindependence}
Let $k$ be a positive integer and let $\mathcal{F}$ be a family of functions from $U$ to $R$.
We say that~$\mathcal{F}$ is a \mbox{\emph{$k$-independent}} family of functions if,
for every choice of $k$ distinct keys~${x_{1}, \dots, x_{k}}$ and arbitrary values~${y_{1}, \dots, y_{k}}$, 
then, for $f$ selected uniformly at random from $\mathcal{F}$ we have that 
\begin{equation*}
\Pr[f(x_{1}) = y_{1} \land f(x_{2}) = y_{2} \land \dots \land f(x_{k}) = y_{k}] = |R|^{-k}.
\end{equation*}
We say that $f$ is $k$-independent when it is selected uniformly at random from a family of $k$-independent functions.
\end{definition}
Simple tabulation functions are an important tool in our constructions. 
Our data structures can be made to consist entirely of simple tabulation functions and our evaluation algorithms can be viewed as a sequence of adaptive calls to this collection of simple tabulation functions.  
\begin{definition}
Let $(R, \oplus)$ denote an abelian group.
A \emph{simple tabulation function} $h : \bit{n}^{c} \to R$ with $k$-independent character tables is defined by
\begin{equation*}
h(x) = \bigoplus_{i=1}^{c}h_{i}(x_{i}) 
\end{equation*}
where each \emph{character table} $h_{i} : \bit{n} \to R$ is a $k$-independent function. 
\end{definition}
In this paper we consider simple tabulation functions with character tables that operate either on bit strings under the exclusive-or operation, $R = (\bit{m}, \oplus)$, 
or on sets of non-negative integers modulo some integer $r$, $R = ([r], +)$.
For our constructions of $k$-independent functions we will generally assume that $r = u^{O(1)}$.
\subsection{From k-uniqueness to k-independence}
In his seminal paper Siegel~\cite{siegel2004} showed how a $k$-unique function can be combined with a table of random elements in order to define a $k$-independent family of functions.
In his paper on the expansion properties of tabulation hash functions, 
Thorup~\cite[Lemma 2]{thorup2013} used a slight variation of Siegel's technique that makes use of the position-sensitive structure of the bipartite graph defined by $\Gamma : \bit{n}^{c} \to \bit{m}^{d}$.
This is the version we state here.
\begin{lemma}[Siegel {\cite{siegel2004}}, Thorup {\cite{thorup2013}}] \label{lem:expanderhashing}
Let $\Gamma : \bit{n}^{c} \to \bit{m}^{d}$ be $k$-unique and let $h : \bit{m}^{d} \to R$ be a simple tabulation function with $k$-independent character tables.
Then $h \circ \Gamma$ defines a family of $k$-independent functions.
We sample a function from the family by sampling the character tables of $h$.
\end{lemma}
\subsection{From k-independence to k-uniqueness}
A $k$-independent function has the same properties as a fully random function when considering $k$-subsets of inputs.
We can therefore use the standard analysis of randomized constructions of bipartite expanders to show that, for the right choice of parameters, a $k$-independent function is likely to be $k$-unique.
For completeness we provide an analysis here.
In our exposition it will be convenient to parameterize the $k$-uniqueness or $k$-majority-uniqueness of our constructions in terms of a positive integer~$\kappa$ such that $k = 2^{\kappa}$.
\begin{lemma} \label{lem:expanderparameters}
For every choice of positive integers $c$, $n$, $\kappa$ let~$\Gamma : \bit{n}^{c} \to \bit{m}^{d}$ be a $2^{\kappa}$-independent function.
Then,
\begin{itemize}
\item[--] for $m \geq n + \kappa + 1$ and $d \geq 4c$ we have that $\Gamma$ is $2^{\kappa}$-unique with probability at least $1 - 2^{-dn/2}$.
\item[--] for $m \geq n + \kappa + 4$ and $d \geq 8c$ we have that $\Gamma$ is $2^{\kappa}$-majority-unique with probability at least $1 - 2^{-dn/4}$.
\end{itemize}
\end{lemma}
\begin{proof}
We will give the proof for $k$-majority-uniqueness. The proof for $k$-uniqueness uses the same technique.
By a standard argument based on the pigeonhole principle, for~$\Gamma$ to be $k$-majority-unique it suffices that for all $S \subseteq \bit{n}^{c}$ with~$|S| \leq k$ we have that $|\Gamma(S)| > (3/4)d|S|$.
Given that $\Gamma$ is $k$-independent, we will now bound the probability that there exists a subset $S$ with $|S| \leq k$ such that $|\Gamma(S)| \leq (3/4)d|S|$.
For every pair of sets $(S, B)$ satisfying that $S \subseteq \bit{n}^{c}$ with $|S| \leq k$ and $B \subseteq \{1, 2, \dots, d \} \times \bit{m}$ with $|B| = (3/4)d|S|$, 
the probability that $\Gamma(S) \subseteq B$ is given by $\prod_{i = 1}^{d}(|B_{i}|/2^{m})^{|S|}$ where $B_{i} = \{ (i, y) \in B \}$.
By the inequality of the arithmetic and geometric means we have that
\begin{equation*}
\prod_{i = 1}^{d}\left(\frac{|B_{i}|}{2^{m}}\right)^{|S|} \leq \left(\frac{|B|}{d2^{m}}\right)^{d|S|}.
\end{equation*}
This allows us to ignore the structure of $B$, and obtain a union bound that matches that of the standard non-compartmentalized probabilistic construction of bipartite expanders.
The probability that $\Gamma$ fails to be $k$-majority-unique is upper bounded by
\begin{equation*}
\sum^{k}_{i=2} \binom{2^{cn}}{i} \binom{d2^{m}}{(3/4)di} \left( \frac{(3/4)di}{d2^{m}} \right)^{di}. 
\end{equation*}
For every choice of positive integers $c$, $n$, $\kappa$, for ${m \geq  n + \kappa + 4}$ and $d \geq 8 c $ we get a probability of failure less than $2^{-2cn}$.
\end{proof}
\subsection{A recursive construction of a k-unique function}
In this section we introduce a recursive construction of a $k$-unique function of the form $\Gamma : \bit{n}^{c} \to \bit{m}^{d}$.
We obtain $\Gamma$ as the last in a sequence $\Gamma_{1}, \Gamma_{2}, \dots, \Gamma_{c}$ of $k$-unique functions~${\Gamma_{i}: \bit{n}^{i} \to \bit{m}^{d}}$.
Each $\Gamma_{i}$ for $i>1$ is defined in terms of $\Gamma_{i-1}$.
At the bottom of the recursion we tabulate a $k$-independent function $\Gamma_{1} : \bit{n} \to \bit{m}^{d}$.
In the general step we apply $\Gamma_{i-1}$ to the length $i-1$ prefix of the key $(x_{1}, x_{2}, \dots, x_{i-1})$, concatenate the result vector component-wise with the $i$th character $x_{i}$, 
and apply a simple tabulation function $h_{i} : \bit{m + n}^{d} \to \bit{m}^{d}$.
The recursion is therefore given by
\begin{equation}
\Gamma_{i}((x_1, x_2, \dots, x_i)) = h_{i}(\Gamma_{i-1}((x_1, x_2, \dots, x_{i-1})) \conc (x_i, x_i, \dots, x_i)) \label{eq:simple}
\end{equation}
where $(x_i, x_i, \dots, x_i)$ denotes the string of $x_i$ repeated $d$ times.
The following theorem summarizes the properties of $\Gamma$ in the word RAM model.

\begin{theorem} \label{thm:simple}
There exists a randomized data structure that takes as input positive integers $c$, $n$, $\kappa$ 
and initializes a function $\Gamma : \bit{n}^{c} \to \bit{n + \kappa + 1}^{4c}$.
In the word RAM model with word size $w$ the data structure satisfies the following:
\begin{itemize}
	\item The space usage is $O(2^{2n + \kappa}c^{3}(n + \kappa)/w)$ words.  
	\item The evaluation time of $\Gamma$ is $O(c^{2} + c^{3}(n + \kappa)/w)$.
	\item $\Gamma$ is $2^{\kappa}$-unique with probability at least $1 - 2^{-cn}$.
\end{itemize}
\end{theorem}
\begin{proof}
Set $m = n + \kappa + 1$ and $d = 4c$.
We initialize $\Gamma$ by tabulating a $k$-independent function $\Gamma_{1} : \bit{n} \to \bit{m}^{d}$ and simple tabulation functions $h_{2}, h_{3}, \dots, h_{c} \colon \bit{m + n}^{d} \to \bit{m}^{c}$.
In total we need to store $c$ functions that each have $O(c)$ character tables with $O(2^{2n + \kappa})$ entries of $O(c(n + \kappa))$ bits.
The space usage is therefore $O(2^{2n + \kappa}c^{3}(n + \kappa)/w)$ words. 
The same bound holds for the time to initialize the data structure. 

The evaluation time of $\Gamma$ can be found by considering the recursion given in equation \eqref{eq:simple}. 
At each of the $c$ steps we perform $O(c)$ lookups and take the exclusive-or of $O(c)$ bit strings of length $O(c(n + \kappa))$.
The total evaluation time is therefore $O(c^{2} + c^{3}(n + \kappa)/w)$.

Conditioned on $\Gamma_{i-1}$ being $k$-unique, it is easy to see that $\Gamma_{i}$ is $k$-unique, and by Lemma \ref{lem:expanderhashing} we have that $\Gamma_{i}$ is $k$-independent.
For our choice of parameters, according to Lemma \ref{lem:expanderparameters} the probability that $\Gamma_{i}$ fails to be $k$-unique is less than $2^{-2cn}$.
Therefore, $\Gamma$ is $k$-unique if $\Gamma_{1}, \Gamma_{2}, \dots, \Gamma_{c}$ are $k$-unique.
This happens with probability at least $1 - c2^{-2cn} \geq 1 - 2^{-cn}$. 
\end{proof}

Combining Theorem \ref{thm:simple} and Lemma \ref{lem:expanderhashing}, we get $k$-independent hashing in the word RAM model.
We state our result in terms of a data structure that represents a family of functions~$\mathcal{F}$.
The family is defined as in Lemma \ref{lem:expanderhashing} and is represented by a particular instance of a function $\Gamma$, constructed using Theorem \ref{thm:simple}, 
together with the parameters of a family of simple tabulation functions.

\begin{corollary} \label{cor:simple}
\family
\begin{itemize}
	\item The space used to represent $\mathcal{F}$, as well as a function $f \in \mathcal{F}$, is $O(ku^{1/t}t^{2}(\log u + t \log k)/w)$ words.
	\item The evaluation time of $f$ is $O(t^{2} + t^{2}(\log u +  t \log k)/w)$.
	\item $\mathcal{F}$ is a $k$-independent family with probability at least $1 - 1/u$.
\end{itemize}
\end{corollary}
\begin{proof}
We apply Theorem \ref{thm:simple}, setting $c = 2t$, $n = \lceil (\log u)/2t \rceil$, $\kappa = \lceil \log k \rceil$.
This gives a function ${\Gamma : \bit{n}^{c} \to \bit{n + \kappa + 1}^{4c}}$ that is $k$-unique over $[u]$ with probability at least $1 - 1/u$.
To sample a function from the family we follow the approach of Lemma \ref{lem:expanderhashing} and compose $\Gamma$ with a simple tabulation function $h : \bit{n + \kappa + 1}^{4c} \to [r]$.
The space used to store $\Gamma$ follows directly from Theorem \ref{thm:simple} and dominates the space used by $h$.
Similarly, the evaluation time of $h \circ \Gamma$ is dominated by the time it takes to evaluate $\Gamma$.
\end{proof}

\begin{remark}
For every integer $\tau \geq 1$ we can construct a family $\mathcal{F}^{(\tau)}$ that is $k$-independent with probability at least ${1 - u^{-\tau}}$ at the cost of increasing the space usage and evaluation time by a factor $\tau$.
The family is defined by
\begin{equation*}
\mathcal{F}^{(\tau)} = \{ f = \bigoplus_{i = 1}^{\tau} f_{i} \mid f_{i} \in \mathcal{F}_{i} \}
\end{equation*} 
where each $\mathcal{F}_{i}$ is constructed independently.
\end{remark}

\begin{remark}
The recursion in equation \eqref{eq:simple} is well suited for sequential evaluation where the task is to evaluate $\Gamma$ in an interval of $[u]$ in order to generate a sequence of $k$-independent random variables.
To see this, note that once we have evaluated $\Gamma$ on a key $x = (x_{1}, x_{2}, \dots, x_{c})$, a change in the last character only changes the last step of the recursion.   
It follows that we can generate $k$-independent variables using amortized time $O(t)$ and space close to $O(ku^{1/t})$.
To our knowledge, this presents the best space-time tradeoff for the generation of $k$-independent variables in the case where we do not have access to multiplication over a suitable finite field as in \cite{christiani2014}.  
\end{remark}
\subsection{A divide and conquer recursion}
In this section we introduce a data structure for representing a $k$-majority-unique function that offers a faster evaluation time at the cost of using more space.
As in the construction from Theorem~\ref{thm:simple} we use the technique of alternating between expansion and independence, 
but rather than reading a single character at the time, we view the key as composed of two characters $x =(x_{1}, x_{2})$ and recurse on each. 
In the previous section we increased the size of the domain of our $k$-unique function by concatenating part of the key.
If we use only a few large characters this approach becomes very costly in terms of the space required to store the simple tabulation function $h_i$ in the composition $h_{i}(\Gamma_{i-1}((x_1, x_2, \dots, x_{i-1})) \conc (x_i, x_i, \dots, x_i))$.
To be able to efficiently recurse on large characters we show that the function $\Upsilon((x_{1}, x_{2})) = \Gamma(x_{1}) \conc \Gamma(x_{2})$ is $k$-unique when $\Gamma$ is $k$-majority-unique.   
\begin{lemma} \label{lem:interleaving}
Let $\Gamma : \bit{n}^{c} \to \bit{m}^{d}$ be a $k$-majority-unique function.
Then the function $\Upsilon : \bit{n}^{c} \times  \bit{n}^{c} \to \bit{2m}^{d}$ defined by $\Upsilon((x_{1}, x_{2})) = \Gamma(x_{1}) \conc \Gamma(x_{2})$ is $k$-unique.
\end{lemma}
\begin{proof}
Let $x = (x_{1}, x_{2})$ denote an element of $\bit{n}^{c} \times \bit{n}^{c}$. 
For $S \subseteq \bit{n}^{c} \times  \bit{n}^{c}$ define $S_{1,a} = \{ x \in S \mid x_{1} = a \}$.
The following holds for every $x = (x_{1}, x_{2}) \in S$.
\begin{align} 
|\Upsilon(\{x\}) {\setminus} \Upsilon(S {\setminus} \{x\})| &= |\Upsilon(\{x\}) {\setminus} (\Upsilon(S {\setminus} S_{1,x_{1}}) \cup \Upsilon(S_{1,x_{1}} {\setminus} \{x\}))| \nonumber \\
&= |(\Upsilon(\{x\}) {\setminus} \Upsilon(S {\setminus} S_{1,x_{1}})) \cap (\Upsilon(\{x\}) {\setminus} \Upsilon(S_{1,x_{1}} {\setminus} \{x\}))| \nonumber \\
&\geq |\Upsilon(\{x\}) {\setminus} \Upsilon(S {\setminus} S_{1,x_{1}})| + |\Upsilon(\{x\}) {\setminus} \Upsilon(S_{1,x_{1}} {\setminus} \{x\})| \nonumber \\
   &\qquad- |\Upsilon(\{x\})|. \label{eq:interleavingbound}
\end{align}
We will show that for every $S \subseteq \bit{n}^{c} \times \bit{n}^{c}$ with $|S| \leq k$ there exists a key $(x_{1}, x_{2}) \in S$ such that $|\Upsilon(\{x\}) {\setminus} \Upsilon(S {\setminus} \{x\})| > 0$. 
We begin by choosing the first component of $x$.
Let $\pi_{j}(S) = \{ x_{j} \mid x \in S \}$ denote the set of $j$th components of elements of $S$. 
By the $k$-majority-uniqueness of $\Gamma$, considering the set $\pi_{1}(S)$, we have that
\begin{equation*}
\exists x_{1} \in \pi_{1}(S) : \forall x \in S_{1,x_{1}} : |\Upsilon(\{x\}) {\setminus} \Upsilon(S {\setminus} S_{1,x_{1}})| > d/2.
\end{equation*}
Fix $x_{1}$ with this property and consider the choice of $x_{2}$.
By the $k$-majority-uniqueness of $\Gamma$, considering the set $\pi_{2}(S)$, we have that
\begin{equation*}
 \forall x_{1} \in \pi_{1}(S) : \exists x_{2} \in \pi_{2}(S_{1,x_{1}}) : |\Upsilon(\{x\}) {\setminus} \Upsilon(S_{1,x_{1}} {\setminus} \{x\})| > d/2.
\end{equation*}
We can therefore always find a key $(x_{1}, x_{2}) \in S$ such that both $|\Upsilon(\{x\}) {\setminus} \Upsilon(S {\setminus} S_{1,x_{1}})| > d/2$, $|\Upsilon(\{x\}) {\setminus} \Upsilon(S_{1,x_{1}} {\setminus} \{x\})| > d/2$ are satisfied.
The result follows from equation \eqref{eq:interleavingbound} where we use the fact that $|\Upsilon(\{x\})| = d$.
\end{proof}

We will give a recursive construction of a $k$-majority-unique function of the form $\Gamma_{i} : \bit{n}^{2^{i}} \to \bit{m}^{2^{i + 3}}$.
Let $h_{i} : \bit{2m}^{2^{i +2}} \to \bit{m}^{2^{i+3}}$ be a simple tabulation function.
For $i > 0$ the recursion takes the following form.  
\begin{equation*}
	\Gamma_{i}((x_1, \dots, x_{2^i})) = h_{i}(\Gamma_{i-1}((x_1, \dots, x_{2^{i-1}})) \conc \Gamma_{i-1}((x_{2^{i-1} + 1}, \dots, x_{2^i}))).
\end{equation*}
At the bottom of the recursion we tabulate a $k$-independent function $\Gamma_{0}$.

\begin{theorem} \label{thm:majorityrecursion}
There exists a randomized data structure that takes as input positive integers $\lambda$, $n$, $\kappa$ and initializes a function $\Gamma : \bit{n}^{2^{\lambda}} \to \bit{n + \kappa + 4}^{2^{\lambda + 3}}$.
In the word RAM model with word size $w$ the data structure satisfies the following:
\begin{itemize}
	\item The space usage is $O(2^{2(n + \kappa + \lambda)}(n + \kappa)/w)$ words. 
	\item The evaluation time of $\Gamma$ is $O(2^{\lambda}(\lambda + 2^{\lambda}(n + \kappa)/w))$.
	\item $\Gamma$ is $2^{\kappa}$-majority-unique with probability at least $1 - 2^{-2n + 1}$.
\end{itemize}
\end{theorem}
\begin{proof}
Let $m = n + \kappa + 4$.
We initialize $\Gamma$ by tabulating $\Gamma_{0}$ and the character tables of the simple tabulation functions $h_{1}, h_{2}, \dots, h_{\lambda}$ where $h_{i} : \bit{2m}^{2^{i +2}} \to \bit{m}^{2^{i+3}}$.
In total we have $O(2^{\lambda})$ character tables with $O(2^{2(n + \kappa)})$ entries of $O(2^{\lambda}(n + \kappa))$ bits, resulting in the space bound.

Let $T(i)$ denote the evaluation time of $\Gamma_{i}$.
For $i = 0$ we can evaluate $\Gamma_{0}$ by performing a single lookup in $O(1)$ time.
For $i > 0$ evaluating $h_{i} \circ (\Gamma_{i-1} \conc \Gamma_{i-1})$ takes two evalutions of $\Gamma_{i-1}$ followed by evaluating $h_{i}$ on their concatenated output using $O(2^{i}(1 + 2^{i}(n + \kappa)/w))$ operations.
The recurrence takes the form
\begin{equation*}
T(i) \leq 
\begin{cases} 2T(i-1) + O(2^{i}(1 + 2^{i}(n + \kappa)/w)) & \mbox{if } i > 0 \\ 
O(1) & \mbox{if } i = 0 
\end{cases}
\end{equation*}
The solution to the recurrence is $O(2^{i}(i + 2^{i}(n + \kappa)/w))$.

We now turn our attention to the probability that $\Gamma_{i} = h_{i} \circ (\Gamma_{i-1} \conc \Gamma_{i-1})$ fails to be $k$-majority-unique.
Conditional on $\Gamma_{i-1}$ being $k$-majority-unique, by Lemma \ref{lem:interleaving} we have that $(\Gamma_{i-1} \conc \Gamma_{i-1})$ is $k$-unique and composing it with $h_{i}$ gives us a $k$-independent function.
For our choice of parameters, according to Lemma \ref{lem:expanderparameters} the probability that $\Gamma_{i}$ fails to be $k$-majority-unique is less than $2^{-2^{i + 1}n}$.
Therefore, $\Gamma$ is $k$-majority-unique if $\Gamma_{0}, \Gamma_{1}, \dots, \Gamma_{\lambda}$ are $k$-majority-unique.
This happens with probability at least $1 - \sum_{i=0}^{\lambda} 2^{-2^{i + 1}n} \geq 1 - 2^{-2n + 1}$. 
\end{proof}

\begin{remark}
	The recursion behind Lemma \ref{thm:majorityrecursion} is well suited for parallelization. 
If we have $c$ processors working in lock-step with some small shared memory we can evaluate $\Gamma$ with domain $\bit{n}^{c}$ in time~$O(c)$. 
\end{remark}

\begin{corollary} \label{cor:majorityrecursion}
\family
\begin{itemize}
	\item The space used to represent $\mathcal{F}$, as well as a function $f \in \mathcal{F}$, is $O(k^{2}u^{1/t}t(\log u + t\log k)/w)$ words.
	\item The evaluation time of $f$ is $O(t \log t + t (\log u + t \log k)/w)$.  
	\item $\mathcal{F}$ is $k$-independent with probability at least $1 - u^{-1/t}$.
\end{itemize}
\end{corollary}
\begin{proof}
Apply Theorem \ref{thm:majorityrecursion} with parameters $\lambda = \lceil \log t \rceil + 1$, $n = \lceil (\log u) / 2t \rceil + 1$, and $\kappa = \lceil \log k \rceil$.
This gives is a function $\Gamma$ that is $k$-unique over $[u]$ with probability at least $1 - u^{-1/t}$.
The family $\mathcal{F}$ is defined by the composition of $\Gamma$ with a suitable simple tabulation function following the approach of Lemma \ref{lem:expanderhashing}.
\end{proof}
\begin{remark}
If we have access to a standard $t$-independent polynomial hash function that outputs elements of $[r]$ then the error probability can be reduced to $1/u$ at no additional cost provided that the evaluation time of such a function does not exceed the evaluation time of $f \in \mathcal{F}$. 
The trick is to add (modulo $r$) the output of such a guaranteed $t$-independent function to the output of our original function from $\mathcal{F}$.
The reason is that most of the error probability comes from the ``small sets'' in the union bound in Lemma \ref{lem:expanderparameters} and these will now be dealt with by our $t$-independent function.
\end{remark}
\subsection{Balancing time and space}
Theorem \ref{thm:simple} yielded a $k$-unique function over $\bit{n}^{c}$ with an evaluation time of about $O(c^{2})$ while using linear space in $k$.
Theorem \ref{thm:majorityrecursion} resulted in an evaluation time of about~$O(c \log c)$, using quadratic space in $k$.
Under a mild restriction on $k$, the two techniques can be combined to obtain an evaluation time of $O(c \log c)$ and linear space in $k$.
We take the construction from Theorem \ref{thm:majorityrecursion} as our starting point, 
but instead of tabulating the character tables of $h_{1}, \dots, h_{\lambda}$ we replace them with more space efficient $k$-independent functions that we construct using Theorem \ref{thm:simple}.

\begin{theorem} \label{thm:linear}
There exists a randomized data structure that takes as input positive integers $\lambda$, $n$, $\kappa = O(n)$ 
and initializes a function $\Gamma : \bit{n}^{2^{\lambda}} \to \bit{n + \kappa + 4}^{2^{\lambda + 3}}$.
In the word RAM model with word size $w$ the data structure satisfies the following:
\begin{itemize}
	\item The space usage is $O(2^{n + \kappa + 2\lambda}n/w)$ words.  
	\item The evaluation time of $\Gamma$ is $O(2^{\lambda}(\lambda + 2^{\lambda}n/w))$.
	\item $\Gamma$ is $2^{\kappa}$-majority-unique with probability at least $1 - 2^{-n + 1}$.
\end{itemize}
\end{theorem}
\begin{proof}
At the top level, the recursion underlying $\Gamma$ takes the same form as in Theorem \ref{thm:majorityrecursion}.
\begin{equation*}
\Gamma_{i} = h_{i} \circ (\Gamma_{i-1} \conc \Gamma_{i-1}).
\end{equation*}
The functions $h_{i} : \bit{2m}^{2^{i+2}} \to \bit{m}^{2^{i+3}}$ are simple tabulation functions with $m = n + \kappa + 4$.
Each $h_{i}$ is constructed from~$2^{i+2}$ character tables $h_{i,j} : \bit{2m} \to \bit{m}^{2^{i+3}}$.
Theorem~\ref{thm:majorityrecursion} only assumes that the character tables $h_{i,j}$ are $k$-independent functions.
We will apply Theorem~\ref{thm:simple} to construct a function~$\Upsilon$ that we for each character table $h_{i,j}$ compose with a simple tabulation function $g_{i,j}$ in order to construct $h_{i,j}$.
By the restriction that $\kappa = O(n)$ we have that $m = O(n)$.
We set the parameters of~$\Upsilon$ to $\hat{c} = O(1)$, $\hat{n} = \lceil n/2 \rceil$, $\hat{\kappa} = \kappa$ such that $\bit{2m}$ can be embedded in $\bit{\hat{n}}^{\hat{c}}$.
Furthermore,~$\Upsilon$ uses~$O(2^{n + \kappa}n/w)$ words of space, can be evaluated in $O(1)$ operations, and is $k$-unique with probability at last ${1 - 2^{-n-1}}$.
Because $\Upsilon$ has $O(1)$ output characters, the time to evaluate~$h_{i,j} = g_{i,j} \circ \Upsilon$ is no more than a constant times the word length of the output of $h_{i,j}$.
The time to evaluate $\Gamma$ therefore only increases by a constant factor compared to the evaluation time in Theorem~\ref{thm:majorityrecursion}.

The probability of failure of $\Gamma$ to be $k$-majority-unique is the same as in Theorem \ref{thm:majorityrecursion}, provided that $\Upsilon$ does not fail to be $k$-unique.
This gives a total probability of failure of less than $2^{-2n + 1} + 2^{-n -1} < 2^{-n + 1}$.

We only store a single $\Upsilon$ and the character tables of $g_{i,j}$ that we use to simulate the character tables $h_{i,j}$.
From the parameters of $\Upsilon$ we have that $g_{i,j}$ uses $O(1)$ character tables with $O(2^{n+\kappa})$ entries of $O(2^{i}n/w)$ words.
The space usage is dominated by the $O(2^{\lambda})$ character tables of $h_{\lambda}$ that use space $O(2^{n + \kappa + 2\lambda}n/w)$ in total.
\end{proof}

\begin{corollary} \label{cor:linear}
There exists a randomized data structure that takes as input positive integers $u$, $r$, $t$, $k = u^{O(1/t)}$ and selects a family of functions $\mathcal{F}$ from $[u]$ to $[r]$. 
In the word RAM model with word length $w$ the data structure satisfies the following:
\begin{itemize}
	\item The space used to represent $\mathcal{F}$, as well as a function $f \in \mathcal{F}$, is $O(ku^{1/t}t(\log u)/w)$ words.
	\item The evaluation time of $f$ is $O(t \log t + t(\log u)/w)$.
	\item $\mathcal{F}$ is $k$-independent with probability at least $1 - u^{-1/t}$.
\end{itemize}
\end{corollary}
\begin{proof}
Apply Theorem \ref{thm:linear} with parameters $\lambda = \lceil \log t \rceil$, $n = \lceil (\log u) / t \rceil + 1$, and $\kappa = \lceil \log k \rceil$.
This gives is a function $\Gamma$ that is $k$-unique over $[u]$ with probability at least $1 - u^{-1/t}$.
The family $\mathcal{F}$ is defined by the composition of $\Gamma$ with a suitable simple tabulation function following the approach of Lemma \ref{lem:expanderhashing}.
\end{proof}
\subsection{An improvement for space close to $k$}\label{sec:close-to-k-improvement}
In this section we present a different space efficient version of the divide-and-conquer recursion.
The new recursion is based on an extension of the ideas behind the graph product from Lemma \ref{lem:interleaving}.
In Lemma \ref{lem:interleaving} we use expansion properties over subsets of size $k$ and concatenate the output characters of $\Gamma$, resulting in an output domain of size at least $k^{2}$.
By using stronger expansion properties and modifying our graph concatenation product to fit the structure of the key set, we are able to reduce the space usage at the cost of using more time.
We now introduce a property of $d$-regular bipartite graph that we call $k$-super-majority-uniqueness: 
More than half the input vertices have more than $d/2$ unique neighbors.
\begin{definition} \label{def:k-super-majority-unique}
Let $\Gamma : \bit{n}^{c} \to \bit{m}^{d}$ be a function satisfying the following property:   
\begin{equation*}
	\forall S \subseteq \bit{n}^{c}, |S| \leq k \colon |\{ x \in S \mid  |\Gamma(\{x\}) {\setminus} \Gamma(S {\setminus} \{ x \})| > d/2\}| > |S|/2.
\end{equation*}
Then we say that $\Gamma$ is \emph{$k$-super-majority-unique}.
\end{definition}
\paragraph{Overview of approach.}
The exposition is quite technical but the underlying idea is simple:
Suppose we have a set $S$ of $k$ keys of the form $(x_1, x_2)$.
Then the set of first components $\pi_1(S) = \{ x_1 \mid (x_1, x_2) \in S \}$ has size $|\pi_{1}(S) = k^\varepsilon$ for some $\varepsilon \in (0, 1)$.
By the pigeonhole principle there must exist a subset $A \subseteq \pi_{1}(S)$ of size at at least $k^{\varepsilon}/2$ such that for every $a \in A$ we have that $|\{ (a, x_2) \in S\}| \leq 2k^{1-\varepsilon}$, otherwise we would have more than $k$ keys!
Therefore, the function $\Gamma((x_1, x_2)) = \Gamma_1(x_1) \conc \Gamma_2(x_2)$ is $k$-unique on the set $S$ provided that $\Gamma_1$ is $k^{\varepsilon}$-super-majority-unique and $\Gamma_2$ is $2k^{1-\varepsilon}$-majority-unique.
By forming these combinations for $O(\log k)$ choices of $\varepsilon$ we can cover all possible structures of two-character keys, ensuring that $\Gamma$ is $k$-unique while keeping the size of the right-hand side close to $O(k)$.
For a given choice of $\varepsilon$ we obtain $\Gamma_1$ and $\Gamma_2$ by componentwise concatenation of differently sized prefixes of the output a $k$-independent function.
Based on this idea we can then build a divide-and-conquer recursion similar to the one behind Theorem \ref{thm:majorityrecursion}.
The technical details behind this idea is deferred to Appendix \ref{app:prefix}
\section{Conclusion}
We have presented new constructions of $k$-independent hash functions that come close to Siegel's lower bound on the space-time tradeoff for such functions.
An interesting open problem is whether the gap to the lower bound can be closed.
From the perspective of efficient expanders it would be very interesting to achieve space $o(k)$ while preserving computational efficiency.
Of course, such a result is not possible via $k$-independence.

\section*{Acknowledgements}
We thank the STOC reviewers for insightful comments that helped us improve the exposition.

\section{Appendix: Details behind the prefix technique} \label{app:prefix}

The following lemma shows how we can construct a $k$-unique function over $U^{2}$ from a set of $k$-super-majority-unique functions over $U$.
\begin{lemma} \label{lem:superconcatenation}
Let $q$ be a positive integer. 
For $j = 1, 2, \dots, q$ let $\Gamma_{j} : \bit{\kappa}^{c} \to \bit{m_{j}}^{d}$ be $\min(2k^{j/q}, k)$-super-majority-unique and set $m = \max_{j}(m_{j} + m_{q-j+1})$.
Then the function $\Gamma : \bit{\kappa}^{c} \times \bit{\kappa}^{c} \to \bit{m}^{dq}$ defined for $(j,l) \in \{1, \dots, q\} \times \{1,\dots, d\}$ by
\begin{equation*}
	\Gamma(x_{1}, x_{2})_{(j-1)d + l} = (\Gamma_{j}(x_{1})_{l} \conc \Gamma_{q-j+1}(x_{2})_{l})[m] 
\end{equation*}
is $k$-unique. 
\end{lemma}
\begin{proof}
Consider a set of keys $S \subseteq \bit{n}^{c} \times \bit{n}^{c}$ with $|S| \leq k$.
We will show that there exists an index $j \in \{ 1, \dots, q \}$ and a key $x = (x_{1},x_{2}) \in S$ 
such that $x$ has a unique neighbor with respect to $S$ and $\Gamma_{j} \conc \Gamma_{q-i+1}$.
Consider the set of first components of the set of keys $\pi_{1}(S)$. 
For some $j \in \{1, \dots, q\}$ we must have that $k^{(j-1)/q} \leq |\pi_{1}(S)| \leq k^{j/q}$.
By the super-majority-uniqueness properties of $\Gamma_{j}$ there must exist more than $k^{(j-1)/q}/2$ first components $x_{1} \in \pi_{1}(S)$ such that $\Gamma_{j}(x_{1})$ has more than $d/2$ unique neighbors with respect to $\pi_{1}(S)$.
Furthermore, because $|S| \leq k$, there exists at least one such $x_{1}$ that is a component of at most $\min(2k^{(q-j+1)/q}, k)$ keys.
Following a similar argument to the proof of Lemma \ref{lem:interleaving}, by the majority-uniqueness properties of $\Gamma_{q-j+1}$ there exists $x_{2} \in S_{1,x_{1}}$ such that we get a unique neighbor. 
\end{proof}
In the following lemma we use a single $k$-independent function to represent a set of $k$-super-majority-unique functions such that the concatenated product of these functions is $k$-unique.
The proof of the lemma is omitted since it follows from using the approach of Lemma \ref{lem:expanderparameters} to obtain expansion $|\Gamma(S)| > (7/8)d|S|$, 
and then applying Lemma \ref{lem:superconcatenation} to obtain the $k$-uniqueness property. 
\begin{lemma} \label{lem:superexpanderparameters}
For every choice of positive integers $c$, $q$, $\kappa$, let~$f : \bit{\kappa}^{c} \to \bit{2\kappa + 12}^{16cq}$ be a~${2^{\kappa}}$-independent function.
For $j = 1, \dots, q$ define $\Gamma_{j} : \bit{\kappa}^{c} \to \bit{\lceil ((j + 1)/q)\kappa \rceil + 12}^{16cq}$ by
\begin{equation*}
\Gamma_{j}(x)_{l} = f(x)_{l}[\lceil ((j + 1)/q)\kappa \rceil + 12] \textnormal{ for } l \in \{1, \dots, 16cq \}.
\end{equation*}
Let $m = \lceil (1 + 3/q)\kappa \rceil + 26$. 
Then the function $\Gamma : \bit{\kappa}^{c} \times \bit{\kappa}^{c} \to \bit{m}^{16cq^{2}}$ defined for $(j,l) \in \{1, \dots, q\} \times \{1,\dots, 16cq\}$ by
\begin{equation*}
\Gamma(x_{1}, x_{2})_{(j-1)16cq + l} = (\Gamma_{j}(x_{1})_{l} \conc \Gamma_{q-j+1}(x_{2})_{l})[m] 
\end{equation*}
is $k$-unique with probability at least $1 - 2^{-2c\kappa}$ .
\end{lemma}
\noindent We remind the reader that the notation $x[m]$ is used to denote the zero-padded $m$-bit prefix of $x$.
Taking the prefix of the concatenated output characters of $\Gamma_{j}$ and $\Gamma_{q-j+1}$ is done with the sole purpose of padding the output characters of $\Gamma$ to uniform length.

We now define a randomized recursive construction of a $k$-unique function similar to the one in Theorem \ref{thm:majorityrecursion}.
The parameters of the data structure are $\lambda$, $\kappa$, and $q$.
The parameters $\lambda$ and $\kappa$ determine the size of the universe and the desired $k$-uniqueness.
The parameter $q$ controls the space-time tradeoff of the character tables used in the recursion.
At the outer level of the recursion, for $i = 1, \dots, \lambda$, we repeatedly square the size of the domain, 
constructing $k$-unique functions of the form $\Gamma_{i} : \bit{\kappa}^{2^{i}} \to \bit{2\kappa + 26}^{144 \cdot 2^{i}}$.	
At level $i$ of the recursion, we obtain a $k$-independent function by composing $\Gamma_{i}$ with a simple tabulation function $h_{i+1} : \bit{2\kappa + 26}^{144 \cdot 2^{i}} \to \bit{2 \kappa + 12}^{48 \cdot 2^{i+1}}$.
The output of this function is then used to construct $\Gamma_{i+1}$, following the approach of Lemma \ref{lem:superexpanderparameters} with the parameter $q$ set to $3$.
For $i = 1, 2, \dots, \lambda$ the recursion is described by the following set of equations
\begin{equation*}
\begin{aligned}
\Gamma_{i}(x_{1}, x_{2})_{(j-1)48 \cdot 2^{i} + l} &= (\Gamma_{i,j}(x_{1})_{l} \conc \Gamma_{i,4-j}(x_{2})_{l})[2\kappa + 26]\\
\Gamma_{i,j}(x_{s})_{l} &= h_{i}(\Gamma_{i-1}(x_{s}))_{l}[\lceil ((j + 1)/3)\kappa \rceil + 12]\\
\Gamma_{0}(x_{s})_{l} &= \id^{(48)}(x_{s})_{l}[2\kappa + 26]
\end{aligned}
\end{equation*}
where the indices are $j \in \{1,2,3\}$, $l \in \{1, \dots, 48 \cdot 2^{i} \}$, and $s \in \{1, 2\}$.
We have defined $\Gamma_{0}$ by simply repeating the input $48$ times, padded to length $2\kappa + 26$, to ensure that it fits into the recursion. 
In practice we only require $h_{1} \circ \Gamma_{0}$ be be $k$-independent over domain $\bit{\kappa}$.  

To further reduce the space usage we apply the technique from Lemma \ref{lem:superexpanderparameters} to implement the character tables of $h_{i}$.
Each character table has domain $\bit{2\kappa + 26}$. 
We view this domain as consisting of two characters of length $\kappa' = \kappa + 13$.
We apply Lemma \ref{lem:superexpanderparameters} with parameters $c = 1$, $q$, and $\kappa = \kappa'$ to construct a function 
$\Upsilon : \bit{\kappa'}^{2} \to \bit{\lceil (1 + 3/q)\kappa' \rceil + 26}^{16q^{2}}$ that is $k$-unique with probability at least $1 - 2^{-2\kappa'}$. 
To facilitate fast evaluation we tabulate the $k$-independent function $f' : \bit{\kappa'} \to \bit{\lceil (1 + 1/q)\kappa' \rceil + 12}^{16q}$ used to construct $\Upsilon$.
The $j$th character table of $h_{i}$ is constructed by composing $\Upsilon$ with an appropriate simple tabulation function,
\begin{equation*}
h_{i,j} = \Upsilon \circ g_{i,j},
\end{equation*}
where $g_{i,j} : \bit{\lceil (1 + 1/q)\kappa' \rceil + 12}^{16q} \to \bit{2\kappa + 12}^{48 \cdot 2^{i}}$ is tabulated.

\begin{theorem} \label{thm:prefixrecursion}
There exists a randomized data structure that takes as input positive integers $\lambda$, $\kappa$, $q$, and initializes a function $\Gamma : \bit{\kappa}^{2^{\lambda}} \to \bit{2\kappa + 26}^{48 \cdot 2^{\lambda}}$.
In the word RAM model with word length $w$ the data structure satisifes the following:
\begin{itemize}
	\item The space usage is $O(2^{(1 + 3/q)\kappa + 2\lambda} q^{2} \kappa /w)$. 
	\item The evaluation time of $\Gamma$ is $O(2^{\lambda} q^{2} (\lambda + 2^{\lambda}\kappa/w))$.
	\item With probability at least $1 - 2^{-2(\kappa - 1)}$ we have that $\Gamma$ is $2^{\kappa}$-unique.
\end{itemize}
\end{theorem}
\begin{proof}
The total space usage is dominated by the simple tabulation functions used to implement the character tables of $h_{\lambda}$.
There are $O(2^{\lambda})$ simple tabulation functions $g_{i,j}$. 
Each of these has $O(q^{2})$ character tables with a domain of size $O(2^{(1 + 3/q)\kappa})$ that map to bit strings of length $O(2^{\lambda}\kappa)$.
This gives a total space usage of $O(2^{(1 + 3/q)\kappa + 2\lambda)} q^{2} \kappa /w)$.

Let $T(i)$ denote the evaluation time of $\Gamma_{i}$.
For $i = 1$ we can evaluate $\Gamma_{1}$ by performing a constant number of lookups into $h_{0}$ and combine prefixes of the output in $O(1)$ time.
For $i > 1$ evaluating $\Gamma_{i}$ takes two evaluations of $\Gamma_{i-1}$ and an additional amount of work combining prefixes that is only a constant factor greater than the time required to read the output of $h_{i} \circ \Gamma_{i-1}$.
Evaluating $h_{i}$ is performed by $O(2^{i})$ evaluations of character tables of the form $g_{i,j} \circ \Upsilon$.
The degree of $\Upsilon$ is $O(q^{2})$ and it has an evaluation time that is proportional to the degree.
We therefore perform $O(q^{2})$ lookups into the character tables of $g_{i,j}$ where we read bit strings of length $O(2^{i}\kappa)$.
The recurrence describing the evaluation time of $\Gamma_{i}$ takes the form
\begin{equation*}
T(i) \leq 
\begin{cases} 2T(i-1) + O(2^{i}q^{2}(1 + 2^{i}\kappa/w)) & \mbox{if } i > 1 \\ 
O(1) & \mbox{if } i = 1. 
\end{cases}
\end{equation*}
The solution to the recurrence is $O(2^{i}q^{2}(i + 2^{i}(n + \kappa)/w))$.

The construction fails if $\Upsilon$ fails to be $k$-unique or if $\Gamma_{1}, \dots, \Gamma_{\lambda}$ fails to be $k$-unique. 
According to Lemma \ref{lem:superexpanderparameters} this happens with probability less than $2^{-2\kappa'} + \sum_{i=1}^{\lambda}2^{-2^{i}\kappa} < 2^{-2(\kappa - 1)}$   
\end{proof}

\begin{corollary}\label{cor:prefixcorollary}
There exists a randomized data structure that takes as input positive integers $u$, $r = u^{O(1)}$, $t$, $k$ and selects a family of functions $\mathcal{F}$ from $[u]$ to $[r]$. 
In the word RAM model with word length $w$ the data structure satisfies the following:
\begin{itemize}
\item The space used to represent $\mathcal{F}$, as well as a function $f \in \mathcal{F}$, is $O(ku^{1/t}t^{2}\log(k)/w)$ words.
\item The evaluation time of $f$ is $O(t^{2} (\log(k)/\log u)(\log(\log(u)/\log k) + \log(u)/w))$.
\item $\mathcal{F}$ is $k$-independent with probability at least $1 - k^{-2}$.
\end{itemize}
\end{corollary}
\begin{proof}
Assume without loss of generality that $k \leq u$ and apply Theorem \ref{thm:prefixrecursion} with parameters 
$\lambda = \lceil \log(\log(u)/\log k) \rceil + 1$, $\kappa = \lceil \log k \rceil + 1$, and $q = \lceil 3t \log(k)/\log u \rceil$.
This gives is a function $\Gamma$ that is $k$-unique over $[u]$ with probability at least $1 - k^{-2}$.
We compose $\Gamma$ with a suitable simple tabulation function $h$ that maps to elements of $[r]$.
Implementing $h$ using $\Upsilon$ we get the same bounds on the space usage, evaluation time, and probability of failure as for the data structure used to represent $\Gamma$.
\end{proof}

\begin{remark}
The construction in Corollary \ref{cor:prefixcorollary} presents an improvement in the case where we wish to minimize the space usage. 
For $w = \Theta(\log u)$ and $t = \lceil \log u \rceil$ we get a space usage of $O(k \log(u) \log( k))$ and an evaluation time of $O(\log (u)\log (k)\log (\log (u) /\log (k)))$.
In comparison, for these parameters Corollary \ref{cor:simple} gives a space usage of $O(k \log^{2} u)$ and an evalution time of $O(\log^{2}(u)\log(k))$. 
\end{remark}

\bibliographystyle{abbrv}
\bibliography{thesis}
\end{document}